\documentclass{article}
\usepackage[
	backend=biber,
	datamodel=mydatamodel,
	style=alphabetic,
	natbib=true,
	url=false,
	doi=true,
	eprint=true,
	giveninits=true,
	isbn=false,
	maxbibnames=10,
]{biblatex}
\renewbibmacro*{doi+eprint+url}{%
	\iftoggle{bbx:doi}
	{\printfield{doi}}
	{}%
	\newunit\newblock
	\printfield{mrnumber}%
	\newunit\newblock
	\printfield{zbl}%
	\newunit\newblock
	\iftoggle{bbx:eprint}
	{\usebibmacro{eprint}}
	{}%
	\newunit\newblock
	\iftoggle{bbx:url}
	{\usebibmacro{url+urldate}}
	{}}
\usepackage[a4paper, margin=3cm]{geometry}
\usepackage{hyperref}
\usepackage{authblk}
\usepackage{parskip} 
\usepackage{datatool}
\usepackage{xfp}
\usepackage{etoc}
\usepackage{caption}
\usepackage{float}
\usepackage{amsthm,amsmath,amsfonts,amssymb}
\usepackage{mathtools}
\usepackage{stmaryrd}
\usepackage{enumitem} 
\usepackage{graphicx} 
\usepackage{dsfont} 
\usepackage{tikz} 
\usepackage{booktabs, caption, longtable, colortbl, array, anyfontsize, multirow}
\usepackage{helvet}
\usepackage[capitalize,nameinlink,noabbrev]{cleveref}
\newcommand*{\ms}[1]{\mathsf{#1}}

\newcommand{\N}{\mathbb{N}}

\newcommand{\R}{\mathbb{R}}

\newcommand{\lcb}{\left\lbrace} 
\newcommand{\rcb}{\right\rbrace} 
\newcommand{\cb}[1]{\lcb #1 \rcb} 
\newcommand{\lb}{\left(} 
\newcommand{\rb}{\right)} 
\newcommand{\br}[1]{\lb #1 \rb} 
\newcommand{\brOf}[1]{\!\br{#1}} 
\newcommand{\euclOf}[1]{\left\Vert#1\right\Vert_2}

\newcommand{\sizedMid}[2]{#1 \, \kern-\nulldelimiterspace\mathopen{}\left| \vphantom{#1}\,#2\right.\mathclose{}\kern-\nulldelimiterspace}
\newcommand{\setByEle}[2]{\cb{\sizedMid{#1}{#2}}}

\newcommand{\tr}{^{\!\top}\!} 

\newcommand{\eqcm}{\,,} 
\newcommand{\eqfs}{\,.}
\newcommand{\dl}{\mathrm{d}} 
\newcommand{\stepsize}{\Delta\!t}
\newcommand{\vpt}[1]{\mathsf{VPT}_{#1}}
\newcommand{\labelSinglePrec}{\colorbox{red!30}{32bit}}
\newcommand{\labelDoublePrec}{\colorbox{blue!30}{64bit}}
\newcommand{\labelMultiPrec}{\colorbox{green!30}{512bit}}
\newcommand{\selectLabel}[2]{%
	\ifx#1s
	\def#2{\labelSinglePrec}%
	\else
	\ifx#1d
	\def#2{\labelDoublePrec}%
	\else
	\ifx#1m
	\def#2{\labelMultiPrec}%
	\else
	\def#2{Unknown}%
	\fi
	\fi
}
\newcommand{\dataflowdiagram}[3]{%
	\begingroup
	\renewcommand{\baselinestretch}{1}\normalsize
	\selectLabel{#1}{\labelSolver}%
	\selectLabel{#2}{\labelData}%
	\selectLabel{#3}{\labelEstimator}%
	\begin{tikzpicture}[every node/.style={font=\sffamily}]%
		\node[align=center, anchor=north west] (Aa) at (-2.1,0.9) {\textbf{System}};
		\node[align=center, draw=black] (A) at (-1,-0.5) {RK4 Solver\\\labelSolver};
		\draw[very thick, rounded corners=2ex] (-2.2,1) rectangle (0.2,-1.5);
		\node[align=center, draw=black, thick] (B) at (4.25,-0.5) {Train Data\\\labelData};
		\node[align=center, draw=black, thick] (C) at (3.2,-4) {Test Data\\\labelData};
		\node[align=center, draw=black, thick] (D) at (5.4,-4) {Forecast\\\labelData};
		\draw[very thick, rounded corners=2ex] (1.7,1) rectangle (6.8,-7);
		\node[align=center, anchor=north west] (Ba) at (1.8,0.9) {\textbf{Data}};
		\draw[very thick, rounded corners=2ex] (8.3,1) rectangle (12.7,-5);
		\node[align=center, anchor=north west] (E) at (8.4,0.9) {\textbf{Method}};
		\node[align=center, draw=black] (F) at (10.5,-0.5) {Polynomial Features\\\labelEstimator};
		\node[align=center, draw=black] (G) at (10.5,-2.25) {Coefficients\\\labelEstimator};
		\node[align=center, draw=black] (H) at (10.5,-4) {Forecast\\\labelEstimator};
		\node[align=center, draw=black, thick] (I) at (4.25,-6) {Error and VPT\\\labelData};
		\ifx#1#2
			\draw[->,thick] (A) --node[midway,above,xshift=-7mm]{copy} (B);
			\draw[->,thick] (A) --node[midway,above,sloped,xshift=-5mm]{copy} (C);
		\else
			\draw[->,thick,color=red] (A) --node[midway,above,xshift=-7mm]{round} (B);
			\draw[->,thick,color=red] (A) --node[midway,above,sloped,xshift=-5mm]{round} (C);
		\fi
		\draw[->,thick] (B) -- (F);
		\ifx#3#2
			\draw[->,thick] (H) --node[midway,above,xshift=-5mm]{copy} (D);
		\else
			\draw[->,thick,color=red] (H) --node[midway,above,xshift=-4mm]{round} (D);
		\fi
		\draw[->,thick] (F) --node[midway,right]{fit} (G);
		\draw[->,thick] (G) --node[midway,right]{predict} (H);		
		\draw[->,thick] (D) -- (I);		
		\draw[->,thick] (C) -- (I);		
	\end{tikzpicture}%
	\endgroup
}

\makeatletter
\newcommand{\DefineSupp}[3]{%
	\expandafter\gdef\csname supp@#1\endcsname{#2}%
	\expandafter\gdef\csname supp@lbl@#1\endcsname{#3}%
}
\newcommand{\SuppSec}[1]{%
	\ifcsname supp@#1\endcsname%
	\hyperref[\csname supp@lbl@#1\endcsname]{Supplementary Section~\csname supp@#1\endcsname}%
	\else%
	Supplementary Section~\textbf{??}%
	\PackageWarning{SuppSec}{Undefined supplementary reference: #1}%
	\fi%
}
\makeatother

\DefineSupp{testmode}{1}{app:sec:randVsSeq}
\DefineSupp{examples}{2}{app:sec:example}
\DefineSupp{normalization}{3}{app:sec:normalize}
\DefineSupp{rk4l63poly}{4}{app:sec:poly}
\DefineSupp{xprec}{5}{app:sec:xprec}
\DefineSupp{noise}{6}{app:sec:noise}
\DefineSupp{l96summary}{7}{app:sec:L96:best}
\DefineSupp{additional}{8}{app:sec:additional}
\DefineSupp{lyapunov}{9}{app:sec:lyapunov}
\DefineSupp{vptref}{10}{app:sec:vptref}
\DefineSupp{compute}{11}{app:sec:compute}
\DefineSupp{details}{12}{app:sec:details}

\newbox{\myorcidthanksbox}
\sbox{\myorcidthanksbox}{\large\includegraphics[height=1.8ex]{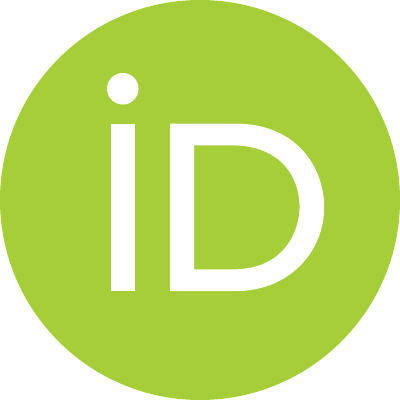}}
\newcommand{\orcidthanks}[1]{%
    \href{https://orcid.org/#1}{\raisebox{-0.5ex}{\usebox{\myorcidthanksbox}}\,#1}}
\usepackage{thmtools}

\def\postBoxSkip{1.0ex}
\def\postBoxSkipCmd{\vskip\postBoxSkip}
\def\preBoxSkip{1.0ex}
\def\preBoxSkipCmd{\vskip\preBoxSkip}
\declaretheoremstyle[
	bodyfont=\normalfont,
	postfoothook={\postBoxSkipCmd},
	preheadhook={\preBoxSkipCmd},
	mdframed={
		backgroundcolor = black!2,
		startcode={},
}]{ruledBoxStyle}
\declaretheoremstyle[
	bodyfont=\normalfont,
	postfoothook={\postBoxSkipCmd},
	preheadhook={\preBoxSkipCmd},
	mdframed={
		backgroundcolor=white,
}]{ruledBoxStyleWhite}
\declaretheoremstyle[
	bodyfont=\normalfont,
	postfoothook={\postBoxSkipCmd},
	preheadhook={\preBoxSkipCmd},
	mdframed={
		backgroundcolor=black!2,
		linecolor = black!2,
		tikzsetting = {
			draw = black,
			line width = 2pt,%
			dashed,%
			dash pattern = on 10pt off 3pt
		},
}]{dashedBoxStyle}
\declaretheoremstyle[
	bodyfont=\normalfont,
	postfoothook={\postBoxSkipCmd},
	preheadhook={\preBoxSkipCmd},
	mdframed={
		linecolor = white,
		startcode={},
		tikzsetting = {
			draw = black,
			line width = 1pt,%
			loosely dotted,
		},
	}
]{dashedStyle}
\declaretheoremstyle[
	bodyfont=\normalfont,
	postfoothook={\postBoxSkipCmd},
	preheadhook={\preBoxSkipCmd},
	mdframed={
		linecolor = black,
		innerlinewidth=1pt,outerlinewidth=1pt,
		middlelinewidth=1pt,
		linecolor=black,middlelinecolor=white,
		startcode={},
	}
]{doubleStyle}
\declaretheoremstyle[
	bodyfont=\normalfont,
	postfoothook={\postBoxSkipCmd},
	preheadhook={\preBoxSkipCmd},
	mdframed={
		backgroundcolor = black!4,
		linecolor = black!4,
		startcode={},
}]{boxStyle}
\declaretheoremstyle[
	headfont=\normalfont\itshape,
	notefont=\normalfont\itshape,
	notebraces={}{},
	bodyfont=\normalfont,
	qed=\qedsymbol,
	numbered=no,
	headindent=0pt,
	postheadspace=1ex,
	name={Proof},
	postheadhook={},
	mdframed={
		hidealllines = true,
		innerrightmargin = 0pt,
		innerleftmargin = 0pt,
		innertopmargin = 0pt,
		innerbottommargin = 0pt,
		leftmargin = 0pt,
		rightmargin = 0pt,
	}
]{proofStyle}
\declaretheoremstyle{basic}
\declaretheoremstyle[
	bodyfont=\normalfont,
	postfoothook={\postBoxSkipCmd},
	preheadhook={\preBoxSkipCmd},
	mdframed={
		backgroundcolor = white,
		linecolor = black,
		startcode={},
		leftline = false,
		rightline = false,
}]{tobBottomStyle}
\declaretheoremstyle[
bodyfont=\normalfont,
]{standardStyle}
\declaretheorem[style=basic,name=Definition, numberwithin=section]{definition}

\declaretheorem[style=basic,name=Theorem,numberlike=definition]{theorem}

\title{Machine-Precision Prediction of Low-Dimensional Chaotic Systems from Noise-Free Data}
\date{}
\author[1,2]{Christof Sch\"otz\thanks{christof.schoetz@tum.de, \orcidthanks{0000-0003-3528-4544}, corresponding author}}
\author[1,2,3]{Niklas Boers\thanks{n.boers@tum.de, \orcidthanks{0000-0002-1239-9034}}}
\affil[1]{Technical University of Munich, Germany; Munich Climate Center; TUM School of Engineering and Design, Department of Aerospace and Geodesy, Earth System Modelling Group}
\affil[2]{Potsdam Institute for Climate Impact Research, Germany}
\affil[3]{University of Exeter, UK; Department of Mathematics}
\begin{document}
\maketitle
\begin{abstract}
	Low-dimensional chaotic systems such as the Lorenz-63 model are commonly used to benchmark system-agnostic methods for learning dynamics from data. This study shows that learning from noise-free observations in such systems can be achieved up to machine precision: using ordinary least squares regression on high-degree polynomial features with 512-bit arithmetic, a system-agnostic method is introduced that matches the accuracy of standard 64-bit numerical ODE solvers using the systems' governing equations. For the Lorenz-63 system, the method achieves valid prediction times of 36 Lyapunov times, and even up to 105 Lyapunov times with favorable precision configurations, dramatically outperforming prior work, which reaches 13 Lyapunov times at most. The results are further validated on Thomas' Cyclically Symmetric Attractor, a non-polynomial chaotic system that is considerably more complex than the Lorenz-63 model, and similar results extend to higher dimensions using the spatiotemporally chaotic Lorenz-96 model. These findings suggest that forecasting low-dimensional chaotic systems from noise-free data is effectively a solved problem.
\end{abstract}
\section{Introduction}
Chaotic dynamical systems represent a wide range of complex phenomena across many natural systems and scientific fields \citep{Ott1993, Strogatz2024}. While predicting their behavior is crucial for understanding these systems, it remains challenging due to their sensitivity to initial conditions. Recently, data-driven, system-agnostic methods using machine learning (ML) have emerged as promising approaches for forecasting chaotic dynamics, ranging from low-dimensional systems such as the Lorenz-63 model \citep{Pathak2017,Pathak2018} to global weather forecasting \citep{Lam2023,Price2024}.
Researchers typically evaluate such methods using synthetic data generated from known chaotic systems. A commonly used metric to assess prediction quality is the valid prediction time (VPT) \citep{Ren2009, Pathak2018}, which measures how long a forecast stays close to the reference trajectory serving as ground truth.  Throughout the article, the term \textit{ground truth} denotes a numerical approximation of a system's analytical solution.

The Lorenz-63 (L63) system \citep{Lorenz1963}---a three-dimensional autonomous system of first-order ordinary differential equations (ODEs) exhibiting chaotic behavior---serves as the most widely used benchmark in this field. 
Prior results on this system are summarized in \cref{tbl:L63:vpt:lit}, where the state-of-the-art VPT is approximately 13 Lyapunov times (about 14 system time units). Crucially, these results were obtained using noise-free training data and initial conditions.

This study introduces \textit{PolyProp}, a system-agnostic method to learn and predict chaotic dynamics from noise-free data. It relies on linear regression to fit high-order polynomials with high numerical precision (512 bit in the main setting). On low-dimensional systems, it is computationally extremely efficient compared to other state-of-the-art machine learning architectures. Despite its simplicity and system-agnostic nature, \textit{PolyProp} not only vastly outperforms previous data-driven approaches but even matches the precision of numerical ODE solvers that utilize the system's governing equations. On the 3-dimensional L63 model, it achieves VPT values of 36 Lyapunov times (and up to 105 with favorable precision configurations). 

To contextualize this performance, the previous state-of-the-art VPT of 13 Lyapunov times corresponds roughly to a normalized root mean squared error (nRMSE) between $10^{-7}$ and $10^{-6}$ during the first Lyapunov time unit of the forecast, with normalization by the system's standard deviation. In the same interval, \textit{PolyProp} achieves an nRMSE of $10^{-16}$ to $10^{-15}$. This is roughly equivalent to 64-bit machine precision, $2^{-52}$, and represents an improvement over the state of the art by a factor of $10^9$.

To confirm that these findings are not artifacts of the numerical setup, the experiments are successfully replicated on Thomas' Cyclically Symmetric Attractor \citep{Thomas99}, a 3-dimensional chaotic system with considerably more complex dynamics than the L63 model. Furthermore, machine precision results are achieved also for higher dimensions, as exemplified with the spatiotemporally chaotic Lorenz-96 model \citep{Lorenz1996}.

These findings can be generalized to data with measurement noise, viewing noise as a precision limitation. In general, PolyProp trained on reduced-precision data yields a forecast accuracy indistinguishable from that of an ODE solver starting from the same reduced-precision initial condition---regardless of whether the reduction stems from rounding, bit-depth, or noise. This implies that PolyProp approximates the underlying dynamics so well that the forecast error is dominated by the initial condition error (and not by the dynamics error), which affects the numerical ODE solver equally. In this sense, PolyProp yields an optimal forecast (among methods that do not apply data assimilation to denoise the initial conditions).

\begin{table}
	\centering
	\renewcommand{\arraystretch}{0.73}
	\setlength{\tabcolsep}{0.3em}
	\rowcolors{2}{gray!20}{white}
	\begin{tabular}{llrrr}
		\toprule
		\textbf{Reference} & \textbf{Family} & \textbf{$n$} & \textbf{$\stepsize$} & \textbf{$\vpt{}$} \\
		\midrule
		Elsner and Tsonis 1992 \citep{Elsner1992} & MLP & 1{,}000 & 0.01 & 1--2 \\
		Dubois et al.\ 2020 \citep{Dubois2020} & LSTM & 15{,}000 & 0.005 & 3 \\
		Jaurigue 2024 \citep{Jaurigue2024} & RC & 10{,}000 & 0.1 & 3 \\
		Griffith et al.\ 2019 \citep{Griffith2019} & RC & 10{,}000 & 0.01 & 4 \\
		K\"oster et al.\ 2023 \citep{Koster2023} & RC & 1{,}000 & 0.1 & 4 \\
		Roberts 2020 \citep{Roberts2020} & LSTM & 1{,}000 & 0.01 & 4 \\
		Yu et al.\ 2019 \citep{Yu2019} & LSTM & & & 4 \\
		Viehweg et al.\ 2023 \citep{Viehweg2023} & RC & 2{,}000 & 0.01 & 5 \\
		Gauthier et al.\ 2021 \citep{Gauthier2021} & RC & 400 & 0.025 & 6 \\
		Pathak et al.\ 2018 \citep{Pathak2018} & RC & 1{,}000 & 0.1 & 6 \\
		Wang et al.\ 2019 \citep{Wang2019} & LSTM & & 0.05 & 6 \\
		Pathak et al.\ 2017 \citep{Pathak2017} & RC & 5{,}000 & 0.02 & 7 \\
		Roque dos Santos and Bollt 2025 \citep{Santos2025} & RC & 5{,}000 & 0.01 & 8 \\
        Silva et al.\ 2020 \citep{deSilva2020} & SINDy & 5{,}000 & 0.002 & 8 \\
        Lu et al.\ 2018 \citep{Lu2018} & RC & 3{,}000 & 0.02 & 8 \\
		Akiyama and Tanaka 2022 \citep{Akiyama2022} & RC & 5{,}000 & 0.02 & 9 \\
		Li et al.\ 2024  \citep{Li2024} & RC & 3{,}000 & 0.02 & 9 \\
		Sch\"otz et al.\ 2025 \citep{Schotz2025} & misc & 10{,}000 & 0.01 & $>$9 \\
		Steinegger and R\"ath 2025  \citep{Steinegger2025} & RC & 2{,}100 & 0.02 & 12 \\
		Brunton et al.\ 2016  \citep{Brunton2015} & SINDy & 100{,}000 & 0.001 & 13 \\
		Mandal and Gottwald 2025	\citep{Mandal2025} & RC & 50{,}000 & 0.01 & 13 \\
		Platt et al.\ 2022 \citep{Platt2022} & RC & 50{,}000 & 0.01 & 13 \\
		RK4 (64 bit) & ODE solver & 1 & 0.001 & 32 \\
		RK4 (512 bit) & ODE solver & 1 & 0.001 & 35 \\
        This Study & PolyProp & 16{,}000 & 0.03 & 36 \\
		This Study & PolyProp & ${}^\ast$33{,}000 & 0.002 & 105
		 \\
		\bottomrule
    	\multicolumn{4}{l}{${}^\ast$ data is provided at 512-bit precision instead of 64-bit}
	\end{tabular}
\caption{\textbf{Valid prediction times in the literature for the Lorenz-63 system.} 
	Forecasting performance on L63 is reported as the valid prediction time (VPT) with threshold parameter $0.5$, $\vpt{0.5}$, given in units of Lyapunov time. For consistency, we estimate $\vpt{0.5}$ from potentially different error measures in the reference. The \textit{Family} column indicates the algorithmic class to which each prediction method belongs. The earliest publication uses a multi-layer perceptron (MLP). Reservoir computing (RC) methods include the Echo State Network \citep{jaeger2001echo}, Nonlinear Vector Autoregression \citep{Gauthier2021}, and related variants. LSTM refers to the Long Short-Term Memory network \citep{Hochreiter1997}. SINDy denotes the Sparse Identification of Nonlinear Dynamical Systems framework \citep{Brunton2015}, which in contrast to the other methods is not fully system-agnostic, as it assumes the system dynamics to be a sparse degree 5 polynomial. For comparison, we show the VPT of a numerical solver (Runge Kutta of order 4, RK4) with access to the system's governing equations. The method we present in this study is PolyProp, short for \textit{Polynomial Propagator}. The number of training data points is given in column $n$, and the temporal resolution of the data is specified by the time step $\stepsize$. Missing entries reflect cases where these details were not clearly reported in the original sources.}		
\label{tbl:L63:vpt:lit}
\end{table}

Contrary to common belief, these findings suggest that chaotic systems can, in principle, be predicted over arbitrarily long timescales---given measurements that are both sufficiently numerous and precise, and a suitable learning algorithm executed with sufficient numerical precision. If the state dimension is moderate, these requirements can be fulfilled in practice to the extent that the system-agnostic, data-driven approach presented here performs at least as well as a numerical integrator with access to the system's governing equations. Thus, the potential disadvantage of not knowing these equations can be overcome with data. In this sense, the problem of learning low-dimensional chaotic systems from noise-free data can be considered solved.
\section{Method and Experimental Setup}\label{sec:setup}

This section describes the methodology underlying the results of the following section. After formalizing the prediction problem and evaluation metric, the three chaotic benchmark systems used in the experiments are introduced, followed by a specification of the PolyProp method, a description of the numerical precision choices central to its performance, and a summary of the experimental design.
\subsection{Problem Setup and Evaluation Metric}\label{ssec:problem}
Consider a dynamical system described by a $d$-dimensional, first order, autonomous ordinary differential equation 
\begin{equation}\label{eq:ode}
    \dot u(t) = f(u(t))
    \eqcm
\end{equation}
where $u \colon \R\to\R^d$ is a \textit{solution} and $f \colon \R^d \to \R^d$ is the vector field describing the dynamics of the system. Given initial conditions $u(0) = u_0$ for $u_0 \in \R^d$ and assuming $f$ is sufficiently smooth (e.g., globally Lipschitz continuous), the solution $u$ of \cref{eq:ode} exists and is unique. In practice, the analytical solution $u$ is typically not available and is approximated numerically as described in \cref{ssec:numerics}.

Assume $n$ noise-free observations $y_i = u(t_i)$, $i = 1,\dots,n$, of the trajectory $u$ are made at time points $t_i$. The observation times start at a negative value and extend up to $t_n = 0$. These times are assumed to be equally spaced. This means $t_i := (i-n)\stepsize$ for a fixed step size $\stepsize > 0$.

For the prediction task, $f$ and $u$ are assumed to be unknown, but the observation time series $(y_i)_{i=1,\dots,n}$ and the step size $\stepsize$ are available. The goal is to predict the future values $z_j := u(j\stepsize)$ for $j \in \N$ given $z_0 = y_n$. In this case the train and test data are \textit{sequential}. Alternatively, we also consider the \textit{random} test mode, in which $z_j := \tilde u(j\stepsize)$ for $j \in \N \cup \{0\}$ for a new solution $\tilde u$ of \cref{eq:ode}, where $\tilde u(0) = z_0$ is chosen randomly from the system's attractor. Note that, after training on $(y_i)_{i=1,\dots,n}$, only a single state $z_0$ is available to initialize the prediction in the \textit{random} test mode.

To evaluate a given forecast $(\hat z_j)_{j=1, \dots, m}$, we define the Valid Prediction Time (VPT) as the maximum time duration over which the normalized Euclidean distance between the prediction and the ground truth remains below a given threshold $\varepsilon>0$. Formally,
\begin{align}
    &\vpt\varepsilon\brOf{(z_j)_{j=1, \dots, m}, (\hat z_j)_{j=1, \dots, m}}
    \\&:= 
    \stepsize \max\setByEle{J \in \{1, \dots, m\}}{\forall j\in\{1,\dots, J\}\colon \frac{\euclOf{z_j - \hat z_j}}{\sigma} \leq \varepsilon}
    \eqcm
\end{align}
where $\sigma$ is the standard deviation of the system, defined as
\begin{align}
    \sigma  &:= \sqrt{\lim_{T\to\infty}\frac{1}{2T}\int_{-T}^{T} \euclOf{u(t) - \mu}^2 \,\dl t}
    \\\text{with}\qquad
    \mu &:= \lim_{T\to\infty}\frac{1}{2T}\int_{-T}^{T} u(t) \,\dl t \eqfs
\end{align}
In practice, $\sigma$ and $\mu$ are approximated empirically from a long ground truth trajectory,
\begin{equation}
    \bar\sigma := \sqrt{\frac1\ell\sum_{j=1}^\ell\euclOf{z_j - \bar\mu}^2}\eqcm\qquad\text{and}\qquad
    \bar\mu := \frac1\ell\sum_{j=1}^\ell z_j \eqfs
\end{equation}
Note that the VPT in this definition is translation and scale invariant.

If the system dynamics are recreated well enough, the forecast error roughly increases by a factor of Euler's number each Lyapunov time. This can be used to translate between VPT values for different thresholds:
Assuming all VPT values are given in units of Lyapunov time,
\begin{equation}
    \vpt{\varepsilon_2} \approx \vpt{\varepsilon_1} + \log\brOf{\frac{\varepsilon_2}{\varepsilon_1}}
    \eqcm
\end{equation}
where $\log$ is the natural logarithm and $\varepsilon_1,\varepsilon_2 \in(0, 1]$. Empirically, in the simulation study presented here, this formula is a rough but valid approximation (showing at least one significant digit to be correctly approximated for well performing systems; see also the slope in \cref{fig:perfect}).
\subsection{Chaotic Benchmark Systems}
\begin{figure}%
    \centering%
    \includegraphics[width=0.45\textwidth]{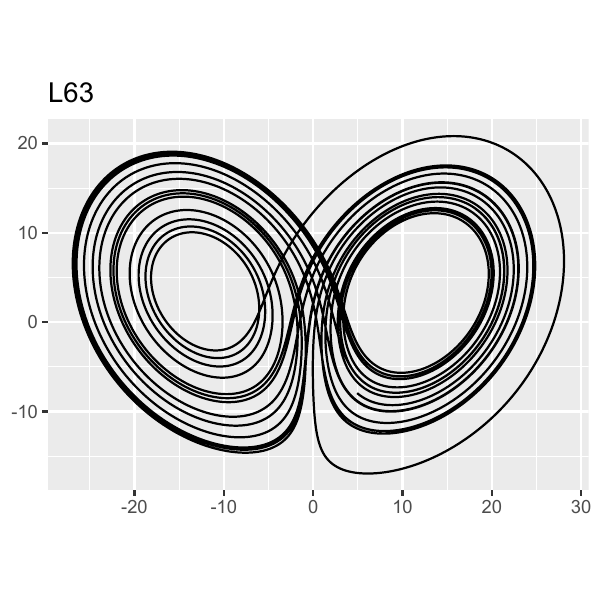}%
    \includegraphics[width=0.45\textwidth]{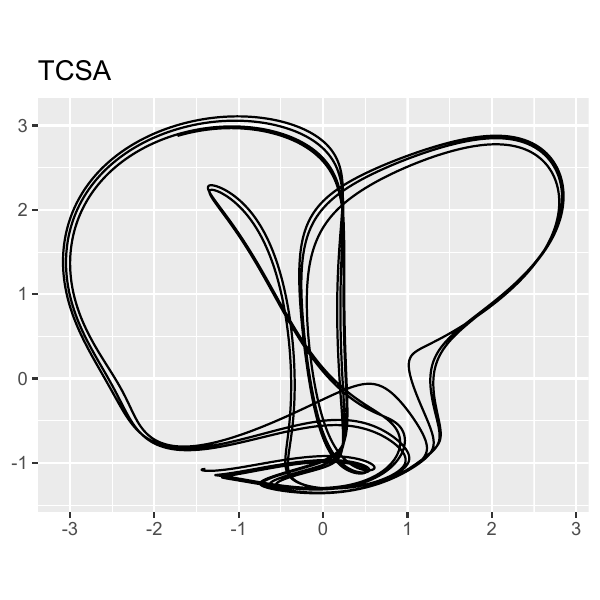}\\%
    \includegraphics[width=0.45\textwidth]{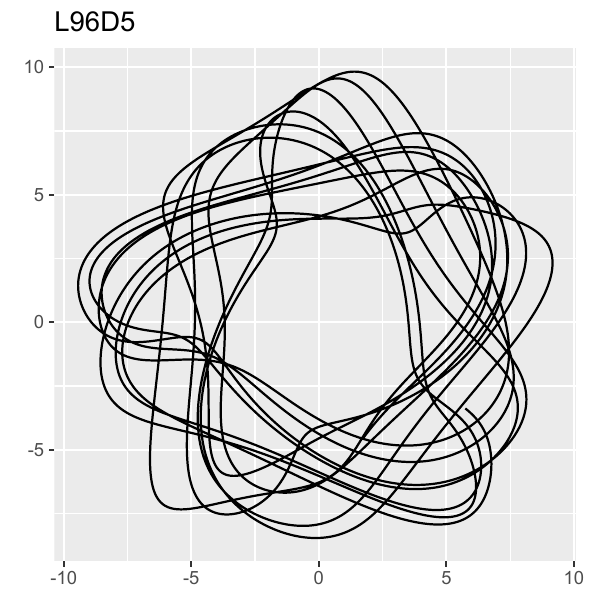}%
    \includegraphics[width=0.45\textwidth]{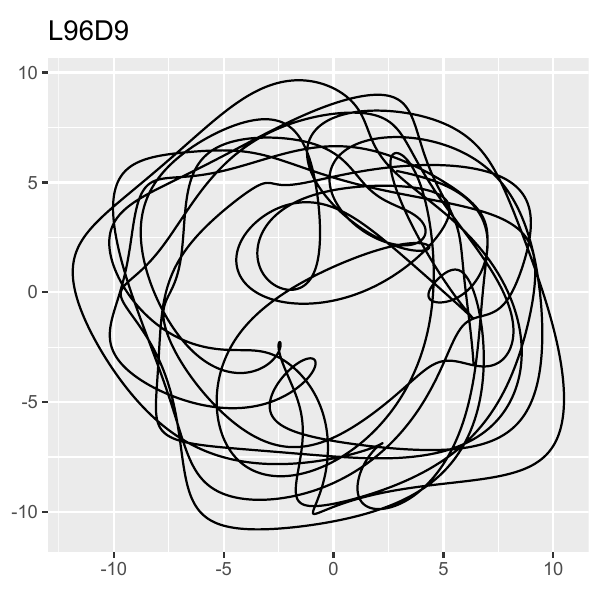}%
    \caption{\textbf{State space view of the Lorenz-63 (L63), Thomas' Cyclically Symmetric Attractor (TCSA), as well as 5-dimensional and 9-dimensional versions of the Lorenz-96 system (L96D5 and L96D9).} L63, L96D5, and L96D9 are integrated for 20 and TCSA for 200 system time units, and the state space is projected to 2 dimensions using principal component analysis.}\label{fig:view}
\end{figure}
We focus on three chaotic dynamical systems: the Lorenz-63 system (L63) \citep{Lorenz1963}, Thomas' cyclically symmetric attractor (TCSA) \citep{Thomas99}, and the Lorenz-96 system (L96) \citep{Lorenz1996}, see \cref{fig:view}. Further systems are shown in \SuppSec{additional}. For L63, the standard parameters are used, so that the dynamics are described by the vector field
\begin{equation}\label{eq:lorenz}
    f\brOf{\begin{bmatrix}
            u_1\\u_2\\u_3
    \end{bmatrix}} =
    \begin{bmatrix}
        10(u_2-u_1)\\
        u_1(28 - u_3) - u_2\\
        u_1 u_2 - \frac{8}{3} u_3
    \end{bmatrix}
    \eqfs
\end{equation}
For TCSA,
\begin{equation}\label{eq:thomas}
    f\brOf{\begin{bmatrix}
            u_1\\u_2\\u_3
    \end{bmatrix}} =
    \begin{bmatrix}
        \sin(u_2) - bu_1\\
        \sin(u_3) - bu_2\\
        \sin(u_1) - bu_3
    \end{bmatrix}
    \eqcm
\end{equation}
with $b = 0.208$.
For L96, the standard forcing value $F=8$ is used, yielding
\begin{equation}
    f\brOf{\begin{bmatrix}
            u_1\\\vdots\\u_i\\\vdots\\u_d
    \end{bmatrix}} =
    \begin{bmatrix}
        (u_{2} - u_{d-1})u_{d} - u_{1} + 8\\
        \vdots\\
        (u_{i+1} - u_{i-2})u_{i-1} - u_{i} + 8\\
        \vdots\\
        (u_{1} - u_{d-2})u_{d-1} - u_{d} + 8\\
    \end{bmatrix}
    \eqfs
\end{equation}
The name L96D$d$, with $d$ replaced by $5,6,7,8$ or $9$, indicates the $d$-dimensional version of L96.

L63 and TCSA are $(d=3)$-dimensional; for L96, $d$ can be chosen freely, with chaotic behavior obtained for $d\geq 5$. The vector fields for L63 and L96 are (sparse) polynomials of degree $2$; for TCSA the vector field is non-polynomial but smooth. The largest Lyapunov exponent of L63 is about $0.906$ \citep{Viswanath04}; for TCSA we estimate it to be $0.0153$. For L96 with $d=5,6,7,8,9$, we estimate the largest Lyapunov exponents to be $0.467$, $0.946$, $1.265$, $1.592$, $1.216$, respectively. See \SuppSec{lyapunov} for details. Time values in units of Lyapunov time are calculated as system time multiplied by the largest Lyapunov exponent.
\subsection{PolyProp}\label{ssec:propagator}
To create a forecast $(\hat z_j)_{j=1, \dots, m}$ from the observations $(y_i)_{i=1,\dots,n}$, the propagator $\Phi_{\stepsize}\colon \R^d\to\R^d,\,u(t) \mapsto u(t + \stepsize)$ is first estimated. The learned propagator is then applied repeatedly starting from a given state $z_0$, which is either the last observed state $z_0 = y_n$ (sequential test mode) or a new randomly chosen state on the attractor of the system (random test mode). 

The propagator is estimated using ordinary least squares (OLS) linear regression with multivariate monomials as features. That is, for a given state vector $s = (s_1, \dots, s_d) \in \R^d$ and a fixed maximal degree $p\in\N$, a feature vector $x \in\R^D$ is created using the feature function $\xi_p\colon\R^d\to\R^D$ such that
\begin{equation}
    x = \xi_p(s) := \br{\prod_{k=1}^ds_k^{\alpha_k} \ \Bigg\vert\ \alpha_1, \dots, \alpha_d\in \{0, \dots, p\},\, \sum_{k=1}^d\alpha_k \leq p}
    \eqfs
\end{equation}
For example, 
\begin{equation}
    \xi_2(s_1, s_2, s_3) = \br{1, s_1, s_2, s_3, s_1s_2, s_2s_3, s_1s_3, s_1^2, s_2^2, s_3^2}
    \eqfs
\end{equation}
The number of features is 
\begin{equation}
    D = 
    \binom{d+p}{d} = 
    \frac{(d+p)!}{d!\,p!}
    \eqfs
\end{equation}
For example, in $d=3$ dimensions, $D = 165$ features are obtained using maximal degree $p = 8$, and $D = 969$ using $p = 16$.
The feature function is applied to the observations $y_1, \dots, y_{n-1} \in \R^d$ to obtain $x_1, \dots, x_{n-1} \in\R^D$. Then OLS linear regression is applied to the data $(x_i, y_{i+1})_{i=1,\dots, n-1}$ to obtain the regression coefficients $\hat\beta\in\R^{D\times d}$. 
By stacking the feature vectors and targets row-wise into matrices $X \in\R^{(n-1)\times D}$ and $Y \in\R^{(n-1)\times d}$, respectively, the coefficients can be expressed as 
\begin{equation}\label{eq:beta}
    \hat\beta = (X\tr X)^{-1}X\tr Y
\end{equation}
assuming the matrix $X\tr X \in \R^{D\times D}$ is invertible. The estimated propagator is 
\begin{equation}
    \hat\Phi_{\stepsize} \colon \R^d \to \R^d,\, s\mapsto \hat \beta\tr\xi_p(s) 
    \eqfs
\end{equation}
The forecast $(\hat z_j)_{j=1, \dots, m}$ is defined by
\begin{equation}
    \hat z_{j} := \hat\Phi_{\stepsize}(\hat z_{j-1})
\end{equation}
with a known initial state $\hat z_{0} := z_0$.

While PolyProp is conceptually simple, executing it for high polynomial degrees presents a numerical challenge, as systems of linear equations given by potentially ill-conditioned matrices must be solved. We address this issue through appropriate data normalization and high precision arithmetic, as described in \cref{ssec:numerics} below.

Polynomial regression for learning dynamical systems has recently been employed \citep{Gauthier2021} with a focus on time delay embedding with relatively low-degree polynomials and Ridge regression (i.e., Tikhonov regularization). In contrast, PolyProp emphasizes high-degree polynomials without time delay embedding and achieves its results without regularization, i.e., using OLS. Note that regularization is typically used to reduce variance of parameters estimated from noisy data at the cost of an induced bias; here the data is noise-free, and we have found that using regularization makes the results worse due to the introduced bias. A predecessor to this method appears in a recent simulation study \citep{Schotz2025} (called \texttt{LinPo6} therein), but it is limited to degree $6$ polynomials and standard 64-bit precision.

The backbone of PolyProp, OLS with polynomial features, can be traced back to the 19th century \citep{Gergonne1815,Stigler1974}. Machine learning for L63 was first explored in the early 90s \citep{Elsner1992}.
Given that 64-bit computing architectures became widely available in the mid-2000s, it is remarkable that---even though the standard-precision version of this method (VPT of $21$ Lyapunov times on L63) was both conceptually and technologically available---it has not been discovered and presented in the literature for nearly two decades.
\subsection{Numerical Precision and Implementation}\label{ssec:numerics}
Executing the procedures described above numerically is necessarily an approximation to the mathematically exact calculations. Its precision is influenced by implementation choices, such as data normalization, the algorithm used for solving linear equations, and the amount of information used for storing a single number. 

On a standard machine with standard operations, a number is represented with 64 bit, which roughly translates to 15 significant digits. By using the C libraries MPLAPACK \citep{mplapack} and MPFR \citep{mpfr}, numbers are represented with 512 bit of information, which allows for more than 150 significant digits. The specific value of 512 bit is not a tight requirement; it is simply the power of 2 closest to a tenfold increase in information over 64 bit, chosen to ensure operation well above the precision threshold at which ill-conditioning of the linear systems involved in polynomial fitting dominates. Any sufficiently high precision would yield qualitatively equivalent results. We run different experiments in which different processing steps are executed with either single precision (32 bit), double precision (64 bit) or multi precision (512 bit).

First, a Runge Kutta ODE solver of order 4 (RK4) is used with a solver time step of $\stepsize_0 = 2^{-10} \approx 0.001$ for L63 and L96 to create a long trajectory of the system. For TCSA, $\stepsize_0 = 2^{-6} \approx 0.016$ is used, as it has a larger characteristic timescale. A starting point is then randomly chosen on the long trajectory, and temporal sub-sampling (taking every $k$-th element of the trajectory) is performed to create the training and test data with the desired step size $\stepsize$, number of observations $n$, and a sufficiently large number $m$ of states in the test set. The term \textit{ground truth} denotes a numerical approximation of this type to the system's analytical solution.

When using single or double precision for the polynomial forecaster, the results depend on whether and how the data is normalized. The training data $y_i$ may first be transformed to $\tilde y_i$ by normalizing it to mean $0$ and covariance matrix equal to the identity matrix, i.e., 
\begin{equation}
    \tilde y_i = \hat\Sigma^{-\frac12}(y_i - \hat\mu)
    ,\ 
    \hat\Sigma = \frac1{n\!-\!1}\sum_{i=1}^n(y_i - \hat\mu)(y_i - \hat\mu)\tr
    ,\ 
    \hat\mu = \frac1n\sum_{i=1}^ny_i
    ,
\end{equation}
where $\hat\Sigma^{\frac12}$ is the matrix square root of the symmetric positive definite matrix $\hat\Sigma$ and $\hat\Sigma^{-\frac12}$ is its inverse.
As this \textit{full} normalization is not common in the literature, we also compare it with the more common \textit{diagonal} normalization, where each dimension is scaled by the inverse of its standard deviation individually without changing correlations between variables, i.e.,
\begin{equation}
    \tilde y_i = \hat S^{-\frac12}(y_i - \hat\mu)
    ,\ 
    \hat S = \begin{pmatrix}
        \hat\sigma_1^2 &&\\
        &\ddots&\\
        &&\hat\sigma_d^2
    \end{pmatrix}
    ,\ 
    \hat \sigma_k^2 = \frac{1}{n\!-\!1}\sum_{i=1}^n \br{y_{i,k} - \hat\mu_k}^2
    .
\end{equation}
When normalizing the training data, the output of the estimated propagator has to be scaled back to the original scale before comparing it with the test data. When using multi-precision arithmetic (512 bit), we do not normalize the data, as numerical instability does not arise in this case.

We implement the simulation study in the R programming language \citep{R}. To evaluate \cref{eq:beta}, systems of linear equations need to be solved. If the matrix $X\tr X$ is numerically close to being singular (not invertible), numerical instabilities can cause inaccurate results. Thus, this is a crucial part of the implementation. While the standard approach in R relies on the LAPACK routine \texttt{DGESV}, we found improved numerical performance by using the Armadillo library \citep{Armadillo} for our single and double precision implementations. Specifically, we employ Armadillo's \texttt{solve} function with the options \texttt{refine}, \texttt{equilibrate}, \texttt{allow\_ugly}, and \texttt{likely\_sympd}, which typically yields more accurate results in this context. For the multi-precision implementation (512 bit), we use the MPLAPACK \citep{mplapack} routine \texttt{RGESV}.
\subsection{Experimental Design}
If not mentioned otherwise, we run experiments with training set size $n = 2^3, 2^4, \dots, 2^{15}$ and polynomial degrees $p = 1, 2, \dots, 16$. For L63, we use the time step sizes $\stepsize = 2^{-10}, 2^{-9}, \dots, 2^{-3}$; for L96, we use $\stepsize = 2^{-9}, \dots, 2^{-5}$ and restrict to $p = 1, 2, \dots, 8$; for TCSA, we use $\stepsize = 2^{-6}, 2^{-5}, \dots, 2^1$. Additionally, L96 is run with $n=2^{16}, 2^{17}$, and TCSA with $\stepsize = 2^{-2}$, $p = 23, 24, 25$, and $n = 2^{16}, 2^{17}$. Each experiment is repeated 100 times with a different random trajectory sample. In each case, $\vpt{0.5}$ is calculated. In the results section, we report the mean over the resulting 100 VPT values.

For reference, we estimate the accuracy of the RK4 solver in terms of $\vpt{0.5}$, starting from differently rounded initial conditions and using calculations of different precision internally. Averaging over $10^4$ repetitions with randomly chosen initial conditions on the systems' attractor yields the values given in \SuppSec{vptref}.

Computational resources used in this study are described in \SuppSec{compute}. Details of all experimental results are given in \SuppSec{details}.
\section{Forecasting Results}\label{sec:results}

This section evaluates PolyProp on the three chaotic benchmark systems introduced in \cref{sec:setup}. PolyProp consistently matches the accuracy of the ODE solvers used to generate the data, thereby surpassing previous limits on prediction accuracy in low-dimensional chaotic systems.
\subsection{Lorenz-63 with Default Setup}
For the L63 model, the dynamics of the three state dimensions $x(t)$, $y(t)$, and $z(t)$ depending on time $t$ are given by the ordinary differential equation (ODE)
\begin{equation}\label{eq:l63}
    \dot x = 10(y - x)
    \eqcm\quad
    \dot y = x(28 - z) - y
    \eqcm\quad
    \dot z = xy - \frac{8}{3} z
    \eqcm
\end{equation}
where the dot above a variable denotes its temporal derivative.

Since the L63 system models a physical process, it is appropriate to use the most accurate available approximation to the exact solution for data generation. Here, this corresponds to the output of a high-precision (512-bit) fourth-order Runge-Kutta (RK4) ODE solver. To reflect typical computational practice, the data is assumed to be stored in double precision (64 bit). The dataset is then split into train and test sets. The train set is used to fit the PolyProp forecasting model.
Given that PolyProp is not computationally intensive, there is no need to artificially limit its computational precision. Therefore, a 512-bit implementation of PolyProp is used for prediction and evaluated on the test set. This default configuration is illustrated in \cref{fig:dataflow:hlh}.

\begin{figure}
    \centering
    \includegraphics[width=0.8\textwidth]{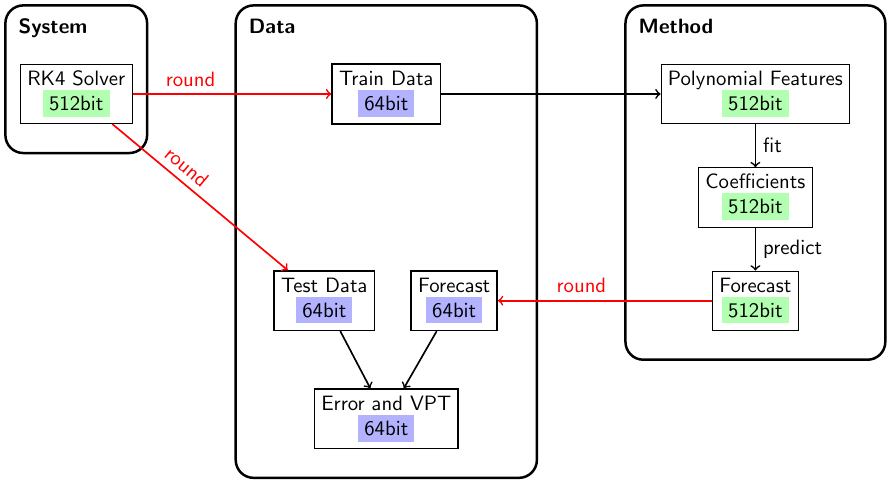}
    \caption{\textbf{Data processing diagram for the default setup.} The ODE system is solved with a multi precision (512-bit) RK4 solver. The train and test data are then stored at double precision (64 bit). The polynomial propagator method internally calculates with multi precision. The evaluation and calculation of the error metric is executed at standard double precision.}\label{fig:dataflow:hlh}
\end{figure}

\begin{figure*}
    \begin{center}
        \includegraphics[width=1.0\textwidth]{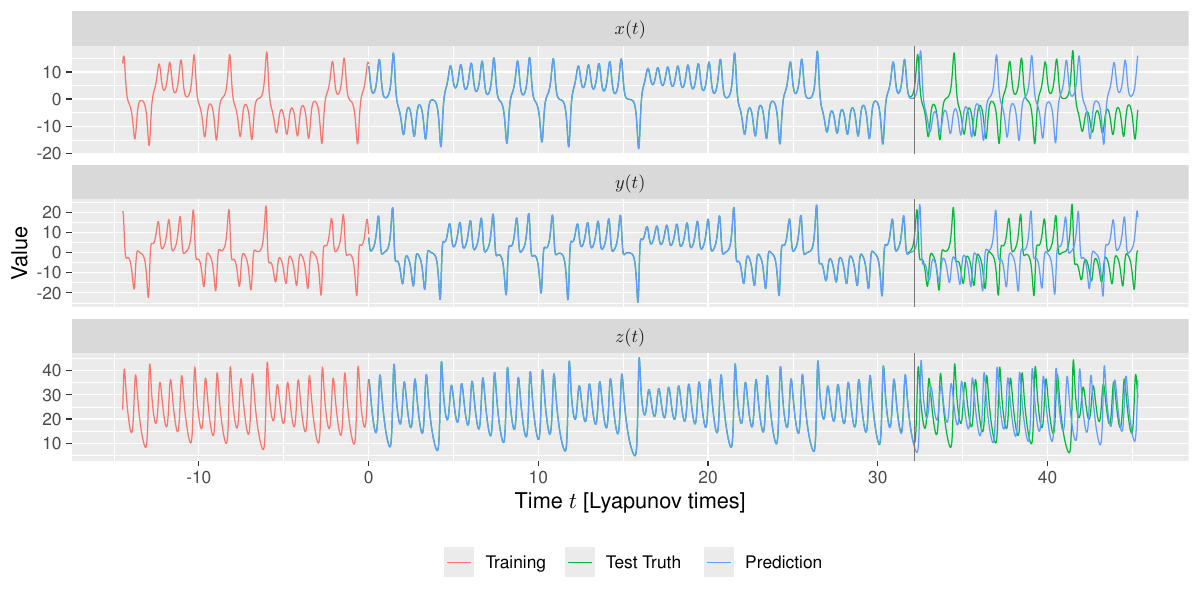}
    \end{center}
    \caption{\textbf{An example illustrating the forecast abilities of the polynomial propagator introduced in this study.} The L63 system is solved with an RK4 ODE solver calculating at 512-bit precision with time step $2^{-10}\approx 0.001$ (system time). The result is stored only at 64-bit precision. It is sub-sampled to obtain a time series with time step $\stepsize = 2^{-6} \approx 0.016$. The first $n = 2^{10} = 1024$ samples (red) are used as training data ($16.0$ system time units or $14.5$ Lyapunov times) to fit a polynomial of degree $9$ to the propagator using 512-bit arithmetic. The resulting prediction (blue) is compared with the ground truth of the ODE solver (green). The prediction is valid for about $32$ Lyapunov times (dark gray vertical line), more than twice the duration of the training data. Further examples, including one with explicit regression coefficients, can be found in \SuppSec{examples}.}
    \label{fig:example}
\end{figure*}

An example trajectory is shown in \cref{fig:example}, while results for the best-performing polynomial degree $p \in \{1, \dots, 16\}$ across different numbers of observations $n$ and time steps $\stepsize$ are presented in \cref{fig:best:hlh}.
Note that this section shows the results for a sequential test mode, i.e., the forecast is started from the last training observation. The results for a random test mode, in which the test set is a new randomly chosen time series on the attractor, are almost identical, see \SuppSec{testmode}.
\begin{figure*}
    \includegraphics[width=\textwidth]{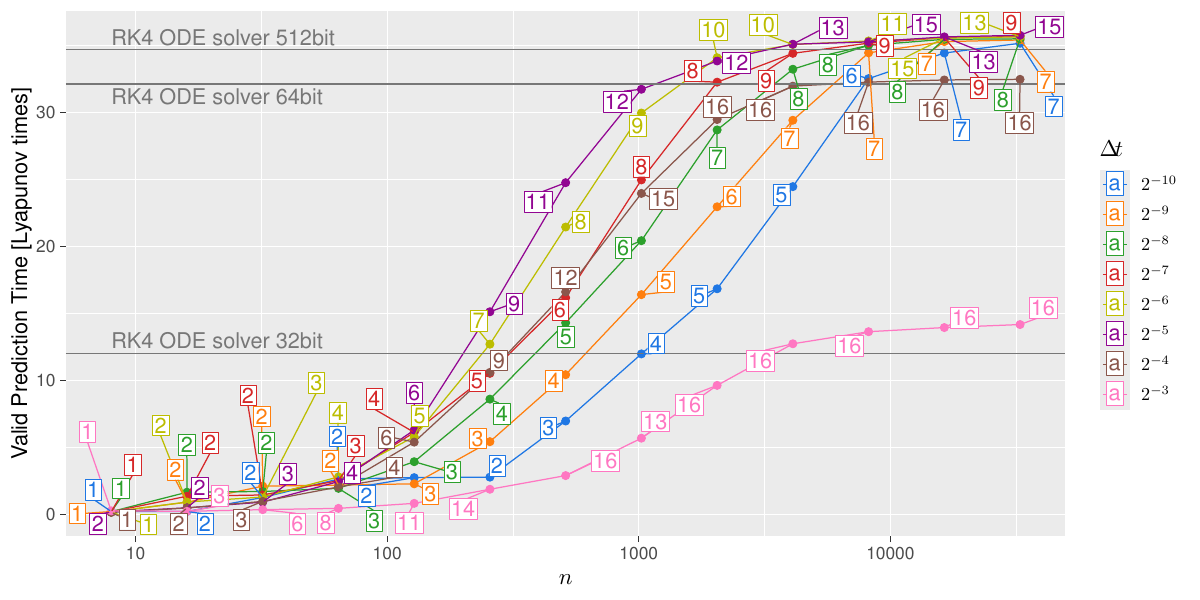}
    \caption{\textbf{Valid prediction times (VPT) for the L63 system using PolyProp with optimal degree under the default setup.} The setup uses a multi-precision solver (512-bit), double-precision data storage (64-bit), and a multi-precision implementation of PolyProp. The data time step $\stepsize$ is indicated by color. Label boxes within the plot denote the degree (between 1 and 16) of the optimal polynomial for each case. Gray horizontal lines show the performance of ODE solvers at various precisions, using the system's governing equations and (rounded) 64-bit initial conditions. The VPT (with threshold $0.5$) values on the vertical axis are given in multiples of the Lyapunov time.}
    \label{fig:best:hlh}
\end{figure*}

The optimal polynomial degree is observed to increase with the number of observations $n$ and step size $\stepsize$. This is expected, as with larger $n$ more complex models can be fitted, and with larger $\stepsize$ the true propagator function increases in complexity. 

There seems to be a trade-off regarding data efficiency in the step size, with the medium value $\stepsize = 2^{-5} \approx 0.03$ showing the best performance for most values of $n$. We interpret this as follows: for smaller step sizes, the target function of the regression---the propagator map---becomes simpler, allowing a precise fit with a small amount of data. In general, this is counteracted to some extent by the need for more prediction steps to achieve the same forecast time. But potentially more importantly, the state space is not sufficiently explored for a fixed $n$ if the step size is too small.

Recall that the ground truth is generated by a 512-bit RK4 solver. Due to the sensitive dependence on initial conditions, a 512-bit RK4 prediction initialized with data rounded to 64-bit precision inevitably diverges from this ground truth. A value of $\vpt{0.5} = 34.7$ Lyapunov times is measured for this effect. Given $n \geq 2^{12}$ noise-free observations of the ground truth trajectory rounded to 64-bit with a suitable time step, PolyProp matches and even slightly exceeds this value. Only the largest time steps do not reach this precision, likely because higher polynomial degrees would be required to fit the higher complexity of the propagator that comes with taking larger steps. In the best case, PolyProp reaches $\vpt{0.5} = 35.7$ on average for $n=2^{15} \approx 33{,}000$ observations, a time step of $\stepsize\leq 2^{-5} \approx 0.03$, and polynomial degree $p = 15$. This suggests that PolyProp successfully learns the underlying dynamics from the rounded data and compensates for the initialization error more effectively than the RK4 ODE solver.

\begin{figure*}
    \includegraphics[width=\textwidth]{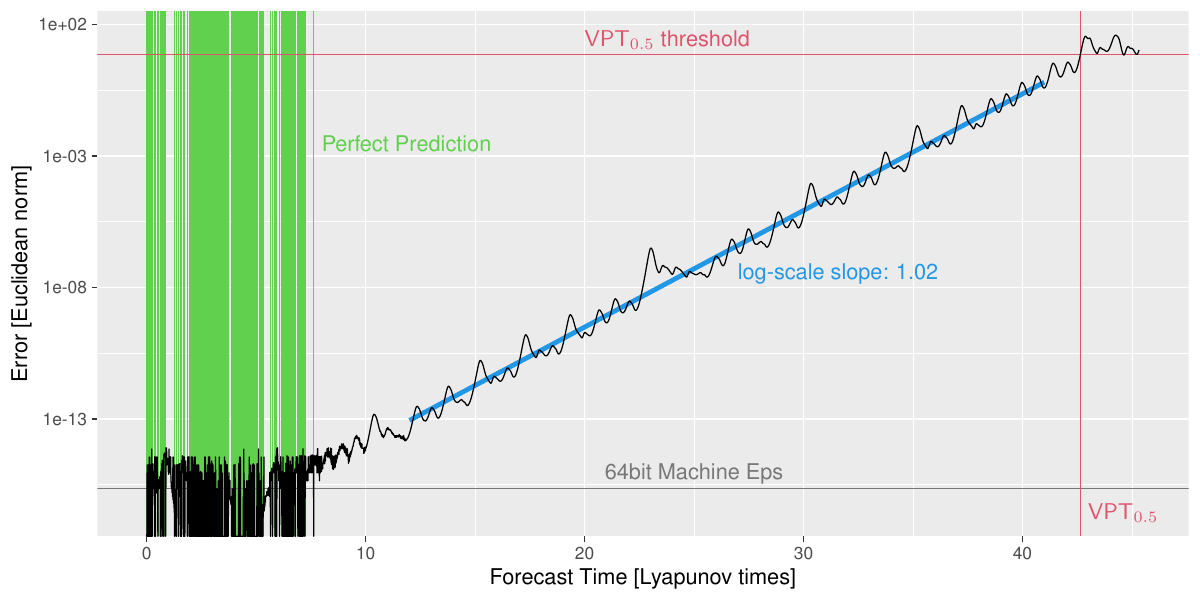}
    \caption{\textbf{Euclidean distance between forecast and ground truth of a single L63 run.} This is the best-performing repetition of the default setup (L63, 512-bit system, 64-bit data, 512-bit method, sequential test mode, $n=2^{15}$, $\stepsize=2^{-7}$, $p=9$). The horizontal locations of the vertical green lines mark time points of the forecast with numerically perfect prediction ($\hat u(t) = u(t)$ in 64-bit arithmetic). The blue line is the linear fit between forecast time in Lyapunov times and the natural logarithm of the Euclidean distance; it has almost the theoretically predicted slope $1$ of a perfectly emulated L63 system. The red lines show the VPT calculation; here, $\vpt{0.5} = 42.6$ is obtained, which is a positive outlier compared to 90\% of runs in this setting yielding $\vpt{0.5}$ between $33.0$ and $39.0$.}\label{fig:perfect}
\end{figure*}

An example of the evolution of the error between prediction and ground truth is shown in \cref{fig:perfect}. In contrast to the defining property of chaos, a constant error is observed for about 8 Lyapunov times when rounded to 64 bit. This, of course, does not break the chaotic behavior. It rather shows that, internally, the polynomial forecaster has a high-precision representation of the system state that is more accurate than the 64 bit at which the train and test data are stored. After the initial period of seemingly constant error, the Euclidean distance grows by a factor of Euler's number each Lyapunov time---the same as it would when using the exact L63 system (by definition of Lyapunov time and Lyapunov exponent).
\subsection{Variations on Numerical Precision}
We explore the effect of single (\textsf{s}, 32-bit), double (\textsf{d}, 64-bit), and multi (\textsf{m}, 512-bit) precision in the system- (the RK4 ODE solver), data- (storage of train and test data), and method- (PolyProp) parts of the processing pipeline for L63. The results are shown in \cref{fig:bbest:L63}.
\begin{figure*}
    \includegraphics[width=\textwidth]{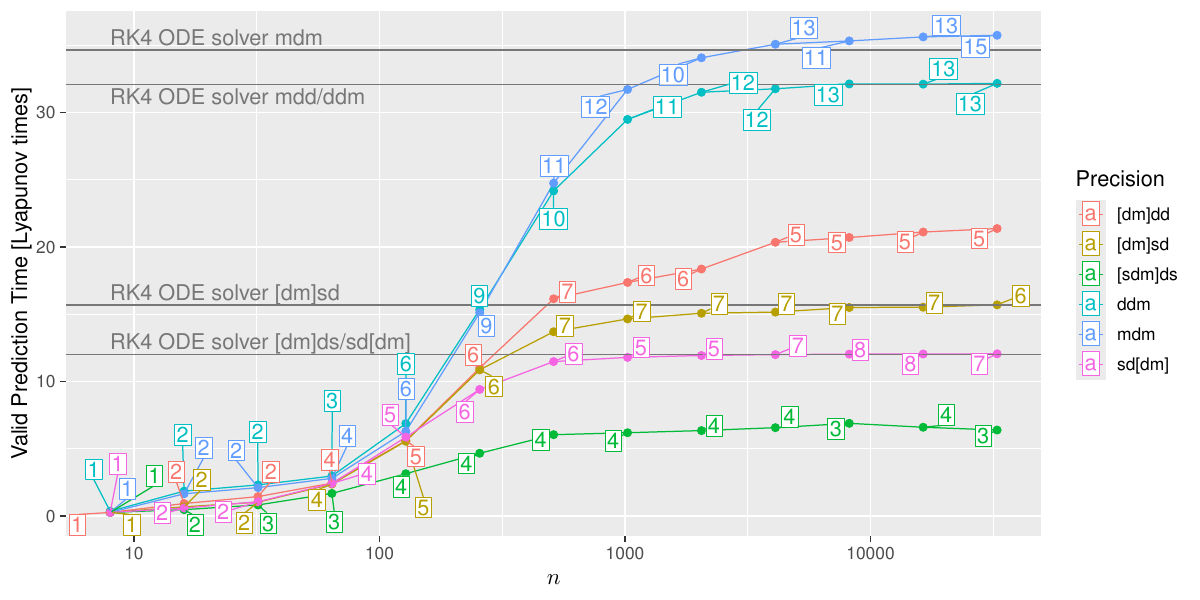}
    \caption{\textbf{Valid prediction time (VPT) for L63 using PolyProp with best polynomial degree and best time step in different precision settings.} The precision is encoded by three letters. The first letter represents the system (ODE solver), the second the data storage, and the third the forecasting method. For each position, s, d, and m indicate single (32-bit), double (64-bit), and multi (512-bit) precision, respectively. Multiple letters in brackets indicate that the choice of any of these precisions for the given part of the processing pipeline yields almost the same results. For each precision setting and value of $n$, the $\vpt{0.5}$ (in Lyapunov times) averaged over 100 repetitions is taken and maximized over the polynomial degree $p$ (given in label boxes) and the data time step $\stepsize \geq 2 \stepsize_0$ (not shown). The gray horizontal lines show the accuracy of RK4 ODE solvers in different precision setups for reference.}
    \label{fig:bbest:L63}
\end{figure*}
If the source data lacks precision, the prediction accuracy is also limited. To achieve the best results for given data, the method precision needs to be larger than the effective data precision (the minimum of system and data precision). This is because solving systems of linear equations to fit polynomial propagators---especially for higher-degree polynomials---can involve ill-conditioned matrices, where numerical errors amplify and lead to a loss of accuracy beyond the nominal precision of the arithmetic used.

\cref{fig:bbest:L63} shows that PolyProp achieves predictive fidelity comparable to, or better than, numerical integration with RK4. Based on 95\% confidence intervals for PolyProp's VPT value from 100 repetitions, we highlight the following specific comparisons:
\begin{enumerate}[label=(\alph*)]
    \item The prediction of double- or multi-precision PolyProp trained on data ($n \geq 2^{10}$) produced by a single-precision RK4 ODE solver diverges at about the same time as the trajectories produced by a single- and a double-precision RK4 ODE solver. (Precision setting \textsf{sdd} or \textsf{sdm}, $\vpt{0.5} = 12.0$ Lyapunov times).
    \item 
    Consider ground truth data produced by a double- or multi-precision RK4 ODE solver. 
    The prediction of double-precision PolyProp trained on ground truth data ($n \geq 2^{11}$) rounded to single precision is about as accurate as that of an RK4 ODE solver starting from single-precision initial conditions.
    (Precision setting \textsf{dsd} or \textsf{msd}, $\vpt{0.5} = 15.7$ Lyapunov times).
    \item 
    The prediction of multi-precision PolyProp trained on data ($n \geq 2^{13}$) produced by a double-precision RK4 ODE solver diverges at about the same time as the trajectories produced by a double- and a multi-precision RK4 ODE solver.
    (Precision setting \textsf{ddm}, $\vpt{0.5} = 32.1$ Lyapunov times).
    \item 
    Consider ground truth data produced by a multi-precision RK4 ODE solver. 
    The prediction of multi-precision PolyProp trained on ground truth data ($n \geq 2^{12}$) rounded to double precision is at least as accurate as that of an RK4 ODE solver starting from double-precision initial conditions.
    (Precision setting \textsf{mdm}, $\vpt{0.5} = 34.7$ Lyapunov times for the ODE solver and $\vpt{0.5} = 35.7$ Lyapunov times for PolyProp).
\end{enumerate}

The impact of numerical accuracy in solving linear systems is also evident when comparing different data normalization schemes for single- and double-precision methods (see \SuppSec{normalization}). In general, normalization enhances results, with more sophisticated schemes outperforming default approaches.

Without artificially limiting the precision of the training data, much higher VPT values can be achieved, see \cref{fig:L63:mmm}. The limit of about $105$ Lyapunov times shown in the plot is likely due to the constraint $p \leq 16$, with even higher VPT values possible for higher degrees. Note that with the data time step equal to the solver time step, i.e., $\stepsize = 2^{-10} = \stepsize_0$, extremely high VPT values are achieved with degree $8$ polynomials (more than $322$ Lyapunov times for $n=2^{15}$; not shown in the plot). 

This can be explained as follows: applying $k$ steps of the RK4 ODE solver to the L63 model yields a polynomial propagator of degree $F_{4k+2}$, where $F_\ell$ is the $\ell$-th Fibonacci number. This is shown in \SuppSec{rk4l63poly}. Thus, one RK4 solver step of the L63 system amounts to a polynomial propagator of degree $F_{6}=8$. Hence, the fits can become near perfect in this case if $p \geq 8$. This is not a concern for larger time steps, as two RK4 solver steps for L63 already require degree $F_{10} = 55$ to be fully represented as a polynomial.

Further experiments using high-precision ODE solvers are detailed in \SuppSec{xprec}. Notably, when using a Taylor Integrator with extreme precision (absolute tolerance $10^{-60}$) as the ground truth, PolyProp---trained and initialized with rounded 64-bit data---achieves a VPT of $35.6$ Lyapunov times. This performance is statistically indistinguishable from that of the Taylor Integrator itself when initialized with a rounded 64-bit state ($35.8$ Lyapunov times).

Finally, \SuppSec{noise} investigates PolyProp's performance on noisy data, viewing observational noise as a stochastic limitation on precision analogous to deterministic rounding. Consistent with the findings in \cref{fig:bbest:L63}, PolyProp's forecast accuracy matches that of the ODE solver when both are initialized with the same noisy state. This demonstrates that PolyProp achieves the optimal predictive fidelity possible when no data assimilation is applied to denoise the initial state.
\begin{figure*}
    \includegraphics[width=\textwidth]{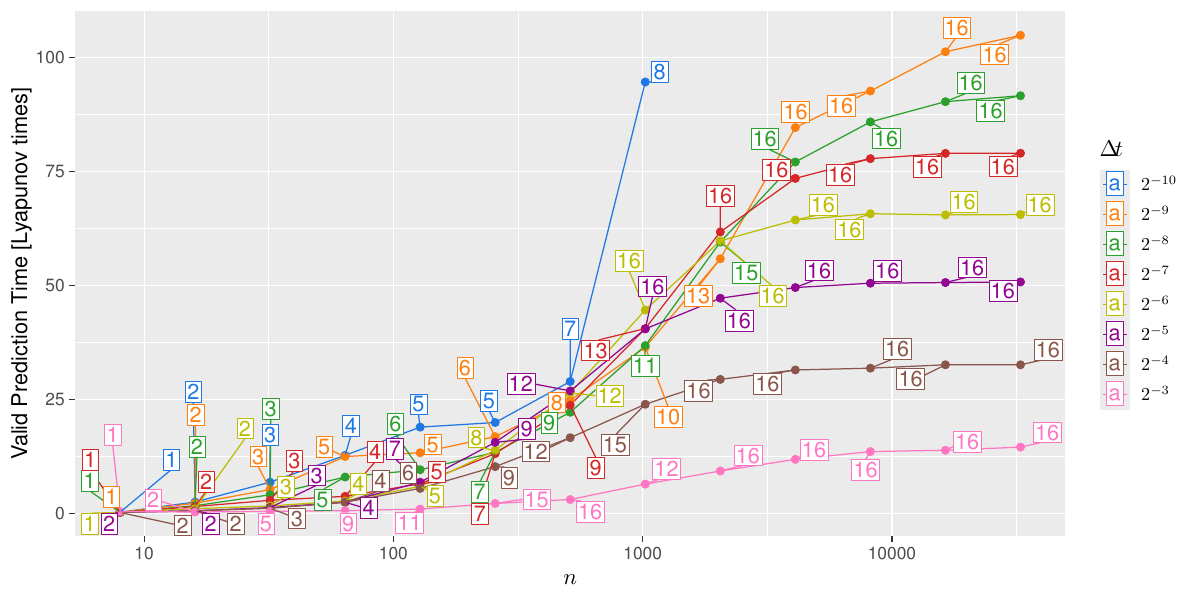}
    \caption{\textbf{Valid Prediction Time (VPT) for L63 using PolyProp with best polynomial degree for maximum precision throughout the processing pipeline.} The setup uses multi-precision (512-bit) for the system, data, and method. For each value of $\stepsize$ (indicated by color) and $n$, the $\vpt{0.5}$ (in Lyapunov times) averaged over 100 repetitions is taken and maximized over the polynomial degree $p$ (given in the label boxes). For the setting $\stepsize = \stepsize_0$ (blue), only results for $n\leq 2^{10}$ are shown. Note that the polynomial degree $p$ is limited to $p \leq 16$.}\label{fig:L63:mmm}
\end{figure*}
\subsection{Thomas' Cyclically Symmetric Attractor}
One may still wonder whether the success of PolyProp is due to the polynomial nature of the L63 vector field and the corresponding RK4 solution (even if the degree in that case is extremely high). To remove such doubts, the machine precision results shown for L63 in the standard setting (512-bit system, 64-bit data, 512-bit method) are replicated for Thomas' Cyclically Symmetric Attractor (TCSA), a chaotic dynamical system, governed by non-polynomial dynamics given by
\begin{equation}
    \dot x = \sin(y) - bx
    \eqcm\quad
    \dot y = \sin(z) - by
    \eqcm\quad
    \dot z = \sin(x) - bz
    \eqcm
\end{equation}
where $b = 0.208$. 

The results are presented in \cref{fig:TCSA:best:hlh}. Using the standard setup with up to $n=2^{15}$ observations and a polynomial degree of up to $p=16$, $\vpt{0.5}=25.9$ Lyapunov times is achieved. This is lower than the reference value of 37.4 obtained from a 512-bit ODE solver initialized with data rounded to 64 bit. By extending the setup to include polynomial degrees $p=23,24,25$ and observation counts $n = 2^{16}, 2^{17}$ for the time step $\stepsize = 2^{-2}$, the result improves to $\vpt{0.5} = 38.0$, with a 95\% confidence interval of $[36.8, 39.3]$, thereby matching the performance of the multi-precision ODE solver and outperforming the double-precision solver ($\vpt{0.5} = 32.2$) when starting on initial conditions rounded to double precision.

\begin{figure*}
    \includegraphics[width=\textwidth]{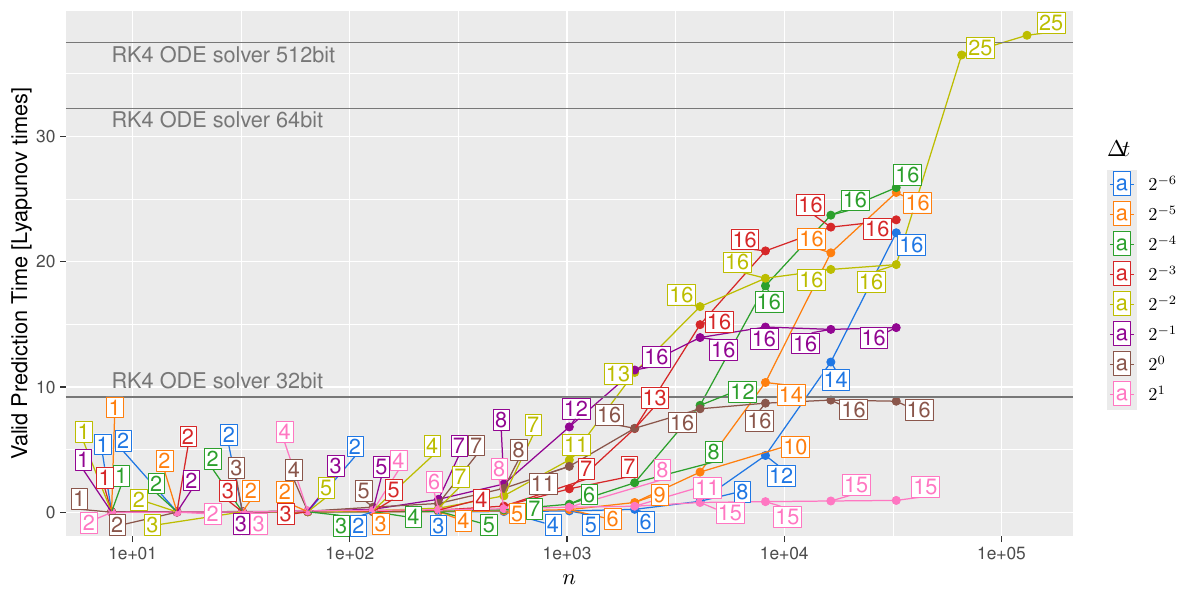}
    \caption{\textbf{Valid prediction times (VPTs) for the TCSA system using PolyProp with optimal degree under the default setup.} The setup uses a multi-precision solver (512-bit), double-precision data storage (64-bit), and a multi-precision (512-bit) implementation of PolyProp. The data time step $\stepsize$ is indicated by color. Label boxes within the plot denote the degree of the optimal polynomial for each case. Gray horizontal lines show the performance of ODE solvers at various precisions, using the system's governing equations and (rounded) 64-bit initial conditions. The $\vpt{0.5}$ values on the vertical axis are given in Lyapunov times. In addition to the default settings ($n \leq 2^{15}$, $p \leq 16$), $p = 23,24,25$ and $n = 2^{16}, 2^{17}$ are added for $\stepsize = 2^{-2}$.}
    \label{fig:TCSA:best:hlh}
\end{figure*}
\subsection{Lorenz-96 System}
To investigate how PolyProp performs in higher dimensions, consider the Lorenz-96 (L96) system \citep{Lorenz1996},
\begin{equation}
    \dot x_i = (x_{i+1} - x_{i-2})x_{i-1} - x_{i} + 8
\end{equation}
for $i=1,\dots,d$ with $x_{-1} = x_{d-1}$, $x_{0} = x_d$, $x_{d+1} = x_1$.
Instead of using 512-bit arithmetic to achieve double-precision (64-bit) accuracy as in the default setting, the goal here is single-precision (32-bit) machine accuracy on the Lorenz-96 system, using standard double-precision computations for PolyProp. 

Single-precision machine accuracy is demonstrated in two ways: 1) a double-precision ODE solver is used to generate the ground truth, and the data is rounded to single precision before training; and 2) the data is generated directly using a single-precision solver. In both cases, the training data is single-precision, but the underlying dynamics are computed at either single or double precision. 

To evaluate prediction quality, the VPT of PolyProp (executed in double precision) is compared with that of a double-precision ODE solver initialized with the single-precision data. In all tested settings, the best results from PolyProp are statistically indistinguishable from the solver's, as shown in \cref{fig:L96:n17} and \SuppSec{l96summary}. This demonstrates that PolyProp reaches single-precision machine accuracy.

This holds across all tested dimensions $d=5,\dots,9$. Higher dimensions require more training data, but no consistent trend in the required polynomial degree is observed. Notably, for fixed polynomial degree, higher dimensions naturally yield more input features.

\begin{figure}
    \centering
    \includegraphics[width=0.95\textwidth]{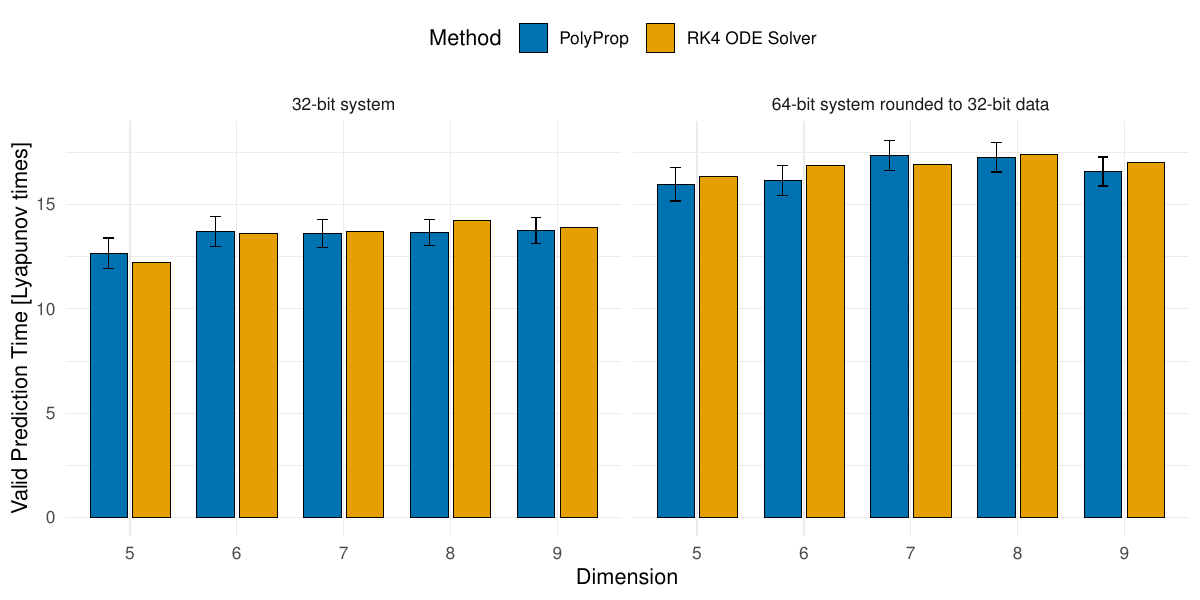}
    \caption{\textbf{Valid prediction time (VPT) for the Lorenz-96 (L96) system using PolyProp and RK4 ODE solvers.} For the two precision settings and dimensions $d=5, \dots, 9$ for L96, the VPT (averaged over 100 repetitions) of PolyProp trained on $n=2^{17}$ observations and time step $\stepsize = 2^{-7}$ is reported. The mean and 95\% confidence interval for the optimal polynomial degree $p \in \{1, \dots, 8\}$ are shown. These are compared to the mean VPT from $10^4$ trajectories computed with a 64-bit RK4 solver initialized from 32-bit initial conditions randomly sampled from the system's attractor. The RK4 mean VPT values have 95\% confidence intervals all smaller than 0.3 Lyapunov times and are not shown. All VPT values are expressed in units of the system's Lyapunov time.}
    \label{fig:L96:n17}
\end{figure}
\subsection{Additional Results}
We evaluate PolyProp also on seven additional chaotic dynamical systems to illustrate the generality of our results, see \SuppSec{additional}. For each setting we achieve machine precision in the same precision setups as previously used for the Lorenz-96 system.
%
\section{Enabling Factors, Comparisons, and Limitations}\label{sec:discussion}

The results in the previous section raise three natural questions. Why do existing data-driven methods not reach machine precision? What allows PolyProp to achieve it? And how far does the approach scale? These are addressed in turn.

\subsection{Limitations of Existing Methods}

To place PolyProp in the context of existing work, three classes of approaches for forecasting low-dimensional chaotic dynamical systems that have been widely used in the literature are identified: Sparse Identification of Nonlinear Dynamical Systems (SINDy) \citep{Brunton2015}, Long Short-Term Memory networks (LSTMs) \citep{Hochreiter1997}, and Reservoir Computers \citep{jaeger2001echo}.

To predict a system $\dot u = f(u)$, the SINDy algorithm uses estimates of the derivative $\dot u(t)$ to fit the model function $f$ and then uses an ODE solver to create a forecast. The model function $f$ is assumed to be a sparse linear combination of a predefined set of features. For example, in the standard version of SINDy \citep{Brunton2015}, this set of features is the monomials up to degree $5$. Thus, the algorithm is not fully system-agnostic, as some information about the target system must be encoded in the feature set. (In contrast, PolyProp fits the typically non-polynomial propagator $\Phi_{\stepsize}$, rather than fitting the model function $f$, and is effective for any sufficiently smooth ODE.) SINDy achieves a VPT of up to 13 Lyapunov times for L63 (\cref{tbl:L63:vpt:lit}), but would not be able to accurately recover systems such as TCSA, where $f$ includes non-polynomial components like sine functions. Furthermore, the necessity of estimating derivatives from data introduces error (increasing with larger time steps $\stepsize$). Even if the feature set is sufficiently expressive, these derivative estimation errors prevent SINDy from achieving forecasts at or near machine precision.

LSTMs, a class of recurrent neural networks, are trained using gradient-based optimization to minimize prediction error over sequences of data. While they can, in principle, approximate complex temporal dynamics, matching the precision of ODE solvers would require their parameters to be tuned with near-exact accuracy---a level of precision that is unrealistic for gradient descent methods, which are inherently approximate and prone to local minima, saddle points, and vanishing gradients. Accordingly, the best VPT that has been achieved with LSTMs for the L63 system is 6 Lyapunov times (\cref{tbl:L63:vpt:lit}). 

Unlike LSTMs, most Reservoir Computers, such as the Echo State Network (ESN) \citep{jaeger2001echo}, are typically trained by solving a linear system of equations, which allows for more precise weight estimation. Similar to LSTMs, Reservoir Computers predict the next state $u(t+\stepsize)$ based on the current state $u(t)$ and an internal memory state $r(t)$, known as the reservoir. Up to 13 Lyapunov times VPT have been achieved with Reservoir Computers (\cref{tbl:L63:vpt:lit}). However, in the setup considered here, where the system evolves deterministically as $u(t+\stepsize)=\Phi_{\stepsize}(u(t))$, the use of an auxiliary memory is unnecessary and complicates the learning task.
Moreover, RCs commonly use regularization during training to reduce variance and improve numerical stability, although some can be trained without regularization in well-conditioned regimes \citep{Santos2025}. While beneficial with noisy data, regularization introduces bias that limits forecasting accuracy, preventing performance at machine precision.
Finally, some RCs such as the ESN rely on randomly initialized neural networks to construct features. Compared to structured polynomial bases, this random approach is less efficient for precise function approximation and, without regularization, would lead to numerical instability when the resulting features are nearly linearly dependent.

\subsection{Factors Enabling Machine Precision Forecasting}\label{ssec:enabling_factors}

The propagator $\Phi_{\stepsize}$ of a smooth ODE is itself a smooth function, and smooth functions can be approximated by polynomials---locally via Taylor's theorem, and uniformly on compact sets via the Weierstrass approximation theorem. The best-fit polynomial of a given degree can be found by least squares. Hence, with sufficient data and sufficient polynomial degree, the propagator of any smooth ODE can be approximated to arbitrary accuracy. Two key ingredients are needed to realize this in practice: high-degree polynomials and multi-precision arithmetic.

High-degree polynomials are traditionally viewed with skepticism in machine learning, often associated with overfitting. However, in the absence of noise, this concern does not apply: since the training data is noise-free, there is no noise to overfit to. In the simulation study presented here, polynomials of degree up to 25 are used effectively.

Fitting such high-degree polynomials, in turn, requires multi-precision arithmetic. To exploit the full range of significant digits provided by noise-free data, the linear systems involved in polynomial fitting must be solved in a numerically stable way. For high-degree polynomials, the matrix $X^\top X$ in \eqref{eq:beta} becomes severely ill-conditioned \citep{Santos2025}, so that standard double-precision solvers fail to preserve these digits. Using 512-bit arithmetic provides enough margin that, after the conditioning-induced loss, a useful number of digits remains. To a smaller but non-negligible extent, 512-bit arithmetic also reduces the per-step rounding error in iterated state propagation, which contributes to the near-perfect short-term predictions visible in \cref{fig:perfect}. The specific value 512 is not a tight requirement---it is simply the power of 2 closest to a tenfold increase in precision over 64 bits. Any sufficiently high precision would yield qualitatively equivalent results.

Note that modern machine learning frameworks, such as PyTorch and TensorFlow, do not support numerical precision beyond double precision (64 bit) on either CPU or GPU. Moreover, they are typically optimized for single precision (32 bit) on GPU, rendering them unsuitable for the multi-precision (512-bit) experiments conducted in this study.

\subsection{Scalability and Robustness}

PolyProp faces challenges related to the curse of dimensionality. For a system with dimension $d$, the number of parameters $D$ in a multivariate polynomial of degree $p$ scales as $D = \mathcal{O}(d^p)$. Since fitting PolyProp requires solving a linear system with a computational complexity of $\mathcal{O}(D^3)$, the method becomes computationally prohibitive for high-dimensional settings, even with moderate polynomial degrees. Dimension reduction techniques may help to alleviate this scalability issue. Alternatively, in the context of Partial Differential Equations, one may apply localization strategies---using only local spatial patches as input---to reduce the effective input dimension \citep{Mandal2025}.

While this study demonstrates how extremely accurate predictions of chaotic systems can be achieved, the main simulation setup is clearly artificial. In practical scenarios, noise-free observations with the 15-digit precision characteristic of 64-bit floating-point arithmetic are rare, found only in specialized domains such as atomic clock measurements (up to $10^{-18}$ relative precision) \cite{Nicholson2015}. Far more commonly, measurements are corrupted by noise, rendering even leading digits unreliable. 

In such settings, PolyProp yields an optimal result among non-denoising methods, matching the forecast accuracy of numerical ODE solvers initialized with noisy data, as shown in \SuppSec{noise}. Yet, despite this relative optimality, forecast accuracy deteriorates rapidly in the presence of noise---a limitation shared by many other learning techniques for dynamical systems \citep{Schotz2025}.
\section{Conclusion}

This study has shown that system-agnostic polynomial regression, trained with suitable numerical precision on sufficiently accurate data, can match the accuracy of ODE solvers when predicting chaotic dynamics from the same initial state. This result was explicitly demonstrated for the 3-dimensional polynomial Lorenz-63 system, the 3-dimensional non-polynomial Thomas' Cyclically Symmetric Attractor, and the Lorenz-96 system in dimensions up to $d=9$. These findings suggest that PolyProp generalizes to a broad class of dynamical systems governed by autonomous first-order ODEs. From this perspective, the problem of forecasting chaotic systems from noise-free data may be considered effectively solved, and standard benchmarks such as the noise-free Lorenz-63 model are no longer adequate for distinguishing methods---they have become trivial under these conditions.

The cornerstones of PolyProp---multi-precision arithmetic and high-degree polynomials---are well suited to the noise-free setting but ill-suited to overcome noise, which remains the dominant obstacle in practical applications. Matching forecast accuracy under noise requires denoising of the initial conditions rather than higher precision or higher polynomial degree. Given that noise is ubiquitous in real-world data, future research on learning dynamical systems should focus on methods that inherently integrate denoising capabilities.
%
%
%
%
%
\paragraph{Data and Code Availability}
The simulation results and source code supporting the findings of this study are archived at Zenodo (\url{https://doi.org/10.5281/zenodo.15863304}). The source code is also available via GitHub at \url{https://github.com/chroetz/PolyProp}.
\printbibliography
\paragraph*{Acknowledgments:}
The authors gratefully acknowledge the Ministry of Research, Science and Culture (MWFK) of Land Brandenburg for supporting this project by providing resources on the high performance computer system at the Potsdam Institute for Climate Impact Research. N.B.\ acknowledges funding by the Volkswagen Foundation. This is ClimTip contribution \#184; the ClimTip project has received funding from the European Union's Horizon Europe research and innovation programme under grant agreement No. 101137601. 

\paragraph{Author Contribution}
C.S.\ developed the methodology, designed and carried out the simulation study, analyzed the simulation results, and wrote the original draft of the manuscript. N.B.\ contributed to the design of the simulation study, to structuring the presentation of the method and simulation results, and revised the manuscript. Both authors discussed the results.

\paragraph{Competing interests}
The authors declare no competing interests.
\clearpage
\section*{Supplementary Information} 
The Supplementary Information provides further details on the simulation experiments and their results as well as a mathematical proof on the polynomial degree of the RK4 ODE solver applied to the L63 model.


\clearpage

\pagenumbering{arabic} 
\setcounter{page}{1}

\vspace*{2cm}
\begin{center}
    \begin{minipage}{\textwidth}
        \centering\LARGE\bfseries
        \setlength{\baselineskip}{1.25\baselineskip}
        Supplementary Information for:\\[0.3em]
        Machine-Precision Prediction of Low-Dimensional Chaotic Systems from Noise-Free Data
    \end{minipage}\\[1.5em]
    {\large Christof Sch\"otz\textsuperscript{1,2} and Niklas Boers\textsuperscript{1,2,3}}\\[1em]
    \begin{minipage}{0.9\textwidth}
        \centering\small
        \textsuperscript{1}Technical University of Munich, Germany; Munich Climate Center; TUM School of Engineering and Design, Department of Aerospace and Geodesy, Earth System Modelling Group\\
        \textsuperscript{2}Potsdam Institute for Climate Impact Research, Germany\\
        \textsuperscript{3}University of Exeter, UK; Department of Mathematics\\[0.8em]
        christof.schoetz@tum.de, \orcidthanks{0000-0003-3528-4544}, corresponding author\\
        n.boers@tum.de, \orcidthanks{0000-0002-1239-9034}
    \end{minipage}
\end{center}
\vspace{2em}

\renewcommand{\theequation}{S\arabic{equation}}
\renewcommand{\thefigure}{S\arabic{figure}}
\renewcommand{\thetable}{S\arabic{table}}
\renewcommand{\thesection}{S\arabic{section}}
\renewcommand{\thesubsection}{\thesection.\arabic{subsection}}
\renewcommand{\thetheorem}{\thesection.\arabic{theorem}}
\renewcommand{\thepage}{S\arabic{page}}

\renewcommand{\theHequation}{S\arabic{equation}}
\renewcommand{\theHfigure}{S\arabic{figure}}
\renewcommand{\theHtable}{S\arabic{table}}
\renewcommand{\theHsection}{S\arabic{section}}
\renewcommand{\theHsubsection}{\thesection.\arabic{subsection}}
\renewcommand{\theHtheorem}{\thesection.\arabic{theorem}}
\def\theHpage{S\arabic{page}}

\setcounter{equation}{0}
\setcounter{figure}{0}
\setcounter{table}{0}
\setcounter{section}{0}
\setcounter{subsection}{0}
\setcounter{theorem}{0}

\makeatletter
\renewcommand*\l@section[2]{%
  \ifnum \c@tocdepth >\z@
    \addpenalty\@secpenalty
    \addvspace{1.0em \@plus\p@}%
    \setlength\@tempdima{2.6em}
    \begingroup
      \parindent \z@ \rightskip \@pnumwidth
      \parfillskip -\@pnumwidth
      \leavevmode \bfseries
      \advance\leftskip\@tempdima
      \hskip -\leftskip
      #1\nobreak\hfil \nobreak\hb@xt@\@pnumwidth{\hss #2}\par
    \endgroup
  \fi}
\renewcommand*\l@subsection{\@dottedtocline{2}{2.6em}{3.5em}}
\renewcommand*\l@subsubsection{\@dottedtocline{3}{6.1em}{4.4em}}
\renewcommand\@pnumwidth{2.2em}
\renewcommand\@tocrmarg{3.2em}
\makeatother
\tableofcontents
\clearpage
\section{Random Test Mode}\label{app:sec:randVsSeq}
As shown in \cref{fig:bbest:L63:randVsSeq}, there is virtually no difference between testing on a time series that continues from the training data and testing with randomly chosen initial conditions on the attractor. From a theoretical standpoint, this is expected: the system's propagator, which we estimate, is a function $\Phi_{\stepsize}$ satisfying $u(t + \stepsize) = \Phi_{\stepsize}(u(t))$, meaning it only requires the current state to predict the next one.

As seen in Table 1 of the main text, Reservoir Computers and Recurrent Neural Networks (such as LSTMs) are commonly used in the field. These methods effectively implement a function $\Psi_{\stepsize}$ of the form $[u(t+\stepsize), r(t+\stepsize)] = \Psi_{\stepsize}(u(t), r(t))$, where $r$ represents the reservoir or memory. While the initial state $u(0)$ is typically given, there is no straightforward way to set $r(0)$. In sequential test settings, $r$ is often initialized using its final value from the training phase. In random test settings starting from an arbitrary initial state, it is usually assumed that a longer segment in the immediate past of the test data is available to "warm up" the reservoir.

In contrast, the polynomial propagator method does not require such a warm-up period. It performs optimally using only a single state vector as the initial condition.

\begin{figure}[b!]
	\includegraphics[width=\textwidth]{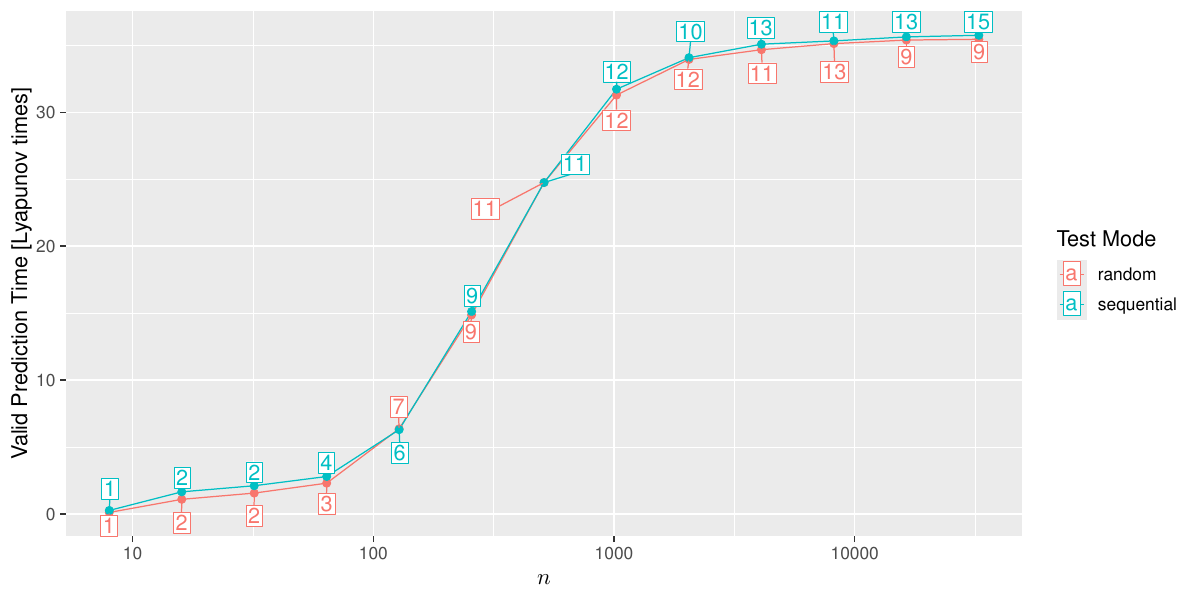}
	\caption{\textbf{Valid prediction time for L63 with best polynomial degree and best time step in different test modes.} We use the default setting of multi-precision system, double-precision data, and multi-precision method and compare random (start of test set is randomly chosen on the attractor) and sequential test mode (test set is continuation of train set). For both settings and each value of $n$, we take the $\vpt{0.5}$ (in Lyapunov times) averaged over 100 repetitions and maximize over the polynomial degree $p$ (given in label boxes) and the data time step $\stepsize \geq 2 \stepsize_0$ (not shown).}
	\label{fig:bbest:L63:randVsSeq}
\end{figure}

\clearpage
\section{Explicit Examples}\label{app:sec:example}

We show further examples of explicit training trajectories, prediction and reference trajectory (test truth) on L63. In \cref{fig:example:512:n1024,fig:example:512:n32768} we illustrate the results in the pure multi-precision (512-bit) setup. In \cref{fig:example:64:n8192}, we see the pure double-precision (64-bit) setup. In the latter case, we use a degree $p=5$ polynomial, which is a model with only 56 parameters per output dimension. To demonstrate the relative simplicity of this model, we write it out in full in \cref{tbl:example:64:n8192}.

\begin{figure}[ht!]
	\begin{center}
		\includegraphics[width=1.0\textwidth]{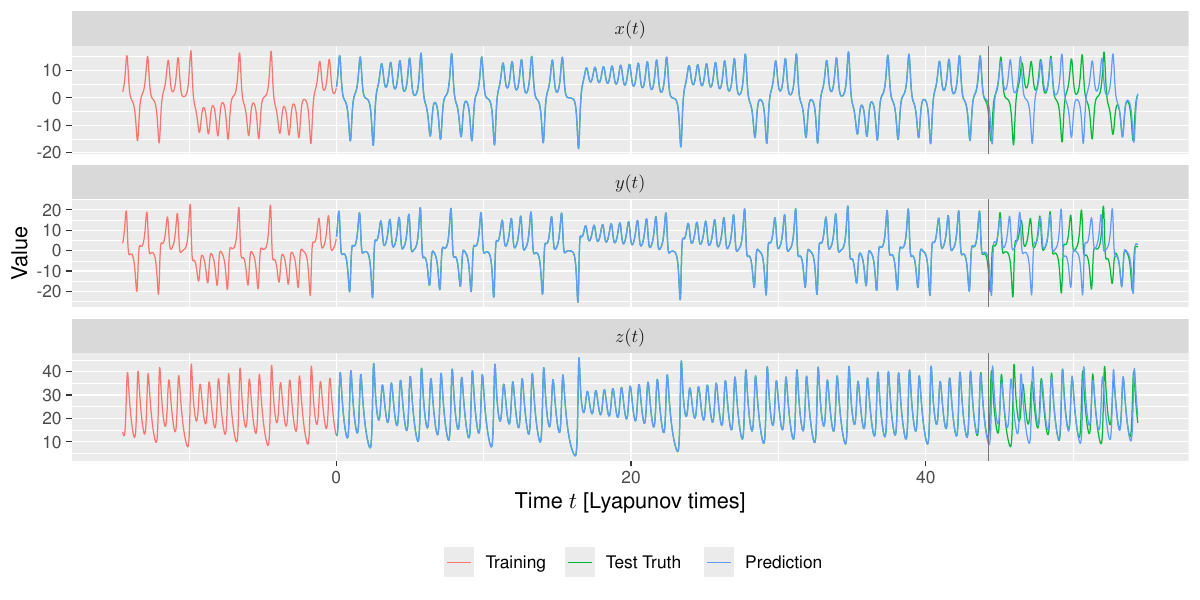}
	\end{center}
	\caption{\textbf{Example trajectory for PolyProp with degree $p = 16$ fitted via $n = 2^{10} = 1024$ observations of L63 with time step $\stepsize = 2^{-6}$.} The ODE solver, the data storage, and the implementation of PolyProp are all in multi precision (512-bit). Normalization was turned off (\textit{none}). The gray vertical line marks the valid prediction time, which is $\vpt{0.5} \approx 44$ Lyapunov times.}
	\label{fig:example:512:n1024}
\end{figure}

\begin{figure}[ht!]
	\begin{center}
		\includegraphics[width=1.0\textwidth]{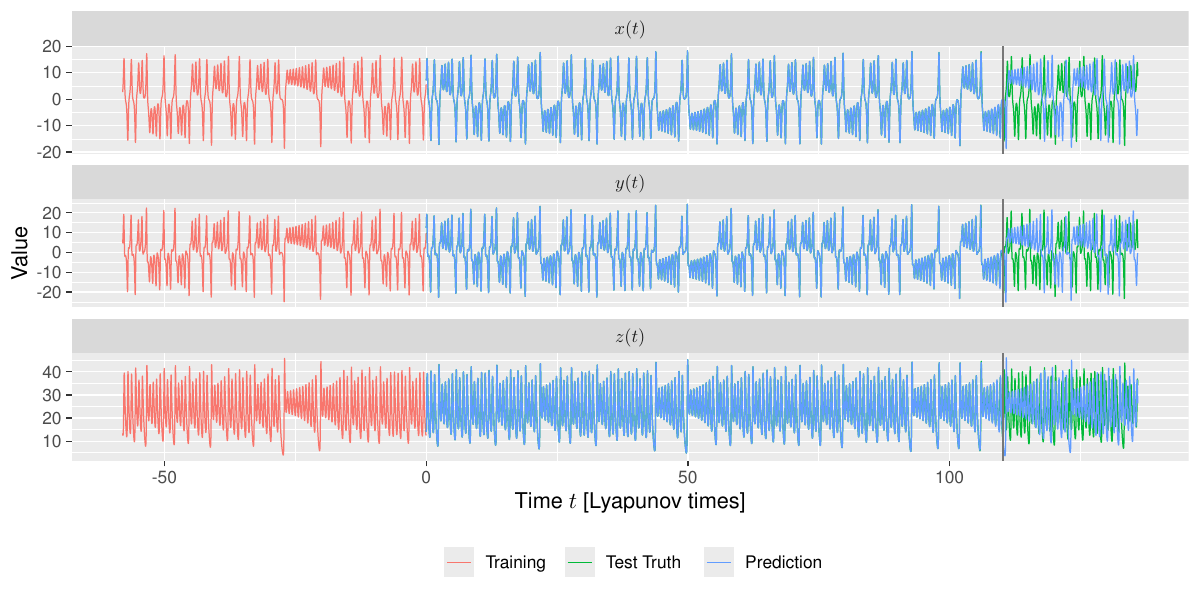}
	\end{center}
	\caption{\textbf{Example trajectory for PolyProp with degree $p = 16$ fitted via $n = 2^{15} = 32768$ observations of L63 with time step $\stepsize = 2^{-9}$.} The ODE solver, the data storage, and the implementation of PolyProp are all in multi precision (512-bit). Normalization was turned off (\textit{none}). The gray vertical line marks the valid prediction time, which is $\vpt{0.5} \approx 110$ Lyapunov times.}
	\label{fig:example:512:n32768}
\end{figure}

\begin{figure}[ht!]
    \begin{center}
        \includegraphics[width=1.0\textwidth]{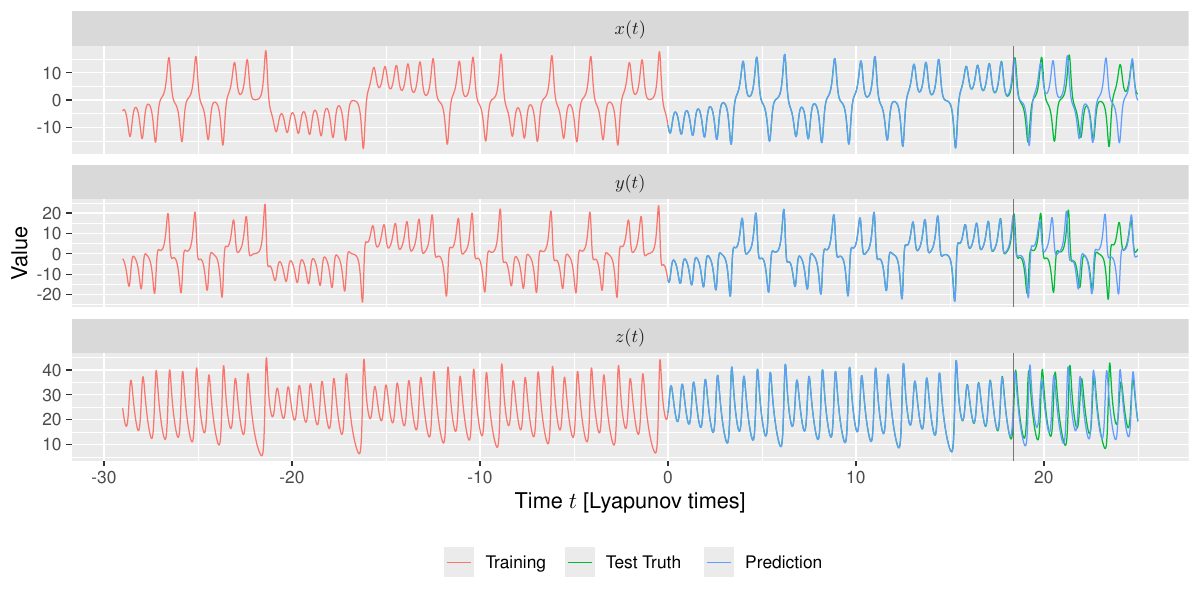}
    \end{center}
    \caption{\textbf{Example trajectory for PolyProp with degree $p = 5$ fitted via $n = 2^{13} = 8192$ observations of L63 with time step $\stepsize = 2^{-8}$.} Normalization was turned off (\textit{none}) to obtain a direct transformation. The ODE solver, the data storage, and the implementation of PolyProp are all in double precision (64-bit). The gray vertical line marks the valid prediction time, which is $\vpt{0.5} \approx 18.4$ Lyapunov times.}
    \label{fig:example:64:n8192}
\end{figure}

\begin{table}[ht!]
    \begin{center}
\fontsize{8.0pt}{10pt}\selectfont
\fontfamily{phv}\selectfont
\renewcommand{\arraystretch}{1.05}
\setlength{\tabcolsep}{0.5em}
\rowcolors{2}{gray!20}{white}
\begin{tabular}{rrcrrr}
\toprule
\multicolumn{3}{c}{Input} & \multicolumn{3}{c}{Coefficients for Output} \\ 
\cmidrule(lr){1-3} \cmidrule(lr){4-6}
\# & deg. & Monomial & $x(t + \Delta\!t)$ & $y(t + \Delta\!t)$ & $z(t + \Delta\!t)$ \\ 
\midrule\addlinespace[2.5pt]
1 & 0 & $1$ & 5.975058972711318E-08 & -3.659966520529594E-08 & 5.555270023681420E-07 \\ 
2 & 1 & $x$ & 9.637723836718638E-01 & 1.071326995765153E-01 & -2.240750866271492E-07 \\ 
3 & 1 & $y$ & 3.826031034542072E-02 & 9.982049452072156E-01 & 1.064607234374484E-07 \\ 
4 & 1 & $z$ & -2.457296755710719E-08 & 1.724623426763240E-08 & 9.896371552647957E-01 \\ 
5 & 2 & $x^{2}$ & -3.288438379714001E-08 & 1.318367624328826E-08 & 2.045868044380153E-04 \\ 
6 & 2 & $xy$ & 2.669956637948852E-08 & -1.616985610021400E-08 & 3.814471403140647E-03 \\
\bottomrule
\end{tabular}
\end{center}
    \caption{\textbf{Coefficients for PolyProp with degree $p = 5$ fitted via $n = 2^{13} = 8192$ observations of L63 with time step $\stepsize = 2^{-8}$.} Normalization was turned off (\textit{none}) to obtain a direct transformation. The ODE solver, the data storage, and the implementation of PolyProp are all in double precision (64-bit). This setting yields a valid prediction time $\vpt{0.5} \approx 18$ Lyapunov times on average.}
    \label{tbl:example:64:n8192}
\end{table}

\clearpage
\section{Effect of Data Normalization}\label{app:sec:normalize}
We compare different data normalization schemes in \cref{fig:bbest:L63:normalization}. See \cref{ssec:numerics} of the main text for a description of the different normalization schemes. Full normalization---where data is linearly transformed to obtain an identity covariance matrix---yields the best performance. Since the underlying function being approximated remains the same, normalization affects only the numerical stability of the linear system solved during the least squares fitting of the propagator. Given that numerical accuracy appears to be the primary limiting factor in this study, the benefits of improved stability through better normalization are clearly reflected in the results. However, when using multi-precision arithmetic with double-precision input data, the increased numerical accuracy inherently stabilizes the computations, making normalization unnecessary.

\begin{figure}[b!]
	\includegraphics[width=\textwidth]{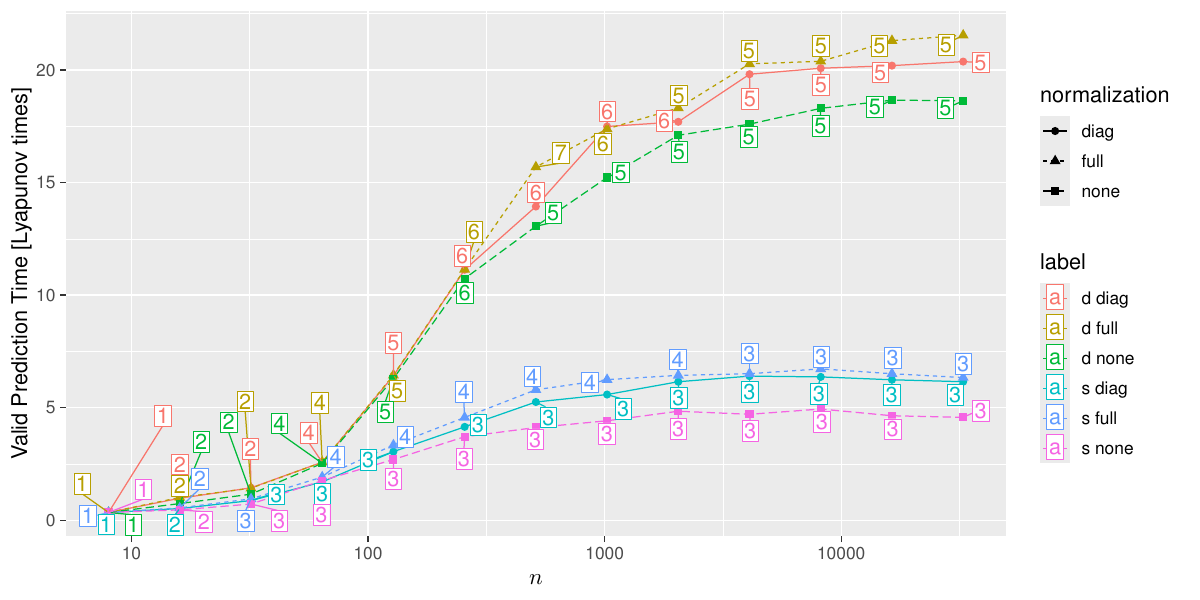}
	\caption{\textbf{Valid prediction time for L63 with best polynomial degree and best time step for single and double precision as well as different normalizations.} We use the setting of multi-precision system and double-precision data and sequential test mode. We compare the polynomial propagator method with single (\textsf{s}, 32-bit) and double (\textsf{d}, 64-bit) arithmetic and different data normalizations. For all settings and each value of $n$, we take the $\vpt{0.5}$ (in Lyapunov times) averaged over 100 repetitions and maximize over the polynomial degree $p$ (given in label boxes) and the data time step $\stepsize \geq 2 \stepsize_0$ (not shown).}\label{fig:bbest:L63:normalization}
\end{figure}

\clearpage
\section{Polynomial ODE Solution}\label{app:sec:poly}

In this section, we show that applying an RK4 ODE solver to the L63 system for $k$ steps yields a polynomial of degree $F_{4k+2}$ in the initial state, where $F_\ell$ is the $\ell$-th Fibonacci number. In particular, one RK4 step is a degree $F_{6} = 8$ polynomial and two steps result in degree $F_{10} = 55$.

Define a single step in the 4th order Runge--Kutta procedure $\mathsf{RK4}\brOf{f, u_0, h, 1}$ for an ODE of the form $\dot u = f(u)$ with initial conditions $u_0\in \R^d$ and time step $h := \stepsize\in\R$ as
\begin{equation*}
	\begin{aligned}
		w_1 &= f(u_0), \\
		w_2 &= f\brOf{u_0 + \frac{1}{2}h w_1}, \\
		w_3 &= f\brOf{u_0 + \frac{1}{2}h w_2}, \\
		w_4 &= f\brOf{u_0 + h w_3}, \\
		\mathsf{RK4}\brOf{f, u_0, h, 1} &= u_0 + \frac{1}{6} h (w_1 + 2w_2 + 2w_3 + w_4)
		\eqfs
	\end{aligned}
\end{equation*}
The $k$-th step in the RK4 procedure is defined recursively by $\mathsf{RK4}\brOf{f, u_0, h, 0} = u_0$ and
\begin{equation*}
	\mathsf{RK4}\brOf{f, u_0, h, k} = \mathsf{RK4}\brOf{f, \mathsf{RK4}\brOf{f, u_0, h, k-1}, h, 1}
	\eqfs
\end{equation*}

The L63 system is described by the ODE $\dot u = f_{\ms{L63}}(u)$, where
\begin{equation*}
	f_{\mathsf{L63}}\brOf{\begin{bmatrix}
			x\\y\\z
	\end{bmatrix}} :=
	\begin{bmatrix}
		a(y-x)\\
		x(b - z) - y\\
		xy - c z
	\end{bmatrix}
\end{equation*}
with $a,b,c\in\R$. The default parameter values are $a = 10$, $b = 28$, $c = \frac83$.

The following theorem presents explicit expressions for the leading terms of the polynomial obtained by applying one RK4 step to the L63 system.
\begin{theorem}\label{thm:rk4l63:one}
	Let
	\begin{equation*}
		\begin{bmatrix}
			\zeta_{1}\\
			\zeta_{2}\\
			\zeta_{3}
		\end{bmatrix}
		:=
		\mathsf{RK4}\brOf{f_{\mathsf{L63}}, \begin{bmatrix}
				x\\
				y\\
				z
			\end{bmatrix}, h, 1}
		\eqfs
	\end{equation*}
	Define $\alpha := 2h^{-1}a^{-1} - 1$.
	Then
	\begin{align*}
		\zeta_{1} &= -\frac{h^7a^3}{192}x^2z^2(y+\alpha x) + q_1(x,y,z)\eqcm\\
		\zeta_{2} &= \frac{h^{11}a^4}{1536}x^3y^2z(y+\alpha x)^2 + q_2(x,y,z)\eqcm\\
		\zeta_{3} &= \frac{h^{11}a^4}{1536}x^3yz^2(y+\alpha x)^2 + q_3(x,y,z)\eqcm\\
	\end{align*}
	where $q_1, q_2, q_3$ are polynomials with $\deg(q_1) \leq 4$ and $\deg(q_2), \deg(q_3) \leq 7$.
\end{theorem}
\begin{proof}
	We calculate the polynomial results of each intermediate step of the RK4 procedure. We use $P_\ell$ in an expression for an arbitrary polynomial in $x,y,z$ of degree at most $\ell$. Each occurrence of an expression $P_\ell$ may refer to a different polynomial. Set $u_0 := (x,y,z)$. We calculate:
	\begin{equation*}
		w_1 = f_{\mathsf{L63}}(u_0) = \begin{bmatrix}
			a(y-x)\\
			x(b - z) - y\\
			xy - c z
		\end{bmatrix} =
		\begin{bmatrix}
			a(y-x)\\
			-xz + P_1\\
			xy + P_1
		\end{bmatrix}
		\eqcm
	\end{equation*}
	\begin{equation*}
		w_2 = f_{\mathsf{L63}}\brOf{u_0 + \frac{1}{2}h w_1}
		=
		f_{\mathsf{L63}}\brOf{\begin{bmatrix}
				2^{-1}ha(y + \alpha x) \\
				-2^{-1}hxz + P_1\\
				2^{-1}hxy + P_1
		\end{bmatrix}}
		=
		\begin{bmatrix}
			-2^{-1}haxz + P_1\\
			-2^{-2}h^2axy(y+\alpha x) + P_2\\
			-2^{-2}h^2axz(y+\alpha x) + P_2
		\end{bmatrix}
		\eqcm
	\end{equation*}
	\begin{equation*}
		w_3 = f_{\mathsf{L63}}\brOf{u_0 + \frac{1}{2}h w_2}
		=
		f_{\mathsf{L63}}\brOf{
		\begin{bmatrix}
			-2^{-2}h^2axz + P_1\\
			-2^{-3}h^3axy(y+\alpha x) + P_2\\
			-2^{-3}h^3axz(y+\alpha x) + P_2
		\end{bmatrix}
		}
		=
		\begin{bmatrix}
			-2^{-3}h^3a^2xy(y+\alpha x) + P_2\\
			-2^{-5}h^5a^2x^2z^2(y+\alpha x) + P_4\\
			2^{-5}h^5a^2x^2yz(y+\alpha x) + P_4
		\end{bmatrix}
		\eqcm
	\end{equation*}
	\begin{equation*}
		w_4 = f_{\mathsf{L63}}\brOf{u_0 + h w_3}
		=
		f_{\mathsf{L63}}\brOf{
			\begin{bmatrix}
				-2^{-3}h^4a^2xy(y+\alpha x) + P_2\\
				-2^{-5}h^6a^2x^2z^2(y+\alpha x) + P_4\\
				2^{-5}h^6a^2x^2yz(y+\alpha x) + P_4
			\end{bmatrix}
		}
		=
		\begin{bmatrix}
			-2^{-5}h^6a^3x^2z^2(y+\alpha x) + P_4\\
			2^{-8}h^{10}a^4x^3y^2z(y+\alpha x)^2 + P_7\\
			2^{-8}h^{10}a^4x^3yz^2(y+\alpha x)^2 + P_7
		\end{bmatrix}
		\eqcm
	\end{equation*}
	\begin{equation*}
		\mathsf{RK4}\brOf{f_{\mathsf{L63}}, u_0, h, 1} = u_0 + \frac{1}{6} h (w_1 + 2w_2 + 2w_3 + w_4)
		=
		\frac{1}{6} h
		\begin{bmatrix}
			-2^{-5}h^6a^3x^2z^2(y+\alpha x) + P_4\\
			2^{-8}h^{10}a^4x^3y^2z(y+\alpha x)^2 + P_7\\
			2^{-8}h^{10}a^4x^3yz^2(y+\alpha x)^2 + P_7
		\end{bmatrix}
		\eqfs
	\end{equation*}
	With $2^5 \cdot 6 = 192$ and $2^8 \cdot 6 = 1536$, we have shown the theorem.
\end{proof}

Using \cref{thm:rk4l63:one}, we can give explicit expressions in terms of Fibonacci numbers for the degree of the polynomial obtained by applying $k$ RK4 steps to the L63 equations. 
Denote the Fibonacci sequence as $(F_k)_{k\in\N_0}$, where $F_{k+1} = F_k + F_{k-1}$ with $F_0 = 0$ and $F_1 = 1$. 

\begin{theorem}
	Let $k \in \N$.
	Let
	\begin{equation*}
		\begin{bmatrix}
			\zeta_1\\
			\zeta_2\\
			\zeta_3
		\end{bmatrix}
		:=
		\mathsf{RK4}\brOf{f_{\mathsf{L63}}, \begin{bmatrix}
				x\\
				y\\
				z
			\end{bmatrix}, h, k}
		\eqfs
	\end{equation*}
	Assume $a,h\neq 0$.
	Then
	\begin{equation*}
		\deg(\zeta_1) = F_{4k+1}
		\qquad\text{and}\qquad
		\deg(\zeta_2) = \deg(\zeta_3) = F_{4k+2}
		\eqfs
	\end{equation*}
\end{theorem}
\begin{proof}
	First note that, by \cref{thm:rk4l63:one}, one RK4 step has the form
	\begin{equation*}
		\mathsf{RK4}\brOf{f_{\mathsf{L63}}, \begin{bmatrix}
				x\\
				y\\
				z
			\end{bmatrix}, h, 1} = 
		\begin{bmatrix}
			\gamma_1 x^2 y z^2 + q_1(x,y,z)\\
			\gamma_2 x^3 y^4 z + q_2(x,y,z)\\
			\gamma_3 x^3 y^3 z^2 + q_3(x,y,z)
		\end{bmatrix}
	\end{equation*}
	with $\gamma_1,\gamma_2, \gamma_3 \in\R$ and polynomials $q_1$, $q_2$, $q_3$ that cannot cancel their respective first term as they contain different monomials. As we assume $h,a\neq 0$, we have $\gamma_1,\gamma_2, \gamma_3\neq0$.
	Now, we prove the statement of the theorem by induction over $k$:
	The induction base with $k=1$ follows directly from \cref{thm:rk4l63:one} with the arguments given above:
	\begin{align*}
		\deg(\gamma_1 x^2 y   z^2 + q_1(x,y,z))  &= 5 = F_{5}\eqcm\\
		\deg(\gamma_2 x^3 y^4 z   + q_2(x,y,z))  &= 8 = F_{6}\eqcm\\
		\deg(\gamma_3 x^3 y^3 z^2 + q_3(x,y,z))  &= 8 = F_{6}\eqfs\\
	\end{align*} 
	For the induction step, let
	\begin{equation*}
		\begin{bmatrix}
			\zeta_{k,1}\\
			\zeta_{k,2}\\
			\zeta_{k,3}
		\end{bmatrix}
		:=
		\mathsf{RK4}\brOf{f_{\mathsf{L63}}, \begin{bmatrix}
				x\\
				y\\
				z
			\end{bmatrix}, h, k}
	\end{equation*}
	and assume $\deg(\zeta_{k-1,1}) = F_{4(k-1)+1}$ and $\deg(\zeta_{k-1,2}) = \deg(\zeta_{k-1,3}) = F_{4(k-1)+2}$.
	Applying \cref{thm:rk4l63:one} to the $k$-th RK4 solver step yields
	\begin{equation*}
		\mathsf{RK4}\brOf{f_{\mathsf{L63}}, \begin{bmatrix}
				x\\
				y\\
				z
			\end{bmatrix}, h, k} = 
		\mathsf{RK4}\brOf{f_{\mathsf{L63}}, \begin{bmatrix}\zeta_{k-1,1}\\\zeta_{k-1,2}\\\zeta_{k-1,3}\end{bmatrix}, h, 1} = 
		\begin{bmatrix}
			\gamma_1 \zeta_{k-1,1}^2\zeta_{k-1,2}\zeta_{k-1,3}^2 + \tilde q_1(x,y,z)\\
			\gamma_2 \zeta_{k-1,1}^3\zeta_{k-1,2}^4\zeta_{k-1,3} + \tilde q_2(x,y,z)\\
			\gamma_3 \zeta_{k-1,1}^3\zeta_{k-1,2}^3\zeta_{k-1,3}^2 + \tilde q_3(x,y,z)
		\end{bmatrix}
		\eqfs
	\end{equation*}
	The polynomials $\tilde q_1, \tilde q_2, \tilde q_3$ are of lower degree than their respective $\zeta$-terms and, thus, cannot cancel the leading monomials in the $\zeta$-terms.
	Therefore, using the induction hypothesis,
	\begin{align*}
		\deg(\zeta_{k,1}) 
		&= 
		2\deg(\zeta_{k-1,1}) + \deg(\zeta_{k-1,2}) + 2 \deg(\zeta_{k-1,3})
		\\&=
		2  F_{4(k-1)+1} + 3  F_{4(k-1)+2}
		\\&=
		2  F_{4(k-1)+3} +  F_{4(k-1)+2}
		\\&=
		F_{4(k-1)+3} +  F_{4(k-1)+4}
		\\&=
		F_{4k+1}
	\end{align*}
	and
	\begin{align*}
		\deg(\zeta_{k,2}) 
		&= 
		3\deg(\zeta_{k-1,1}) + 4\deg(\zeta_{k-1,2}) +  \deg(\zeta_{k-1,3})
		\\&=
		3  F_{4(k-1)+1} + 5 F_{4(k-1)+2}
		\\&=
		3  F_{4(k-1)+3} + 2 F_{4(k-1)+2}
		\\&=
		F_{4(k-1)+3} + 2 F_{4(k-1)+4}
		\\&=
		F_{4(k-1)+4} + F_{4(k-1)+5}
		\\&=
		F_{4k+2}
		\eqfs
	\end{align*}
	The calculation for $\deg(\zeta_{k,3}) = F_{4k+2}$ is almost the same as for $\deg(\zeta_{k,2})$.
\end{proof}
\clearpage
\section{Extreme Precision ODE Solvers}\label{app:sec:xprec}

In this section, we compare the forecasting accuracy of PolyProp against standard numerical ODE solvers on the Lorenz-63 (L63) system. Since L63 lacks an analytical solution, we must rely on high-precision numerical approximations to serve as the ground truth. The following ODE solver configurations are considered:
\begin{itemize}
	\item \textbf{RK4-m:} RK4 with a time step of $\stepsize_0 = 2^{-10}$ using 512-bit multiprecision arithmetic (same as in the main text).
	\item \textbf{RK4-y:} RK4 using 512-bit arithmetic with a finer time step of $\stepsize_{\mathsf{y}} = 2^{-24}$.
	\item \textbf{TI-x:} A 28th-order Taylor Integration solver using 256-bit arithmetic with an absolute tolerance of $10^{-60}$. This is implemented in the Julia programming language using the \texttt{TaylorIntegration.jl} package (\url{https://perezhz.github.io/TaylorIntegration.jl}).
\end{itemize}

We treat the numerical solutions of L63 from these solvers as ``ground truth'' trajectories. For the comparison, we initialize the target models (PolyProp or a numerical ODE solver) with initial conditions rounded to 64-bit precision. Consequently, the divergence between the ground truth and the forecast stems from two sources:
\begin{enumerate}
	\item \textbf{Initial Condition Error:} The deviation caused by rounding the initial state to 64-bit.
	\item \textbf{Dynamics Error:} The deviation caused by discrepancies in the state update algorithm.
\end{enumerate}

For PolyProp executed at 512-bit precision (denoted PP-m), we use polynomial degree $p=15$, $n=2^{15}$ observations, and observation time step $\stepsize = 2^{-5}$. We employ the Valid Prediction Time (VPT) with a threshold of $0.5$, i.e., $\vpt{0.5}$, as our metric.
The results of the comparison are detailed in \cref{tbl:xprec}.

\begin{table}[b!]
	\centering
	\begin{center}
\fontsize{8.0pt}{10pt}\selectfont
\fontfamily{phv}\selectfont
\renewcommand{\arraystretch}{1.05}
\setlength{\tabcolsep}{0.5em}
\rowcolors{2}{gray!20}{white}
\begin{tabular}{lllrlc}
\toprule
System & Data & Method & VPT & 95\% CI & Reps. \\ 
\midrule\addlinespace[2.5pt]
\multicolumn{6}{l}{Extreme Precision} \\[2.5pt] 
\midrule\addlinespace[2.5pt]
TI-x & 64 bit & TI-x & 35.78 & [35.66, 35.91] & $10^{3}$ \\ 
TI-x & 64 bit & RK4-y & 35.80 & [35.34, 36.26] & $10^{2}$ \\ 
RK4-y & 64 bit & RK4-y & 35.75 & [35.30, 36.20] & $10^{2}$ \\ 
\midrule\addlinespace[2.5pt]
\multicolumn{6}{l}{High Precision} \\[2.5pt] 
\midrule\addlinespace[2.5pt]
RK4-m & 64 bit & RK4-m & 34.66 & [34.63, 34.70] & $10^{4}$ \\ 
TI-x & 64 bit & RK4-m & 21.56 & [21.53, 21.59] & $10^{4}$ \\ 
RK4-y & 64 bit & RK4-m & 21.60 & [21.56, 21.63] & $10^{4}$ \\ 
\midrule\addlinespace[2.5pt]
\multicolumn{6}{l}{PolyProp} \\[2.5pt] 
\midrule\addlinespace[2.5pt]
RK4-m & 64 bit & PolyProp & 35.75 & [35.38, 36.12] & $10^{2}$ \\ 
TI-x & 64 bit & PolyProp & 35.63 & [35.22, 36.04] & $10^{2}$ \\ 
RK4-y & 64 bit & PolyProp & 35.82 & [35.48, 36.16] & $10^{2}$ \\ 
\bottomrule
\end{tabular}
\end{center}
	\caption{\textbf{Comparison of Valid Prediction Times (VPT) for ODE Solvers and PolyProp on L63.} Values represent the mean VPT (specifically $\vpt{0.5}$) given in units of Lyapunov time, averaged over the specified number of repetitions, with 95\% confidence intervals.}\label{tbl:xprec}
\end{table}

By comparing the high-precision solvers TI-x and RK4-y each against themselves (with rounded initial conditions), we observe a mean VPT of $\approx 35.8$ Lyapunov times. 
As the dynamics in these comparisons are identical, only Initial Condition Error is present. 
Comparing TI-x and RK4-y yields a similar VPT of $\approx 35.8$ Lyapunov times. Thus, we can assume that the RK4-y dynamics are similar enough to the TI-x dynamics that again the Initial Condition Error is dominating. 

When comparing the highest precision solutions (TI-x and RK4-y) to RK4-m, the VPT drops substantially to $\approx 21.6$ Lyapunov times, which is also much less than the $\approx 34.6$ Lyapunov times VPT of RK4-m against itself (with rounded initial conditions). This indicates that the Dynamics Error dominates in the comparison of RK4-m to TI-x and RK4-y.

In all cases PolyProp, when trained on data from the ground truth system rounded to 64-bit, achieves at least the Initial Condition Error limit (up to some statistical uncertainty). This demonstrates the high precision and adaptability of PolyProp.
\clearpage
\section{Noisy Measurements}\label{app:sec:noise}

This study primarily focuses on noise-free data, where we have shown that PolyProp trained on data with limited precision achieves the same accuracy as a numerical ODE solver initialized with values of the same precision. Interpreting finite precision as a form of observational noise, we extend this finding to stochastic noise by replacing the rounding operation with additive Gaussian noise. Specifically, we utilize training data $y_i = u(t_i) + \xi_i$ on L63, where $\xi_i$ are independent and identically distributed centered Gaussian random vectors with covariance matrix $\sigma_{\ms{noise}}^2 I_3$ with the $3\times 3$ identity matrix $I_3$. The results, presented in \cref{fig:bbest:L63:noise}, show that while PolyProp's performance deteriorates with increasing noise levels, it remains equivalent---given sufficient data---to that of an ODE solver started from the same noisy initial conditions. This means that the Initial Condition Error that the numerical solver shares with PolyProp dominates PolyProp's Dynamics Error so that both methods yield an indistinguishable total forecast error, see also Supplementary \cref{app:sec:xprec}.

\begin{figure}[b!]
	\includegraphics[width=\textwidth]{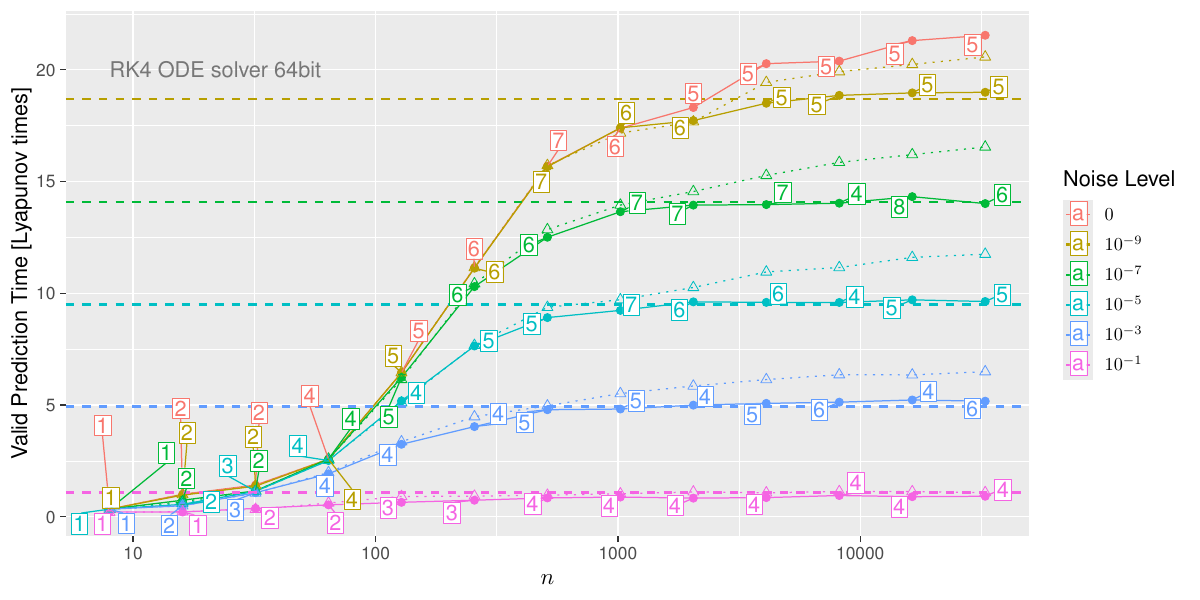}
	\caption{\textbf{Valid prediction time for L63 with best polynomial degree and best time step for different measurement noise levels.} We generate data from the L63 system using an RK4 ODE solver with time step $\stepsize_0 = 2^{-10}$. We use the setting of double-precision system, double-precision data, and double-precision method with full normalization and polynomial degree $p \in\{1,\dots, 8\}$. We add centered Gaussian noise with different variances $\sigma_{\ms{noise}}^2$ to the training data. The \textit{noise level} indicated by color in the plot is the ratio of the standard deviation of the noise and the system,  $\sigma_{\ms{noise}} / \sigma_{\ms{L63}}$, where the standard deviation of the system is calculated as described in \cref{ssec:problem} of the main text yielding $\sigma_{\ms{L63}} \approx 14.78$. The test mode is sequential with the forecast starting either from noisy initial conditions with the same noise level as for training (large dots with solid thin lines for PolyProp; dashed horizontal lines for the ODE solver) or from noise-free initial conditions (dotted lines with triangle marks---only for PolyProp). 
	For each noise level and each value of $n$, we take the $\vpt{0.5}$ (in Lyapunov times) averaged over 100 repetitions and maximize over the polynomial degree $p$ of PolyProp (given in label boxes) and the data time step $\stepsize \in \{2^{-9}, 2^{-8}, \dots, 2^{-3}\}$ (not shown). For comparison, the horizontal dashed lines mark the VPT value, obtained as the mean of $10^4$ repetitions, of the RK4 ODE solver used for generating the ground truth but started at the noised initial conditions.}
	\label{fig:bbest:L63:noise}
\end{figure}

\clearpage
\section{L96 Summary Plots}\label{app:sec:L96:best}

\begin{figure}[b!]
	\includegraphics[width=\textwidth]{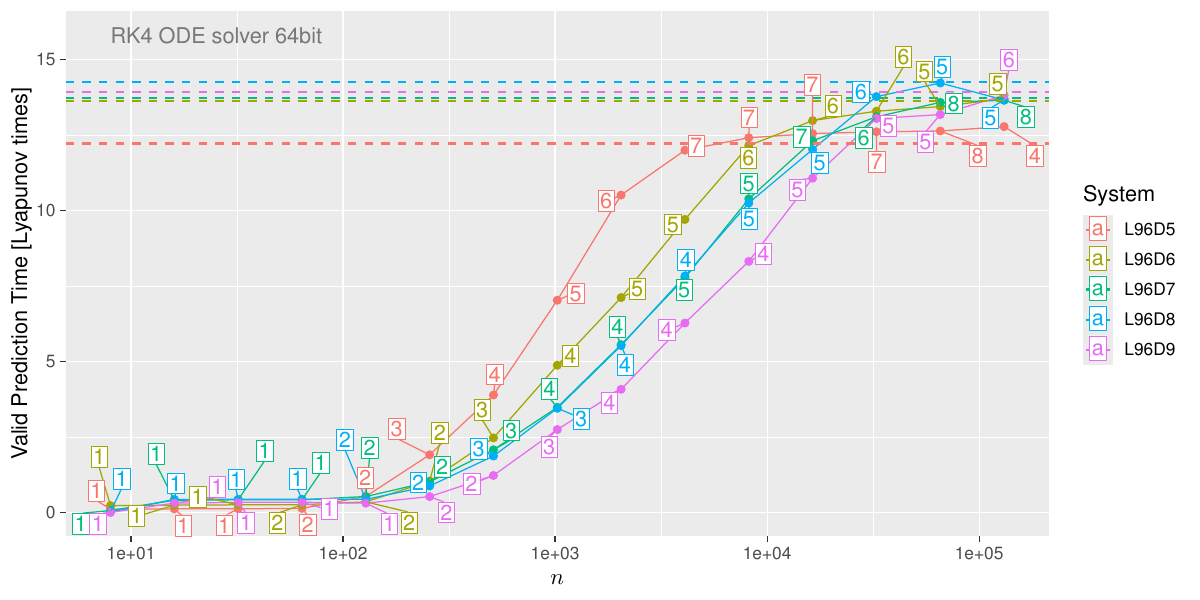}
	\caption{\textbf{Valid prediction time for the L96 model with single precision solver optimized over polynomial degree and time step.} We use a single-precision RK4 ODE solver (32-bit), double-precision data storage (64-bit), and a double-precision polynomial propagator. The system dimension of L96 is indicated by color. Label boxes within the plot denote the degree of the optimal polynomial for each case. The results are also optimized over the time step $\stepsize\geq 2\stepsize_0$ and only the best are depicted. Colored dashed horizontal lines show the performance of ODE solvers with access to the system's governing equations for the respective dimension in the same precision setting. The $\vpt{0.5}$ values on the vertical axis are given in units of Lyapunov time.}
	\label{fig:L96:best:sdd}
\end{figure}

\begin{figure}[b!]
	\includegraphics[width=\textwidth]{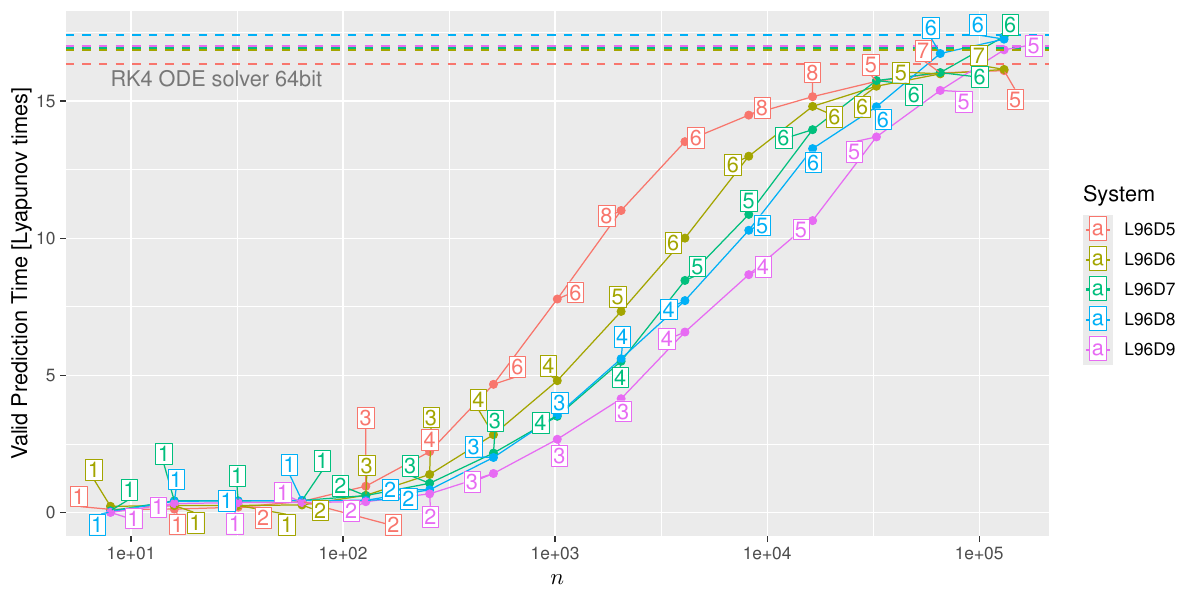}
	\caption{\textbf{Valid prediction time for the L96 model with double precision solver rounded to single precision optimized over polynomial degree and time step.}  We use a double-precision RK4 ODE solver (64-bit), single-precision data storage (32-bit), and a double-precision polynomial propagator. Otherwise, the plot is created with the same setup as in \cref{fig:L96:best:sdd}.}
	\label{fig:L96:best:dsd}
\end{figure}
\clearpage
\section{Results on Further Systems}\label{app:sec:additional}

We perform the experiments listed in \cref{tbl:add:experis} on the additional systems in \cref{tbl:add:systems} with Lyapunov times given in \cref{tbl:add:lyapunov}. A state space view of one trajectory for each system is shown in \cref{fig:add:attr}. The results, see \cref{fig:add:n15}, show that we can achieve machine precision on these systems, supporting the claim that the results in the main text focusing on L63, L96, and TCSA generalize.

\begin{table}[b!]
    \begin{center}
        \fontsize{8.0pt}{10pt}\selectfont
        \fontfamily{phv}\selectfont
        \renewcommand{\arraystretch}{1.05}
        \setlength{\tabcolsep}{0.5em}
        \rowcolors{3}{gray!20}{white}
        \begin{tabular}{llcl}
            \hline
            Shorthand & System & $d$ & Reference \\
            \hline
            AIZ   & Aizawa                             & $3$   & \cite{Sprott2003} \\
            CHEN  & Chen                               & $3$   & \cite{Chen1999} \\
            HALV  & Halvorsen                          & $3$   & \cite{Sprott2003} \\
            HCHEN & Hyperchaotic Chen                  & $4$   & \cite{Li2005} \\
            HROSS & Hyperchaotic Rössler               & $4$   & \cite{Rossler1979} \\
            LU    & Lü                                 & $3$   & \cite{Lu2002} \\
            ROSS  & Rössler                            & $3$   & \cite{Rossler1976} \\
            \hline
        \end{tabular}
    \end{center}
    \caption{\textbf{Additional Dynamical Systems.}}
    \label{tbl:add:systems}
\end{table}

\begin{table}[b!]
    \begin{center}
\fontsize{8.0pt}{10pt}\selectfont
\fontfamily{phv}\selectfont
\renewcommand{\arraystretch}{1.05}
\setlength{\tabcolsep}{0.5em}
\rowcolors{2}{gray!20}{white}
\begin{tabular}{lrc}
\toprule
System & $\lambda_{\mathsf{max}}$ & 95\%-Conf.Int. \\ 
\midrule\addlinespace[2.5pt]
AIZ & 0.096029 & [0.095917, 0.096142] \\ 
CHEN & 2.0362 & [2.0352, 2.0371] \\ 
HALV & 0.67358 & [0.67322, 0.67395] \\ 
HCHEN & 0.50265 & [0.50188, 0.50341] \\ 
HROSS & 0.11274 & [0.11266, 0.11283] \\ 
LU & 1.3531 & [1.352, 1.3541] \\ 
ROSS & 0.070700 & [0.070668, 0.070732] \\ 
\bottomrule
\end{tabular}
\end{center}
    \caption{\textbf{Estimates for the largest Lyapunov exponent for additional systems.} See \cref{app:sec:lyapunov} for more details.}
    \label{tbl:add:lyapunov}
\end{table}

\begin{table}[b!]
    \begin{center}
        \fontsize{8.0pt}{10pt}\selectfont
        \fontfamily{phv}\selectfont
        \renewcommand{\arraystretch}{1.05}
        \setlength{\tabcolsep}{0.5em}
        \rowcolors{3}{gray!20}{white}
        \begin{tabular}{lrrrccccccc}
            \toprule
            System & \multicolumn{3}{c}{Precision} & $\stepsize_0$ & $\stepsize$ & $n$ & normalization & $p$ & test mode & Reps.\\
            \cmidrule(lr){2-4}
            & system & data & method & & & & & & & \\
            \midrule
            AIZ & 32 bit & 64 bit & 64 bit & $2^{-9}$ & $2^3 \cdot \stepsize_0$ & $2^{15}$ & full & $1, \dots, 8$ & sequential & 100\\
            CHEN & 32 bit & 64 bit & 64 bit & $2^{-12}$ & $2^3 \cdot \stepsize_0$ & $2^{15}$ & full & $1, \dots, 8$ & sequential & 100\\
            HALV & 32 bit & 64 bit & 64 bit & $2^{-10}$ & $2^3 \cdot \stepsize_0$ & $2^{15}$ & full & $1, \dots, 8$ & sequential & 100\\
            HCHEN & 32 bit & 64 bit & 64 bit & $2^{-12}$ & $2^3 \cdot \stepsize_0$ & $2^{15}$ & full & $1, \dots, 8$ & sequential & 100\\
            HROSS & 32 bit & 64 bit & 64 bit & $2^{-8}$ & $2^3 \cdot \stepsize_0$ & $2^{15}$ & full & $1, \dots, 8$ & sequential & 100\\
            LU & 32 bit & 64 bit & 64 bit & $2^{-12}$ & $2^3 \cdot \stepsize_0$ & $2^{15}$ & full & $1, \dots, 8$ & sequential & 100\\
            ROSS & 32 bit & 64 bit & 64 bit & $2^{-7}$ & $2^3 \cdot \stepsize_0$ & $2^{15}$ & full & $1, \dots, 8$ & sequential & 100\\
            AIZ & 64 bit & 32 bit & 64 bit & $2^{-9}$ & $2^3 \cdot \stepsize_0$ & $2^{15}$ & full & $1, \dots, 8$ & sequential & 100\\
            CHEN & 64 bit & 32 bit & 64 bit & $2^{-12}$ & $2^3 \cdot \stepsize_0$ & $2^{15}$ & full & $1, \dots, 8$ & sequential & 100\\
            HALV & 64 bit & 32 bit & 64 bit & $2^{-10}$ & $2^3 \cdot \stepsize_0$ & $2^{15}$ & full & $1, \dots, 8$ & sequential & 100\\
            HCHEN & 64 bit & 32 bit & 64 bit & $2^{-12}$ & $2^3 \cdot \stepsize_0$ & $2^{15}$ & full & $1, \dots, 8$ & sequential & 100\\
            HROSS & 64 bit & 32 bit & 64 bit & $2^{-8}$ & $2^3 \cdot \stepsize_0$ & $2^{15}$ & full & $1, \dots, 8$ & sequential & 100\\
            LU & 64 bit & 32 bit & 64 bit & $2^{-12}$ & $2^3 \cdot \stepsize_0$ & $2^{15}$ & full & $1, \dots, 8$ & sequential & 100\\
            ROSS & 64 bit & 32 bit & 64 bit & $2^{-7}$ & $2^3 \cdot \stepsize_0$ & $2^{15}$ & full & $1, \dots, 8$ & sequential & 100\\
            \bottomrule
        \end{tabular}
    \end{center}
    \caption{\textbf{Settings of additional noise-free experiments.} See \cref{tbl:experiments} for a description.}
    \label{tbl:add:experis}
\end{table}

\begin{figure}%
    \centering%
    \begin{minipage}{0.45\textwidth}
        \includegraphics[width=0.95\textwidth]{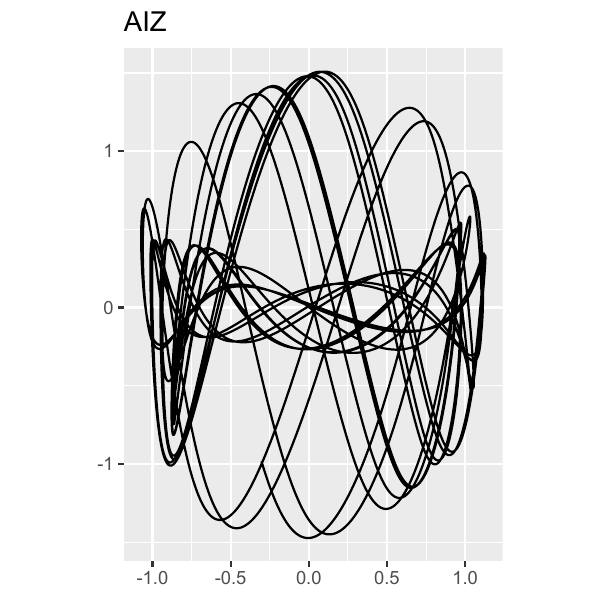}\\%
        \includegraphics[width=0.95\textwidth]{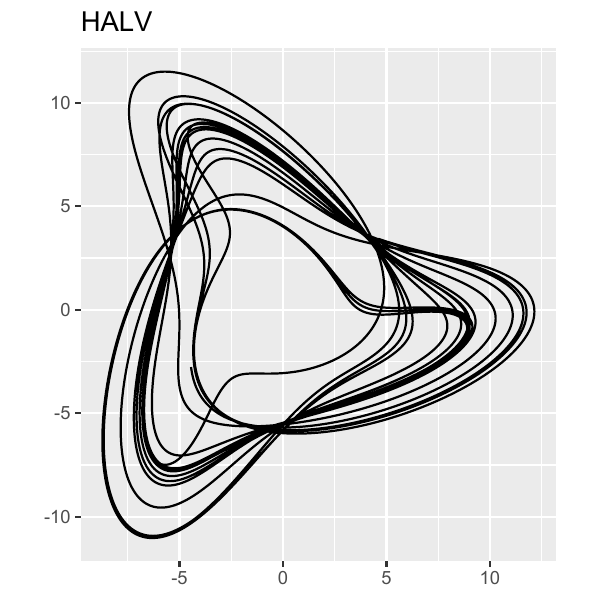}\\%
        \includegraphics[width=0.95\textwidth]{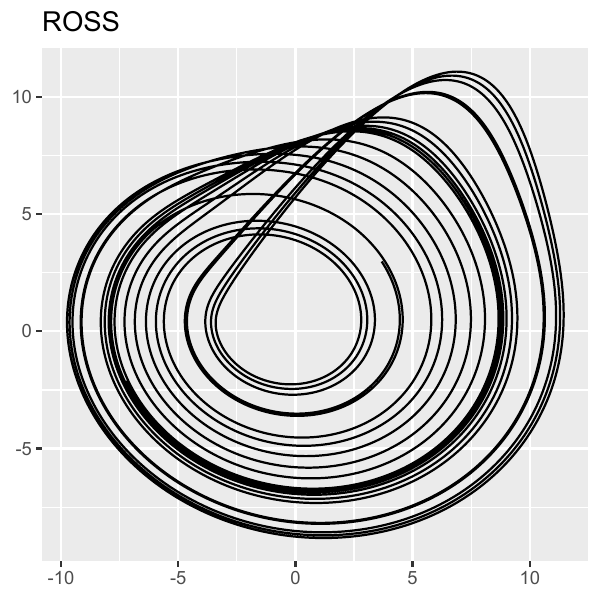}%
    \end{minipage}%
    \begin{minipage}{0.45\textwidth}
        \vspace*{-0.5cm}
        
        \includegraphics[width=0.95\textwidth]{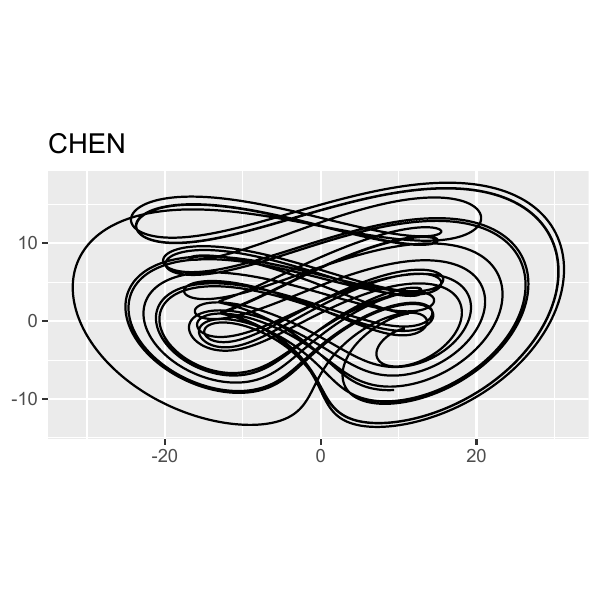}\\%
        
        \vspace*{-2.5cm}
        
        \includegraphics[width=0.95\textwidth]{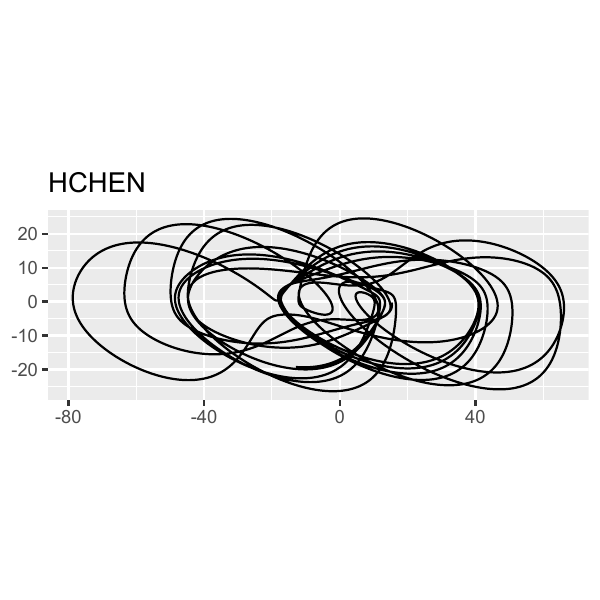}\\%
        
        \vspace*{-2.5cm}
        
        \includegraphics[width=0.95\textwidth]{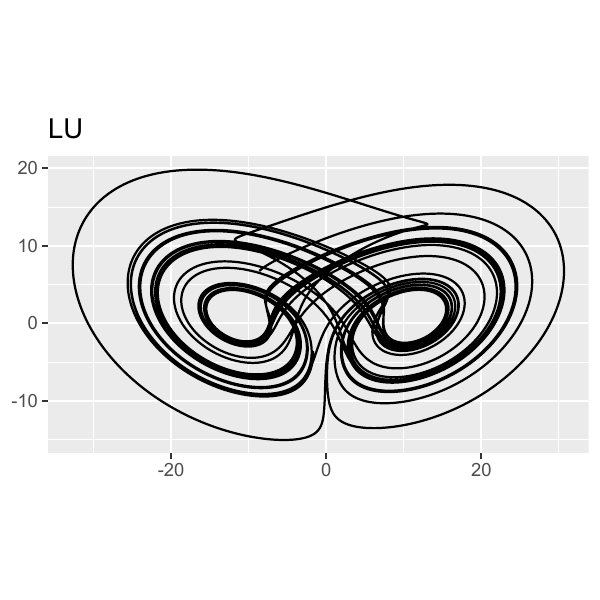}\\%
        
        \vspace*{-1.5cm}
        
        \includegraphics[width=0.95\textwidth]{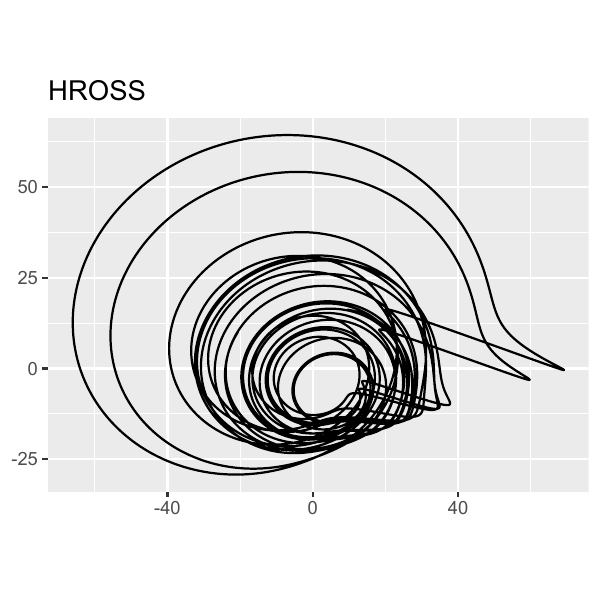}%
    \end{minipage}%
    \caption{\textbf{State space view of the systems in the additional experiments.}}
    \label{fig:add:attr}
\end{figure}

\begin{figure}
    \centering
    \includegraphics[width=0.95\textwidth]{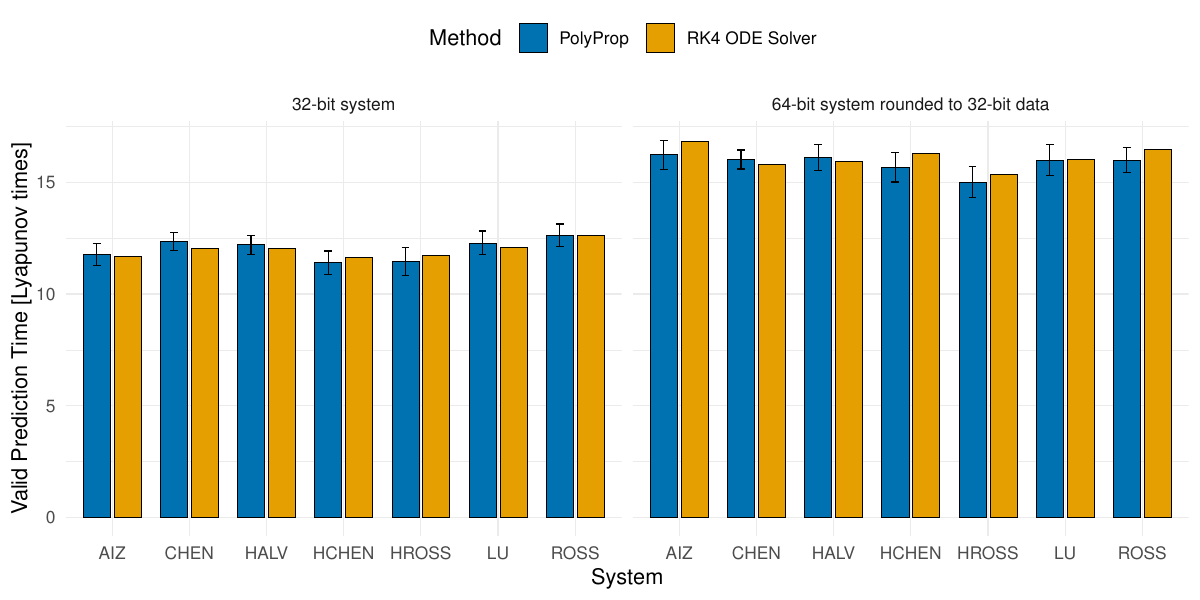}
    \caption{\textbf{Valid prediction time (VPT) for further chaotic dynamical systems using PolyProp and RK4 ODE solvers.} For the two precision settings and different chaotic dynamical systems, the VPT (averaged over 100 repetitions) of PolyProp trained on $n=2^{15}$ observations and time step $2^3 \cdot \stepsize_0$ is reported. The mean and 95\% confidence interval for the optimal polynomial degree $p \in \{1, \dots, 8\}$ are shown. These are compared to the mean VPT from $10^4$ trajectories computed with a 64-bit RK4 solver initialized from initial conditions of the respective data precision, randomly sampled from the system's attractor. The RK4 mean VPT values have 95\% confidence intervals all smaller than 0.3 Lyapunov times and are not shown. All VPT values are expressed in units of the system's Lyapunov time.}
    \label{fig:add:n15}
\end{figure}

\clearpage
\section{Lyapunov Exponent}\label{app:sec:lyapunov}

We convert system time to Lyapunov time by multiplying the time values by the estimated largest Lyapunov exponent of the respective system. To estimate the largest Lyapunov exponent for TCSA with $b=0.208$, for L63 with default parameters, and for L96 with constant forcing $F=8$ in dimensions $d=5,6,7,8,9$, we follow the method of \citet{Benettin1980a, Benettin1980b}.

Estimates are performed in double precision, using the same ODE solver time steps as those used to generate the train and test data: $\stepsize_0 = 2^{-10}$ for L63 and L96, and $\stepsize_0 = 2^{-6}$ for TCSA. To measure error growth, we maintain a perturbed state trajectory at a fixed separation of $10^{-8}$ from the reference trajectory and record the error evolution by an RK4 solver. Each estimate consists of $10^5$ integration steps and is additionally averaged over $10^5$ repetitions starting from randomly chosen states on the system's attractor.

Using this method, we obtain a largest Lyapunov exponent of $0.90642$ for L63, with a 95\% confidence interval of $[0.90628, 0.90655]$, which is close to values reported in the literature (e.g., $0.905630$ in \citet{Viswanath04}). For the estimates in other systems, see \cref{tbl:lyapunov}.

\begin{table}[b!]
	\begin{center}
\fontsize{8.0pt}{10pt}\selectfont
\fontfamily{phv}\selectfont
\renewcommand{\arraystretch}{1.05}
\setlength{\tabcolsep}{0.5em}
\rowcolors{2}{gray!20}{white}
\begin{tabular}{lrc}
\toprule
System & $\lambda_{\mathsf{max}}$ & 95\%-Conf.Int. \\ 
\midrule\addlinespace[2.5pt]
L63 & 0.90642 & [0.90628, 0.90655] \\ 
TCSA & 0.015273 & [0.015256, 0.01529] \\ 
L96D5 & 0.46705 & [0.46651, 0.46758] \\ 
L96D6 & 0.94633 & [0.94579, 0.94687] \\ 
L96D7 & 1.2652 & [1.2646, 1.2658] \\ 
L96D8 & 1.5918 & [1.5912, 1.5924] \\ 
L96D9 & 1.2164 & [1.2158, 1.217] \\ 
\bottomrule
\end{tabular}
\end{center}
	\caption{\textbf{Estimates for the largest Lyapunov exponent.} The table shows the mean estimate and the respective 95\% confidence interval for the largest Lyapunov exponent $\lambda_{\mathsf{max}}$.}
	\label{tbl:lyapunov}
\end{table}

\clearpage
\section{Valid Prediction Time Reference}\label{app:sec:vptref}
For reference, we estimate the accuracy, in terms of $\vpt{0.5}$, of the RK4 solver starting from differently rounded initial conditions and using calculations of different precision internally. Averaging over $10^4$ repetitions with randomly chosen initial conditions on the system's attractor, we arrive at the values given in \cref{tbl:vpt:solver}.

The reference values in \cref{tbl:vpt:solver} can also be approximated analytically. Assuming an initial error $\epsilon_0$ that grows by a factor of Euler's number per Lyapunov time, the valid prediction time for a threshold error $\varepsilon$ can be estimated as $\vpt{\varepsilon} = \log(\varepsilon \sigma / \epsilon_0)$, where $\sigma$ is the standard deviation of the system. In the L63 system, state values range up to about $48$ in absolute value. The spacing between adjacent double-precision floating-point numbers in the range $2^5 = 32$ to $2^6 = 64$ is $2^{-47}$. Using $\varepsilon = 0.5$, $\sigma \approx 14.8$ (an empirical estimate of the L63 standard deviation), and $\epsilon_0 = 2^{-48}$ (the average rounding error in the $[2^5, 2^6]$ range for double precision), we compute $\vpt{0.5} \approx \log(0.5 \cdot 14.8 \cdot 2^{48}) \approx 35.3$. This closely matches the empirical value of $34.7$ obtained using multi-precision RK4 solvers with initial conditions rounded to double precision (\cref{tbl:vpt:solver}). The same approach for rounding to single precision yields $\vpt{0.5} \approx \log(0.5 \cdot 14.8 \cdot 2^{19}) \approx 15.2$, which is close to the empirical value $15.6$. In the TCSA system, state values range up to about $4$ in absolute value. The spacing between adjacent double-precision floating-point numbers in the range $2^1 = 2$ to $2^2 = 4$ is $2^{-51}$. Using $\varepsilon = 0.5$, $\sigma \approx 2.1$ (an empirical estimate of the TCSA standard deviation), and $\epsilon_0 = 2^{-52}$ (the average rounding error in the $[2^1, 2^2]$ range for double precision), we compute $\vpt{0.5} \approx \log(0.5 \cdot 2.1 \cdot 2^{52}) \approx 36.1$. The empirically determined value is $37.4$.

\begin{table}
	\begin{center}
		\begin{center}
\fontsize{8.0pt}{10pt}\selectfont
\fontfamily{phv}\selectfont
\renewcommand{\arraystretch}{1.05}
\setlength{\tabcolsep}{0.5em}
\rowcolors{2}{gray!20}{white}
\begin{tabular}{rrrcr}
\toprule
\multicolumn{3}{c}{Precision} &  &  \\ 
\cmidrule(lr){1-3}
system & data & method & $\stepsize_0$ & VPT \\ 
\midrule\addlinespace[2.5pt]
\multicolumn{5}{l}{L63} \\[2.5pt] 
\midrule\addlinespace[2.5pt]
32 bit & 64 bit & 512 bit & $2^{-10}$ & 12.0 \\ 
64 bit & 64 bit & 32 bit & $2^{-10}$ & 12.0 \\ 
32 bit & 64 bit & 64 bit & $2^{-10}$ & 12.0 \\ 
512 bit & 64 bit & 32 bit & $2^{-10}$ & 12.0 \\ 
64 bit & 32 bit & 64 bit & $2^{-10}$ & 15.6 \\ 
512 bit & 32 bit & 64 bit & $2^{-10}$ & 15.7 \\ 
64 bit & 64 bit & 512 bit & $2^{-10}$ & 32.0 \\ 
512 bit & 64 bit & 64 bit & $2^{-10}$ & 32.1 \\ 
512 bit & 64 bit & 512 bit & $2^{-10}$ & 34.7 \\ 
\midrule\addlinespace[2.5pt]
\multicolumn{5}{l}{L96D5} \\[2.5pt] 
\midrule\addlinespace[2.5pt]
32 bit & 64 bit & 64 bit & $2^{-10}$ & 12.2 \\ 
64 bit & 32 bit & 64 bit & $2^{-10}$ & 16.3 \\ 
\midrule\addlinespace[2.5pt]
\multicolumn{5}{l}{L96D6} \\[2.5pt] 
\midrule\addlinespace[2.5pt]
32 bit & 64 bit & 64 bit & $2^{-10}$ & 13.6 \\ 
64 bit & 32 bit & 64 bit & $2^{-10}$ & 16.9 \\ 
\midrule\addlinespace[2.5pt]
\multicolumn{5}{l}{L96D7} \\[2.5pt] 
\midrule\addlinespace[2.5pt]
32 bit & 64 bit & 64 bit & $2^{-10}$ & 13.7 \\ 
64 bit & 32 bit & 64 bit & $2^{-10}$ & 16.9 \\ 
\midrule\addlinespace[2.5pt]
\multicolumn{5}{l}{L96D8} \\[2.5pt] 
\midrule\addlinespace[2.5pt]
32 bit & 64 bit & 64 bit & $2^{-10}$ & 14.3 \\ 
64 bit & 32 bit & 64 bit & $2^{-10}$ & 17.4 \\ 
\midrule\addlinespace[2.5pt]
\multicolumn{5}{l}{L96D9} \\[2.5pt] 
\midrule\addlinespace[2.5pt]
32 bit & 64 bit & 64 bit & $2^{-10}$ & 13.9 \\ 
64 bit & 32 bit & 64 bit & $2^{-10}$ & 17.0 \\ 
\midrule\addlinespace[2.5pt]
\multicolumn{5}{l}{TCSA} \\[2.5pt] 
\midrule\addlinespace[2.5pt]
512 bit & 64 bit & 32 bit & $2^{-6}$ & 9.2 \\ 
512 bit & 64 bit & 64 bit & $2^{-6}$ & 32.2 \\ 
512 bit & 64 bit & 512 bit & $2^{-6}$ & 37.4 \\ 
\bottomrule
\end{tabular}
\end{center}
	\end{center}
	\caption{\textbf{Valid Prediction Times achieved by RK4 ODE Solvers.} We generate ground truth data using the same procedure as in \cref{tbl:experiments}. Starting from initial conditions specified at the precision level indicated in the \textit{Precision-data} column, we use an RK4 ODE solver with time step $\stepsize_0$---identical to that used for the ground truth generation---to produce a forecast. This RK4 solver uses the precision level indicated in the column \textit{Precision-method}. The forecast is evaluated using the $\vpt{0.5}$ metric.}\label{tbl:vpt:solver}
\end{table}
\clearpage
\section{Computational Resources}\label{app:sec:compute}
All simulations were conducted on AMD EPYC 9354 CPUs. The average total computation time---including both training (polynomial feature generation and least squares fitting) and inference (forecasting time series)---is shown in \cref{fig:compute:s,fig:compute:d,fig:compute:m,fig:compute:L96}. For the L63 and L96 systems, forecasts span 50 system time units, except in the case of multi-precision data storage, where 500 time units are used. For TCSA, forecasts cover 5000 system time units.

Almost all experiments run on 6\,GByte of DDR5 memory. A few configurations are allocated more memory; these allocations are upper bounds rather than minimal requirements. Specifically, L63 with 512-bit input data and a 512-bit PolyProp forecaster is allocated 10\,GByte. L96 in dimensions $8$ and above, as well as dimension $7$ with $2^{16}$ or more observations is allocated more memory---in the extreme case L96D9 with $2^{17}$ observations is allocated 50\,GByte. The high-degree TCSA experiments are allocated more memory---the extreme case of degree $25$ and $2^{17}$ observations is allocated 50\,GByte.

Multi-precision computations (using MPLAPACK and MPFR) are substantially more expensive, taking up to one hour for configurations with $d=3$, $p \leq 16$, and $n \leq 2^{15}$. In comparison, single- and double-precision computations using OpenBLAS and Armadillo typically complete in under one second for the same parameter ranges.

The larger variability observed in some single- and double-precision compute times is likely due to Armadillo internally attempting to approximate a solution when the matrix involved in the least squares problem is singular. In such cases, additional computations are triggered, leading to occasional increases in runtime.
 
\begin{figure}[b!]
	\includegraphics[width=\textwidth]{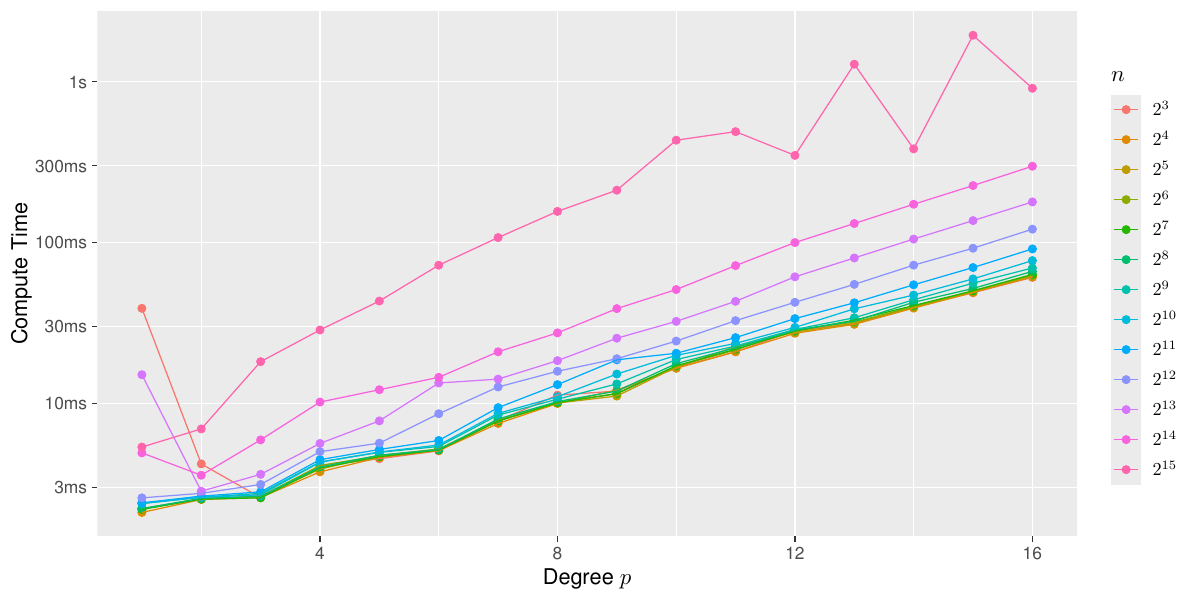}
	\caption{Average compute time for single-precision polynomial propagator estimation in dimension $d=3$.}
	\label{fig:compute:s}
\end{figure}

\begin{figure}[b!]
	\includegraphics[width=\textwidth]{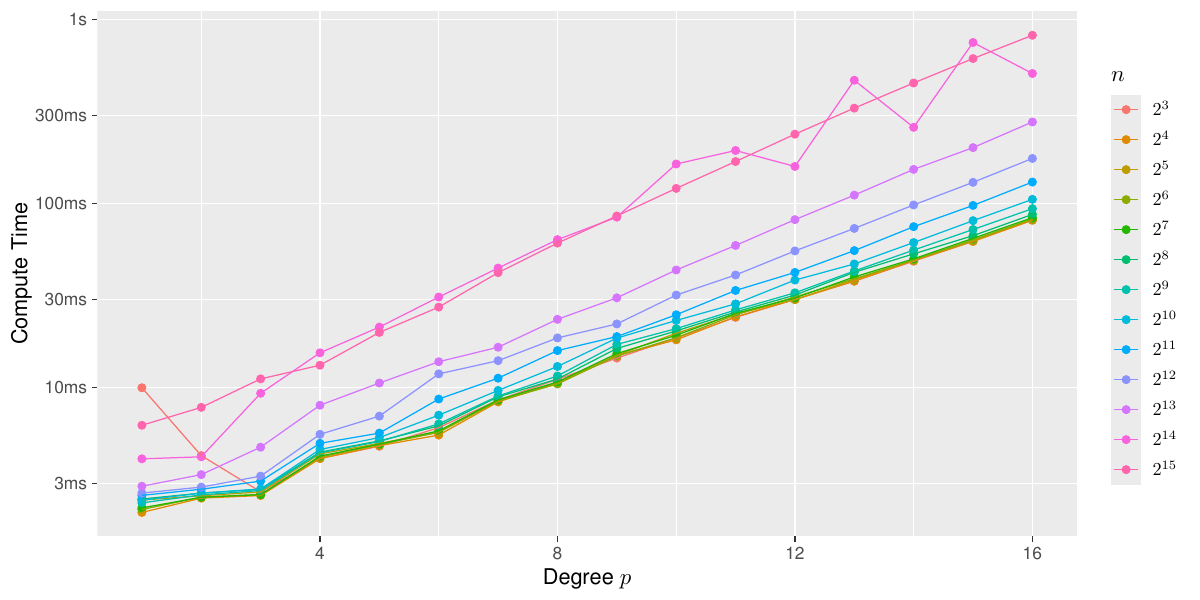}
	\caption{Average compute time for double-precision polynomial propagator estimation in dimension $d=3$.}
	\label{fig:compute:d}
\end{figure}

\begin{figure}[b!]
	\includegraphics[width=\textwidth]{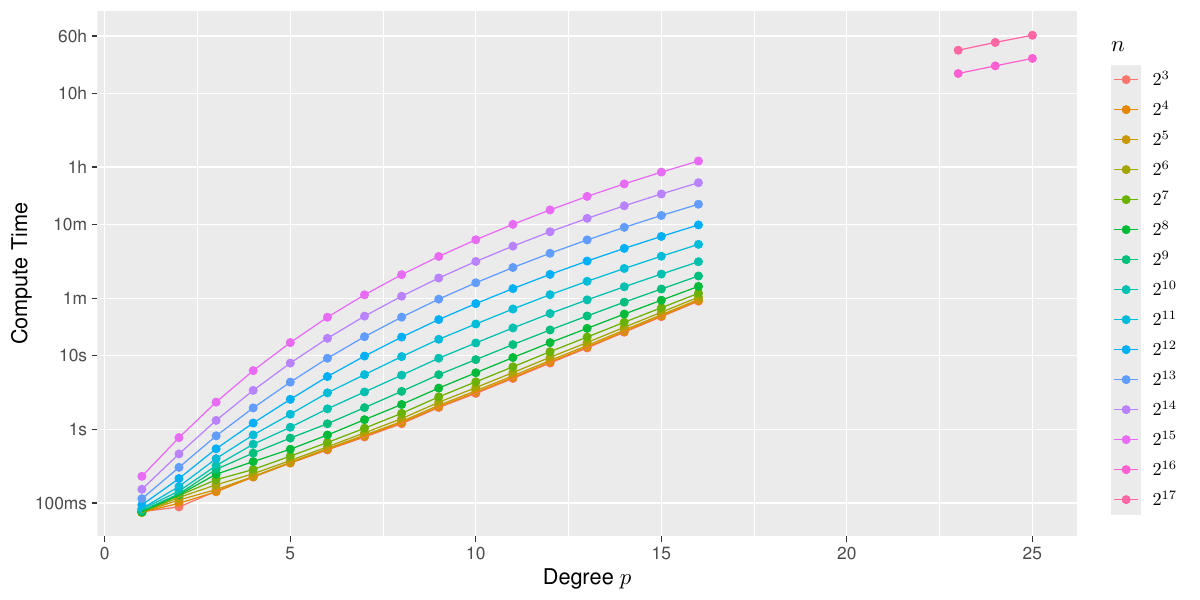}
	\caption{Average compute time for multi-precision polynomial propagator estimation in dimension $d=3$.}
	\label{fig:compute:m}
\end{figure}

\begin{figure}[b!]
	\includegraphics[width=\textwidth]{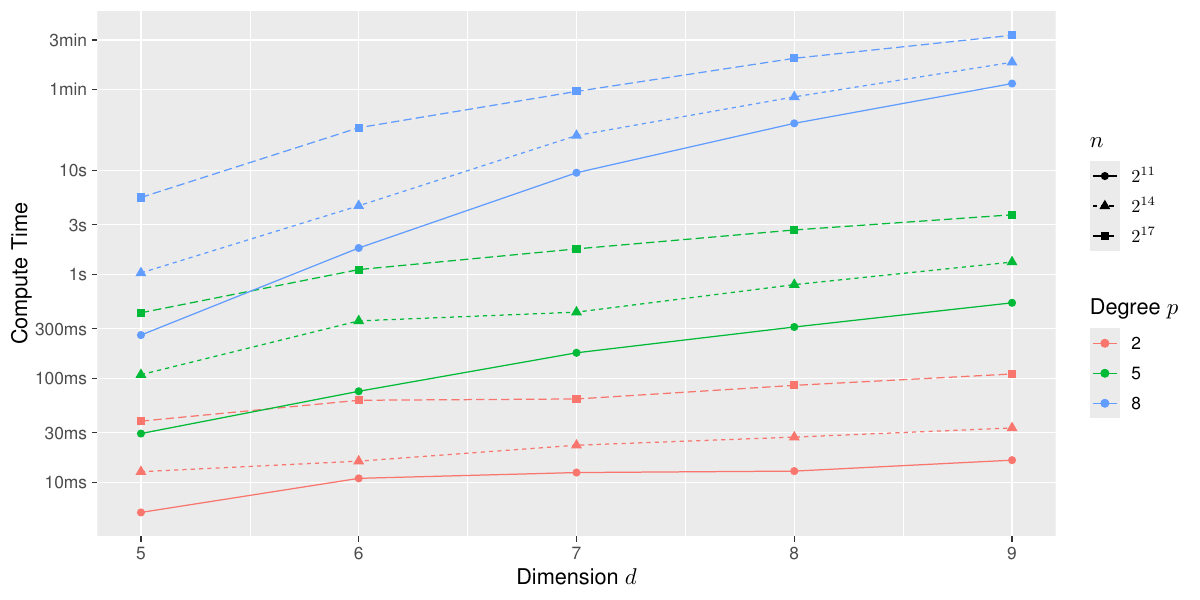}
	\caption{Average compute time for double-precision polynomial propagator estimation on L96.}
	\label{fig:compute:L96}
\end{figure}
\clearpage
\section{Further Details on the Results}\label{app:sec:details}
In the following subsections, we present the results for each experiment listed in \cref{tbl:experiments,tbl:experiments:noisy}. The title of each subsection describes the experimental setting. First, the system is specified. This is followed by a three-letter code indicating the numerical precision used for the ODE solver, the data, and the forecast method, respectively: `s' denotes single (32-bit) precision, `d' double (64-bit) precision, and `m' multi (512-bit) precision. Next, the applied data normalization method is indicated (see description in \cref{ssec:numerics} of the main text). Finally, the test mode is stated. For experiments with noise, it is additionally indicated whether the forecast initial conditions were noisy and the noise level.

\begin{table}[b!]
	\begin{center}
		\fontsize{8.0pt}{10pt}\selectfont
		\fontfamily{phv}\selectfont
		\renewcommand{\arraystretch}{1.05}
		\setlength{\tabcolsep}{0.5em}
		\rowcolors{3}{gray!20}{white}
		\begin{tabular}{lrrrcccccc}
			\toprule
			System & \multicolumn{3}{c}{Precision} & $\stepsize$ & $n$ & normalization & $p$ & test mode & Reps.\\
			\cmidrule(lr){2-4}
			& system & data & method & & & & & & \\
			\midrule
			L63 & 32 bit & 64 bit & 32 bit & $2^{-10}, \dots, 2^{-3}$ & $2^3, \dots, 2^{15}$ & full & $1, \dots, 16$ & sequential & 100\\
			L63 & 32 bit & 64 bit & 64 bit & $2^{-10}, \dots, 2^{-3}$ & $2^3, \dots, 2^{15}$ & full & $1, \dots, 16$ & sequential & 100\\
			L63 & 32 bit & 64 bit & 512 bit & $2^{-10}, \dots, 2^{-3}$ & $2^3, \dots, 2^{15}$ & none & $1, \dots, 16$ & sequential & 100\\
			L63 & 64 bit & 64 bit & 32 bit & $2^{-10}, \dots, 2^{-3}$ & $2^3, \dots, 2^{15}$ & none & $1, \dots, 16$ & sequential & 100\\
			L63 & 64 bit & 64 bit & 32 bit & $2^{-10}, \dots, 2^{-3}$ & $2^3, \dots, 2^{15}$ & diagonal & $1, \dots, 16$ & sequential & 100\\
			L63 & 64 bit & 64 bit & 32 bit & $2^{-10}, \dots, 2^{-3}$ & $2^3, \dots, 2^{15}$ & full & $1, \dots, 16$ & sequential & 100\\
			L63 & 64 bit & 64 bit & 64 bit & $2^{-10}, \dots, 2^{-3}$ & $2^3, \dots, 2^{15}$ & none & $1, \dots, 16$ & sequential & 100\\
			L63 & 64 bit & 64 bit & 64 bit & $2^{-10}, \dots, 2^{-3}$ & $2^3, \dots, 2^{15}$ & diagonal & $1, \dots, 16$ & sequential & 100\\
			L63 & 64 bit & 64 bit & 64 bit & $2^{-10}, \dots, 2^{-3}$ & $2^3, \dots, 2^{15}$ & full & $1, \dots, 16$ & sequential & 100\\
			L63 & 64 bit & 64 bit & 512 bit & $2^{-10}, \dots, 2^{-3}$ & $2^3, \dots, 2^{15}$ & none & $1, \dots, 16$ & sequential & 100\\
			L63 & 64 bit & 32 bit & 64 bit & $2^{-10}, \dots, 2^{-3}$ & $2^3, \dots, 2^{15}$ & full & $1, \dots, 16$ & sequential & 100\\
			L63 & 512 bit & 32 bit & 64 bit & $2^{-10}, \dots, 2^{-3}$ & $2^3, \dots, 2^{15}$ & full & $1, \dots, 16$ & sequential & 100\\
			L63 & 512 bit & 64 bit & 32 bit & $2^{-10}, \dots, 2^{-3}$ & $2^3, \dots, 2^{15}$ & full & $1, \dots, 16$ & sequential & 100\\
			L63 & 512 bit & 64 bit & 64 bit & $2^{-10}, \dots, 2^{-3}$ & $2^3, \dots, 2^{15}$ & full & $1, \dots, 16$ & sequential & 100\\
			L63 & 512 bit & 64 bit & 512 bit & $2^{-10}, \dots, 2^{-3}$ & $2^3, \dots, 2^{15}$ & none & $1, \dots, 16$ & sequential & 100\\
			L63 & 512 bit & 64 bit & 512 bit & $2^{-10}, \dots, 2^{-3}$ & $2^3, \dots, 2^{15}$ & none & $1, \dots, 16$ & random & 100\\
			L63 & 512 bit & 512 bit & 512 bit & $2^{-10}, \dots, 2^{-3}$ & $2^3, \dots, 2^{15}$ & none & $1, \dots, 16$ & sequential & 100\\
			L63 & x & 64 bit & 512 bit & $2^{-10}, \dots, 2^{-3}$ & $2^3, \dots, 2^{15}$ & none & $1, \dots, 16$ & sequential & 100\\
			L63 & y & 64 bit & 512 bit & $2^{-10}, \dots, 2^{-3}$ & $2^3, \dots, 2^{15}$ & none & $1, \dots, 16$ & sequential & 100\\
			L96D5 & 64 bit & 32 bit & 64 bit & $2^{-9}, \dots, 2^{-5}$ & $2^3, \dots, 2^{17}$ & full & $1, \dots, 8$ & sequential & 100\\
			L96D6 & 64 bit & 32 bit & 64 bit & $2^{-9}, \dots, 2^{-5}$ & $2^3, \dots, 2^{17}$ & full & $1, \dots, 8$ & sequential & 100\\
			L96D7 & 64 bit & 32 bit & 64 bit & $2^{-9}, \dots, 2^{-5}$ & $2^3, \dots, 2^{17}$ & full & $1, \dots, 8$ & sequential & 100\\
			L96D8 & 64 bit & 32 bit & 64 bit & $2^{-9}, \dots, 2^{-5}$ & $2^3, \dots, 2^{17}$ & full & $1, \dots, 8$ & sequential & 100\\
			L96D9 & 64 bit & 32 bit & 64 bit & $2^{-9}, \dots, 2^{-5}$ & $2^3, \dots, 2^{17}$ & full & $1, \dots, 8$ & sequential & 100\\
			L96D5 & 32 bit & 64 bit & 64 bit & $2^{-9}, \dots, 2^{-5}$ & $2^3, \dots, 2^{17}$ & full & $1, \dots, 8$ & sequential & 100\\
			L96D6 & 32 bit & 64 bit & 64 bit & $2^{-9}, \dots, 2^{-5}$ & $2^3, \dots, 2^{17}$ & full & $1, \dots, 8$ & sequential & 100\\
			L96D7 & 32 bit & 64 bit & 64 bit & $2^{-9}, \dots, 2^{-5}$ & $2^3, \dots, 2^{17}$ & full & $1, \dots, 8$ & sequential & 100\\
			L96D8 & 32 bit & 64 bit & 64 bit & $2^{-9}, \dots, 2^{-5}$ & $2^3, \dots, 2^{17}$ & full & $1, \dots, 8$ & sequential & 100\\
			L96D9 & 32 bit & 64 bit & 64 bit & $2^{-9}, \dots, 2^{-5}$ & $2^3, \dots, 2^{17}$ & full & $1, \dots, 8$ & sequential & 100\\
			TCSA & 512 bit & 64 bit & 512 bit & $2^{-6}, \dots, 2^{1}$ & $2^3, \dots, 2^{15}$ & none & $1, \dots, 16$ & sequential & 100\\
			TCSA & 512 bit & 64 bit & 512 bit & $2^{-2}$ & $2^{16}, 2^{17}$ & none & $23, 24, 25$ & sequential & 100\\
			\bottomrule
		\end{tabular}
	\end{center}
	\caption{\textbf{Settings of noise-free experiments in this study.} Data is generated by an RK4 ODE solver with time step $\stepsize_0 = 2^{-10}$ for L63 and L96, and $\stepsize_0 = 2^{-6}$ for TCSA using the precision given in the column \textit{Precision-system}; the letters \texttt{x} and \texttt{y} in this column instead indicate extreme precision solutions using a Taylor Integrator or an RK4 solver with very small step size, see Supplementary \cref{app:sec:xprec}. The data is stored at the precision given in the column \textit{Precision-data}. If this number is smaller than the system precision, then data is rounded. The stored data is sub-sampled to arrive at a time step $\stepsize$. A time series of length $n$ is used as training data. It is normalized if indicated in the column \textit{normalization}. Then the training data is presented to the polynomial propagator with degree $p$, which internally calculates at the precision given in the column \textit{Precision-method}. To measure its performance, we create a forecast starting either from the last training state in the sequential \textit{test mode} or from a randomly chosen state in the random \textit{test mode}. Each experiment is repeated 100 times with randomly chosen initial conditions of the training and testing data. Additionally, we conduct an experiment with noisy observations as described in \SuppSec{noise}.}
	\label{tbl:experiments}
\end{table}

\clearpage

\begin{table}[t!]
	\begin{center}
		\fontsize{8.0pt}{10pt}\selectfont
		\fontfamily{phv}\selectfont
		\renewcommand{\arraystretch}{1.05}
		\setlength{\tabcolsep}{0.5em}
		\rowcolors{3}{gray!20}{white}
		\begin{tabular}{lrrrcccccccc}
			\toprule
			Sys. & \multicolumn{3}{c}{Precision} & $\stepsize$ & $n$ & norm. & $p$ & test mode & \multicolumn{2}{c}{Noise} & Reps.\\
			\cmidrule(lr){2-4}
			& system & data & method & & & & & & ini.cond. & level & \\
			\midrule
			L63 & 64 bit & 64 bit & 64 bit & $2^{-10}, \dots, 2^{-3}$ & $2^3, \dots, 2^{15}$ & full & $1, \dots, 8$ & sequential & FALSE & $10^{-9}$ & 100\\
			L63 & 64 bit & 64 bit & 64 bit & $2^{-10}, \dots, 2^{-3}$ & $2^3, \dots, 2^{15}$ & full & $1, \dots, 8$ & sequential & FALSE & $10^{-7}$ & 100\\
			L63 & 64 bit & 64 bit & 64 bit & $2^{-10}, \dots, 2^{-3}$ & $2^3, \dots, 2^{15}$ & full & $1, \dots, 8$ & sequential & FALSE & $10^{-5}$ & 100\\
			L63 & 64 bit & 64 bit & 64 bit & $2^{-10}, \dots, 2^{-3}$ & $2^3, \dots, 2^{15}$ & full & $1, \dots, 8$ & sequential & FALSE & $10^{-3}$ & 100\\
			L63 & 64 bit & 64 bit & 64 bit & $2^{-10}, \dots, 2^{-3}$ & $2^3, \dots, 2^{15}$ & full & $1, \dots, 8$ & sequential & FALSE & $10^{-1}$ & 100\\
			L63 & 64 bit & 64 bit & 64 bit & $2^{-10}, \dots, 2^{-3}$ & $2^3, \dots, 2^{15}$ & full & $1, \dots, 8$ & sequential & TRUE & $10^{-9}$ & 100\\
			L63 & 64 bit & 64 bit & 64 bit & $2^{-10}, \dots, 2^{-3}$ & $2^3, \dots, 2^{15}$ & full & $1, \dots, 8$ & sequential & TRUE & $10^{-7}$ & 100\\
			L63 & 64 bit & 64 bit & 64 bit & $2^{-10}, \dots, 2^{-3}$ & $2^3, \dots, 2^{15}$ & full & $1, \dots, 8$ & sequential & TRUE & $10^{-5}$ & 100\\
			L63 & 64 bit & 64 bit & 64 bit & $2^{-10}, \dots, 2^{-3}$ & $2^3, \dots, 2^{15}$ & full & $1, \dots, 8$ & sequential & TRUE & $10^{-3}$ & 100\\
			L63 & 64 bit & 64 bit & 64 bit & $2^{-10}, \dots, 2^{-3}$ & $2^3, \dots, 2^{15}$ & full & $1, \dots, 8$ & sequential & TRUE & $10^{-1}$ & 100\\
			\bottomrule
		\end{tabular}
	\end{center}
	\caption{\textbf{Settings of noisy experiments in this study.} Same as \cref{tbl:experiments} with the following additions: The column \textit{Noise-ini.cond.} declares whether forecasts are started from noisy initial conditions (or noise-free ones). The column \textit{Noise-level} shows the noise level. See Supplementary \cref{app:sec:noise} for more details.}
	\label{tbl:experiments:noisy}
\end{table}

Each subsection includes three elements:
\begin{itemize}
	\item
	The \textit{Best Plot} displays the average $\vpt{0.5}$ (in Lyapunov times) over 100 repetitions for the best-performing polynomial propagator, plotted against the number of training samples $n$ (horizontal axis). The data time step $\stepsize$ (see \cref{tbl:timesteps}) is indicated by color. The degree $p$ of the optimal polynomial is shown in the label boxes.
	\item
	The \textit{Best Table} presents the same results in tabular form. Each cell in the body of the table reports the $\vpt{0.5}$ value (in Lyapunov times), with the corresponding optimal degree $p$ in parentheses.
	\item
	The \textit{All Plot} shows the average $\vpt{0.5}$ values (in Lyapunov times), averaged over 100 repetitions, for all tested polynomial degrees $p$ (color) and training set sizes $n$ (horizontal axis), across different time steps $\stepsize$ (separate subplots).
\end{itemize}

\etocsettocstyle{\subsection*{Subsections in This Section}}{}
\localtableofcontents

\begin{table}[h!]
    \begin{center}
        \begin{tabular}{c|cccccccc}
        $\stepsize = $ & $2^{-10}$ & $2^{-9}$ & $2^{-8}$ & $2^{-7}$ & $2^{-6}$ & $2^{-5}$ & $2^{-4}$ & $2^{-3}$\\
        \hline
        $\stepsize \approx $ & $0.00098$ & $0.0020$ & $0.0039$ & $0.0078$ & $0.016$ & $0.031$ & $0.063$ & $0.13$
    \end{tabular}
    \end{center}
    \caption{Approximations of the powers of $2$ used as time steps.}\label{tbl:timesteps}
\end{table}

\clearpage
\subsection{L63, dds, normalize none, test sequential}
\begin{figure}[ht!]
\begin{center}
\includegraphics[width=\textwidth]{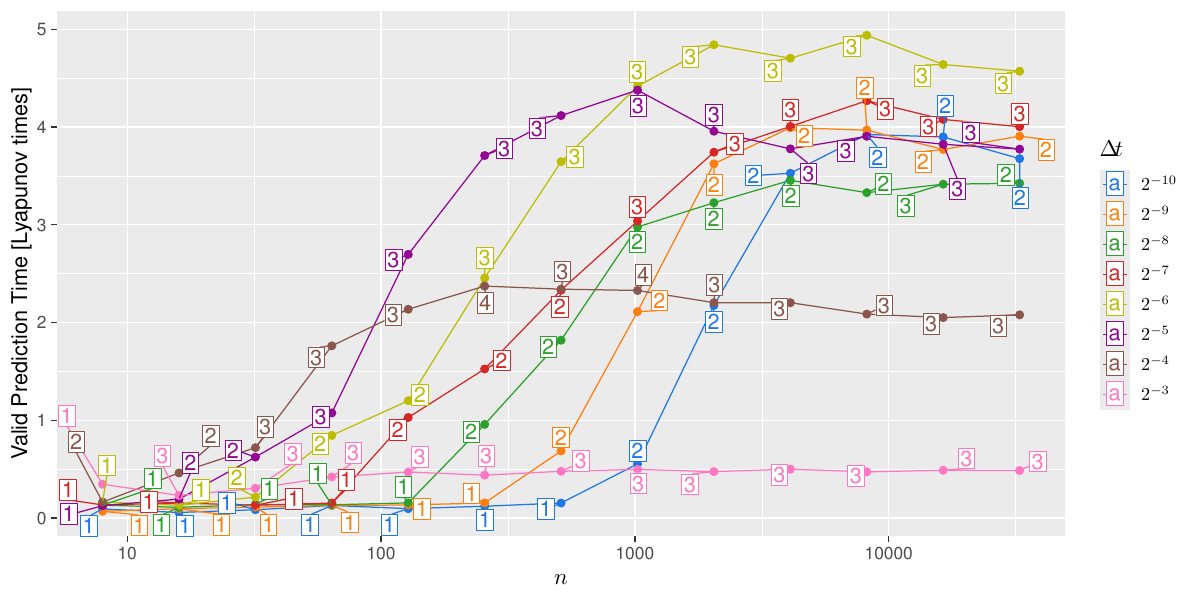}
\end{center}
\caption{\textbf{Best Plot for L63, dds, normalize none, test sequential}. See the beginning of \cref{app:sec:details} for a description.}
\end{figure}
\begin{table}[ht!]
\begin{center}
\fontsize{8.0pt}{10pt}\selectfont
\fontfamily{phv}\selectfont
\renewcommand{\arraystretch}{1.05}
\setlength{\tabcolsep}{0.5em}
\rowcolors{2}{gray!20}{white}
\begin{tabular}{crrrrrrrr}
\toprule
$n$ & \multicolumn{8}{c}{$\stepsize$} \\ 
\cmidrule(lr){1-1} \cmidrule(lr){2-9}
 & $2^{-10}$ & $2^{-9}$ & $2^{-8}$ & $2^{-7}$ & $2^{-6}$ & $2^{-5}$ & $2^{-4}$ & $2^{-3}$ \\ 
\midrule\addlinespace[2.5pt]
$2^{3} = 8$ & {\cellcolor[HTML]{46085B}{\textcolor[HTML]{FFFFFF}{0.1 (1)}}} & {\cellcolor[HTML]{450659}{\textcolor[HTML]{FFFFFF}{0.1 (1)}}} & {\cellcolor[HTML]{460B5E}{\textcolor[HTML]{FFFFFF}{0.1 (1)}}} & {\cellcolor[HTML]{460B5E}{\textcolor[HTML]{FFFFFF}{0.1 (1)}}} & {\cellcolor[HTML]{470D60}{\textcolor[HTML]{FFFFFF}{0.2 (1)}}} & {\cellcolor[HTML]{460A5D}{\textcolor[HTML]{FFFFFF}{0.1 (1)}}} & {\cellcolor[HTML]{470D60}{\textcolor[HTML]{FFFFFF}{0.2 (2)}}} & {\cellcolor[HTML]{481B6D}{\textcolor[HTML]{FFFFFF}{0.3 (1)}}} \\ 
$2^{4} = 16$ & {\cellcolor[HTML]{450559}{\textcolor[HTML]{FFFFFF}{0.1 (1)}}} & {\cellcolor[HTML]{46085C}{\textcolor[HTML]{FFFFFF}{0.1 (1)}}} & {\cellcolor[HTML]{46095D}{\textcolor[HTML]{FFFFFF}{0.1 (1)}}} & {\cellcolor[HTML]{470D60}{\textcolor[HTML]{FFFFFF}{0.2 (1)}}} & {\cellcolor[HTML]{460B5E}{\textcolor[HTML]{FFFFFF}{0.1 (1)}}} & {\cellcolor[HTML]{471063}{\textcolor[HTML]{FFFFFF}{0.2 (2)}}} & {\cellcolor[HTML]{482274}{\textcolor[HTML]{FFFFFF}{0.5 (2)}}} & {\cellcolor[HTML]{471365}{\textcolor[HTML]{FFFFFF}{0.2 (3)}}} \\ 
$2^{5} = 32$ & {\cellcolor[HTML]{46075B}{\textcolor[HTML]{FFFFFF}{0.1 (1)}}} & {\cellcolor[HTML]{46095D}{\textcolor[HTML]{FFFFFF}{0.1 (1)}}} & {\cellcolor[HTML]{460B5E}{\textcolor[HTML]{FFFFFF}{0.1 (1)}}} & {\cellcolor[HTML]{460B5E}{\textcolor[HTML]{FFFFFF}{0.1 (1)}}} & {\cellcolor[HTML]{471164}{\textcolor[HTML]{FFFFFF}{0.2 (2)}}} & {\cellcolor[HTML]{472D7B}{\textcolor[HTML]{FFFFFF}{0.6 (2)}}} & {\cellcolor[HTML]{46337F}{\textcolor[HTML]{FFFFFF}{0.7 (3)}}} & {\cellcolor[HTML]{48186A}{\textcolor[HTML]{FFFFFF}{0.3 (3)}}} \\
\bottomrule
\end{tabular}
\end{center}
\caption{\textbf{Best Table for L63, dds, normalize none, test sequential}. See the beginning of \cref{app:sec:details} for a description.}
\end{table}
\begin{figure}[ht!]
\begin{center}
\includegraphics[width=\textwidth]{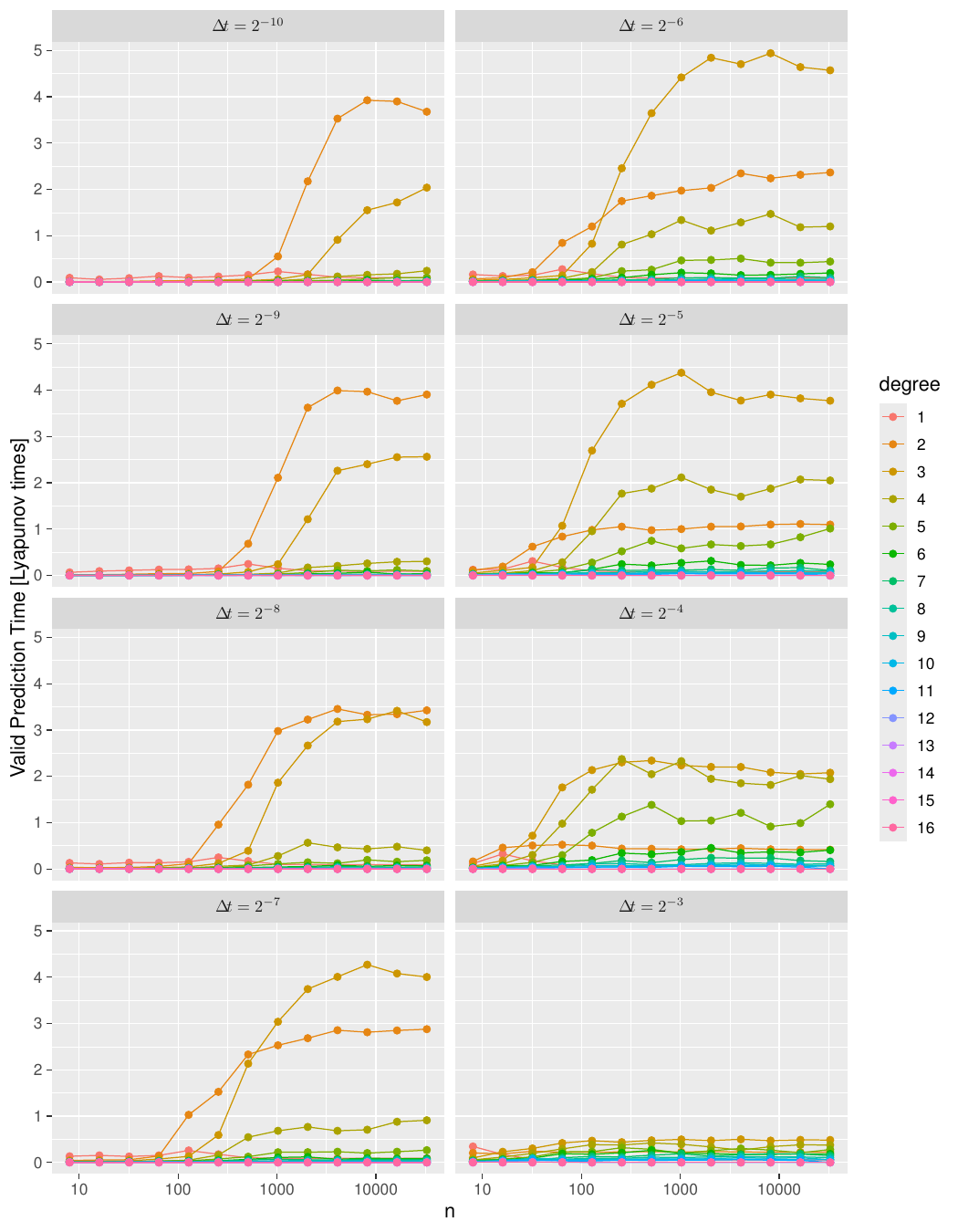}
\end{center}
\caption{\textbf{All Plot for L63, dds, normalize none, test sequential}. See the beginning of \cref{app:sec:details} for a description.}
\end{figure}

\clearpage
\subsection{L63, dds, normalize diag, test sequential}
\begin{figure}[ht!]
\begin{center}
\includegraphics[width=\textwidth]{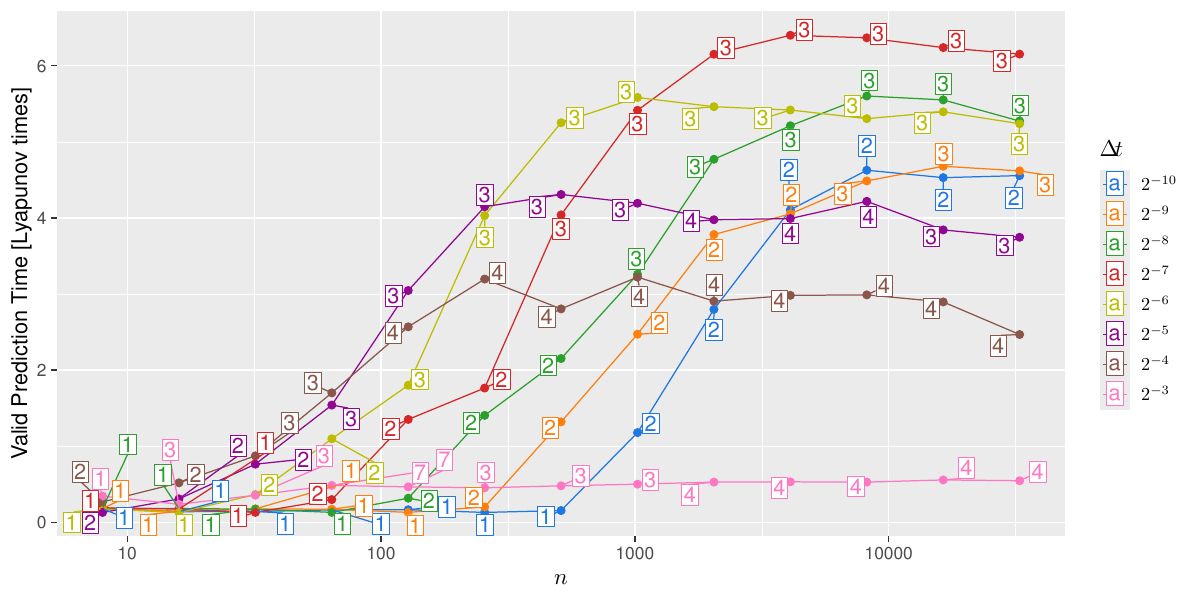}
\end{center}
\caption{\textbf{Best Plot for L63, dds, normalize diag, test sequential}. See the beginning of \cref{app:sec:details} for a description.}
\end{figure}
\begin{table}[ht!]
\begin{center}
\fontsize{8.0pt}{10pt}\selectfont
\fontfamily{phv}\selectfont
\renewcommand{\arraystretch}{1.05}
\setlength{\tabcolsep}{0.5em}
\rowcolors{2}{gray!20}{white}
\begin{tabular}{crrrrrrrr}
\toprule
$n$ & \multicolumn{8}{c}{$\stepsize$} \\ 
\cmidrule(lr){1-1} \cmidrule(lr){2-9}
 & $2^{-10}$ & $2^{-9}$ & $2^{-8}$ & $2^{-7}$ & $2^{-6}$ & $2^{-5}$ & $2^{-4}$ & $2^{-3}$ \\ 
\midrule\addlinespace[2.5pt]
$2^{3} = 8$ & {\cellcolor[HTML]{460B5E}{\textcolor[HTML]{FFFFFF}{0.2 (1)}}} & {\cellcolor[HTML]{460B5E}{\textcolor[HTML]{FFFFFF}{0.2 (1)}}} & {\cellcolor[HTML]{460B5E}{\textcolor[HTML]{FFFFFF}{0.2 (1)}}} & {\cellcolor[HTML]{470C5F}{\textcolor[HTML]{FFFFFF}{0.2 (1)}}} & {\cellcolor[HTML]{460B5E}{\textcolor[HTML]{FFFFFF}{0.2 (1)}}} & {\cellcolor[HTML]{46085C}{\textcolor[HTML]{FFFFFF}{0.1 (2)}}} & {\cellcolor[HTML]{471164}{\textcolor[HTML]{FFFFFF}{0.3 (2)}}} & {\cellcolor[HTML]{481568}{\textcolor[HTML]{FFFFFF}{0.3 (1)}}} \\ 
$2^{4} = 16$ & {\cellcolor[HTML]{46095D}{\textcolor[HTML]{FFFFFF}{0.1 (1)}}} & {\cellcolor[HTML]{460A5D}{\textcolor[HTML]{FFFFFF}{0.1 (1)}}} & {\cellcolor[HTML]{460B5E}{\textcolor[HTML]{FFFFFF}{0.2 (1)}}} & {\cellcolor[HTML]{460B5E}{\textcolor[HTML]{FFFFFF}{0.2 (1)}}} & {\cellcolor[HTML]{46085C}{\textcolor[HTML]{FFFFFF}{0.1 (1)}}} & {\cellcolor[HTML]{471365}{\textcolor[HTML]{FFFFFF}{0.3 (2)}}} & {\cellcolor[HTML]{481E70}{\textcolor[HTML]{FFFFFF}{0.5 (2)}}} & {\cellcolor[HTML]{470F62}{\textcolor[HTML]{FFFFFF}{0.2 (3)}}} \\ 
$2^{5} = 32$ & {\cellcolor[HTML]{46095D}{\textcolor[HTML]{FFFFFF}{0.1 (1)}}} & {\cellcolor[HTML]{460B5E}{\textcolor[HTML]{FFFFFF}{0.2 (1)}}} & {\cellcolor[HTML]{460B5E}{\textcolor[HTML]{FFFFFF}{0.2 (1)}}} & {\cellcolor[HTML]{46085C}{\textcolor[HTML]{FFFFFF}{0.1 (1)}}} & {\cellcolor[HTML]{481668}{\textcolor[HTML]{FFFFFF}{0.4 (2)}}} & {\cellcolor[HTML]{472B7A}{\textcolor[HTML]{FFFFFF}{0.8 (2)}}} & {\cellcolor[HTML]{46307E}{\textcolor[HTML]{FFFFFF}{0.9 (3)}}} & {\cellcolor[HTML]{481668}{\textcolor[HTML]{FFFFFF}{0.4 (3)}}} \\
\bottomrule
\end{tabular}
\end{center}
\caption{\textbf{Best Table for L63, dds, normalize diag, test sequential}. See the beginning of \cref{app:sec:details} for a description.}
\end{table}
\begin{figure}[ht!]
\begin{center}
\includegraphics[width=\textwidth]{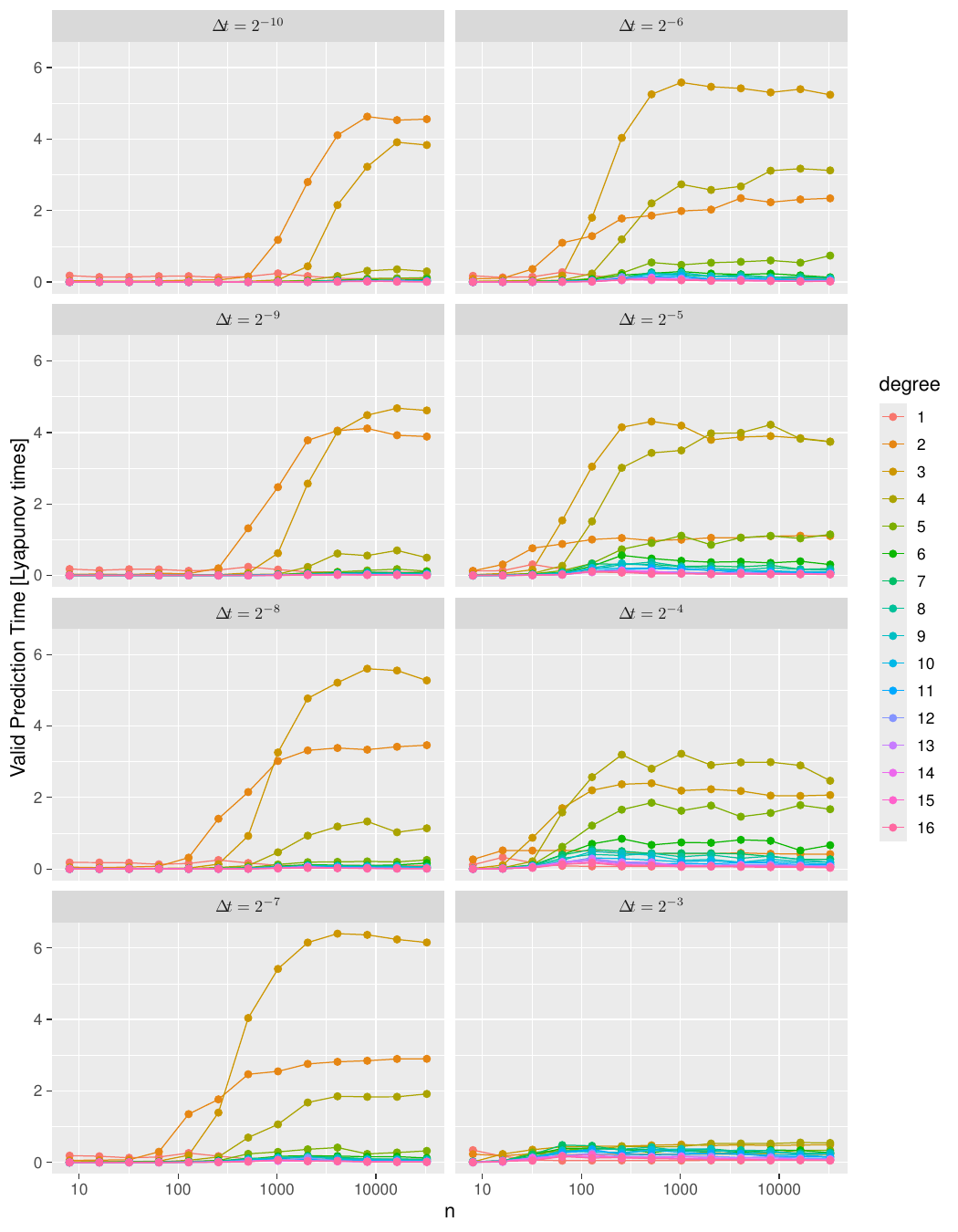}
\end{center}
\caption{\textbf{All Plot for L63, dds, normalize diag, test sequential}. See the beginning of \cref{app:sec:details} for a description.}
\end{figure}

\clearpage
\subsection{L63, dds, normalize full, test sequential}
\begin{figure}[ht!]
\begin{center}
\includegraphics[width=\textwidth]{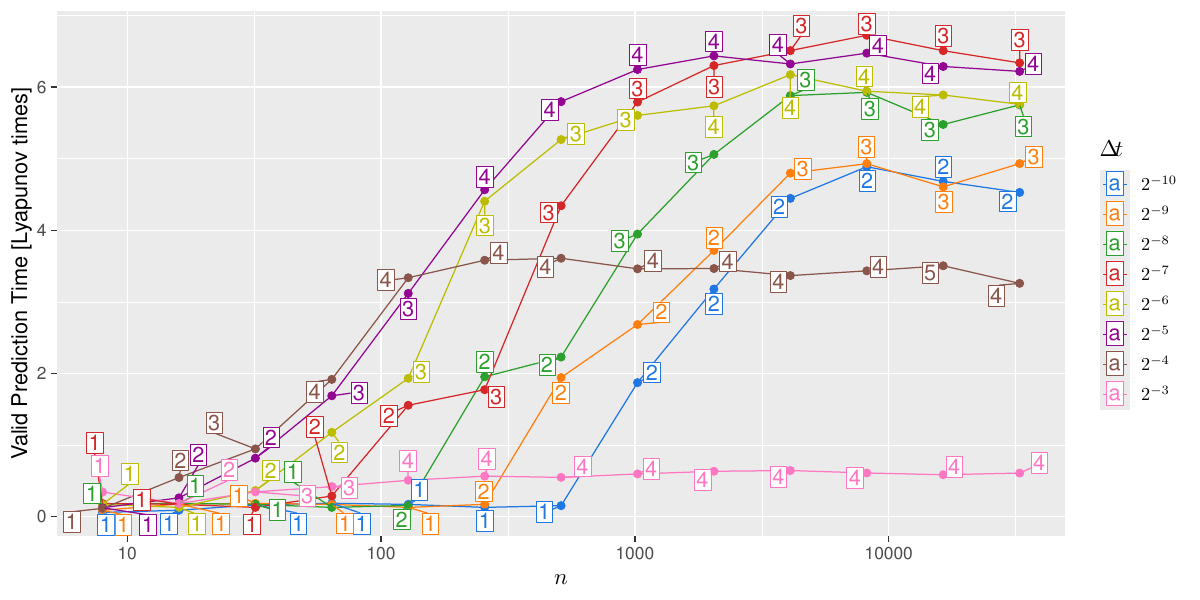}
\end{center}
\caption{\textbf{Best Plot for L63, dds, normalize full, test sequential}. See the beginning of \cref{app:sec:details} for a description.}
\end{figure}
\begin{table}[ht!]
\begin{center}
\fontsize{8.0pt}{10pt}\selectfont
\fontfamily{phv}\selectfont
\renewcommand{\arraystretch}{1.05}
\setlength{\tabcolsep}{0.5em}
\rowcolors{2}{gray!20}{white}
\begin{tabular}{crrrrrrrr}
\toprule
$n$ & \multicolumn{8}{c}{$\stepsize$} \\ 
\cmidrule(lr){1-1} \cmidrule(lr){2-9}
 & $2^{-10}$ & $2^{-9}$ & $2^{-8}$ & $2^{-7}$ & $2^{-6}$ & $2^{-5}$ & $2^{-4}$ & $2^{-3}$ \\ 
\midrule\addlinespace[2.5pt]
$2^{3} = 8$ & {\cellcolor[HTML]{450458}{\textcolor[HTML]{FFFFFF}{0.1 (1)}}} & {\cellcolor[HTML]{450659}{\textcolor[HTML]{FFFFFF}{0.1 (1)}}} & {\cellcolor[HTML]{460B5E}{\textcolor[HTML]{FFFFFF}{0.2 (1)}}} & {\cellcolor[HTML]{460B5E}{\textcolor[HTML]{FFFFFF}{0.2 (1)}}} & {\cellcolor[HTML]{460B5E}{\textcolor[HTML]{FFFFFF}{0.2 (1)}}} & {\cellcolor[HTML]{46085C}{\textcolor[HTML]{FFFFFF}{0.1 (1)}}} & {\cellcolor[HTML]{46075B}{\textcolor[HTML]{FFFFFF}{0.1 (1)}}} & {\cellcolor[HTML]{481467}{\textcolor[HTML]{FFFFFF}{0.3 (1)}}} \\ 
$2^{4} = 16$ & {\cellcolor[HTML]{450659}{\textcolor[HTML]{FFFFFF}{0.1 (1)}}} & {\cellcolor[HTML]{460B5E}{\textcolor[HTML]{FFFFFF}{0.2 (1)}}} & {\cellcolor[HTML]{460B5E}{\textcolor[HTML]{FFFFFF}{0.2 (1)}}} & {\cellcolor[HTML]{460B5E}{\textcolor[HTML]{FFFFFF}{0.2 (1)}}} & {\cellcolor[HTML]{46085C}{\textcolor[HTML]{FFFFFF}{0.1 (1)}}} & {\cellcolor[HTML]{471063}{\textcolor[HTML]{FFFFFF}{0.3 (2)}}} & {\cellcolor[HTML]{481F70}{\textcolor[HTML]{FFFFFF}{0.5 (2)}}} & {\cellcolor[HTML]{460B5E}{\textcolor[HTML]{FFFFFF}{0.2 (2)}}} \\ 
$2^{5} = 32$ & {\cellcolor[HTML]{460A5D}{\textcolor[HTML]{FFFFFF}{0.2 (1)}}} & {\cellcolor[HTML]{460B5E}{\textcolor[HTML]{FFFFFF}{0.2 (1)}}} & {\cellcolor[HTML]{460B5E}{\textcolor[HTML]{FFFFFF}{0.2 (1)}}} & {\cellcolor[HTML]{46085C}{\textcolor[HTML]{FFFFFF}{0.1 (1)}}} & {\cellcolor[HTML]{481668}{\textcolor[HTML]{FFFFFF}{0.4 (2)}}} & {\cellcolor[HTML]{472C7A}{\textcolor[HTML]{FFFFFF}{0.8 (2)}}} & {\cellcolor[HTML]{46327E}{\textcolor[HTML]{FFFFFF}{0.9 (3)}}} & {\cellcolor[HTML]{481467}{\textcolor[HTML]{FFFFFF}{0.3 (3)}}} \\
\bottomrule
\end{tabular}
\end{center}
\caption{\textbf{Best Table for L63, dds, normalize full, test sequential}. See the beginning of \cref{app:sec:details} for a description.}
\end{table}
\begin{figure}[ht!]
\begin{center}
\includegraphics[width=\textwidth]{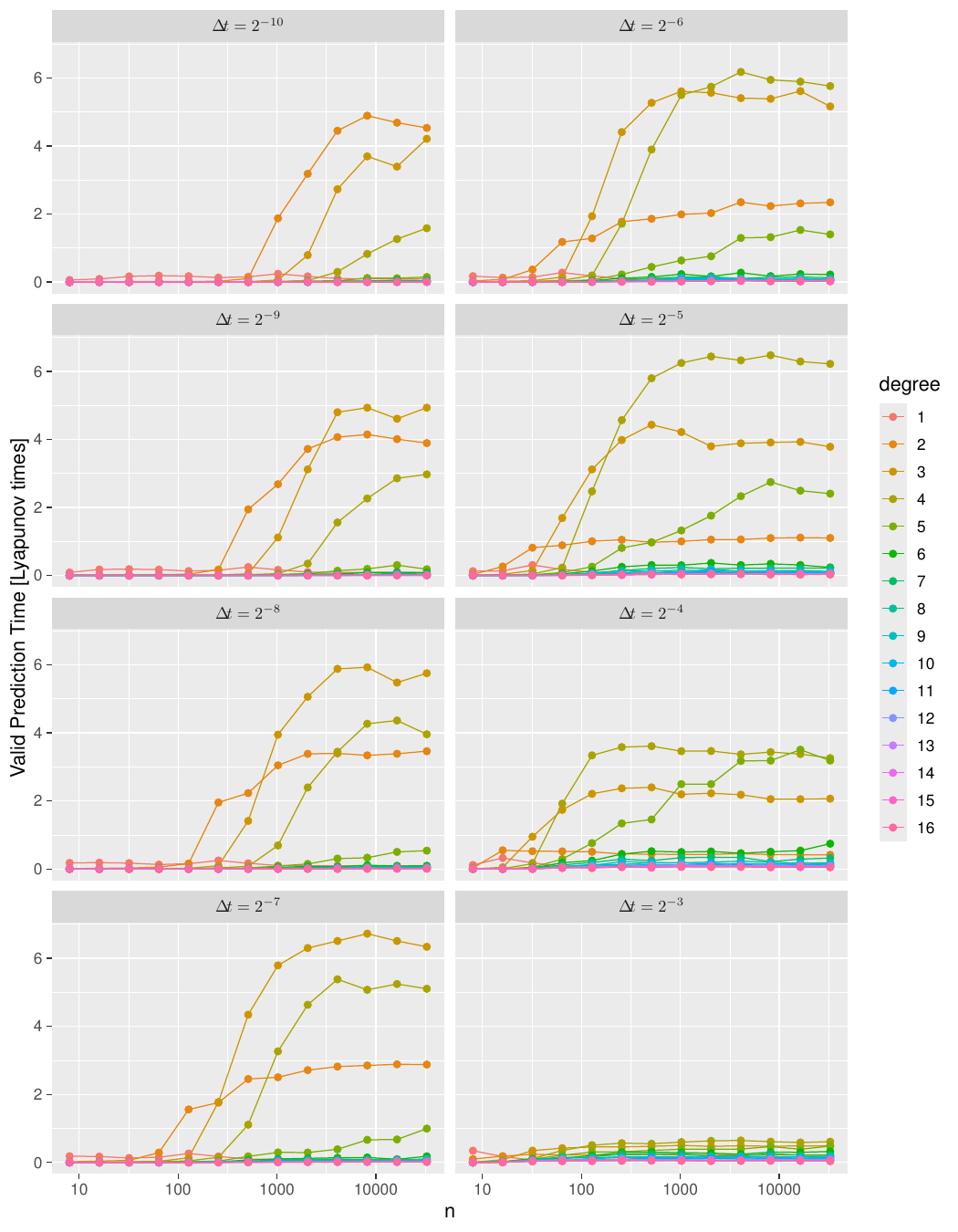}
\end{center}
\caption{\textbf{All Plot for L63, dds, normalize full, test sequential}. See the beginning of \cref{app:sec:details} for a description.}
\end{figure}

\clearpage
\subsection{L63, ddd, normalize none, test sequential}
\begin{figure}[ht!]
\begin{center}
\includegraphics[width=\textwidth]{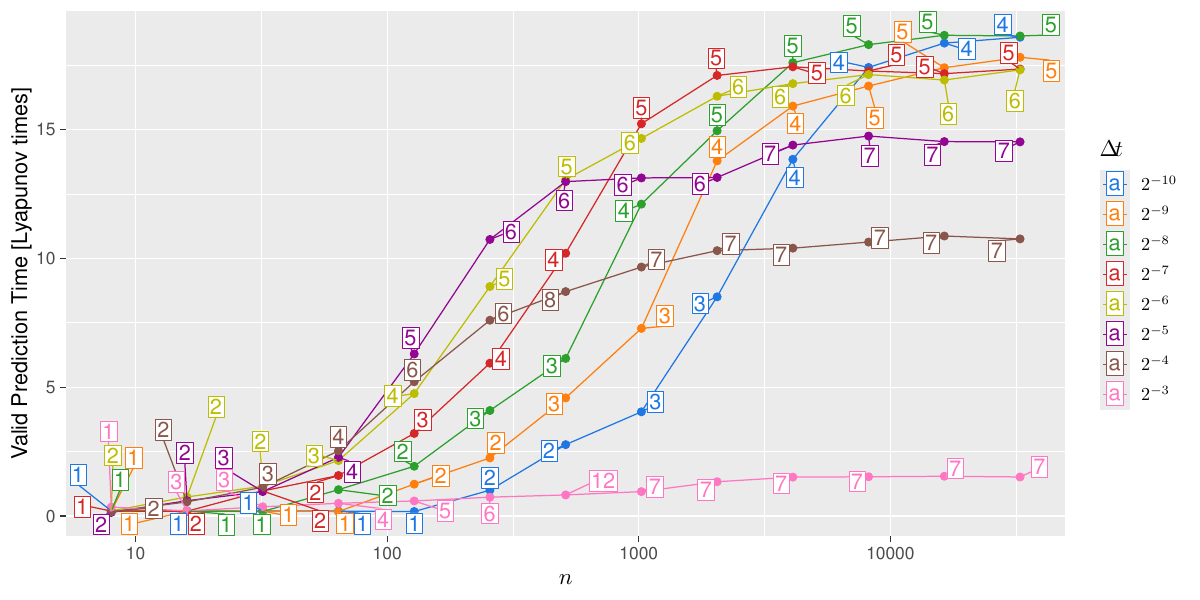}
\end{center}
\caption{\textbf{Best Plot for L63, ddd, normalize none, test sequential}. See the beginning of \cref{app:sec:details} for a description.}
\end{figure}
\begin{table}[ht!]
\begin{center}
\fontsize{8.0pt}{10pt}\selectfont
\fontfamily{phv}\selectfont
\renewcommand{\arraystretch}{1.05}
\setlength{\tabcolsep}{0.5em}
\rowcolors{2}{gray!20}{white}
\begin{tabular}{crrrrrrrr}
\toprule
$n$ & \multicolumn{8}{c}{$\stepsize$} \\ 
\cmidrule(lr){1-1} \cmidrule(lr){2-9}
 & $2^{-10}$ & $2^{-9}$ & $2^{-8}$ & $2^{-7}$ & $2^{-6}$ & $2^{-5}$ & $2^{-4}$ & $2^{-3}$ \\ 
\midrule\addlinespace[2.5pt]
$2^{3} = 8$ & {\cellcolor[HTML]{450558}{\textcolor[HTML]{FFFFFF}{0.2 (1)}}} & {\cellcolor[HTML]{450558}{\textcolor[HTML]{FFFFFF}{0.2 (1)}}} & {\cellcolor[HTML]{450558}{\textcolor[HTML]{FFFFFF}{0.2 (1)}}} & {\cellcolor[HTML]{450558}{\textcolor[HTML]{FFFFFF}{0.2 (1)}}} & {\cellcolor[HTML]{450558}{\textcolor[HTML]{FFFFFF}{0.2 (2)}}} & {\cellcolor[HTML]{450457}{\textcolor[HTML]{FFFFFF}{0.1 (2)}}} & {\cellcolor[HTML]{450458}{\textcolor[HTML]{FFFFFF}{0.2 (2)}}} & {\cellcolor[HTML]{46085B}{\textcolor[HTML]{FFFFFF}{0.3 (1)}}} \\ 
$2^{4} = 16$ & {\cellcolor[HTML]{450558}{\textcolor[HTML]{FFFFFF}{0.2 (1)}}} & {\cellcolor[HTML]{450558}{\textcolor[HTML]{FFFFFF}{0.2 (1)}}} & {\cellcolor[HTML]{450558}{\textcolor[HTML]{FFFFFF}{0.2 (1)}}} & {\cellcolor[HTML]{450558}{\textcolor[HTML]{FFFFFF}{0.2 (2)}}} & {\cellcolor[HTML]{471063}{\textcolor[HTML]{FFFFFF}{0.7 (2)}}} & {\cellcolor[HTML]{470D60}{\textcolor[HTML]{FFFFFF}{0.6 (2)}}} & {\cellcolor[HTML]{460C5F}{\textcolor[HTML]{FFFFFF}{0.5 (2)}}} & {\cellcolor[HTML]{450558}{\textcolor[HTML]{FFFFFF}{0.2 (3)}}} \\ 
$2^{5} = 32$ & {\cellcolor[HTML]{450558}{\textcolor[HTML]{FFFFFF}{0.2 (1)}}} & {\cellcolor[HTML]{450558}{\textcolor[HTML]{FFFFFF}{0.2 (1)}}} & {\cellcolor[HTML]{450458}{\textcolor[HTML]{FFFFFF}{0.2 (1)}}} & {\cellcolor[HTML]{481567}{\textcolor[HTML]{FFFFFF}{1.0 (2)}}} & {\cellcolor[HTML]{48186A}{\textcolor[HTML]{FFFFFF}{1.2 (2)}}} & {\cellcolor[HTML]{481467}{\textcolor[HTML]{FFFFFF}{0.9 (3)}}} & {\cellcolor[HTML]{481769}{\textcolor[HTML]{FFFFFF}{1.1 (3)}}} & {\cellcolor[HTML]{46085C}{\textcolor[HTML]{FFFFFF}{0.4 (3)}}} \\
\bottomrule
\end{tabular}
\end{center}
\caption{\textbf{Best Table for L63, ddd, normalize none, test sequential}. See the beginning of \cref{app:sec:details} for a description.}
\end{table}
\begin{figure}[ht!]
\begin{center}
\includegraphics[width=\textwidth]{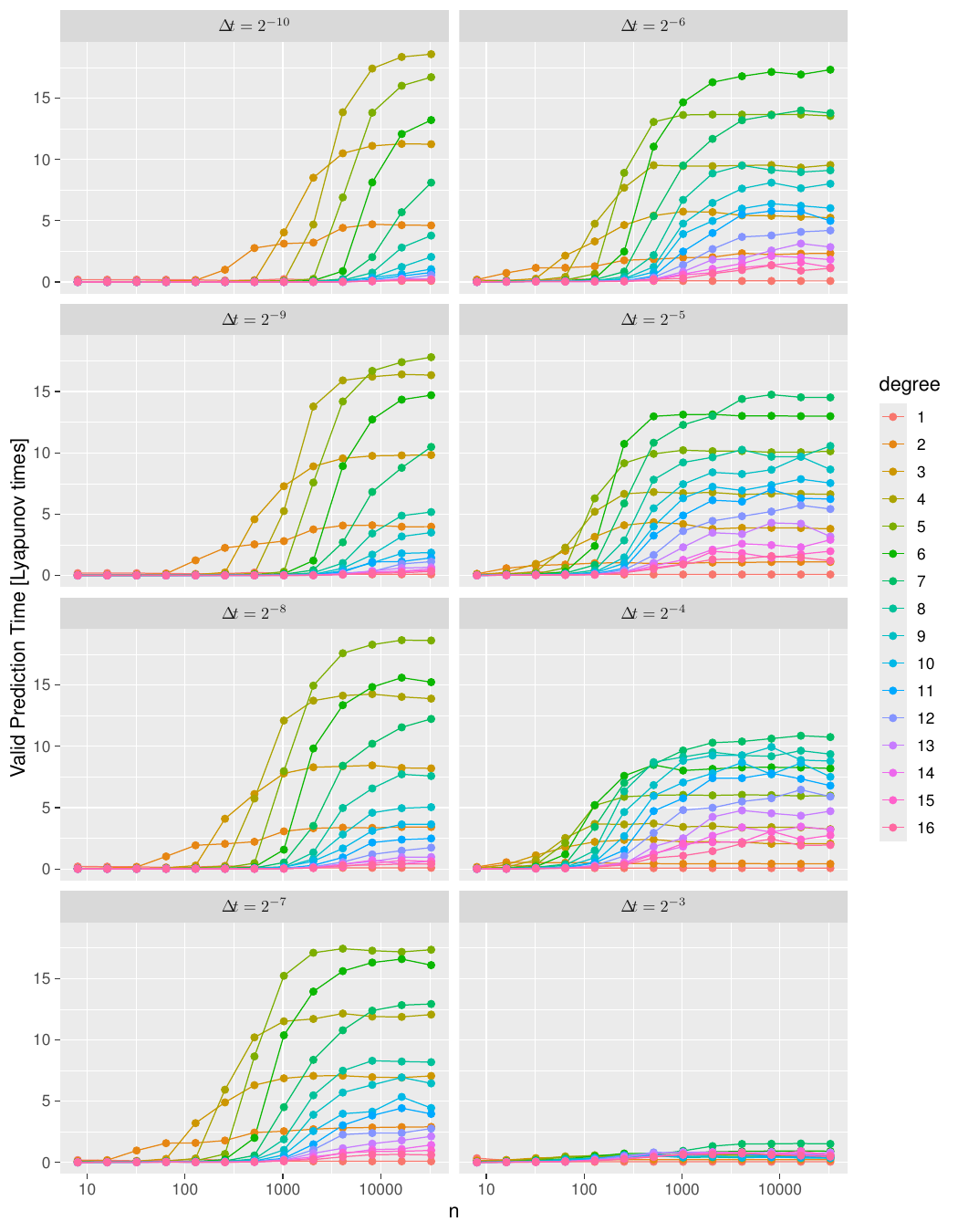}
\end{center}
\caption{\textbf{All Plot for L63, ddd, normalize none, test sequential}. See the beginning of \cref{app:sec:details} for a description.}
\end{figure}

\clearpage
\subsection{L63, ddd, normalize diag, test sequential}
\begin{figure}[ht!]
\begin{center}
\includegraphics[width=\textwidth]{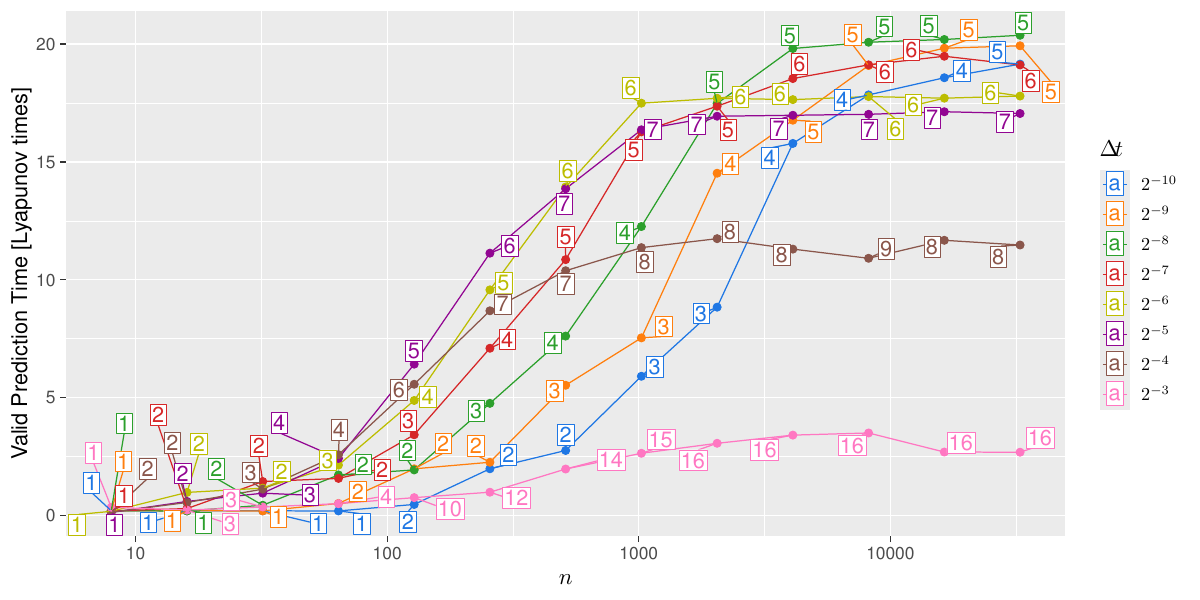}
\end{center}
\caption{\textbf{Best Plot for L63, ddd, normalize diag, test sequential}. See the beginning of \cref{app:sec:details} for a description.}
\end{figure}
\begin{table}[ht!]
\begin{center}
\fontsize{8.0pt}{10pt}\selectfont
\fontfamily{phv}\selectfont
\renewcommand{\arraystretch}{1.05}
\setlength{\tabcolsep}{0.5em}
\rowcolors{2}{gray!20}{white}
\begin{tabular}{crrrrrrrr}
\toprule
$n$ & \multicolumn{8}{c}{$\stepsize$} \\ 
\cmidrule(lr){1-1} \cmidrule(lr){2-9}
 & $2^{-10}$ & $2^{-9}$ & $2^{-8}$ & $2^{-7}$ & $2^{-6}$ & $2^{-5}$ & $2^{-4}$ & $2^{-3}$ \\ 
\midrule\addlinespace[2.5pt]
$2^{3} = 8$ & {\cellcolor[HTML]{450458}{\textcolor[HTML]{FFFFFF}{0.2 (1)}}} & {\cellcolor[HTML]{450458}{\textcolor[HTML]{FFFFFF}{0.2 (1)}}} & {\cellcolor[HTML]{450458}{\textcolor[HTML]{FFFFFF}{0.2 (1)}}} & {\cellcolor[HTML]{450458}{\textcolor[HTML]{FFFFFF}{0.2 (1)}}} & {\cellcolor[HTML]{450457}{\textcolor[HTML]{FFFFFF}{0.2 (1)}}} & {\cellcolor[HTML]{450357}{\textcolor[HTML]{FFFFFF}{0.1 (1)}}} & {\cellcolor[HTML]{450357}{\textcolor[HTML]{FFFFFF}{0.1 (2)}}} & {\cellcolor[HTML]{46075B}{\textcolor[HTML]{FFFFFF}{0.3 (1)}}} \\ 
$2^{4} = 16$ & {\cellcolor[HTML]{450458}{\textcolor[HTML]{FFFFFF}{0.2 (1)}}} & {\cellcolor[HTML]{450458}{\textcolor[HTML]{FFFFFF}{0.2 (1)}}} & {\cellcolor[HTML]{450458}{\textcolor[HTML]{FFFFFF}{0.2 (1)}}} & {\cellcolor[HTML]{46065A}{\textcolor[HTML]{FFFFFF}{0.3 (2)}}} & {\cellcolor[HTML]{471365}{\textcolor[HTML]{FFFFFF}{1.0 (2)}}} & {\cellcolor[HTML]{460C5F}{\textcolor[HTML]{FFFFFF}{0.6 (2)}}} & {\cellcolor[HTML]{460B5E}{\textcolor[HTML]{FFFFFF}{0.5 (2)}}} & {\cellcolor[HTML]{450559}{\textcolor[HTML]{FFFFFF}{0.2 (3)}}} \\ 
$2^{5} = 32$ & {\cellcolor[HTML]{450458}{\textcolor[HTML]{FFFFFF}{0.2 (1)}}} & {\cellcolor[HTML]{450458}{\textcolor[HTML]{FFFFFF}{0.2 (1)}}} & {\cellcolor[HTML]{46095C}{\textcolor[HTML]{FFFFFF}{0.4 (2)}}} & {\cellcolor[HTML]{481B6D}{\textcolor[HTML]{FFFFFF}{1.4 (2)}}} & {\cellcolor[HTML]{481668}{\textcolor[HTML]{FFFFFF}{1.2 (2)}}} & {\cellcolor[HTML]{471365}{\textcolor[HTML]{FFFFFF}{0.9 (3)}}} & {\cellcolor[HTML]{481668}{\textcolor[HTML]{FFFFFF}{1.1 (3)}}} & {\cellcolor[HTML]{46075B}{\textcolor[HTML]{FFFFFF}{0.4 (3)}}} \\
\bottomrule
\end{tabular}
\end{center}
\caption{\textbf{Best Table for L63, ddd, normalize diag, test sequential}. See the beginning of \cref{app:sec:details} for a description.}
\end{table}
\begin{figure}[ht!]
\begin{center}
\includegraphics[width=\textwidth]{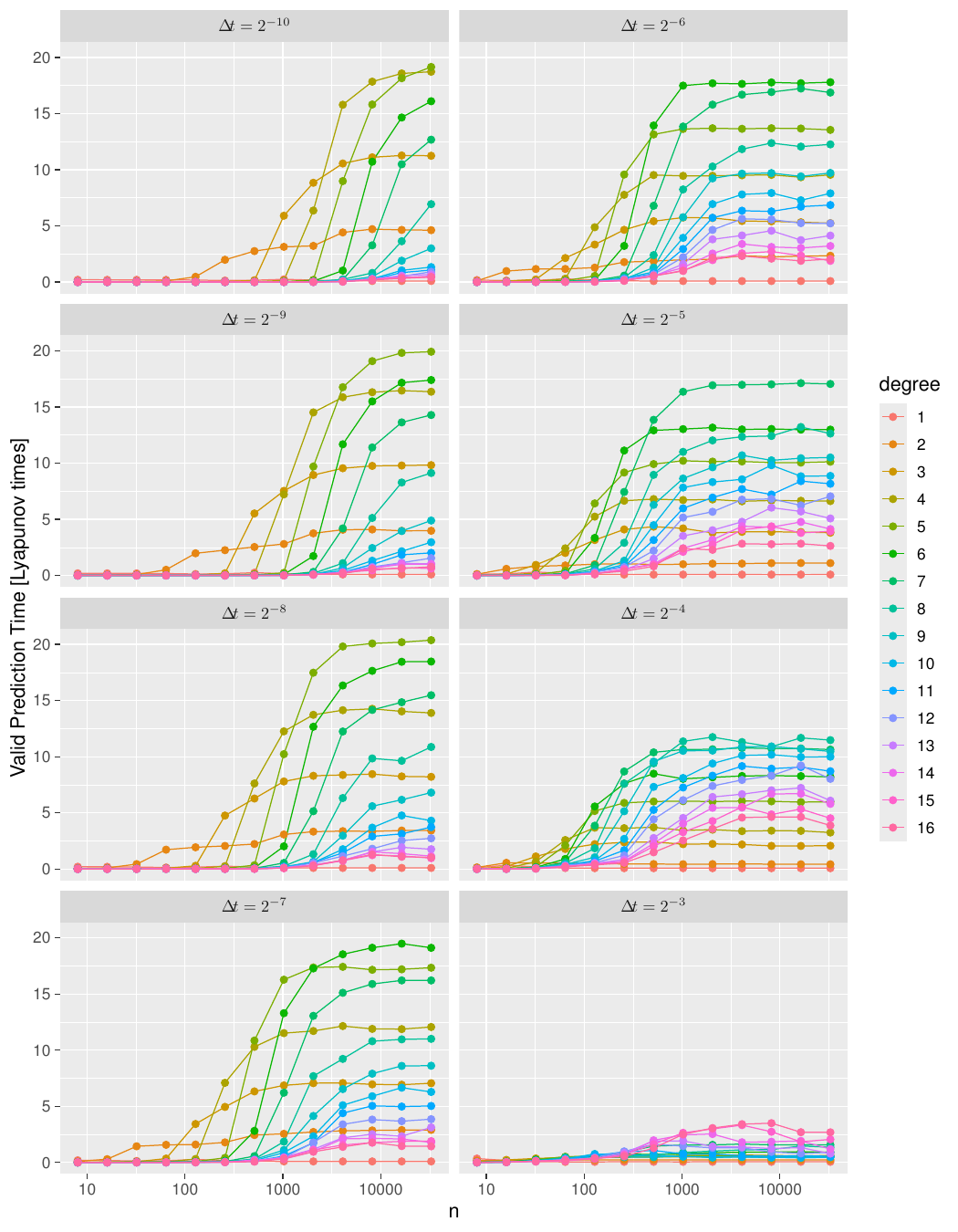}
\end{center}
\caption{\textbf{All Plot for L63, ddd, normalize diag, test sequential}. See the beginning of \cref{app:sec:details} for a description.}
\end{figure}

\clearpage
\subsection{L63, ddd, normalize full, test sequential}
\begin{figure}[ht!]
\begin{center}
\includegraphics[width=\textwidth]{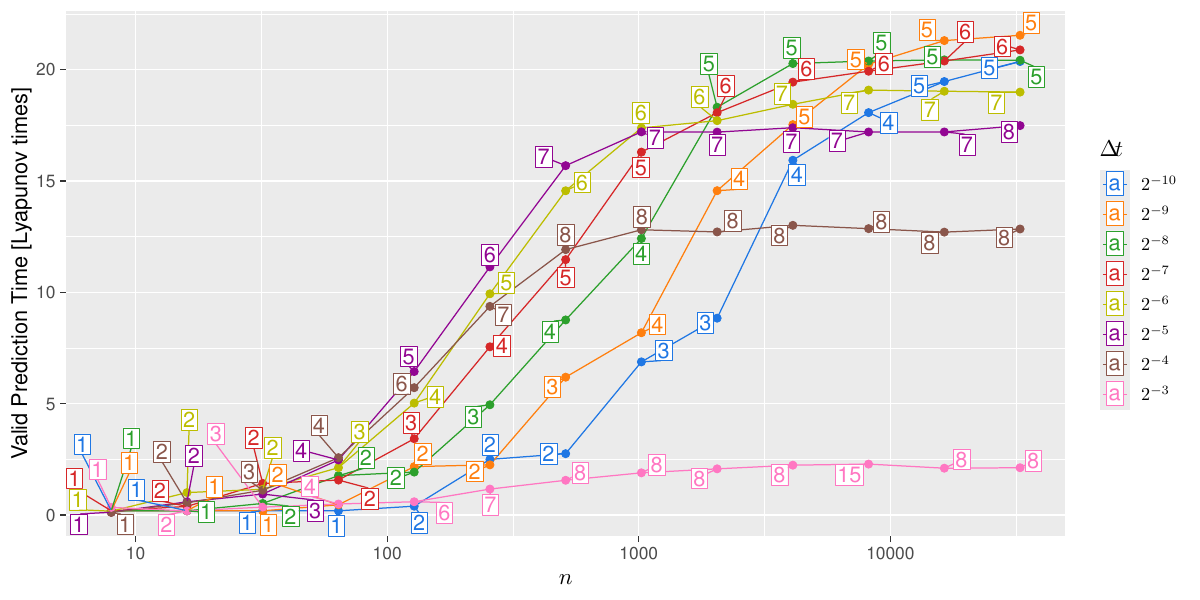}
\end{center}
\caption{\textbf{Best Plot for L63, ddd, normalize full, test sequential}. See the beginning of \cref{app:sec:details} for a description.}
\end{figure}
\begin{table}[ht!]
\begin{center}
\fontsize{8.0pt}{10pt}\selectfont
\fontfamily{phv}\selectfont
\renewcommand{\arraystretch}{1.05}
\setlength{\tabcolsep}{0.5em}
\rowcolors{2}{gray!20}{white}
\begin{tabular}{crrrrrrrr}
\toprule
$n$ & \multicolumn{8}{c}{$\stepsize$} \\ 
\cmidrule(lr){1-1} \cmidrule(lr){2-9}
 & $2^{-10}$ & $2^{-9}$ & $2^{-8}$ & $2^{-7}$ & $2^{-6}$ & $2^{-5}$ & $2^{-4}$ & $2^{-3}$ \\ 
\midrule\addlinespace[2.5pt]
$2^{3} = 8$ & {\cellcolor[HTML]{450458}{\textcolor[HTML]{FFFFFF}{0.2 (1)}}} & {\cellcolor[HTML]{450458}{\textcolor[HTML]{FFFFFF}{0.2 (1)}}} & {\cellcolor[HTML]{450458}{\textcolor[HTML]{FFFFFF}{0.2 (1)}}} & {\cellcolor[HTML]{450457}{\textcolor[HTML]{FFFFFF}{0.2 (1)}}} & {\cellcolor[HTML]{450457}{\textcolor[HTML]{FFFFFF}{0.2 (1)}}} & {\cellcolor[HTML]{440356}{\textcolor[HTML]{FFFFFF}{0.1 (1)}}} & {\cellcolor[HTML]{440356}{\textcolor[HTML]{FFFFFF}{0.1 (1)}}} & {\cellcolor[HTML]{46075A}{\textcolor[HTML]{FFFFFF}{0.3 (1)}}} \\ 
$2^{4} = 16$ & {\cellcolor[HTML]{450458}{\textcolor[HTML]{FFFFFF}{0.2 (1)}}} & {\cellcolor[HTML]{450458}{\textcolor[HTML]{FFFFFF}{0.2 (1)}}} & {\cellcolor[HTML]{450458}{\textcolor[HTML]{FFFFFF}{0.2 (1)}}} & {\cellcolor[HTML]{46085B}{\textcolor[HTML]{FFFFFF}{0.4 (2)}}} & {\cellcolor[HTML]{471365}{\textcolor[HTML]{FFFFFF}{1.0 (2)}}} & {\cellcolor[HTML]{460B5E}{\textcolor[HTML]{FFFFFF}{0.6 (2)}}} & {\cellcolor[HTML]{460A5D}{\textcolor[HTML]{FFFFFF}{0.5 (2)}}} & {\cellcolor[HTML]{450457}{\textcolor[HTML]{FFFFFF}{0.2 (2)}}} \\ 
$2^{5} = 32$ & {\cellcolor[HTML]{450458}{\textcolor[HTML]{FFFFFF}{0.2 (1)}}} & {\cellcolor[HTML]{450457}{\textcolor[HTML]{FFFFFF}{0.2 (1)}}} & {\cellcolor[HTML]{460A5D}{\textcolor[HTML]{FFFFFF}{0.5 (2)}}} & {\cellcolor[HTML]{481A6C}{\textcolor[HTML]{FFFFFF}{1.4 (2)}}} & {\cellcolor[HTML]{481568}{\textcolor[HTML]{FFFFFF}{1.2 (2)}}} & {\cellcolor[HTML]{471164}{\textcolor[HTML]{FFFFFF}{0.9 (3)}}} & {\cellcolor[HTML]{481467}{\textcolor[HTML]{FFFFFF}{1.1 (3)}}} & {\cellcolor[HTML]{46075A}{\textcolor[HTML]{FFFFFF}{0.4 (3)}}} \\
\bottomrule
\end{tabular}
\end{center}
\caption{\textbf{Best Table for L63, ddd, normalize full, test sequential}. See the beginning of \cref{app:sec:details} for a description.}
\end{table}
\begin{figure}[ht!]
\begin{center}
\includegraphics[width=\textwidth]{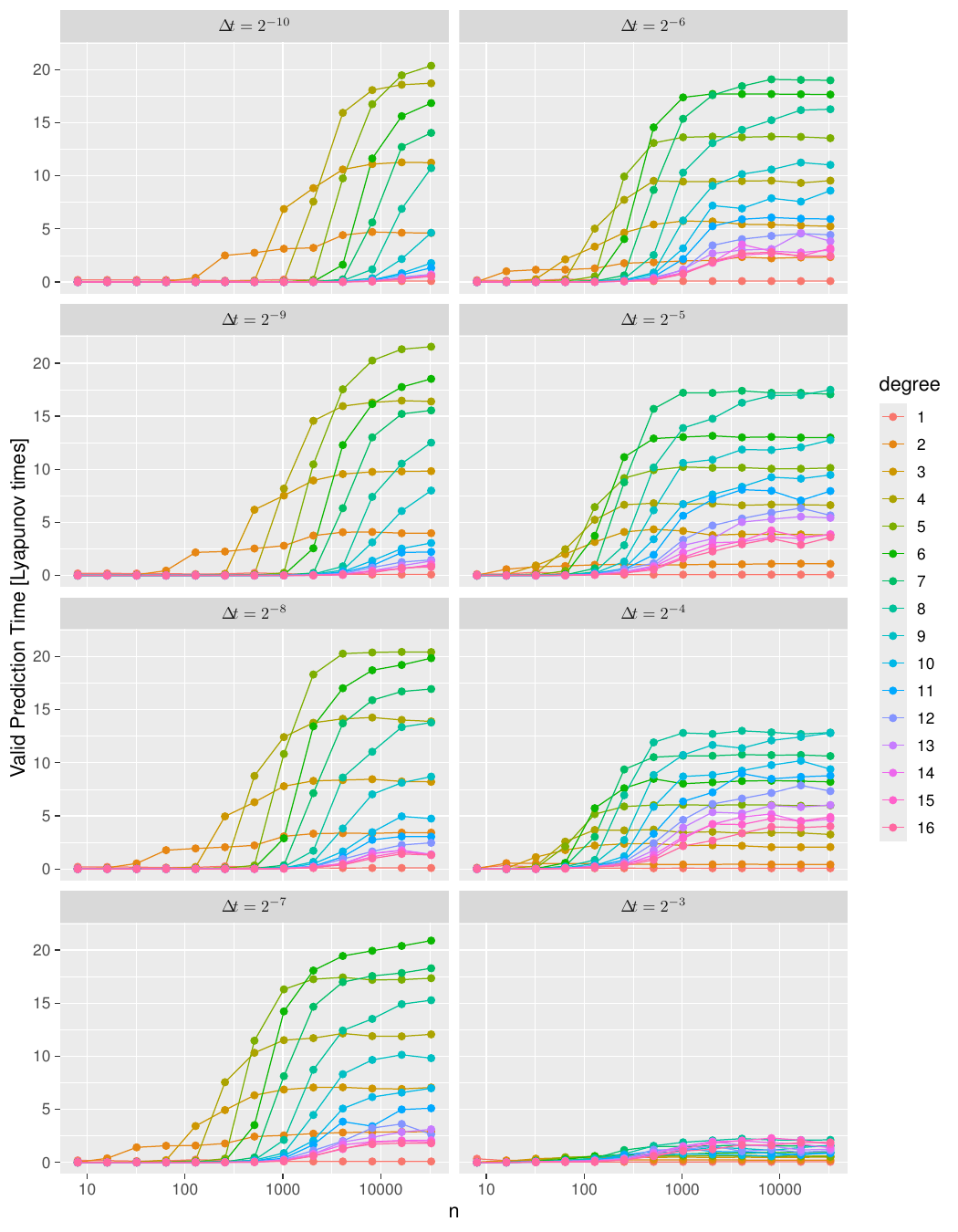}
\end{center}
\caption{\textbf{All Plot for L63, ddd, normalize full, test sequential}. See the beginning of \cref{app:sec:details} for a description.}
\end{figure}

\clearpage
\subsection{L63, sds, normalize full, test sequential}
\begin{figure}[ht!]
\begin{center}
\includegraphics[width=\textwidth]{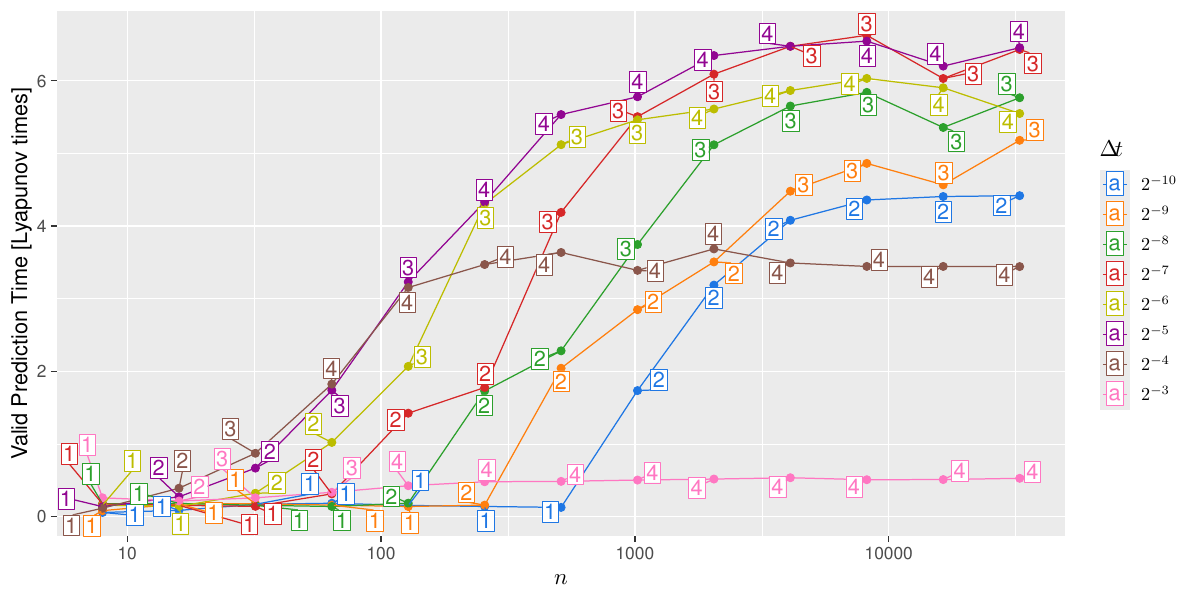}
\end{center}
\caption{\textbf{Best Plot for L63, sds, normalize full, test sequential}. See the beginning of \cref{app:sec:details} for a description.}
\end{figure}
\begin{table}[ht!]
\begin{center}
\fontsize{8.0pt}{10pt}\selectfont
\fontfamily{phv}\selectfont
\renewcommand{\arraystretch}{1.05}
\setlength{\tabcolsep}{0.5em}
\rowcolors{2}{gray!20}{white}
\begin{tabular}{crrrrrrrr}
\toprule
$n$ & \multicolumn{8}{c}{$\stepsize$} \\ 
\cmidrule(lr){1-1} \cmidrule(lr){2-9}
 & $2^{-10}$ & $2^{-9}$ & $2^{-8}$ & $2^{-7}$ & $2^{-6}$ & $2^{-5}$ & $2^{-4}$ & $2^{-3}$ \\ 
\midrule\addlinespace[2.5pt]
$2^{3} = 8$ & {\cellcolor[HTML]{450457}{\textcolor[HTML]{FFFFFF}{0.1 (1)}}} & {\cellcolor[HTML]{450559}{\textcolor[HTML]{FFFFFF}{0.1 (1)}}} & {\cellcolor[HTML]{460B5E}{\textcolor[HTML]{FFFFFF}{0.2 (1)}}} & {\cellcolor[HTML]{460B5E}{\textcolor[HTML]{FFFFFF}{0.2 (1)}}} & {\cellcolor[HTML]{460A5D}{\textcolor[HTML]{FFFFFF}{0.2 (1)}}} & {\cellcolor[HTML]{46095C}{\textcolor[HTML]{FFFFFF}{0.1 (1)}}} & {\cellcolor[HTML]{46075B}{\textcolor[HTML]{FFFFFF}{0.1 (1)}}} & {\cellcolor[HTML]{471063}{\textcolor[HTML]{FFFFFF}{0.3 (1)}}} \\ 
$2^{4} = 16$ & {\cellcolor[HTML]{450659}{\textcolor[HTML]{FFFFFF}{0.1 (1)}}} & {\cellcolor[HTML]{460B5E}{\textcolor[HTML]{FFFFFF}{0.2 (1)}}} & {\cellcolor[HTML]{460B5E}{\textcolor[HTML]{FFFFFF}{0.2 (1)}}} & {\cellcolor[HTML]{460A5D}{\textcolor[HTML]{FFFFFF}{0.2 (1)}}} & {\cellcolor[HTML]{46095C}{\textcolor[HTML]{FFFFFF}{0.1 (1)}}} & {\cellcolor[HTML]{471063}{\textcolor[HTML]{FFFFFF}{0.3 (2)}}} & {\cellcolor[HTML]{481769}{\textcolor[HTML]{FFFFFF}{0.4 (2)}}} & {\cellcolor[HTML]{470D60}{\textcolor[HTML]{FFFFFF}{0.2 (2)}}} \\ 
$2^{5} = 32$ & {\cellcolor[HTML]{460B5D}{\textcolor[HTML]{FFFFFF}{0.2 (1)}}} & {\cellcolor[HTML]{460B5E}{\textcolor[HTML]{FFFFFF}{0.2 (1)}}} & {\cellcolor[HTML]{460A5D}{\textcolor[HTML]{FFFFFF}{0.2 (1)}}} & {\cellcolor[HTML]{46095C}{\textcolor[HTML]{FFFFFF}{0.1 (1)}}} & {\cellcolor[HTML]{471366}{\textcolor[HTML]{FFFFFF}{0.3 (2)}}} & {\cellcolor[HTML]{482576}{\textcolor[HTML]{FFFFFF}{0.7 (2)}}} & {\cellcolor[HTML]{472F7D}{\textcolor[HTML]{FFFFFF}{0.9 (3)}}} & {\cellcolor[HTML]{471063}{\textcolor[HTML]{FFFFFF}{0.3 (3)}}} \\
\bottomrule
\end{tabular}
\end{center}
\caption{\textbf{Best Table for L63, sds, normalize full, test sequential}. See the beginning of \cref{app:sec:details} for a description.}
\end{table}
\begin{figure}[ht!]
\begin{center}
\includegraphics[width=\textwidth]{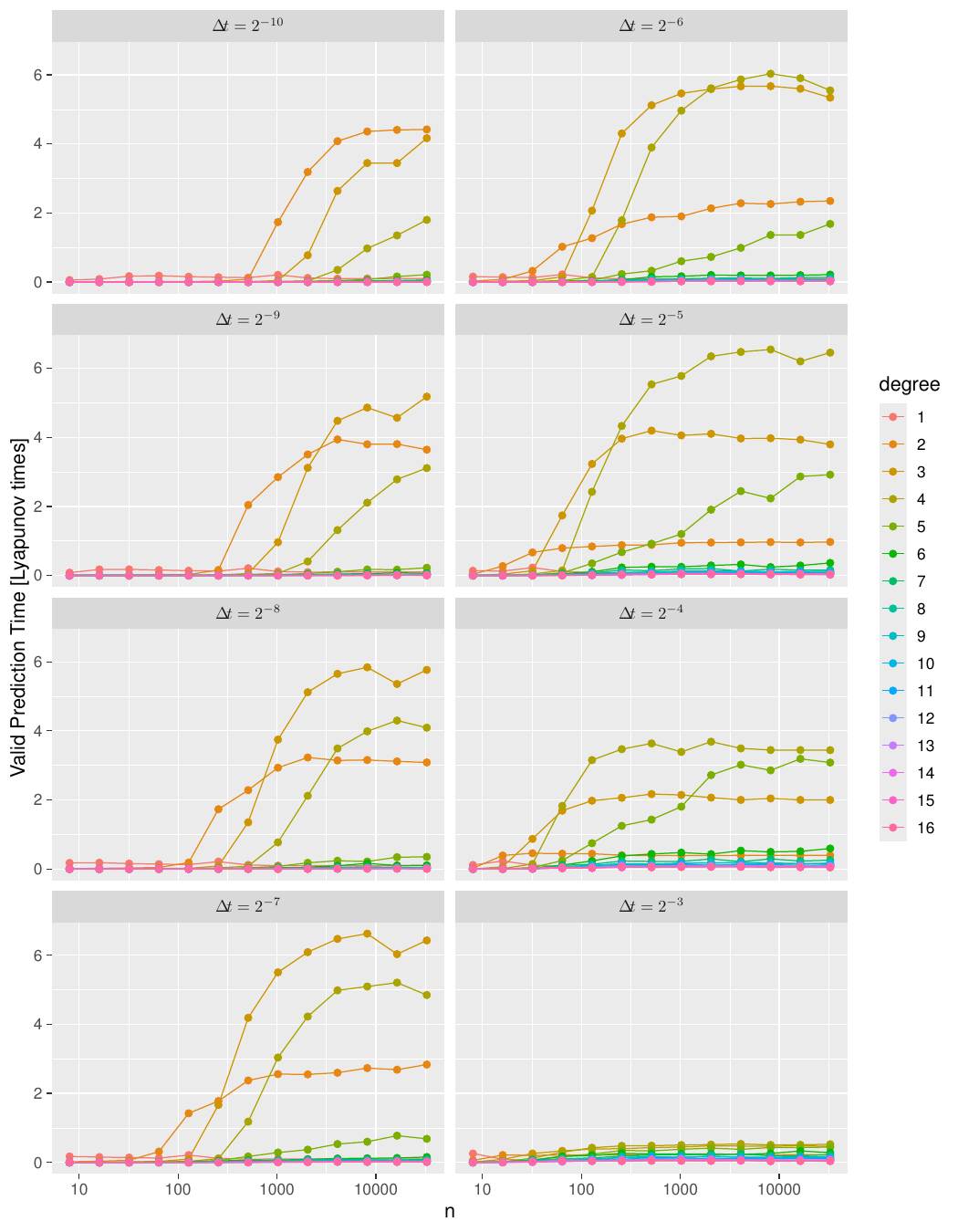}
\end{center}
\caption{\textbf{All Plot for L63, sds, normalize full, test sequential}. See the beginning of \cref{app:sec:details} for a description.}
\end{figure}

\clearpage
\subsection{L63, sdd, normalize full, test sequential}
\begin{figure}[ht!]
\begin{center}
\includegraphics[width=\textwidth]{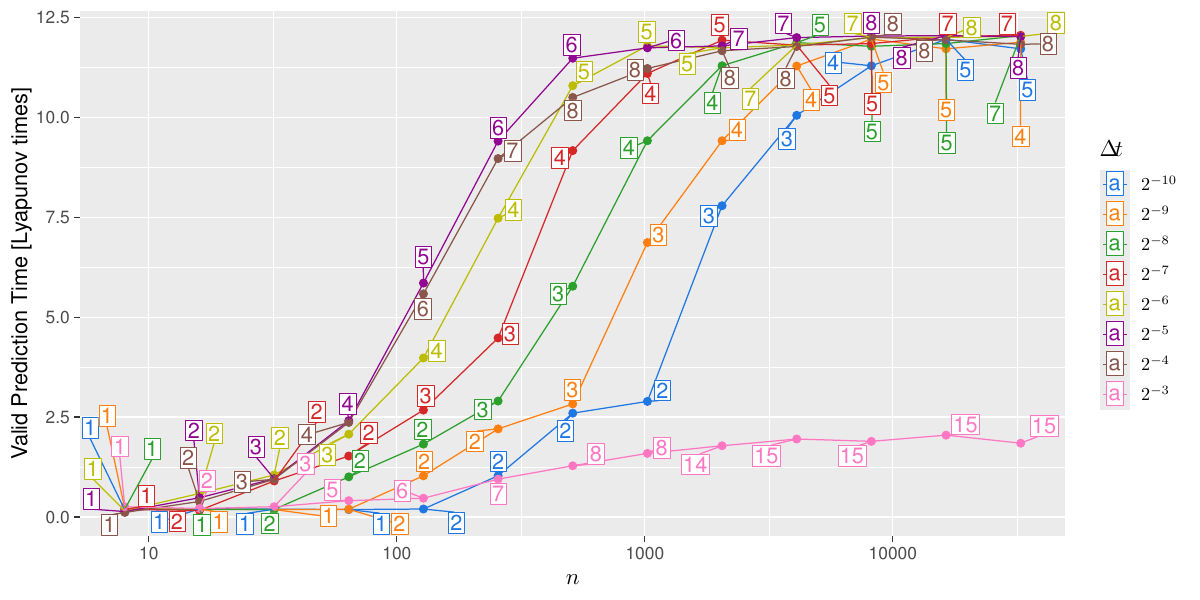}
\end{center}
\caption{\textbf{Best Plot for L63, sdd, normalize full, test sequential}. See the beginning of \cref{app:sec:details} for a description.}
\end{figure}
\begin{table}[ht!]
\begin{center}
\fontsize{8.0pt}{10pt}\selectfont
\fontfamily{phv}\selectfont
\renewcommand{\arraystretch}{1.05}
\setlength{\tabcolsep}{0.5em}
\rowcolors{2}{gray!20}{white}
\begin{tabular}{crrrrrrrr}
\toprule
$n$ & \multicolumn{8}{c}{$\stepsize$} \\ 
\cmidrule(lr){1-1} \cmidrule(lr){2-9}
 & $2^{-10}$ & $2^{-9}$ & $2^{-8}$ & $2^{-7}$ & $2^{-6}$ & $2^{-5}$ & $2^{-4}$ & $2^{-3}$ \\ 
\midrule\addlinespace[2.5pt]
$2^{3} = 8$ & {\cellcolor[HTML]{46075A}{\textcolor[HTML]{FFFFFF}{0.2 (1)}}} & {\cellcolor[HTML]{46075B}{\textcolor[HTML]{FFFFFF}{0.2 (1)}}} & {\cellcolor[HTML]{46075A}{\textcolor[HTML]{FFFFFF}{0.2 (1)}}} & {\cellcolor[HTML]{46075A}{\textcolor[HTML]{FFFFFF}{0.2 (1)}}} & {\cellcolor[HTML]{450659}{\textcolor[HTML]{FFFFFF}{0.2 (1)}}} & {\cellcolor[HTML]{450559}{\textcolor[HTML]{FFFFFF}{0.1 (1)}}} & {\cellcolor[HTML]{450458}{\textcolor[HTML]{FFFFFF}{0.1 (1)}}} & {\cellcolor[HTML]{46095C}{\textcolor[HTML]{FFFFFF}{0.3 (1)}}} \\ 
$2^{4} = 16$ & {\cellcolor[HTML]{46075B}{\textcolor[HTML]{FFFFFF}{0.2 (1)}}} & {\cellcolor[HTML]{46075A}{\textcolor[HTML]{FFFFFF}{0.2 (1)}}} & {\cellcolor[HTML]{46075A}{\textcolor[HTML]{FFFFFF}{0.2 (1)}}} & {\cellcolor[HTML]{46065A}{\textcolor[HTML]{FFFFFF}{0.2 (2)}}} & {\cellcolor[HTML]{471366}{\textcolor[HTML]{FFFFFF}{0.6 (2)}}} & {\cellcolor[HTML]{471063}{\textcolor[HTML]{FFFFFF}{0.5 (2)}}} & {\cellcolor[HTML]{470D60}{\textcolor[HTML]{FFFFFF}{0.4 (2)}}} & {\cellcolor[HTML]{46085B}{\textcolor[HTML]{FFFFFF}{0.2 (2)}}} \\ 
$2^{5} = 32$ & {\cellcolor[HTML]{46075A}{\textcolor[HTML]{FFFFFF}{0.2 (1)}}} & {\cellcolor[HTML]{46075A}{\textcolor[HTML]{FFFFFF}{0.2 (1)}}} & {\cellcolor[HTML]{46075A}{\textcolor[HTML]{FFFFFF}{0.2 (2)}}} & {\cellcolor[HTML]{481C6E}{\textcolor[HTML]{FFFFFF}{0.9 (2)}}} & {\cellcolor[HTML]{482072}{\textcolor[HTML]{FFFFFF}{1.1 (2)}}} & {\cellcolor[HTML]{481E6F}{\textcolor[HTML]{FFFFFF}{1.0 (3)}}} & {\cellcolor[HTML]{481D6F}{\textcolor[HTML]{FFFFFF}{0.9 (3)}}} & {\cellcolor[HTML]{46095C}{\textcolor[HTML]{FFFFFF}{0.3 (3)}}} \\
\bottomrule
\end{tabular}
\end{center}
\caption{\textbf{Best Table for L63, sdd, normalize full, test sequential}. See the beginning of \cref{app:sec:details} for a description.}
\end{table}
\begin{figure}[ht!]
\begin{center}
\includegraphics[width=\textwidth]{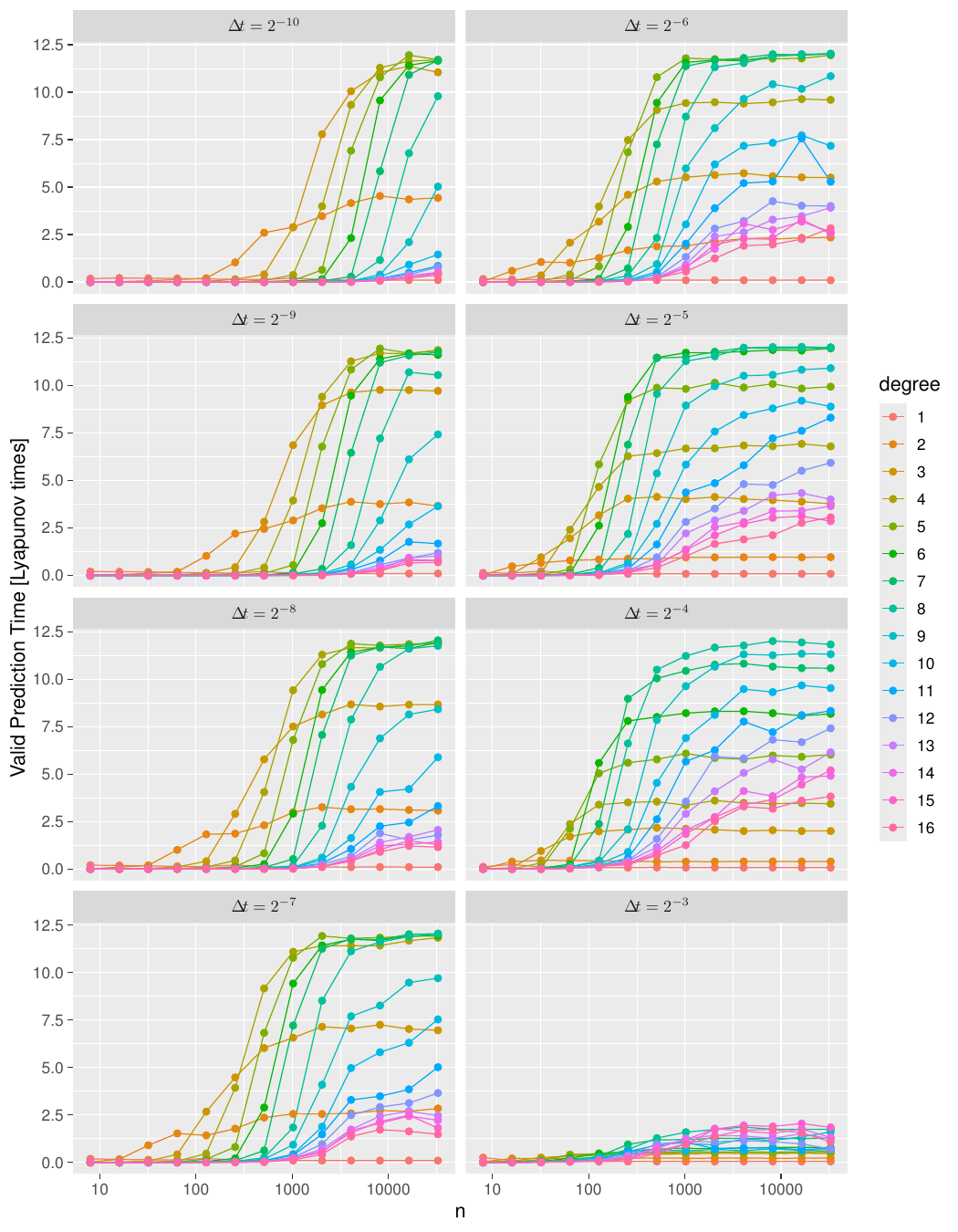}
\end{center}
\caption{\textbf{All Plot for L63, sdd, normalize full, test sequential}. See the beginning of \cref{app:sec:details} for a description.}
\end{figure}

\clearpage
\subsection{L63, mds, normalize full, test sequential}
\begin{figure}[ht!]
\begin{center}
\includegraphics[width=\textwidth]{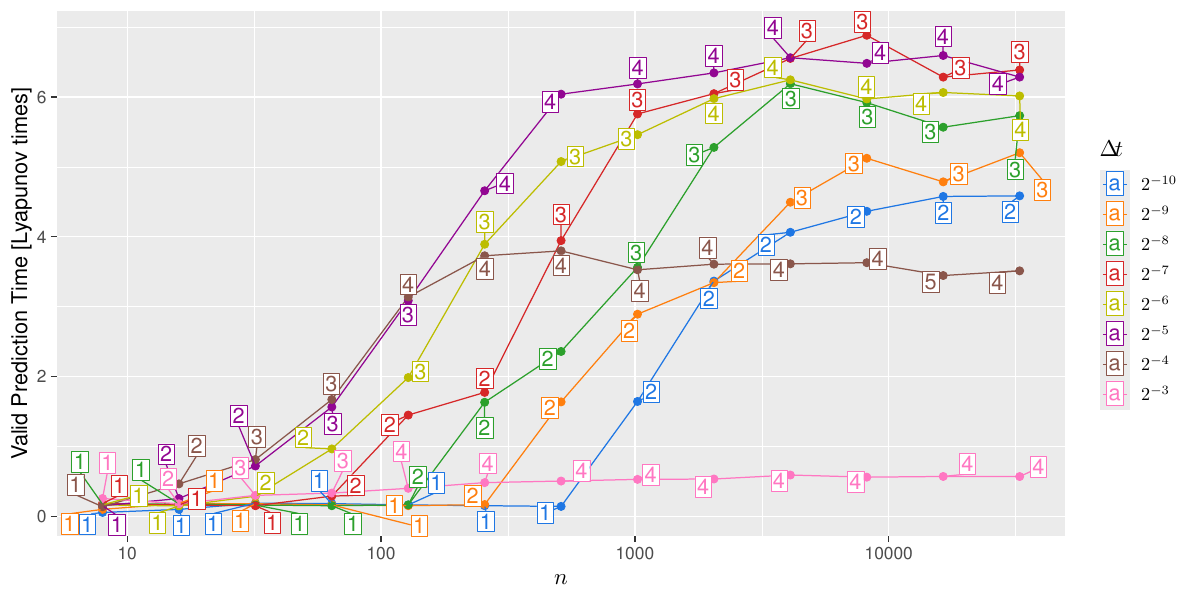}
\end{center}
\caption{\textbf{Best Plot for L63, mds, normalize full, test sequential}. See the beginning of \cref{app:sec:details} for a description.}
\end{figure}
\begin{table}[ht!]
\begin{center}
\fontsize{8.0pt}{10pt}\selectfont
\fontfamily{phv}\selectfont
\renewcommand{\arraystretch}{1.05}
\setlength{\tabcolsep}{0.5em}
\rowcolors{2}{gray!20}{white}
\begin{tabular}{crrrrrrrr}
\toprule
$n$ & \multicolumn{8}{c}{$\stepsize$} \\ 
\cmidrule(lr){1-1} \cmidrule(lr){2-9}
 & $2^{-10}$ & $2^{-9}$ & $2^{-8}$ & $2^{-7}$ & $2^{-6}$ & $2^{-5}$ & $2^{-4}$ & $2^{-3}$ \\ 
\midrule\addlinespace[2.5pt]
$2^{3} = 8$ & {\cellcolor[HTML]{450457}{\textcolor[HTML]{FFFFFF}{0.1 (1)}}} & {\cellcolor[HTML]{46075A}{\textcolor[HTML]{FFFFFF}{0.1 (1)}}} & {\cellcolor[HTML]{460A5D}{\textcolor[HTML]{FFFFFF}{0.2 (1)}}} & {\cellcolor[HTML]{460B5E}{\textcolor[HTML]{FFFFFF}{0.2 (1)}}} & {\cellcolor[HTML]{460A5D}{\textcolor[HTML]{FFFFFF}{0.2 (1)}}} & {\cellcolor[HTML]{46095D}{\textcolor[HTML]{FFFFFF}{0.2 (1)}}} & {\cellcolor[HTML]{46085C}{\textcolor[HTML]{FFFFFF}{0.1 (1)}}} & {\cellcolor[HTML]{470F62}{\textcolor[HTML]{FFFFFF}{0.3 (1)}}} \\ 
$2^{4} = 16$ & {\cellcolor[HTML]{46065A}{\textcolor[HTML]{FFFFFF}{0.1 (1)}}} & {\cellcolor[HTML]{460B5E}{\textcolor[HTML]{FFFFFF}{0.2 (1)}}} & {\cellcolor[HTML]{460B5E}{\textcolor[HTML]{FFFFFF}{0.2 (1)}}} & {\cellcolor[HTML]{460A5D}{\textcolor[HTML]{FFFFFF}{0.2 (1)}}} & {\cellcolor[HTML]{46095D}{\textcolor[HTML]{FFFFFF}{0.2 (1)}}} & {\cellcolor[HTML]{470F62}{\textcolor[HTML]{FFFFFF}{0.3 (2)}}} & {\cellcolor[HTML]{481A6C}{\textcolor[HTML]{FFFFFF}{0.5 (2)}}} & {\cellcolor[HTML]{460B5E}{\textcolor[HTML]{FFFFFF}{0.2 (2)}}} \\ 
$2^{5} = 32$ & {\cellcolor[HTML]{460B5E}{\textcolor[HTML]{FFFFFF}{0.2 (1)}}} & {\cellcolor[HTML]{460B5E}{\textcolor[HTML]{FFFFFF}{0.2 (1)}}} & {\cellcolor[HTML]{460A5D}{\textcolor[HTML]{FFFFFF}{0.2 (1)}}} & {\cellcolor[HTML]{46095D}{\textcolor[HTML]{FFFFFF}{0.2 (1)}}} & {\cellcolor[HTML]{471164}{\textcolor[HTML]{FFFFFF}{0.3 (2)}}} & {\cellcolor[HTML]{482677}{\textcolor[HTML]{FFFFFF}{0.7 (2)}}} & {\cellcolor[HTML]{472A7A}{\textcolor[HTML]{FFFFFF}{0.8 (3)}}} & {\cellcolor[HTML]{471164}{\textcolor[HTML]{FFFFFF}{0.3 (3)}}} \\
\bottomrule
\end{tabular}
\end{center}
\caption{\textbf{Best Table for L63, mds, normalize full, test sequential}. See the beginning of \cref{app:sec:details} for a description.}
\end{table}
\begin{figure}[ht!]
\begin{center}
\includegraphics[width=\textwidth]{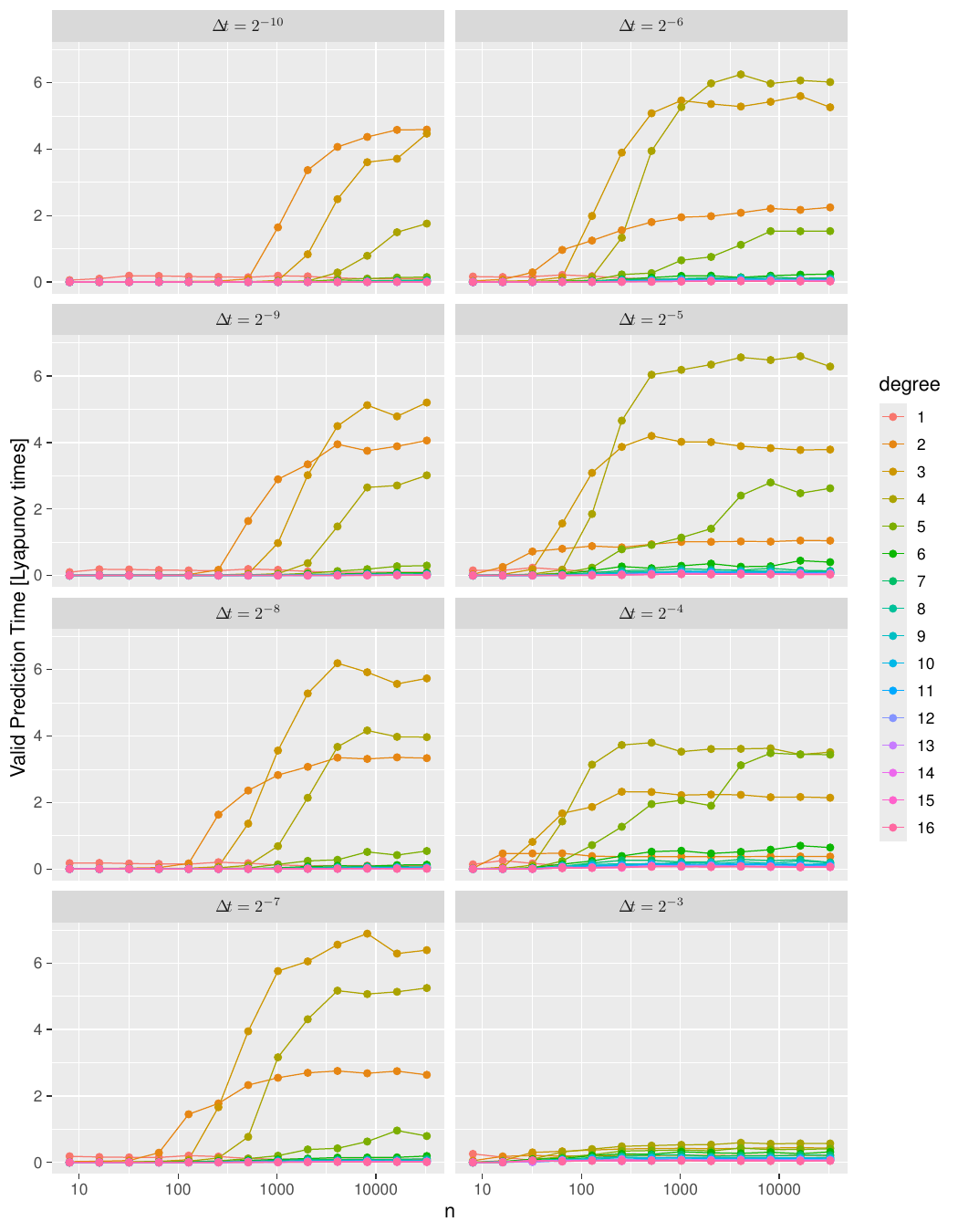}
\end{center}
\caption{\textbf{All Plot for L63, mds, normalize full, test sequential}. See the beginning of \cref{app:sec:details} for a description.}
\end{figure}

\clearpage
\subsection{L63, mdd, normalize full, test sequential}
\begin{figure}[ht!]
\begin{center}
\includegraphics[width=\textwidth]{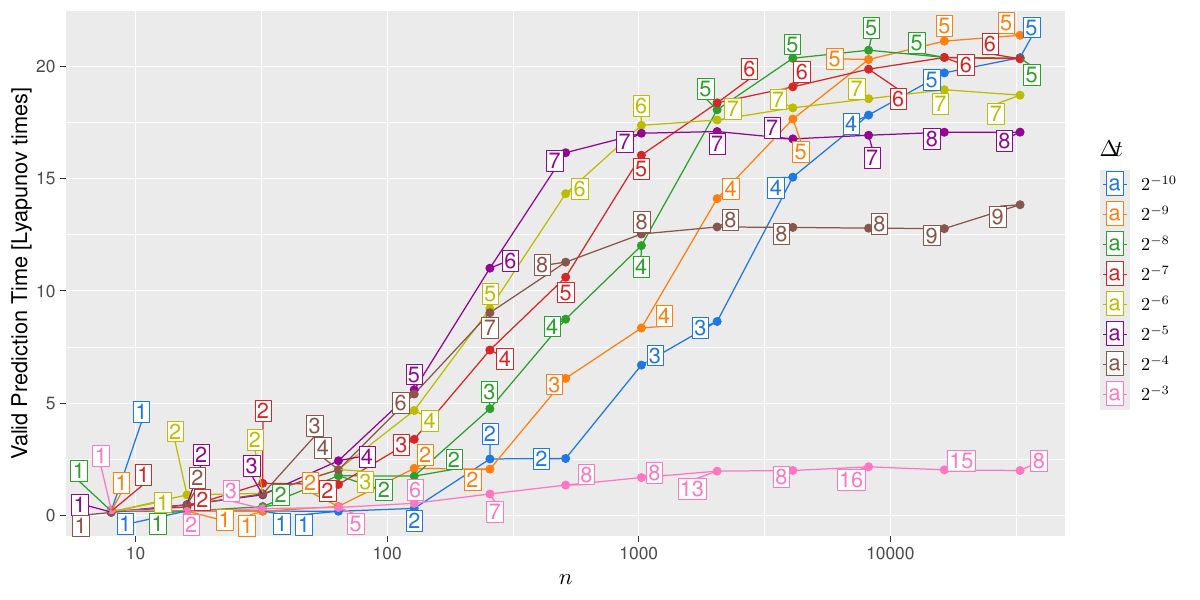}
\end{center}
\caption{\textbf{Best Plot for L63, mdd, normalize full, test sequential}. See the beginning of \cref{app:sec:details} for a description.}
\end{figure}
\begin{table}[ht!]
\begin{center}
\fontsize{8.0pt}{10pt}\selectfont
\fontfamily{phv}\selectfont
\renewcommand{\arraystretch}{1.05}
\setlength{\tabcolsep}{0.5em}
\rowcolors{2}{gray!20}{white}
\begin{tabular}{crrrrrrrr}
\toprule
$n$ & \multicolumn{8}{c}{$\stepsize$} \\ 
\cmidrule(lr){1-1} \cmidrule(lr){2-9}
 & $2^{-10}$ & $2^{-9}$ & $2^{-8}$ & $2^{-7}$ & $2^{-6}$ & $2^{-5}$ & $2^{-4}$ & $2^{-3}$ \\ 
\midrule\addlinespace[2.5pt]
$2^{3} = 8$ & {\cellcolor[HTML]{450458}{\textcolor[HTML]{FFFFFF}{0.2 (1)}}} & {\cellcolor[HTML]{450458}{\textcolor[HTML]{FFFFFF}{0.2 (1)}}} & {\cellcolor[HTML]{450458}{\textcolor[HTML]{FFFFFF}{0.2 (1)}}} & {\cellcolor[HTML]{450458}{\textcolor[HTML]{FFFFFF}{0.2 (1)}}} & {\cellcolor[HTML]{450457}{\textcolor[HTML]{FFFFFF}{0.2 (1)}}} & {\cellcolor[HTML]{450457}{\textcolor[HTML]{FFFFFF}{0.2 (1)}}} & {\cellcolor[HTML]{450357}{\textcolor[HTML]{FFFFFF}{0.1 (1)}}} & {\cellcolor[HTML]{450559}{\textcolor[HTML]{FFFFFF}{0.3 (1)}}} \\ 
$2^{4} = 16$ & {\cellcolor[HTML]{450458}{\textcolor[HTML]{FFFFFF}{0.2 (1)}}} & {\cellcolor[HTML]{450458}{\textcolor[HTML]{FFFFFF}{0.2 (1)}}} & {\cellcolor[HTML]{450458}{\textcolor[HTML]{FFFFFF}{0.2 (1)}}} & {\cellcolor[HTML]{46085B}{\textcolor[HTML]{FFFFFF}{0.4 (2)}}} & {\cellcolor[HTML]{471164}{\textcolor[HTML]{FFFFFF}{0.9 (2)}}} & {\cellcolor[HTML]{460A5D}{\textcolor[HTML]{FFFFFF}{0.5 (2)}}} & {\cellcolor[HTML]{46095D}{\textcolor[HTML]{FFFFFF}{0.5 (2)}}} & {\cellcolor[HTML]{450457}{\textcolor[HTML]{FFFFFF}{0.2 (2)}}} \\ 
$2^{5} = 32$ & {\cellcolor[HTML]{450458}{\textcolor[HTML]{FFFFFF}{0.2 (1)}}} & {\cellcolor[HTML]{450458}{\textcolor[HTML]{FFFFFF}{0.2 (1)}}} & {\cellcolor[HTML]{46085C}{\textcolor[HTML]{FFFFFF}{0.4 (2)}}} & {\cellcolor[HTML]{481A6C}{\textcolor[HTML]{FFFFFF}{1.4 (2)}}} & {\cellcolor[HTML]{471365}{\textcolor[HTML]{FFFFFF}{1.0 (2)}}} & {\cellcolor[HTML]{471164}{\textcolor[HTML]{FFFFFF}{0.9 (3)}}} & {\cellcolor[HTML]{471265}{\textcolor[HTML]{FFFFFF}{1.0 (3)}}} & {\cellcolor[HTML]{46065A}{\textcolor[HTML]{FFFFFF}{0.3 (3)}}} \\
\bottomrule
\end{tabular}
\end{center}
\caption{\textbf{Best Table for L63, mdd, normalize full, test sequential}. See the beginning of \cref{app:sec:details} for a description.}
\end{table}
\begin{figure}[ht!]
\begin{center}
\includegraphics[width=\textwidth]{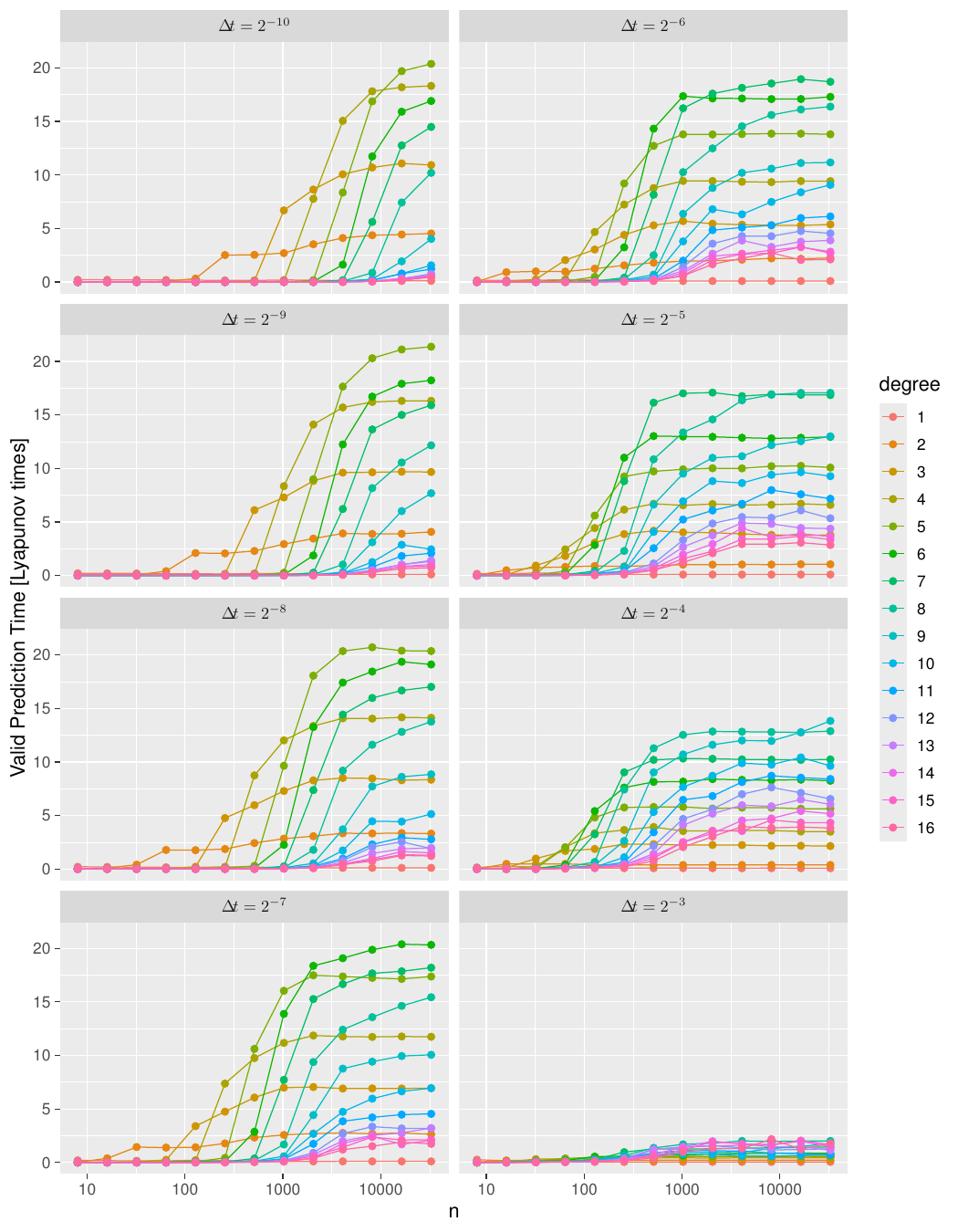}
\end{center}
\caption{\textbf{All Plot for L63, mdd, normalize full, test sequential}. See the beginning of \cref{app:sec:details} for a description.}
\end{figure}

\clearpage
\subsection{L63, dsd, normalize full, test sequential}
\begin{figure}[ht!]
\begin{center}
\includegraphics[width=\textwidth]{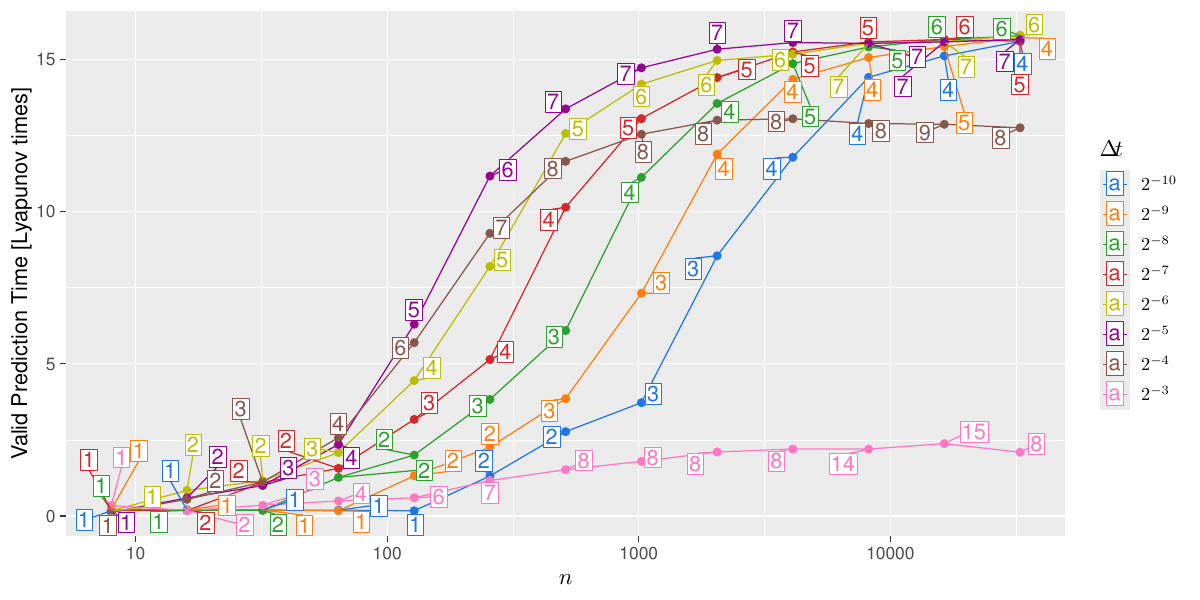}
\end{center}
\caption{\textbf{Best Plot for L63, dsd, normalize full, test sequential}. See the beginning of \cref{app:sec:details} for a description.}
\end{figure}
\begin{table}[ht!]
\begin{center}
\fontsize{8.0pt}{10pt}\selectfont
\fontfamily{phv}\selectfont
\renewcommand{\arraystretch}{1.05}
\setlength{\tabcolsep}{0.5em}
\rowcolors{2}{gray!20}{white}
\begin{tabular}{crrrrrrrr}
\toprule
$n$ & \multicolumn{8}{c}{$\stepsize$} \\ 
\cmidrule(lr){1-1} \cmidrule(lr){2-9}
 & $2^{-10}$ & $2^{-9}$ & $2^{-8}$ & $2^{-7}$ & $2^{-6}$ & $2^{-5}$ & $2^{-4}$ & $2^{-3}$ \\ 
\midrule\addlinespace[2.5pt]
$2^{3} = 8$ & {\cellcolor[HTML]{450559}{\textcolor[HTML]{FFFFFF}{0.2 (1)}}} & {\cellcolor[HTML]{450659}{\textcolor[HTML]{FFFFFF}{0.2 (1)}}} & {\cellcolor[HTML]{450559}{\textcolor[HTML]{FFFFFF}{0.2 (1)}}} & {\cellcolor[HTML]{450559}{\textcolor[HTML]{FFFFFF}{0.2 (1)}}} & {\cellcolor[HTML]{450559}{\textcolor[HTML]{FFFFFF}{0.2 (1)}}} & {\cellcolor[HTML]{450457}{\textcolor[HTML]{FFFFFF}{0.1 (1)}}} & {\cellcolor[HTML]{450457}{\textcolor[HTML]{FFFFFF}{0.1 (1)}}} & {\cellcolor[HTML]{46095D}{\textcolor[HTML]{FFFFFF}{0.3 (1)}}} \\ 
$2^{4} = 16$ & {\cellcolor[HTML]{450659}{\textcolor[HTML]{FFFFFF}{0.2 (1)}}} & {\cellcolor[HTML]{450559}{\textcolor[HTML]{FFFFFF}{0.2 (1)}}} & {\cellcolor[HTML]{450559}{\textcolor[HTML]{FFFFFF}{0.2 (1)}}} & {\cellcolor[HTML]{450559}{\textcolor[HTML]{FFFFFF}{0.2 (2)}}} & {\cellcolor[HTML]{481568}{\textcolor[HTML]{FFFFFF}{0.8 (2)}}} & {\cellcolor[HTML]{470F62}{\textcolor[HTML]{FFFFFF}{0.6 (2)}}} & {\cellcolor[HTML]{470E61}{\textcolor[HTML]{FFFFFF}{0.5 (2)}}} & {\cellcolor[HTML]{450559}{\textcolor[HTML]{FFFFFF}{0.2 (2)}}} \\ 
$2^{5} = 32$ & {\cellcolor[HTML]{450559}{\textcolor[HTML]{FFFFFF}{0.2 (1)}}} & {\cellcolor[HTML]{450559}{\textcolor[HTML]{FFFFFF}{0.2 (1)}}} & {\cellcolor[HTML]{450559}{\textcolor[HTML]{FFFFFF}{0.2 (2)}}} & {\cellcolor[HTML]{481B6D}{\textcolor[HTML]{FFFFFF}{1.1 (2)}}} & {\cellcolor[HTML]{481C6E}{\textcolor[HTML]{FFFFFF}{1.2 (2)}}} & {\cellcolor[HTML]{48186A}{\textcolor[HTML]{FFFFFF}{1.0 (3)}}} & {\cellcolor[HTML]{481B6D}{\textcolor[HTML]{FFFFFF}{1.1 (3)}}} & {\cellcolor[HTML]{46095D}{\textcolor[HTML]{FFFFFF}{0.4 (3)}}} \\
\bottomrule
\end{tabular}
\end{center}
\caption{\textbf{Best Table for L63, dsd, normalize full, test sequential}. See the beginning of \cref{app:sec:details} for a description.}
\end{table}
\begin{figure}[ht!]
\begin{center}
\includegraphics[width=\textwidth]{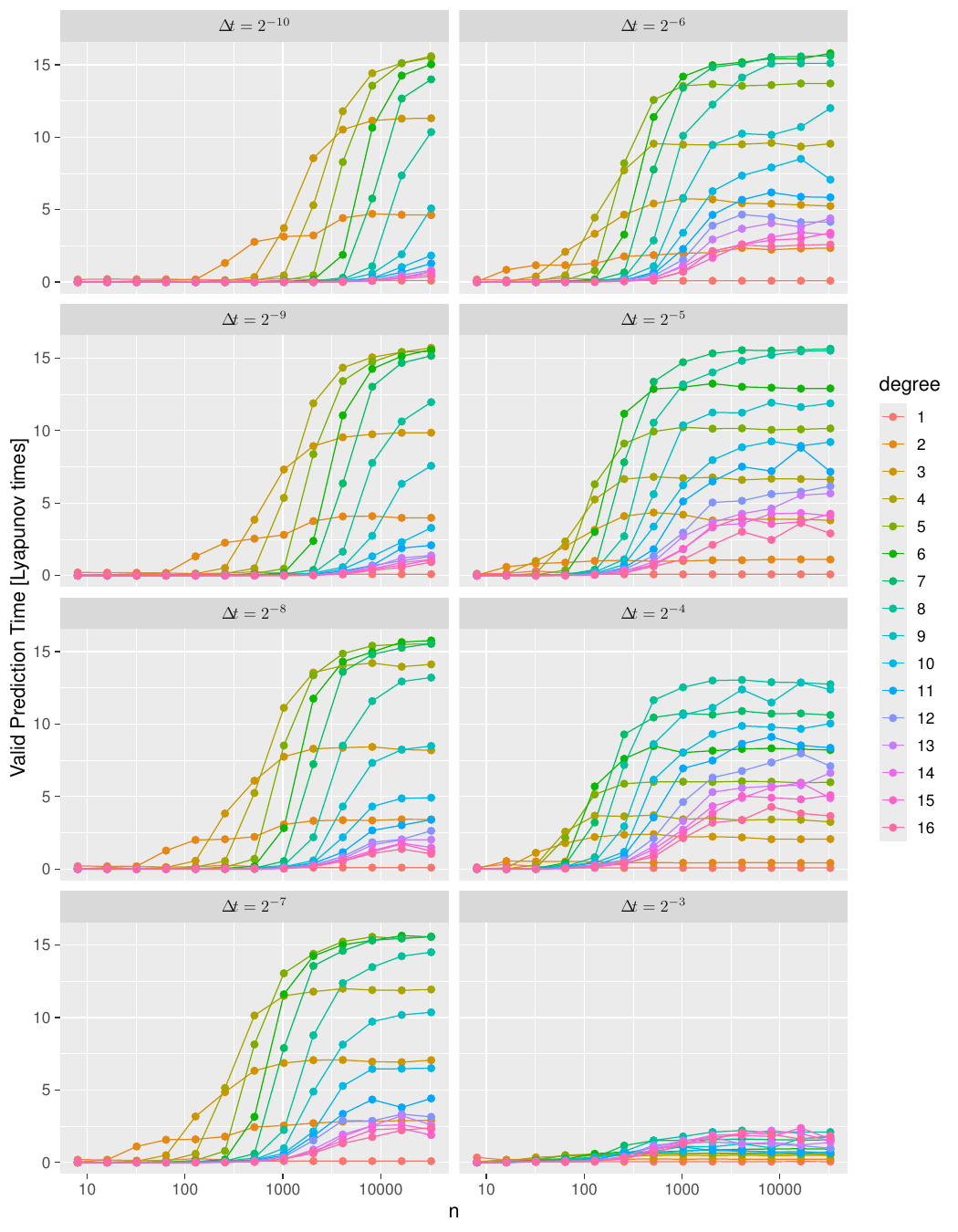}
\end{center}
\caption{\textbf{All Plot for L63, dsd, normalize full, test sequential}. See the beginning of \cref{app:sec:details} for a description.}
\end{figure}

\clearpage
\subsection{L63, msd, normalize full, test sequential}
\begin{figure}[ht!]
\begin{center}
\includegraphics[width=\textwidth]{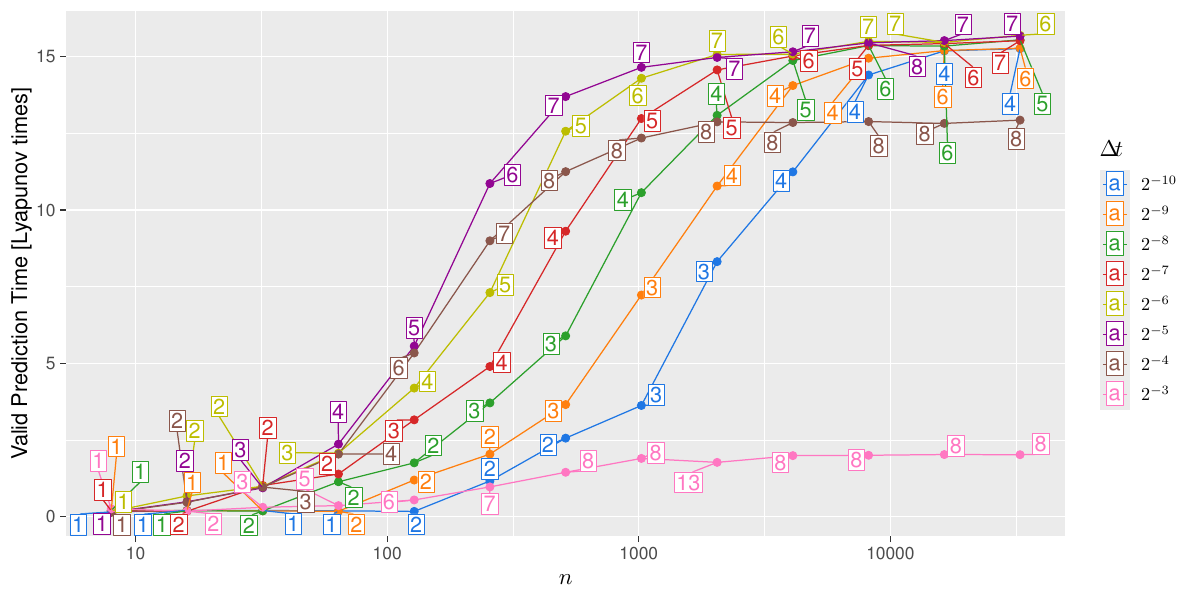}
\end{center}
\caption{\textbf{Best Plot for L63, msd, normalize full, test sequential}. See the beginning of \cref{app:sec:details} for a description.}
\end{figure}
\begin{table}[ht!]
\begin{center}
\fontsize{8.0pt}{10pt}\selectfont
\fontfamily{phv}\selectfont
\renewcommand{\arraystretch}{1.05}
\setlength{\tabcolsep}{0.5em}
\rowcolors{2}{gray!20}{white}
\begin{tabular}{crrrrrrrr}
\toprule
$n$ & \multicolumn{8}{c}{$\stepsize$} \\ 
\cmidrule(lr){1-1} \cmidrule(lr){2-9}
 & $2^{-10}$ & $2^{-9}$ & $2^{-8}$ & $2^{-7}$ & $2^{-6}$ & $2^{-5}$ & $2^{-4}$ & $2^{-3}$ \\ 
\midrule\addlinespace[2.5pt]
$2^{3} = 8$ & {\cellcolor[HTML]{450558}{\textcolor[HTML]{FFFFFF}{0.2 (1)}}} & {\cellcolor[HTML]{450659}{\textcolor[HTML]{FFFFFF}{0.2 (1)}}} & {\cellcolor[HTML]{450659}{\textcolor[HTML]{FFFFFF}{0.2 (1)}}} & {\cellcolor[HTML]{450559}{\textcolor[HTML]{FFFFFF}{0.2 (1)}}} & {\cellcolor[HTML]{450558}{\textcolor[HTML]{FFFFFF}{0.2 (1)}}} & {\cellcolor[HTML]{450458}{\textcolor[HTML]{FFFFFF}{0.2 (1)}}} & {\cellcolor[HTML]{450458}{\textcolor[HTML]{FFFFFF}{0.1 (1)}}} & {\cellcolor[HTML]{46075A}{\textcolor[HTML]{FFFFFF}{0.3 (1)}}} \\ 
$2^{4} = 16$ & {\cellcolor[HTML]{450659}{\textcolor[HTML]{FFFFFF}{0.2 (1)}}} & {\cellcolor[HTML]{450659}{\textcolor[HTML]{FFFFFF}{0.2 (1)}}} & {\cellcolor[HTML]{450559}{\textcolor[HTML]{FFFFFF}{0.2 (1)}}} & {\cellcolor[HTML]{450559}{\textcolor[HTML]{FFFFFF}{0.2 (2)}}} & {\cellcolor[HTML]{471164}{\textcolor[HTML]{FFFFFF}{0.7 (2)}}} & {\cellcolor[HTML]{470D60}{\textcolor[HTML]{FFFFFF}{0.5 (2)}}} & {\cellcolor[HTML]{470C5F}{\textcolor[HTML]{FFFFFF}{0.5 (2)}}} & {\cellcolor[HTML]{450559}{\textcolor[HTML]{FFFFFF}{0.2 (2)}}} \\ 
$2^{5} = 32$ & {\cellcolor[HTML]{450659}{\textcolor[HTML]{FFFFFF}{0.2 (1)}}} & {\cellcolor[HTML]{450559}{\textcolor[HTML]{FFFFFF}{0.2 (1)}}} & {\cellcolor[HTML]{450559}{\textcolor[HTML]{FFFFFF}{0.2 (2)}}} & {\cellcolor[HTML]{48196B}{\textcolor[HTML]{FFFFFF}{1.0 (2)}}} & {\cellcolor[HTML]{48186A}{\textcolor[HTML]{FFFFFF}{1.0 (2)}}} & {\cellcolor[HTML]{481769}{\textcolor[HTML]{FFFFFF}{0.9 (3)}}} & {\cellcolor[HTML]{48186A}{\textcolor[HTML]{FFFFFF}{1.0 (3)}}} & {\cellcolor[HTML]{46085C}{\textcolor[HTML]{FFFFFF}{0.3 (3)}}} \\
\bottomrule
\end{tabular}
\end{center}
\caption{\textbf{Best Table for L63, msd, normalize full, test sequential}. See the beginning of \cref{app:sec:details} for a description.}
\end{table}
\begin{figure}[ht!]
\begin{center}
\includegraphics[width=\textwidth]{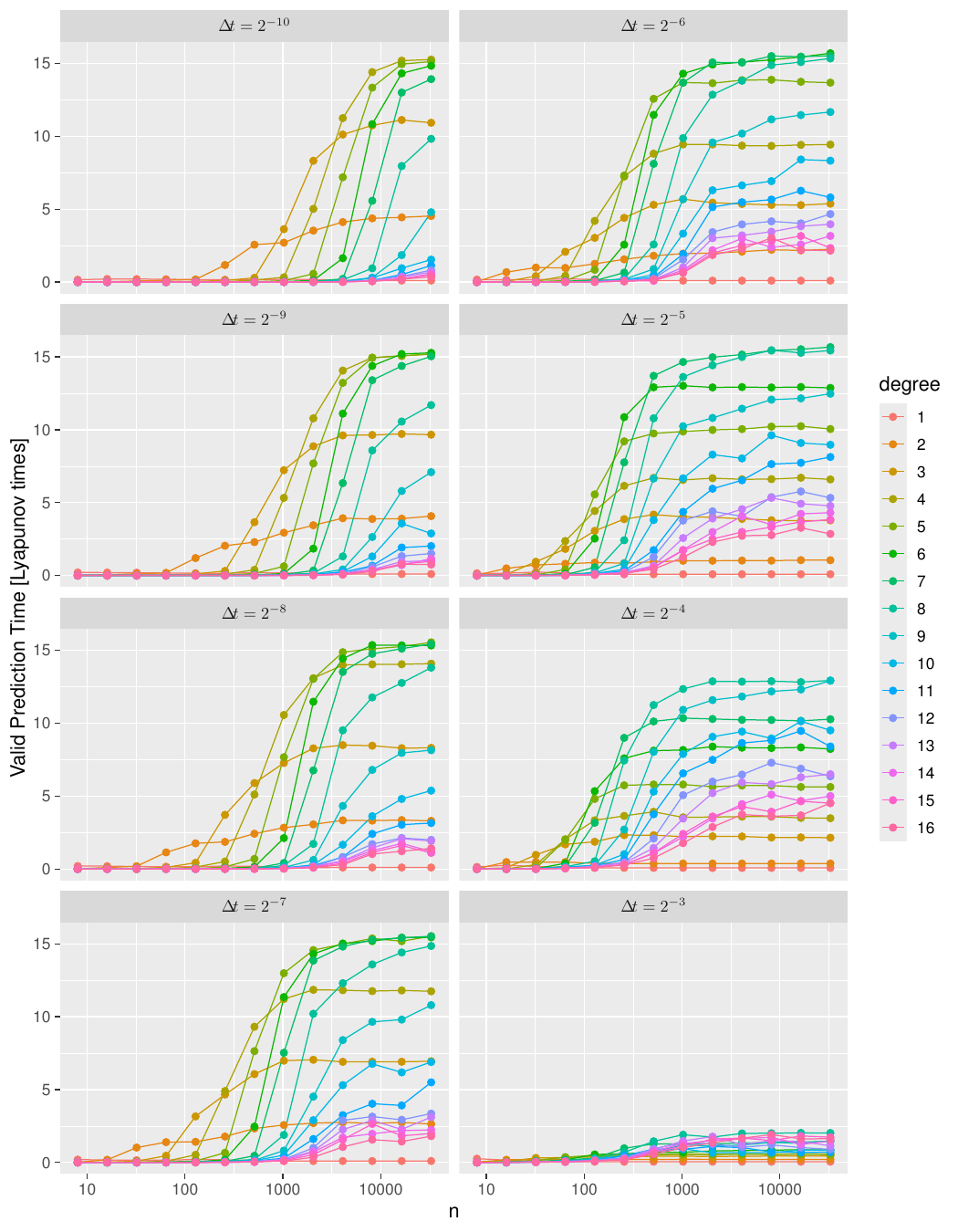}
\end{center}
\caption{\textbf{All Plot for L63, msd, normalize full, test sequential}. See the beginning of \cref{app:sec:details} for a description.}
\end{figure}

\clearpage
\subsection{L63, sdm, normalize none, test sequential}
\begin{figure}[ht!]
\begin{center}
\includegraphics[width=\textwidth]{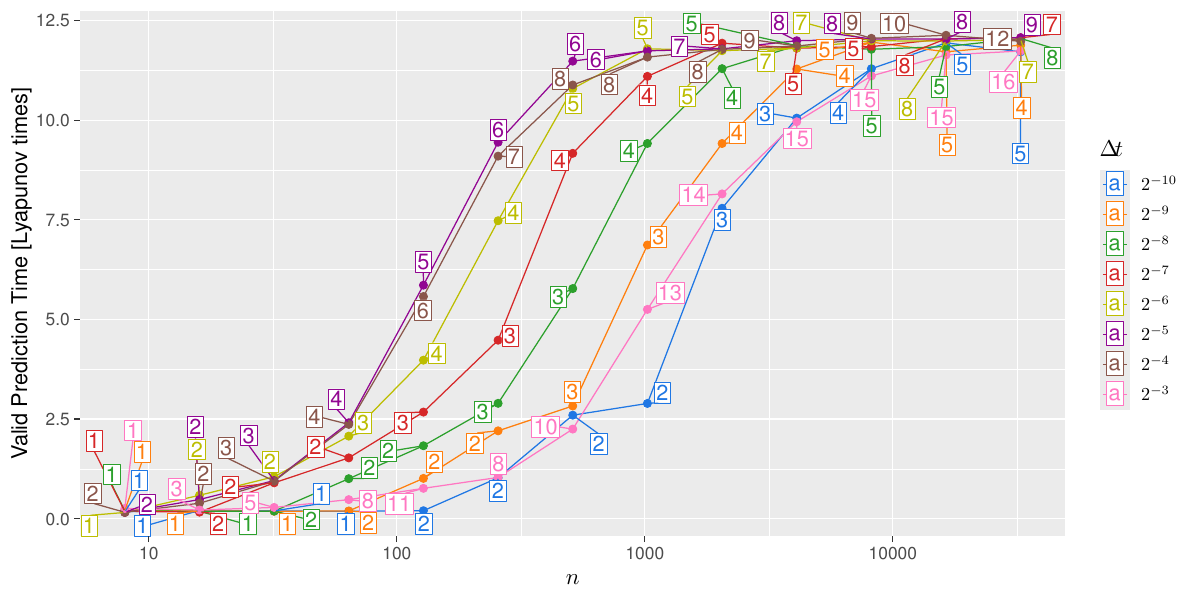}
\end{center}
\caption{\textbf{Best Plot for L63, sdm, normalize none, test sequential}. See the beginning of \cref{app:sec:details} for a description.}
\end{figure}
\begin{table}[ht!]
\begin{center}
\fontsize{8.0pt}{10pt}\selectfont
\fontfamily{phv}\selectfont
\renewcommand{\arraystretch}{1.05}
\setlength{\tabcolsep}{0.5em}
\rowcolors{2}{gray!20}{white}
\begin{tabular}{crrrrrrrr}
\toprule
$n$ & \multicolumn{8}{c}{$\stepsize$} \\ 
\cmidrule(lr){1-1} \cmidrule(lr){2-9}
 & $2^{-10}$ & $2^{-9}$ & $2^{-8}$ & $2^{-7}$ & $2^{-6}$ & $2^{-5}$ & $2^{-4}$ & $2^{-3}$ \\ 
\midrule\addlinespace[2.5pt]
$2^{3} = 8$ & {\cellcolor[HTML]{46065A}{\textcolor[HTML]{FFFFFF}{0.2 (1)}}} & {\cellcolor[HTML]{46075B}{\textcolor[HTML]{FFFFFF}{0.2 (1)}}} & {\cellcolor[HTML]{46075A}{\textcolor[HTML]{FFFFFF}{0.2 (1)}}} & {\cellcolor[HTML]{46075A}{\textcolor[HTML]{FFFFFF}{0.2 (1)}}} & {\cellcolor[HTML]{450659}{\textcolor[HTML]{FFFFFF}{0.2 (1)}}} & {\cellcolor[HTML]{450659}{\textcolor[HTML]{FFFFFF}{0.2 (2)}}} & {\cellcolor[HTML]{450659}{\textcolor[HTML]{FFFFFF}{0.2 (2)}}} & {\cellcolor[HTML]{46095C}{\textcolor[HTML]{FFFFFF}{0.3 (1)}}} \\ 
$2^{4} = 16$ & {\cellcolor[HTML]{46075B}{\textcolor[HTML]{FFFFFF}{0.2 (1)}}} & {\cellcolor[HTML]{46075A}{\textcolor[HTML]{FFFFFF}{0.2 (1)}}} & {\cellcolor[HTML]{46075A}{\textcolor[HTML]{FFFFFF}{0.2 (1)}}} & {\cellcolor[HTML]{450659}{\textcolor[HTML]{FFFFFF}{0.2 (2)}}} & {\cellcolor[HTML]{471366}{\textcolor[HTML]{FFFFFF}{0.6 (2)}}} & {\cellcolor[HTML]{471063}{\textcolor[HTML]{FFFFFF}{0.5 (2)}}} & {\cellcolor[HTML]{470D60}{\textcolor[HTML]{FFFFFF}{0.4 (2)}}} & {\cellcolor[HTML]{46085B}{\textcolor[HTML]{FFFFFF}{0.2 (3)}}} \\ 
$2^{5} = 32$ & {\cellcolor[HTML]{46075A}{\textcolor[HTML]{FFFFFF}{0.2 (1)}}} & {\cellcolor[HTML]{46075A}{\textcolor[HTML]{FFFFFF}{0.2 (1)}}} & {\cellcolor[HTML]{46075A}{\textcolor[HTML]{FFFFFF}{0.2 (2)}}} & {\cellcolor[HTML]{481C6E}{\textcolor[HTML]{FFFFFF}{0.9 (2)}}} & {\cellcolor[HTML]{482071}{\textcolor[HTML]{FFFFFF}{1.1 (2)}}} & {\cellcolor[HTML]{481D6F}{\textcolor[HTML]{FFFFFF}{1.0 (3)}}} & {\cellcolor[HTML]{481D6F}{\textcolor[HTML]{FFFFFF}{0.9 (3)}}} & {\cellcolor[HTML]{460A5D}{\textcolor[HTML]{FFFFFF}{0.3 (5)}}} \\
\bottomrule
\end{tabular}
\end{center}
\caption{\textbf{Best Table for L63, sdm, normalize none, test sequential}. See the beginning of \cref{app:sec:details} for a description.}
\end{table}
\begin{figure}[ht!]
\begin{center}
\includegraphics[width=\textwidth]{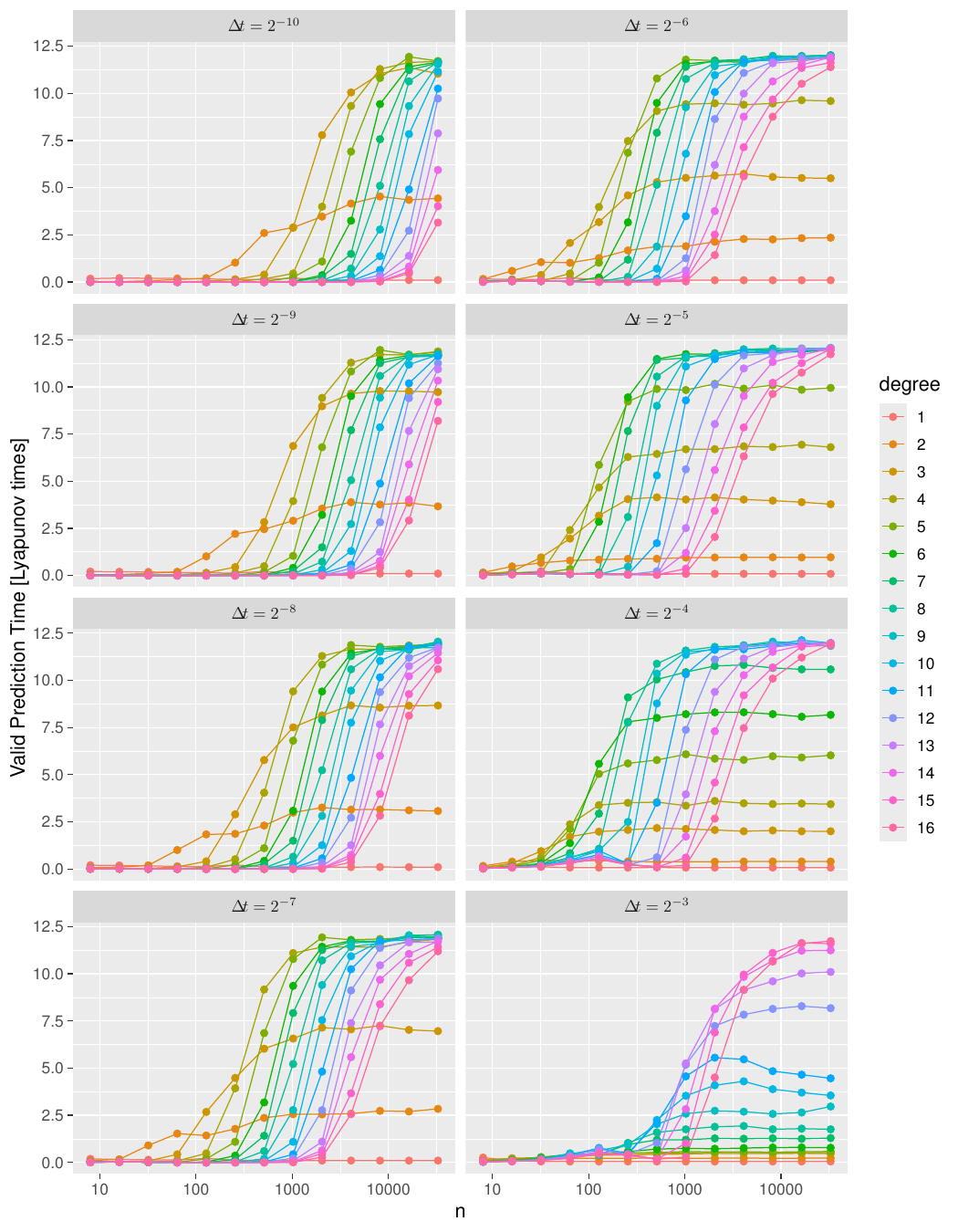}
\end{center}
\caption{\textbf{All Plot for L63, sdm, normalize none, test sequential}. See the beginning of \cref{app:sec:details} for a description.}
\end{figure}

\clearpage
\subsection{L63, ddm, normalize none, test sequential}
\begin{figure}[ht!]
\begin{center}
\includegraphics[width=\textwidth]{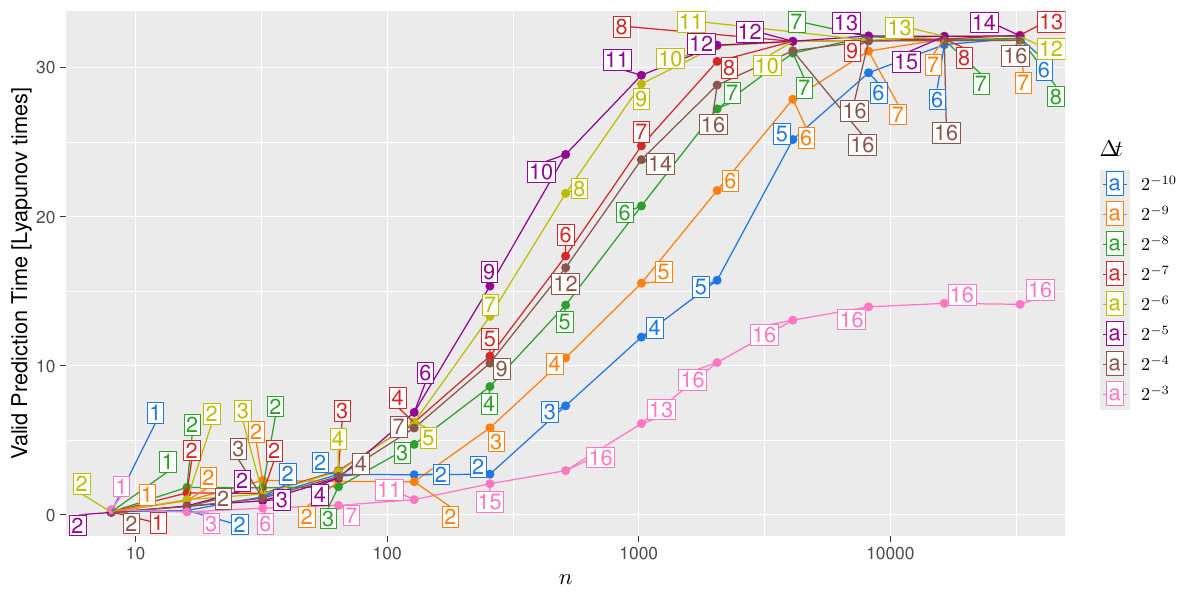}
\end{center}
\caption{\textbf{Best Plot for L63, ddm, normalize none, test sequential}. See the beginning of \cref{app:sec:details} for a description.}
\end{figure}
\begin{table}[ht!]
\begin{center}
\fontsize{8.0pt}{10pt}\selectfont
\fontfamily{phv}\selectfont
\renewcommand{\arraystretch}{1.05}
\setlength{\tabcolsep}{0.5em}
\rowcolors{2}{gray!20}{white}
\begin{tabular}{crrrrrrrr}
\toprule
$n$ & \multicolumn{8}{c}{$\stepsize$} \\ 
\cmidrule(lr){1-1} \cmidrule(lr){2-9}
 & $2^{-10}$ & $2^{-9}$ & $2^{-8}$ & $2^{-7}$ & $2^{-6}$ & $2^{-5}$ & $2^{-4}$ & $2^{-3}$ \\ 
\midrule\addlinespace[2.5pt]
$2^{3} = 8$ & {\cellcolor[HTML]{450357}{\textcolor[HTML]{FFFFFF}{0.2 (1)}}} & {\cellcolor[HTML]{450357}{\textcolor[HTML]{FFFFFF}{0.2 (1)}}} & {\cellcolor[HTML]{450357}{\textcolor[HTML]{FFFFFF}{0.2 (1)}}} & {\cellcolor[HTML]{450357}{\textcolor[HTML]{FFFFFF}{0.2 (1)}}} & {\cellcolor[HTML]{440356}{\textcolor[HTML]{FFFFFF}{0.2 (2)}}} & {\cellcolor[HTML]{440356}{\textcolor[HTML]{FFFFFF}{0.2 (2)}}} & {\cellcolor[HTML]{440356}{\textcolor[HTML]{FFFFFF}{0.2 (2)}}} & {\cellcolor[HTML]{450558}{\textcolor[HTML]{FFFFFF}{0.3 (1)}}} \\ 
$2^{4} = 16$ & {\cellcolor[HTML]{450458}{\textcolor[HTML]{FFFFFF}{0.3 (2)}}} & {\cellcolor[HTML]{470C5F}{\textcolor[HTML]{FFFFFF}{0.9 (2)}}} & {\cellcolor[HTML]{481769}{\textcolor[HTML]{FFFFFF}{1.8 (2)}}} & {\cellcolor[HTML]{471265}{\textcolor[HTML]{FFFFFF}{1.5 (2)}}} & {\cellcolor[HTML]{470D60}{\textcolor[HTML]{FFFFFF}{1.0 (2)}}} & {\cellcolor[HTML]{46085B}{\textcolor[HTML]{FFFFFF}{0.6 (2)}}} & {\cellcolor[HTML]{46075B}{\textcolor[HTML]{FFFFFF}{0.5 (2)}}} & {\cellcolor[HTML]{450357}{\textcolor[HTML]{FFFFFF}{0.2 (3)}}} \\ 
$2^{5} = 32$ & {\cellcolor[HTML]{470F62}{\textcolor[HTML]{FFFFFF}{1.2 (2)}}} & {\cellcolor[HTML]{481B6D}{\textcolor[HTML]{FFFFFF}{2.3 (2)}}} & {\cellcolor[HTML]{481668}{\textcolor[HTML]{FFFFFF}{1.8 (2)}}} & {\cellcolor[HTML]{471264}{\textcolor[HTML]{FFFFFF}{1.4 (2)}}} & {\cellcolor[HTML]{471164}{\textcolor[HTML]{FFFFFF}{1.4 (3)}}} & {\cellcolor[HTML]{460C5F}{\textcolor[HTML]{FFFFFF}{0.9 (3)}}} & {\cellcolor[HTML]{470E61}{\textcolor[HTML]{FFFFFF}{1.1 (3)}}} & {\cellcolor[HTML]{46065A}{\textcolor[HTML]{FFFFFF}{0.4 (6)}}} \\
\bottomrule
\end{tabular}
\end{center}
\caption{\textbf{Best Table for L63, ddm, normalize none, test sequential}. See the beginning of \cref{app:sec:details} for a description.}
\end{table}
\begin{figure}[ht!]
\begin{center}
\includegraphics[width=\textwidth]{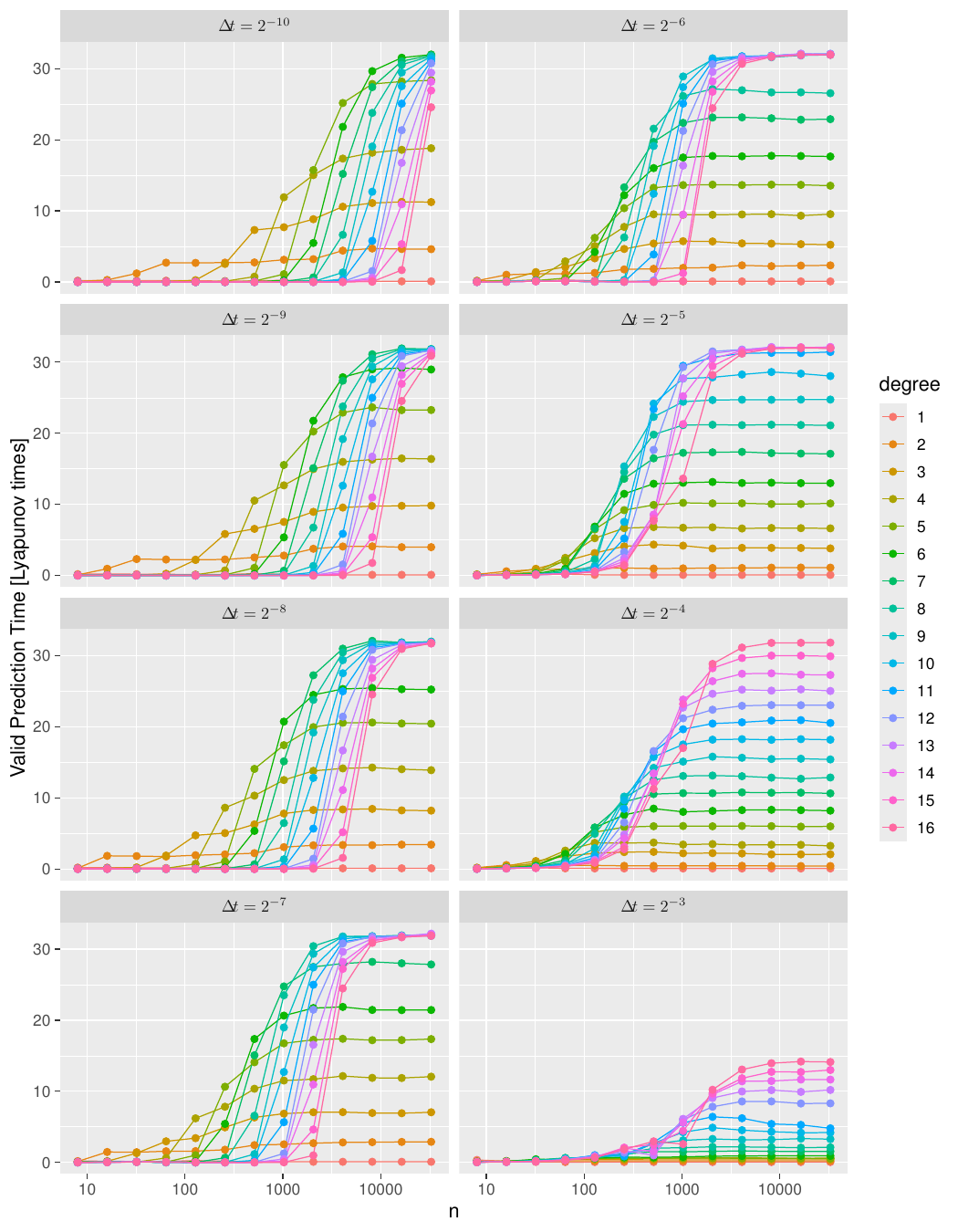}
\end{center}
\caption{\textbf{All Plot for L63, ddm, normalize none, test sequential}. See the beginning of \cref{app:sec:details} for a description.}
\end{figure}

\clearpage
\subsection{L63, mdm, normalize none, test sequential}
\begin{figure}[ht!]
\begin{center}
\includegraphics[width=\textwidth]{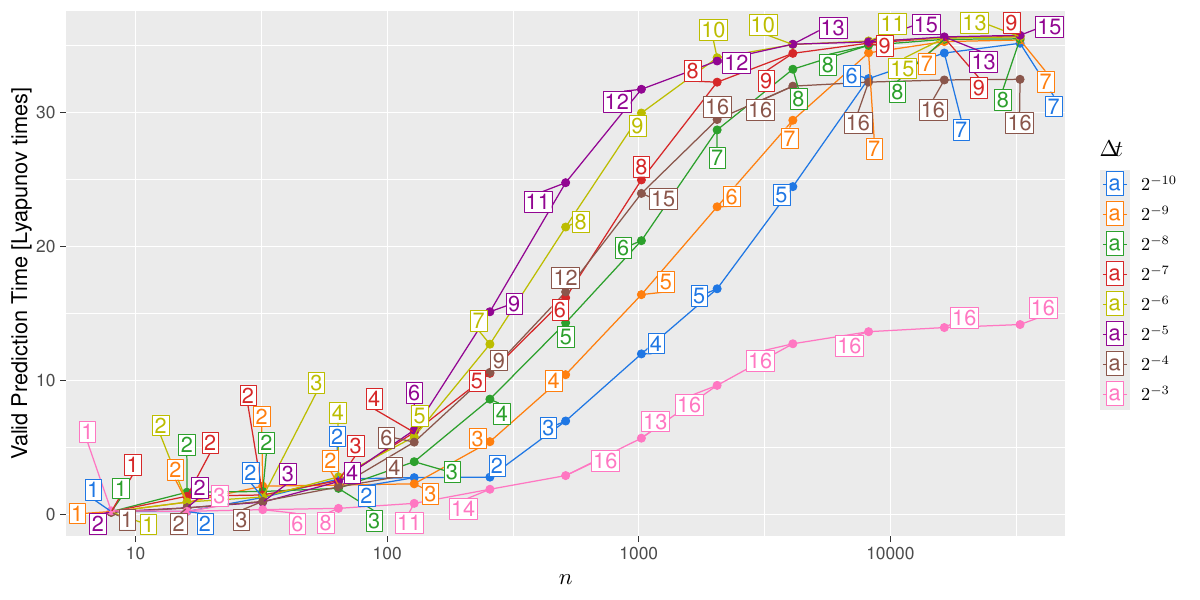}
\end{center}
\caption{\textbf{Best Plot for L63, mdm, normalize none, test sequential}. See the beginning of \cref{app:sec:details} for a description.}
\end{figure}
\begin{table}[ht!]
\begin{center}
\fontsize{8.0pt}{10pt}\selectfont
\fontfamily{phv}\selectfont
\renewcommand{\arraystretch}{1.05}
\setlength{\tabcolsep}{0.5em}
\rowcolors{2}{gray!20}{white}
\begin{tabular}{crrrrrrrr}
\toprule
$n$ & \multicolumn{8}{c}{$\stepsize$} \\ 
\cmidrule(lr){1-1} \cmidrule(lr){2-9}
 & $2^{-10}$ & $2^{-9}$ & $2^{-8}$ & $2^{-7}$ & $2^{-6}$ & $2^{-5}$ & $2^{-4}$ & $2^{-3}$ \\ 
\midrule\addlinespace[2.5pt]
$2^{3} = 8$ & {\cellcolor[HTML]{440356}{\textcolor[HTML]{FFFFFF}{0.2 (1)}}} & {\cellcolor[HTML]{440356}{\textcolor[HTML]{FFFFFF}{0.2 (1)}}} & {\cellcolor[HTML]{440356}{\textcolor[HTML]{FFFFFF}{0.2 (1)}}} & {\cellcolor[HTML]{440356}{\textcolor[HTML]{FFFFFF}{0.2 (1)}}} & {\cellcolor[HTML]{440256}{\textcolor[HTML]{FFFFFF}{0.2 (1)}}} & {\cellcolor[HTML]{440356}{\textcolor[HTML]{FFFFFF}{0.2 (2)}}} & {\cellcolor[HTML]{440256}{\textcolor[HTML]{FFFFFF}{0.1 (1)}}} & {\cellcolor[HTML]{450457}{\textcolor[HTML]{FFFFFF}{0.3 (1)}}} \\ 
$2^{4} = 16$ & {\cellcolor[HTML]{450357}{\textcolor[HTML]{FFFFFF}{0.2 (2)}}} & {\cellcolor[HTML]{460A5D}{\textcolor[HTML]{FFFFFF}{0.9 (2)}}} & {\cellcolor[HTML]{471365}{\textcolor[HTML]{FFFFFF}{1.6 (2)}}} & {\cellcolor[HTML]{470F62}{\textcolor[HTML]{FFFFFF}{1.3 (2)}}} & {\cellcolor[HTML]{460B5E}{\textcolor[HTML]{FFFFFF}{0.9 (2)}}} & {\cellcolor[HTML]{450659}{\textcolor[HTML]{FFFFFF}{0.5 (2)}}} & {\cellcolor[HTML]{450659}{\textcolor[HTML]{FFFFFF}{0.5 (2)}}} & {\cellcolor[HTML]{450357}{\textcolor[HTML]{FFFFFF}{0.2 (3)}}} \\ 
$2^{5} = 32$ & {\cellcolor[HTML]{470E61}{\textcolor[HTML]{FFFFFF}{1.3 (2)}}} & {\cellcolor[HTML]{481769}{\textcolor[HTML]{FFFFFF}{2.1 (2)}}} & {\cellcolor[HTML]{471365}{\textcolor[HTML]{FFFFFF}{1.7 (2)}}} & {\cellcolor[HTML]{471063}{\textcolor[HTML]{FFFFFF}{1.4 (2)}}} & {\cellcolor[HTML]{470E61}{\textcolor[HTML]{FFFFFF}{1.3 (3)}}} & {\cellcolor[HTML]{460A5D}{\textcolor[HTML]{FFFFFF}{0.9 (3)}}} & {\cellcolor[HTML]{460B5E}{\textcolor[HTML]{FFFFFF}{1.0 (3)}}} & {\cellcolor[HTML]{450458}{\textcolor[HTML]{FFFFFF}{0.3 (6)}}} \\
\bottomrule
\end{tabular}
\end{center}
\caption{\textbf{Best Table for L63, mdm, normalize none, test sequential}. See the beginning of \cref{app:sec:details} for a description.}
\end{table}
\begin{figure}[ht!]
\begin{center}
\includegraphics[width=\textwidth]{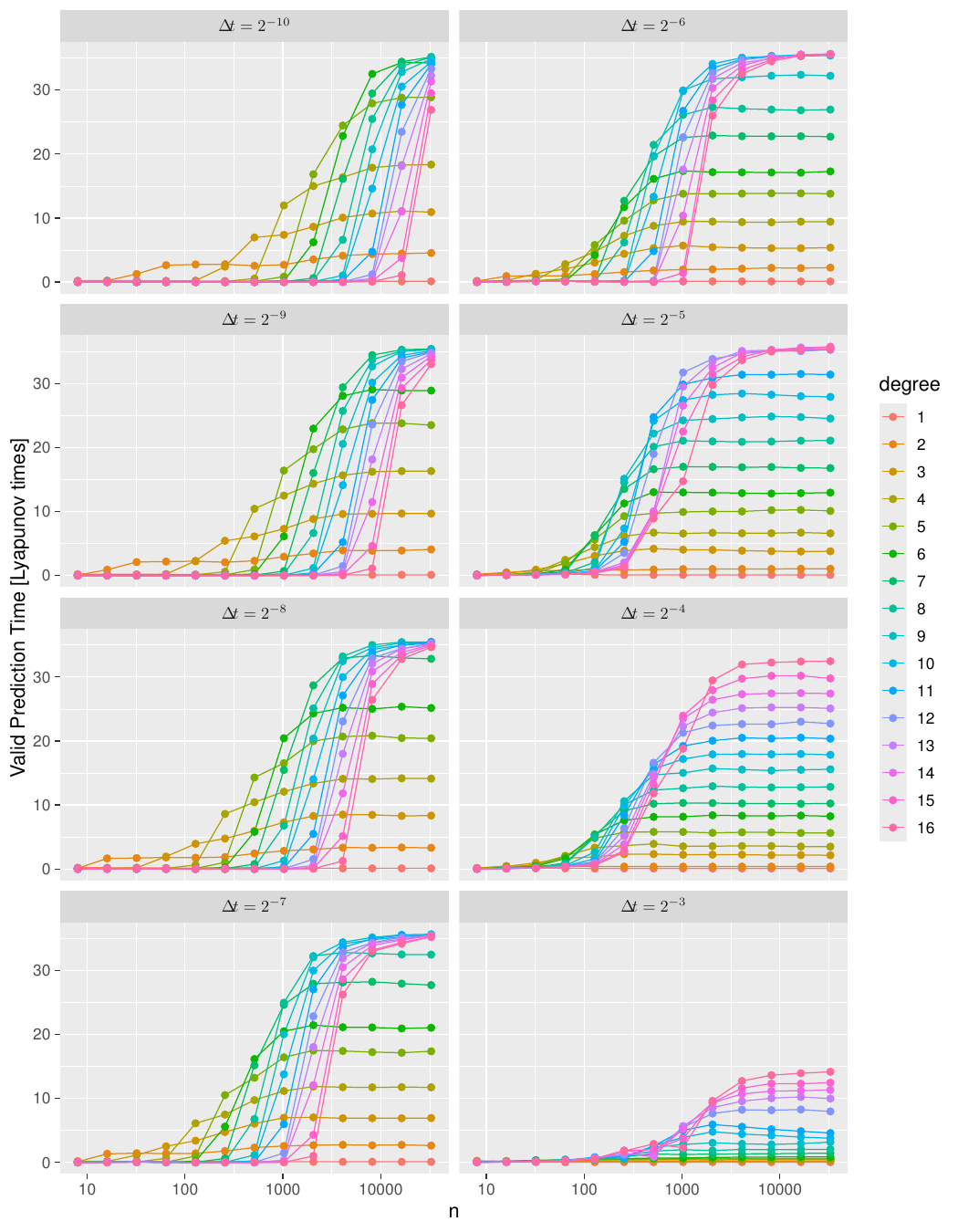}
\end{center}
\caption{\textbf{All Plot for L63, mdm, normalize none, test sequential}. See the beginning of \cref{app:sec:details} for a description.}
\end{figure}

\clearpage
\subsection{L63, mdm, normalize none, test random}
\begin{figure}[ht!]
\begin{center}
\includegraphics[width=\textwidth]{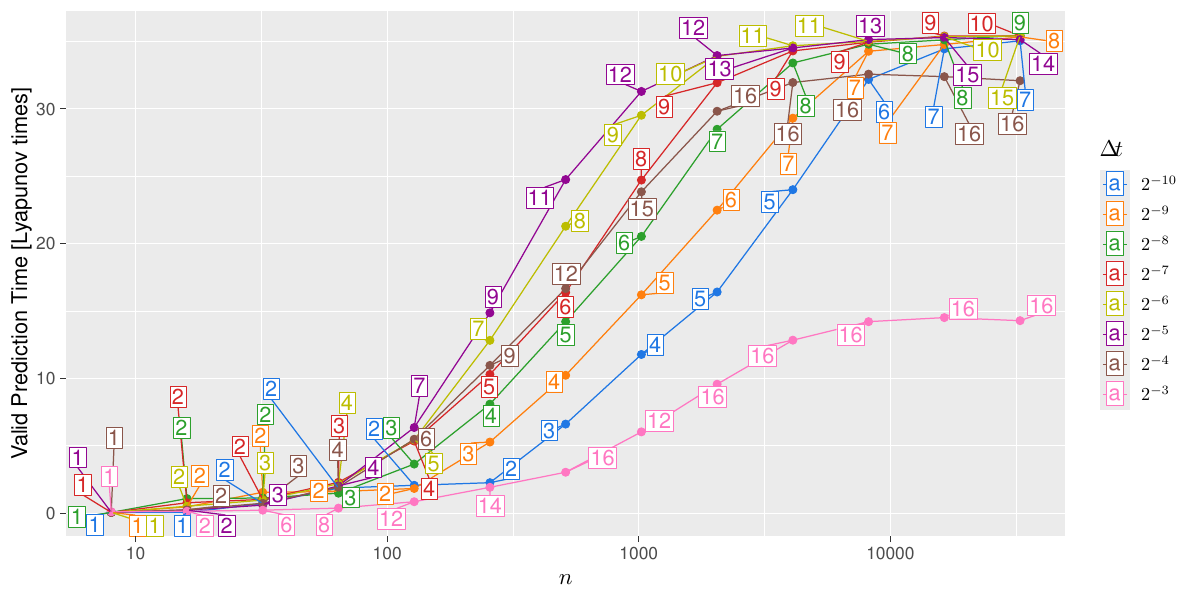}
\end{center}
\caption{\textbf{Best Plot for L63, mdm, normalize none, test random}. See the beginning of \cref{app:sec:details} for a description.}
\end{figure}
\begin{table}[ht!]
\begin{center}
\fontsize{8.0pt}{10pt}\selectfont
\fontfamily{phv}\selectfont
\renewcommand{\arraystretch}{1.05}
\setlength{\tabcolsep}{0.5em}
\rowcolors{2}{gray!20}{white}
\begin{tabular}{crrrrrrrr}
\toprule
$n$ & \multicolumn{8}{c}{$\stepsize$} \\ 
\cmidrule(lr){1-1} \cmidrule(lr){2-9}
 & $2^{-10}$ & $2^{-9}$ & $2^{-8}$ & $2^{-7}$ & $2^{-6}$ & $2^{-5}$ & $2^{-4}$ & $2^{-3}$ \\ 
\midrule\addlinespace[2.5pt]
$2^{3} = 8$ & {\cellcolor[HTML]{440155}{\textcolor[HTML]{FFFFFF}{0.1 (1)}}} & {\cellcolor[HTML]{440155}{\textcolor[HTML]{FFFFFF}{0.1 (1)}}} & {\cellcolor[HTML]{440155}{\textcolor[HTML]{FFFFFF}{0.1 (1)}}} & {\cellcolor[HTML]{440155}{\textcolor[HTML]{FFFFFF}{0.1 (1)}}} & {\cellcolor[HTML]{440155}{\textcolor[HTML]{FFFFFF}{0.1 (1)}}} & {\cellcolor[HTML]{440155}{\textcolor[HTML]{FFFFFF}{0.1 (1)}}} & {\cellcolor[HTML]{440155}{\textcolor[HTML]{FFFFFF}{0.1 (1)}}} & {\cellcolor[HTML]{440256}{\textcolor[HTML]{FFFFFF}{0.1 (1)}}} \\ 
$2^{4} = 16$ & {\cellcolor[HTML]{440155}{\textcolor[HTML]{FFFFFF}{0.1 (1)}}} & {\cellcolor[HTML]{450659}{\textcolor[HTML]{FFFFFF}{0.5 (2)}}} & {\cellcolor[HTML]{470D60}{\textcolor[HTML]{FFFFFF}{1.1 (2)}}} & {\cellcolor[HTML]{46095D}{\textcolor[HTML]{FFFFFF}{0.8 (2)}}} & {\cellcolor[HTML]{46065A}{\textcolor[HTML]{FFFFFF}{0.5 (2)}}} & {\cellcolor[HTML]{450357}{\textcolor[HTML]{FFFFFF}{0.2 (2)}}} & {\cellcolor[HTML]{450457}{\textcolor[HTML]{FFFFFF}{0.3 (2)}}} & {\cellcolor[HTML]{440256}{\textcolor[HTML]{FFFFFF}{0.1 (2)}}} \\ 
$2^{5} = 32$ & {\cellcolor[HTML]{46085C}{\textcolor[HTML]{FFFFFF}{0.7 (2)}}} & {\cellcolor[HTML]{471164}{\textcolor[HTML]{FFFFFF}{1.6 (2)}}} & {\cellcolor[HTML]{470D60}{\textcolor[HTML]{FFFFFF}{1.1 (2)}}} & {\cellcolor[HTML]{460C5F}{\textcolor[HTML]{FFFFFF}{1.0 (2)}}} & {\cellcolor[HTML]{460B5E}{\textcolor[HTML]{FFFFFF}{1.0 (3)}}} & {\cellcolor[HTML]{46075B}{\textcolor[HTML]{FFFFFF}{0.6 (3)}}} & {\cellcolor[HTML]{46095C}{\textcolor[HTML]{FFFFFF}{0.7 (3)}}} & {\cellcolor[HTML]{450357}{\textcolor[HTML]{FFFFFF}{0.2 (6)}}} \\
\bottomrule
\end{tabular}
\end{center}
\caption{\textbf{Best Table for L63, mdm, normalize none, test random}. See the beginning of \cref{app:sec:details} for a description.}
\end{table}
\begin{figure}[ht!]
\begin{center}
\includegraphics[width=\textwidth]{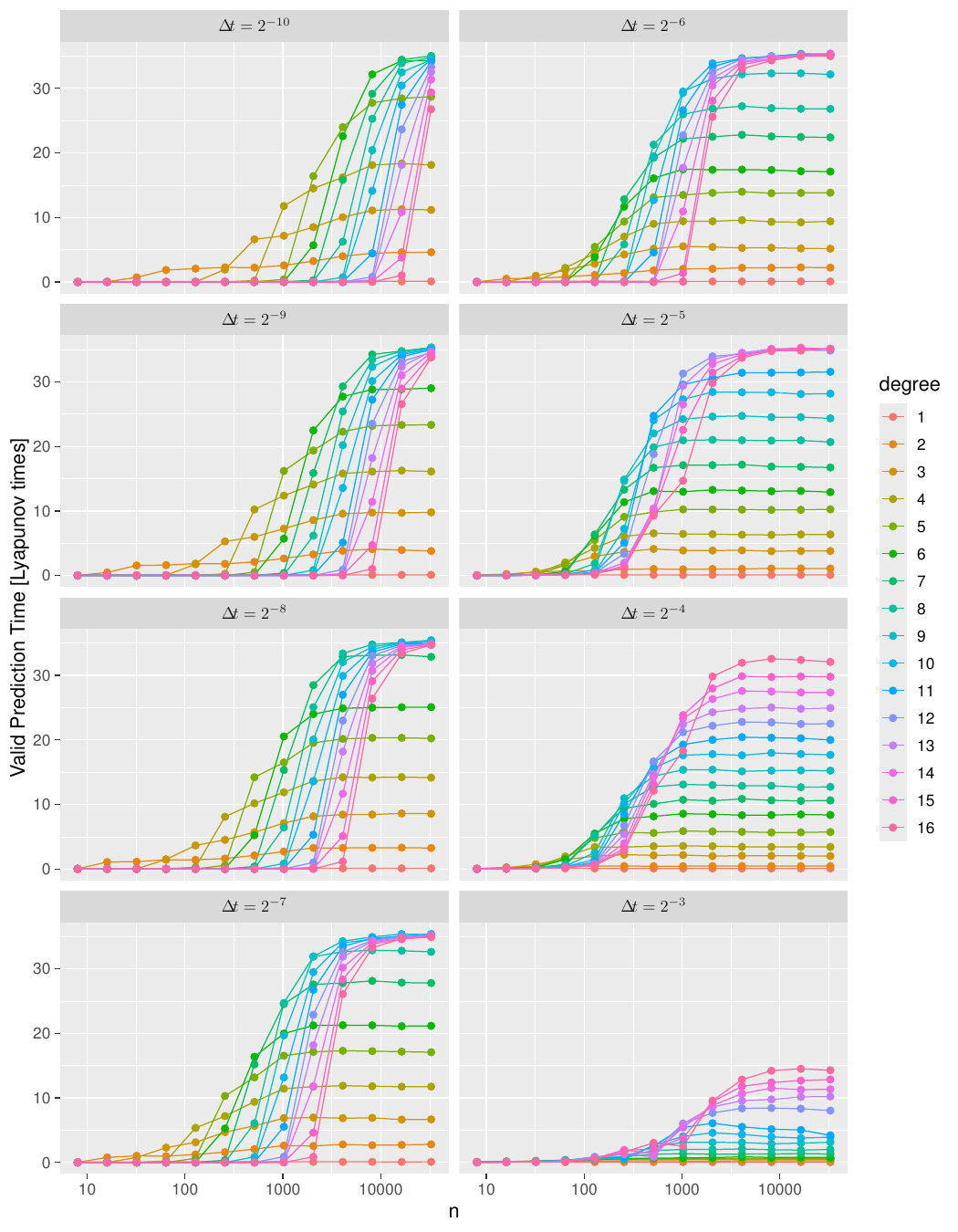}
\end{center}
\caption{\textbf{All Plot for L63, mdm, normalize none, test random}. See the beginning of \cref{app:sec:details} for a description.}
\end{figure}

\clearpage
\subsection{L63, mmm, normalize none, test sequential}
\begin{figure}[ht!]
\begin{center}
\includegraphics[width=\textwidth]{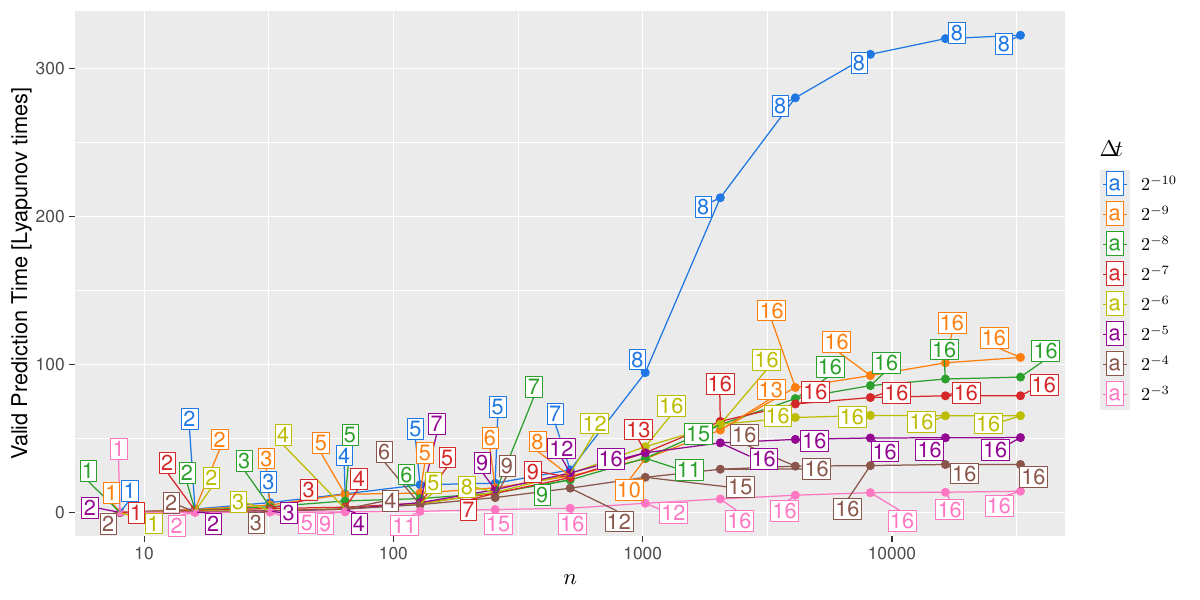}
\end{center}
\caption{\textbf{Best Plot for L63, mmm, normalize none, test sequential}. See the beginning of \cref{app:sec:details} for a description.}
\end{figure}
\begin{table}[ht!]
\begin{center}
\fontsize{8.0pt}{10pt}\selectfont
\fontfamily{phv}\selectfont
\renewcommand{\arraystretch}{1.05}
\setlength{\tabcolsep}{0.5em}
\rowcolors{2}{gray!20}{white}
\begin{tabular}{crrrrrrrr}
\toprule
$n$ & \multicolumn{8}{c}{$\stepsize$} \\ 
\cmidrule(lr){1-1} \cmidrule(lr){2-9}
 & $2^{-10}$ & $2^{-9}$ & $2^{-8}$ & $2^{-7}$ & $2^{-6}$ & $2^{-5}$ & $2^{-4}$ & $2^{-3}$ \\ 
\midrule\addlinespace[2.5pt]
$2^{3} = 8$ & {\cellcolor[HTML]{440154}{\textcolor[HTML]{FFFFFF}{0.2 (1)}}} & {\cellcolor[HTML]{440154}{\textcolor[HTML]{FFFFFF}{0.2 (1)}}} & {\cellcolor[HTML]{440154}{\textcolor[HTML]{FFFFFF}{0.2 (1)}}} & {\cellcolor[HTML]{440154}{\textcolor[HTML]{FFFFFF}{0.2 (1)}}} & {\cellcolor[HTML]{440154}{\textcolor[HTML]{FFFFFF}{0.2 (1)}}} & {\cellcolor[HTML]{440154}{\textcolor[HTML]{FFFFFF}{0.2 (2)}}} & {\cellcolor[HTML]{440154}{\textcolor[HTML]{FFFFFF}{0.2 (2)}}} & {\cellcolor[HTML]{440155}{\textcolor[HTML]{FFFFFF}{0.4 (1)}}} \\ 
$2^{4} = 16$ & {\cellcolor[HTML]{450457}{\textcolor[HTML]{FFFFFF}{2.4 (2)}}} & {\cellcolor[HTML]{450357}{\textcolor[HTML]{FFFFFF}{2.2 (2)}}} & {\cellcolor[HTML]{440356}{\textcolor[HTML]{FFFFFF}{1.6 (2)}}} & {\cellcolor[HTML]{440256}{\textcolor[HTML]{FFFFFF}{1.5 (2)}}} & {\cellcolor[HTML]{440256}{\textcolor[HTML]{FFFFFF}{1.0 (2)}}} & {\cellcolor[HTML]{440155}{\textcolor[HTML]{FFFFFF}{0.5 (2)}}} & {\cellcolor[HTML]{440155}{\textcolor[HTML]{FFFFFF}{0.5 (2)}}} & {\cellcolor[HTML]{440154}{\textcolor[HTML]{FFFFFF}{0.3 (2)}}} \\ 
$2^{5} = 32$ & {\cellcolor[HTML]{46095C}{\textcolor[HTML]{FFFFFF}{6.8 (3)}}} & {\cellcolor[HTML]{46075A}{\textcolor[HTML]{FFFFFF}{5.2 (3)}}} & {\cellcolor[HTML]{450559}{\textcolor[HTML]{FFFFFF}{4.0 (3)}}} & {\cellcolor[HTML]{450457}{\textcolor[HTML]{FFFFFF}{2.8 (3)}}} & {\cellcolor[HTML]{440356}{\textcolor[HTML]{FFFFFF}{1.6 (3)}}} & {\cellcolor[HTML]{440256}{\textcolor[HTML]{FFFFFF}{1.2 (3)}}} & {\cellcolor[HTML]{440256}{\textcolor[HTML]{FFFFFF}{1.1 (3)}}} & {\cellcolor[HTML]{440155}{\textcolor[HTML]{FFFFFF}{0.4 (5)}}} \\
\bottomrule
\end{tabular}
\end{center}
\caption{\textbf{Best Table for L63, mmm, normalize none, test sequential}. See the beginning of \cref{app:sec:details} for a description.}
\end{table}
\begin{figure}[ht!]
\begin{center}
\includegraphics[width=\textwidth]{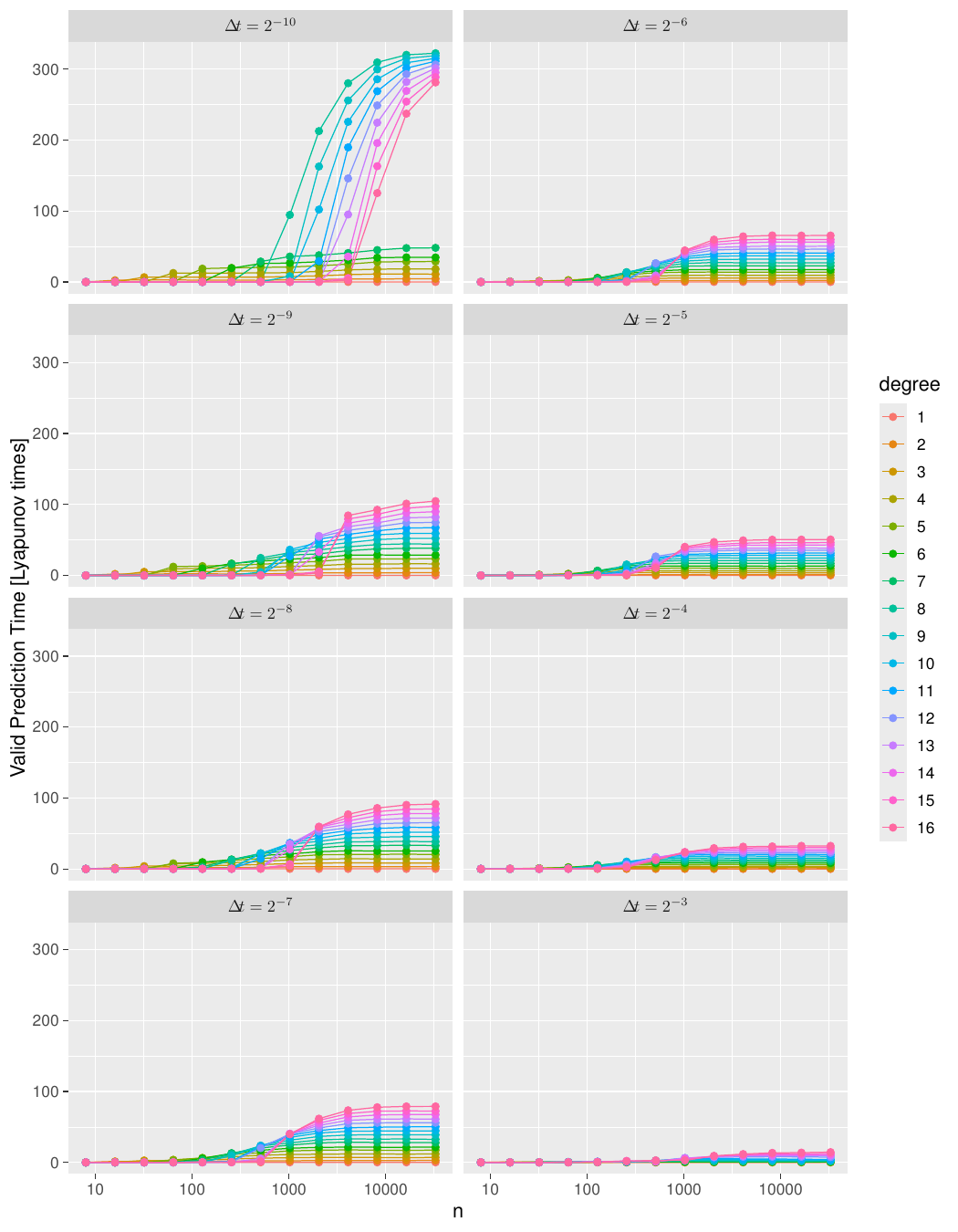}
\end{center}
\caption{\textbf{All Plot for L63, mmm, normalize none, test sequential}. See the beginning of \cref{app:sec:details} for a description.}
\end{figure}

\clearpage
\subsection{L63, xdm, normalize none, test sequential}
\begin{figure}[ht!]
\begin{center}
\includegraphics[width=\textwidth]{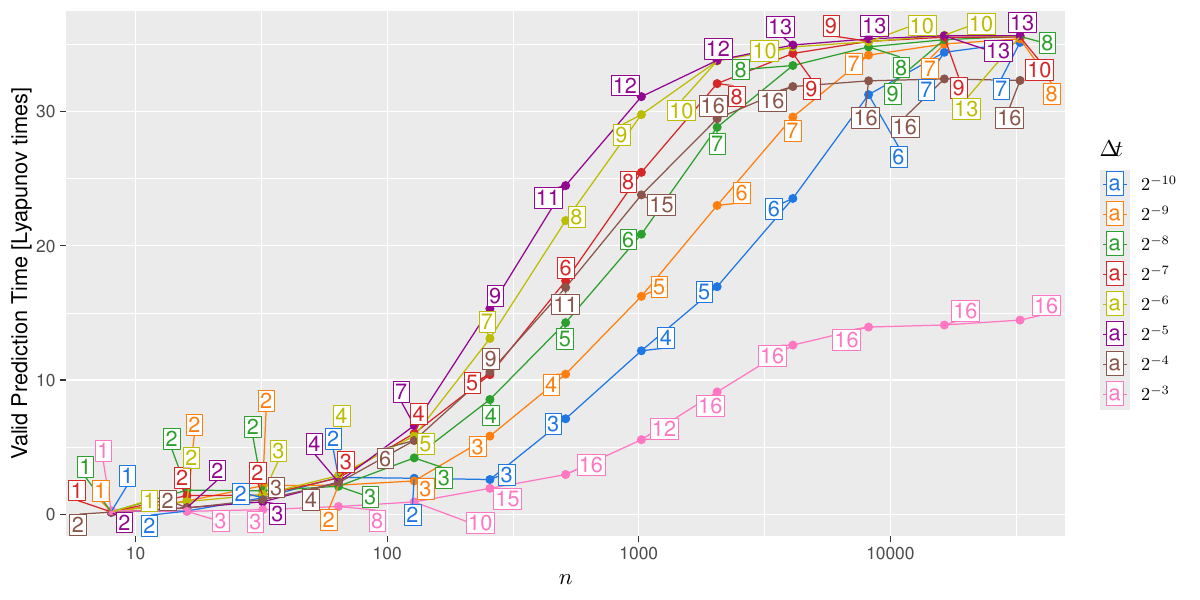}
\end{center}
\caption{\textbf{Best Plot for L63, xdm, normalize none, test sequential}. See the beginning of \cref{app:sec:details} for a description.}
\end{figure}
\begin{table}[ht!]
\begin{center}
\fontsize{8.0pt}{10pt}\selectfont
\fontfamily{phv}\selectfont
\renewcommand{\arraystretch}{1.05}
\setlength{\tabcolsep}{0.5em}
\rowcolors{2}{gray!20}{white}
\begin{tabular}{crrrrrrrr}
\toprule
$n$ & \multicolumn{8}{c}{$\stepsize$} \\ 
\cmidrule(lr){1-1} \cmidrule(lr){2-9}
 & $2^{-10}$ & $2^{-9}$ & $2^{-8}$ & $2^{-7}$ & $2^{-6}$ & $2^{-5}$ & $2^{-4}$ & $2^{-3}$ \\ 
\midrule\addlinespace[2.5pt]
$2^{3} = 8$ & {\cellcolor[HTML]{450357}{\textcolor[HTML]{FFFFFF}{0.2 (1)}}} & {\cellcolor[HTML]{450357}{\textcolor[HTML]{FFFFFF}{0.2 (1)}}} & {\cellcolor[HTML]{440356}{\textcolor[HTML]{FFFFFF}{0.2 (1)}}} & {\cellcolor[HTML]{440356}{\textcolor[HTML]{FFFFFF}{0.2 (1)}}} & {\cellcolor[HTML]{440256}{\textcolor[HTML]{FFFFFF}{0.2 (1)}}} & {\cellcolor[HTML]{440256}{\textcolor[HTML]{FFFFFF}{0.1 (2)}}} & {\cellcolor[HTML]{440256}{\textcolor[HTML]{FFFFFF}{0.2 (2)}}} & {\cellcolor[HTML]{450457}{\textcolor[HTML]{FFFFFF}{0.3 (1)}}} \\ 
$2^{4} = 16$ & {\cellcolor[HTML]{450457}{\textcolor[HTML]{FFFFFF}{0.3 (2)}}} & {\cellcolor[HTML]{460B5E}{\textcolor[HTML]{FFFFFF}{1.0 (2)}}} & {\cellcolor[HTML]{481467}{\textcolor[HTML]{FFFFFF}{1.8 (2)}}} & {\cellcolor[HTML]{471063}{\textcolor[HTML]{FFFFFF}{1.4 (2)}}} & {\cellcolor[HTML]{460B5E}{\textcolor[HTML]{FFFFFF}{1.0 (2)}}} & {\cellcolor[HTML]{46065A}{\textcolor[HTML]{FFFFFF}{0.5 (2)}}} & {\cellcolor[HTML]{450659}{\textcolor[HTML]{FFFFFF}{0.5 (2)}}} & {\cellcolor[HTML]{450357}{\textcolor[HTML]{FFFFFF}{0.2 (3)}}} \\ 
$2^{5} = 32$ & {\cellcolor[HTML]{470E61}{\textcolor[HTML]{FFFFFF}{1.2 (2)}}} & {\cellcolor[HTML]{481769}{\textcolor[HTML]{FFFFFF}{2.1 (2)}}} & {\cellcolor[HTML]{481466}{\textcolor[HTML]{FFFFFF}{1.8 (2)}}} & {\cellcolor[HTML]{471063}{\textcolor[HTML]{FFFFFF}{1.4 (2)}}} & {\cellcolor[HTML]{471164}{\textcolor[HTML]{FFFFFF}{1.5 (3)}}} & {\cellcolor[HTML]{460B5E}{\textcolor[HTML]{FFFFFF}{0.9 (3)}}} & {\cellcolor[HTML]{470D60}{\textcolor[HTML]{FFFFFF}{1.1 (3)}}} & {\cellcolor[HTML]{450558}{\textcolor[HTML]{FFFFFF}{0.4 (3)}}} \\
\bottomrule
\end{tabular}
\end{center}
\caption{\textbf{Best Table for L63, xdm, normalize none, test sequential}. See the beginning of \cref{app:sec:details} for a description.}
\end{table}
\begin{figure}[ht!]
\begin{center}
\includegraphics[width=\textwidth]{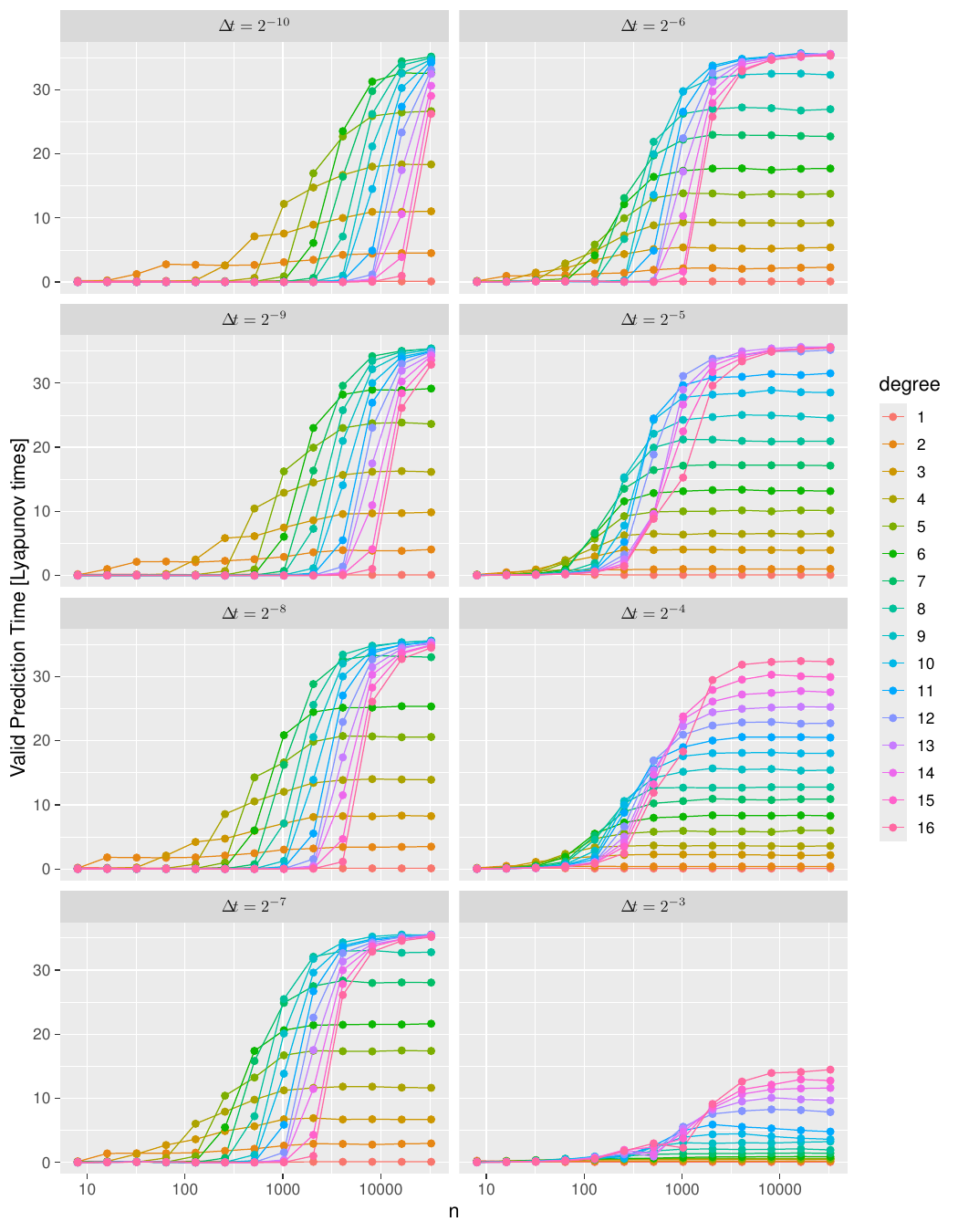}
\end{center}
\caption{\textbf{All Plot for L63, xdm, normalize none, test sequential}. See the beginning of \cref{app:sec:details} for a description.}
\end{figure}

\clearpage
\subsection{L63, ydm, normalize none, test sequential}
\begin{figure}[ht!]
\begin{center}
\includegraphics[width=\textwidth]{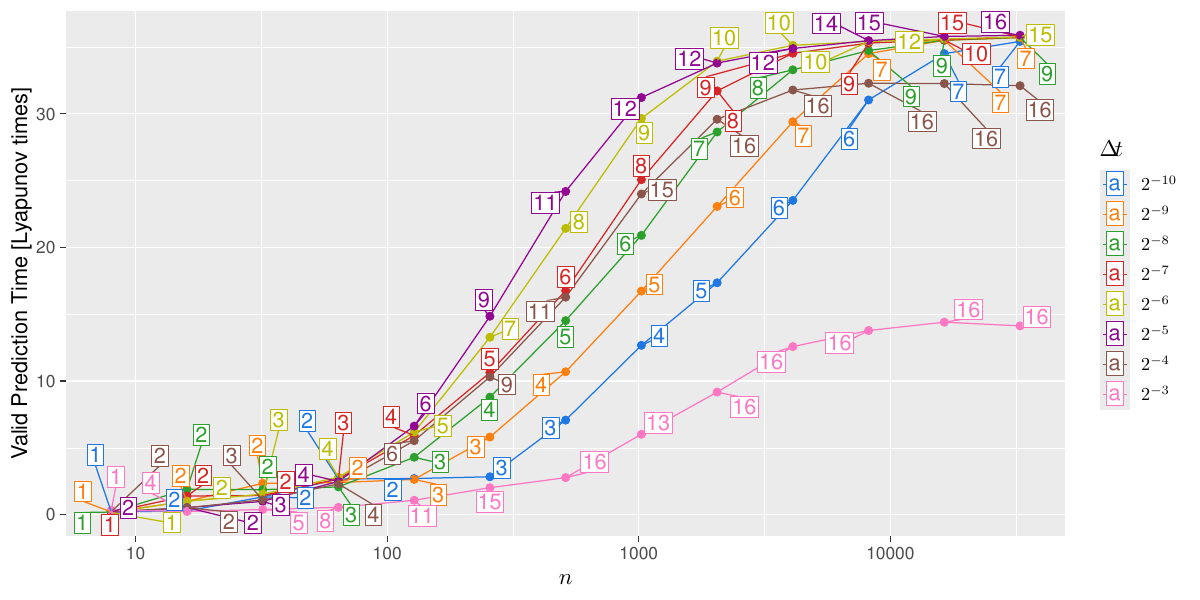}
\end{center}
\caption{\textbf{Best Plot for L63, ydm, normalize none, test sequential}. See the beginning of \cref{app:sec:details} for a description.}
\end{figure}
\begin{table}[ht!]
\begin{center}
\fontsize{8.0pt}{10pt}\selectfont
\fontfamily{phv}\selectfont
\renewcommand{\arraystretch}{1.05}
\setlength{\tabcolsep}{0.5em}
\rowcolors{2}{gray!20}{white}
\begin{tabular}{crrrrrrrr}
\toprule
$n$ & \multicolumn{8}{c}{$\stepsize$} \\ 
\cmidrule(lr){1-1} \cmidrule(lr){2-9}
 & $2^{-10}$ & $2^{-9}$ & $2^{-8}$ & $2^{-7}$ & $2^{-6}$ & $2^{-5}$ & $2^{-4}$ & $2^{-3}$ \\ 
\midrule\addlinespace[2.5pt]
$2^{3} = 8$ & {\cellcolor[HTML]{440356}{\textcolor[HTML]{FFFFFF}{0.2 (1)}}} & {\cellcolor[HTML]{440356}{\textcolor[HTML]{FFFFFF}{0.2 (1)}}} & {\cellcolor[HTML]{440356}{\textcolor[HTML]{FFFFFF}{0.2 (1)}}} & {\cellcolor[HTML]{440356}{\textcolor[HTML]{FFFFFF}{0.2 (1)}}} & {\cellcolor[HTML]{440256}{\textcolor[HTML]{FFFFFF}{0.2 (1)}}} & {\cellcolor[HTML]{440256}{\textcolor[HTML]{FFFFFF}{0.2 (2)}}} & {\cellcolor[HTML]{450357}{\textcolor[HTML]{FFFFFF}{0.2 (2)}}} & {\cellcolor[HTML]{450458}{\textcolor[HTML]{FFFFFF}{0.3 (1)}}} \\ 
$2^{4} = 16$ & {\cellcolor[HTML]{450457}{\textcolor[HTML]{FFFFFF}{0.3 (2)}}} & {\cellcolor[HTML]{460B5E}{\textcolor[HTML]{FFFFFF}{1.0 (2)}}} & {\cellcolor[HTML]{481567}{\textcolor[HTML]{FFFFFF}{1.9 (2)}}} & {\cellcolor[HTML]{471063}{\textcolor[HTML]{FFFFFF}{1.4 (2)}}} & {\cellcolor[HTML]{460B5E}{\textcolor[HTML]{FFFFFF}{1.0 (2)}}} & {\cellcolor[HTML]{46065A}{\textcolor[HTML]{FFFFFF}{0.5 (2)}}} & {\cellcolor[HTML]{46065A}{\textcolor[HTML]{FFFFFF}{0.5 (2)}}} & {\cellcolor[HTML]{450357}{\textcolor[HTML]{FFFFFF}{0.2 (4)}}} \\ 
$2^{5} = 32$ & {\cellcolor[HTML]{470F62}{\textcolor[HTML]{FFFFFF}{1.3 (2)}}} & {\cellcolor[HTML]{48196B}{\textcolor[HTML]{FFFFFF}{2.3 (2)}}} & {\cellcolor[HTML]{481567}{\textcolor[HTML]{FFFFFF}{1.9 (2)}}} & {\cellcolor[HTML]{471063}{\textcolor[HTML]{FFFFFF}{1.4 (2)}}} & {\cellcolor[HTML]{471164}{\textcolor[HTML]{FFFFFF}{1.5 (3)}}} & {\cellcolor[HTML]{460B5E}{\textcolor[HTML]{FFFFFF}{1.0 (3)}}} & {\cellcolor[HTML]{470D60}{\textcolor[HTML]{FFFFFF}{1.1 (3)}}} & {\cellcolor[HTML]{450558}{\textcolor[HTML]{FFFFFF}{0.4 (5)}}} \\
\bottomrule
\end{tabular}
\end{center}
\caption{\textbf{Best Table for L63, ydm, normalize none, test sequential}. See the beginning of \cref{app:sec:details} for a description.}
\end{table}
\begin{figure}[ht!]
\begin{center}
\includegraphics[width=\textwidth]{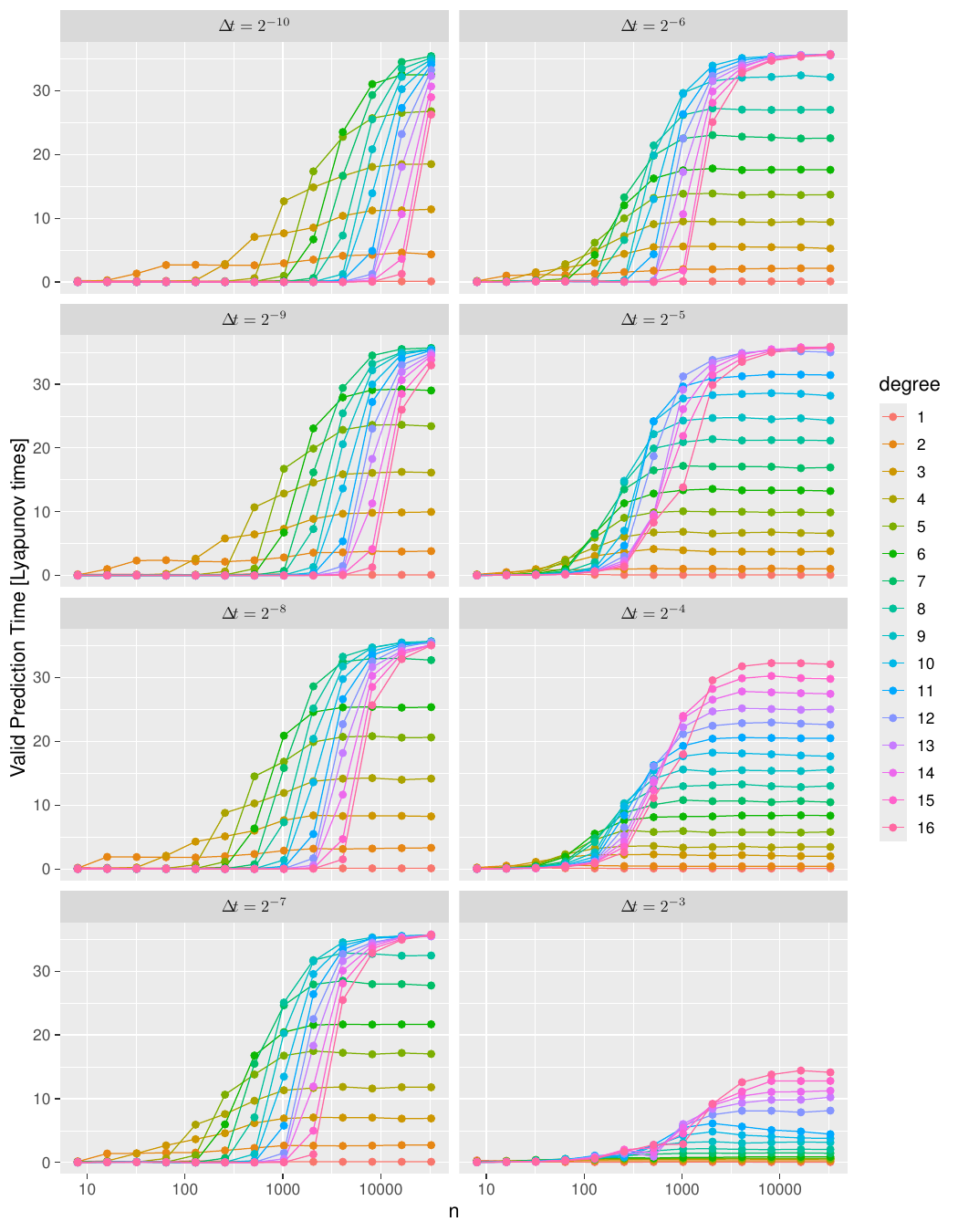}
\end{center}
\caption{\textbf{All Plot for L63, ydm, normalize none, test sequential}. See the beginning of \cref{app:sec:details} for a description.}
\end{figure}

\clearpage
\subsection{L63, ddd, normalize full, test sequential, noise-free init.cond., noise level 1e-09}
\begin{figure}[ht!]
\begin{center}
\includegraphics[width=\textwidth]{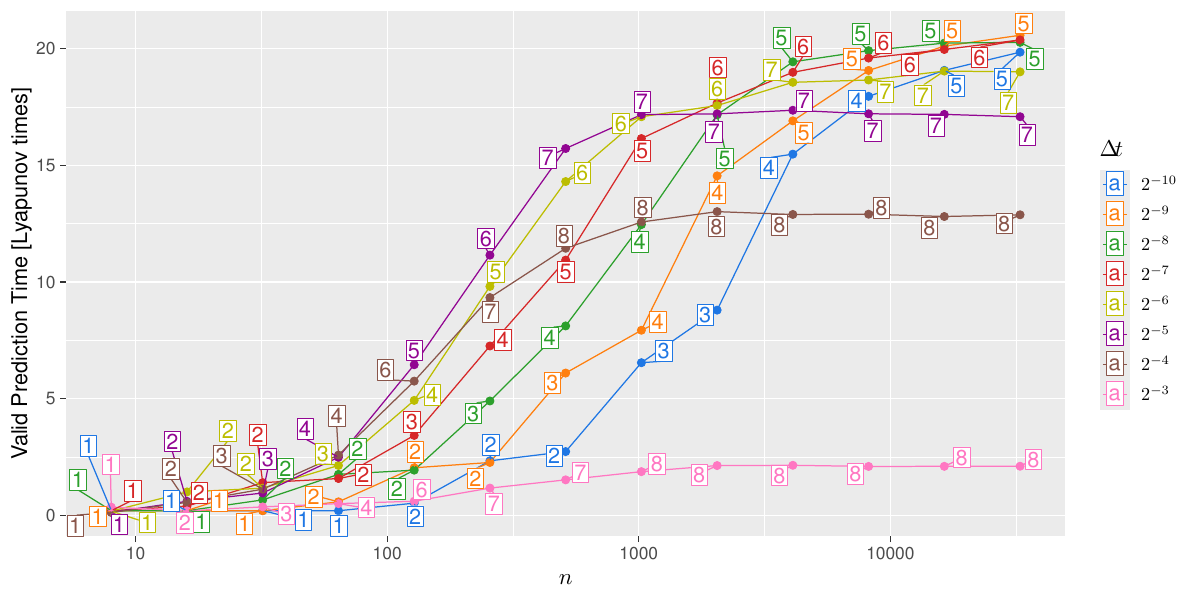}
\end{center}
\caption{\textbf{Best Plot for L63, ddd, normalize full, test sequential, noise-free init.cond., noise level 1e-09}. See the beginning of \cref{app:sec:details} for a description.}
\end{figure}
\begin{table}[ht!]
\begin{center}
\fontsize{8.0pt}{10pt}\selectfont
\fontfamily{phv}\selectfont
\renewcommand{\arraystretch}{1.05}
\setlength{\tabcolsep}{0.5em}
\rowcolors{2}{gray!20}{white}
\begin{tabular}{crrrrrrrr}
\toprule
$n$ & \multicolumn{8}{c}{$\stepsize$} \\ 
\cmidrule(lr){1-1} \cmidrule(lr){2-9}
 & $2^{-10}$ & $2^{-9}$ & $2^{-8}$ & $2^{-7}$ & $2^{-6}$ & $2^{-5}$ & $2^{-4}$ & $2^{-3}$ \\ 
\midrule\addlinespace[2.5pt]
$2^{3} = 8$ & {\cellcolor[HTML]{450558}{\textcolor[HTML]{FFFFFF}{0.2 (1)}}} & {\cellcolor[HTML]{450458}{\textcolor[HTML]{FFFFFF}{0.2 (1)}}} & {\cellcolor[HTML]{450458}{\textcolor[HTML]{FFFFFF}{0.2 (1)}}} & {\cellcolor[HTML]{450458}{\textcolor[HTML]{FFFFFF}{0.2 (1)}}} & {\cellcolor[HTML]{450457}{\textcolor[HTML]{FFFFFF}{0.2 (1)}}} & {\cellcolor[HTML]{450357}{\textcolor[HTML]{FFFFFF}{0.1 (1)}}} & {\cellcolor[HTML]{440356}{\textcolor[HTML]{FFFFFF}{0.1 (1)}}} & {\cellcolor[HTML]{46075B}{\textcolor[HTML]{FFFFFF}{0.3 (1)}}} \\ 
$2^{4} = 16$ & {\cellcolor[HTML]{450458}{\textcolor[HTML]{FFFFFF}{0.2 (1)}}} & {\cellcolor[HTML]{450458}{\textcolor[HTML]{FFFFFF}{0.2 (1)}}} & {\cellcolor[HTML]{450458}{\textcolor[HTML]{FFFFFF}{0.2 (1)}}} & {\cellcolor[HTML]{46085C}{\textcolor[HTML]{FFFFFF}{0.4 (2)}}} & {\cellcolor[HTML]{471366}{\textcolor[HTML]{FFFFFF}{1.0 (2)}}} & {\cellcolor[HTML]{460C5F}{\textcolor[HTML]{FFFFFF}{0.6 (2)}}} & {\cellcolor[HTML]{460B5E}{\textcolor[HTML]{FFFFFF}{0.5 (2)}}} & {\cellcolor[HTML]{450458}{\textcolor[HTML]{FFFFFF}{0.2 (2)}}} \\ 
$2^{5} = 32$ & {\cellcolor[HTML]{450458}{\textcolor[HTML]{FFFFFF}{0.2 (1)}}} & {\cellcolor[HTML]{450458}{\textcolor[HTML]{FFFFFF}{0.2 (1)}}} & {\cellcolor[HTML]{470D60}{\textcolor[HTML]{FFFFFF}{0.7 (2)}}} & {\cellcolor[HTML]{481A6C}{\textcolor[HTML]{FFFFFF}{1.4 (2)}}} & {\cellcolor[HTML]{481668}{\textcolor[HTML]{FFFFFF}{1.2 (2)}}} & {\cellcolor[HTML]{471365}{\textcolor[HTML]{FFFFFF}{0.9 (3)}}} & {\cellcolor[HTML]{481668}{\textcolor[HTML]{FFFFFF}{1.1 (3)}}} & {\cellcolor[HTML]{46075B}{\textcolor[HTML]{FFFFFF}{0.4 (3)}}} \\
\bottomrule
\end{tabular}
\end{center}
\caption{\textbf{Best Table for L63, ddd, normalize full, test sequential, noise-free init.cond., noise level 1e-09}. See the beginning of \cref{app:sec:details} for a description.}
\end{table}
\begin{figure}[ht!]
\begin{center}
\includegraphics[width=\textwidth]{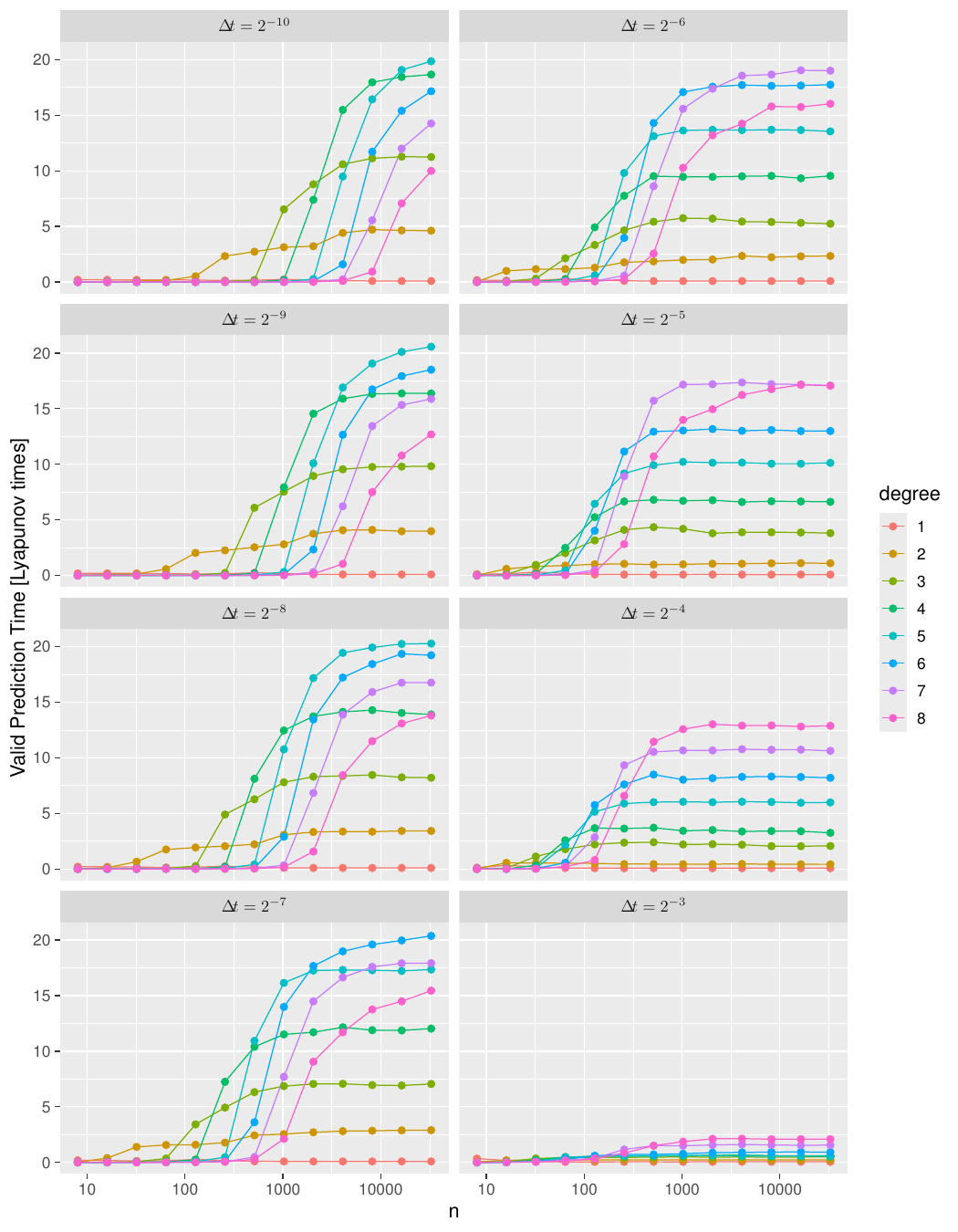}
\end{center}
\caption{\textbf{All Plot for L63, ddd, normalize full, test sequential, noise-free init.cond., noise level 1e-09}. See the beginning of \cref{app:sec:details} for a description.}
\end{figure}

\clearpage
\subsection{L63, ddd, normalize full, test sequential, noise-free init.cond., noise level 1e-07}
\begin{figure}[ht!]
\begin{center}
\includegraphics[width=\textwidth]{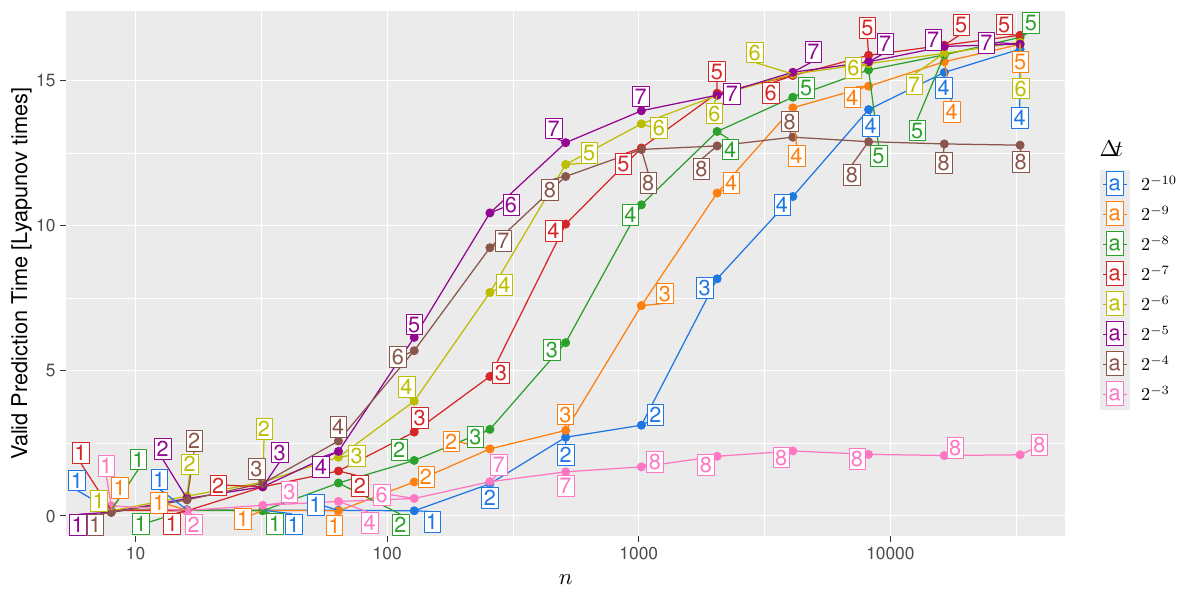}
\end{center}
\caption{\textbf{Best Plot for L63, ddd, normalize full, test sequential, noise-free init.cond., noise level 1e-07}. See the beginning of \cref{app:sec:details} for a description.}
\end{figure}
\begin{table}[ht!]
\begin{center}
\fontsize{8.0pt}{10pt}\selectfont
\fontfamily{phv}\selectfont
\renewcommand{\arraystretch}{1.05}
\setlength{\tabcolsep}{0.5em}
\rowcolors{2}{gray!20}{white}
\begin{tabular}{crrrrrrrr}
\toprule
$n$ & \multicolumn{8}{c}{$\stepsize$} \\ 
\cmidrule(lr){1-1} \cmidrule(lr){2-9}
 & $2^{-10}$ & $2^{-9}$ & $2^{-8}$ & $2^{-7}$ & $2^{-6}$ & $2^{-5}$ & $2^{-4}$ & $2^{-3}$ \\ 
\midrule\addlinespace[2.5pt]
$2^{3} = 8$ & {\cellcolor[HTML]{450559}{\textcolor[HTML]{FFFFFF}{0.2 (1)}}} & {\cellcolor[HTML]{450559}{\textcolor[HTML]{FFFFFF}{0.2 (1)}}} & {\cellcolor[HTML]{450559}{\textcolor[HTML]{FFFFFF}{0.2 (1)}}} & {\cellcolor[HTML]{450559}{\textcolor[HTML]{FFFFFF}{0.2 (1)}}} & {\cellcolor[HTML]{450558}{\textcolor[HTML]{FFFFFF}{0.2 (1)}}} & {\cellcolor[HTML]{450457}{\textcolor[HTML]{FFFFFF}{0.1 (1)}}} & {\cellcolor[HTML]{450457}{\textcolor[HTML]{FFFFFF}{0.1 (1)}}} & {\cellcolor[HTML]{46095C}{\textcolor[HTML]{FFFFFF}{0.3 (1)}}} \\ 
$2^{4} = 16$ & {\cellcolor[HTML]{450559}{\textcolor[HTML]{FFFFFF}{0.2 (1)}}} & {\cellcolor[HTML]{450559}{\textcolor[HTML]{FFFFFF}{0.2 (1)}}} & {\cellcolor[HTML]{450559}{\textcolor[HTML]{FFFFFF}{0.2 (1)}}} & {\cellcolor[HTML]{450558}{\textcolor[HTML]{FFFFFF}{0.2 (1)}}} & {\cellcolor[HTML]{471063}{\textcolor[HTML]{FFFFFF}{0.7 (2)}}} & {\cellcolor[HTML]{470E61}{\textcolor[HTML]{FFFFFF}{0.6 (2)}}} & {\cellcolor[HTML]{470D60}{\textcolor[HTML]{FFFFFF}{0.5 (2)}}} & {\cellcolor[HTML]{450559}{\textcolor[HTML]{FFFFFF}{0.2 (2)}}} \\ 
$2^{5} = 32$ & {\cellcolor[HTML]{450559}{\textcolor[HTML]{FFFFFF}{0.2 (1)}}} & {\cellcolor[HTML]{450559}{\textcolor[HTML]{FFFFFF}{0.2 (1)}}} & {\cellcolor[HTML]{450558}{\textcolor[HTML]{FFFFFF}{0.2 (1)}}} & {\cellcolor[HTML]{481769}{\textcolor[HTML]{FFFFFF}{1.0 (2)}}} & {\cellcolor[HTML]{481B6D}{\textcolor[HTML]{FFFFFF}{1.2 (2)}}} & {\cellcolor[HTML]{481769}{\textcolor[HTML]{FFFFFF}{1.0 (3)}}} & {\cellcolor[HTML]{481A6C}{\textcolor[HTML]{FFFFFF}{1.1 (3)}}} & {\cellcolor[HTML]{46095D}{\textcolor[HTML]{FFFFFF}{0.4 (3)}}} \\
\bottomrule
\end{tabular}
\end{center}
\caption{\textbf{Best Table for L63, ddd, normalize full, test sequential, noise-free init.cond., noise level 1e-07}. See the beginning of \cref{app:sec:details} for a description.}
\end{table}
\begin{figure}[ht!]
\begin{center}
\includegraphics[width=\textwidth]{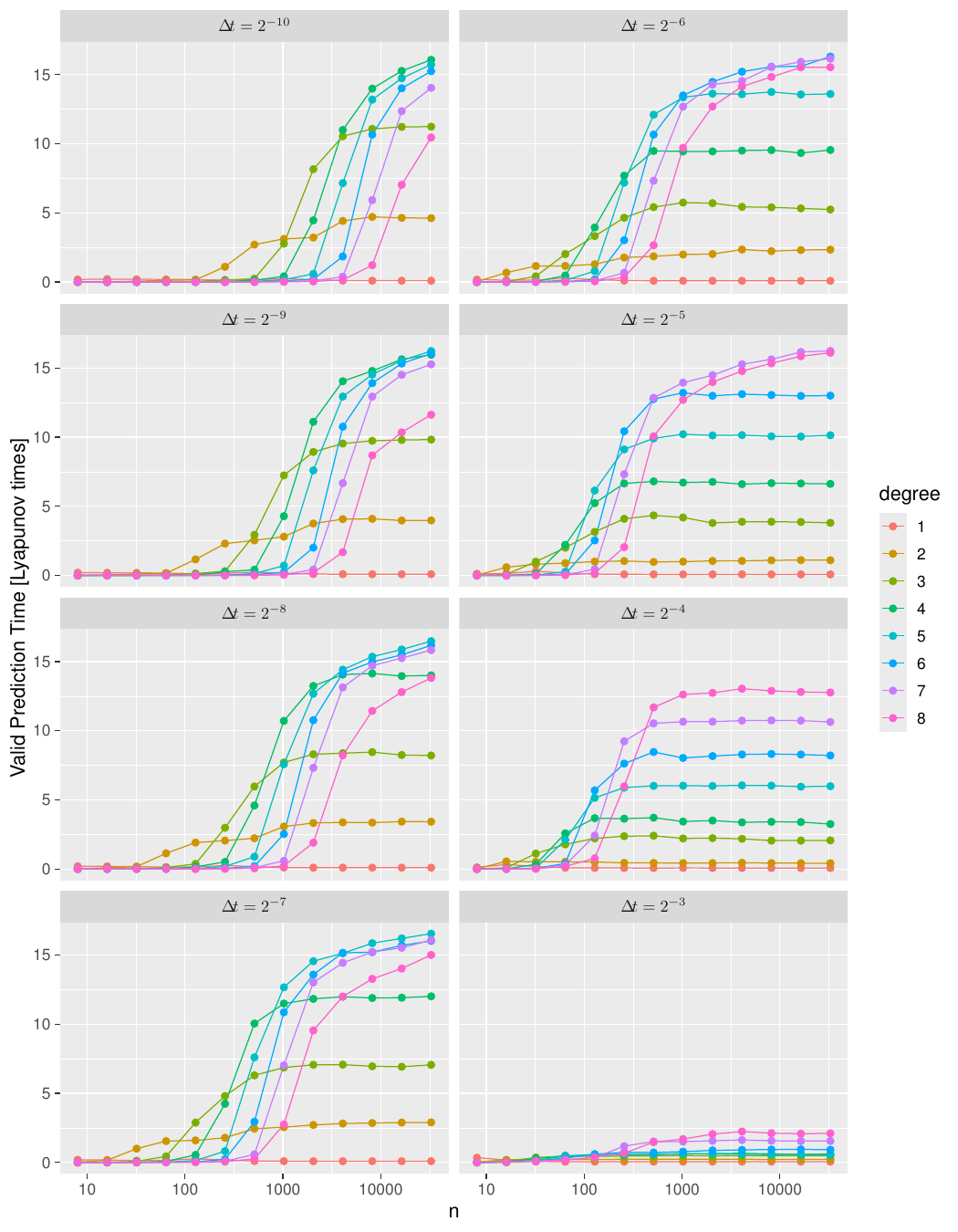}
\end{center}
\caption{\textbf{All Plot for L63, ddd, normalize full, test sequential, noise-free init.cond., noise level 1e-07}. See the beginning of \cref{app:sec:details} for a description.}
\end{figure}

\clearpage
\subsection{L63, ddd, normalize full, test sequential, noise-free init.cond., noise level 1e-05}
\begin{figure}[ht!]
\begin{center}
\includegraphics[width=\textwidth]{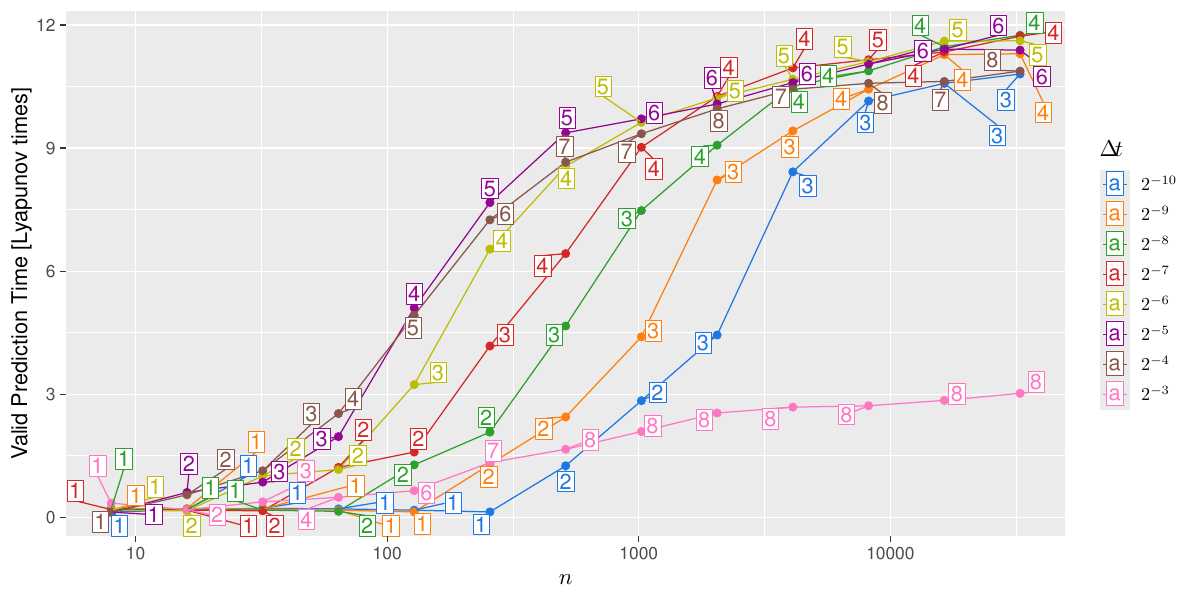}
\end{center}
\caption{\textbf{Best Plot for L63, ddd, normalize full, test sequential, noise-free init.cond., noise level 1e-05}. See the beginning of \cref{app:sec:details} for a description.}
\end{figure}
\begin{table}[ht!]
\input{tbl/L63_ddd_f_s_VPT_FALSE_5_best_table.tex}
\caption{\textbf{Best Table for L63, ddd, normalize full, test sequential, noise-free init.cond., noise level 1e-05}. See the beginning of \cref{app:sec:details} for a description.}
\end{table}
\begin{figure}[ht!]
\begin{center}
\includegraphics[width=\textwidth]{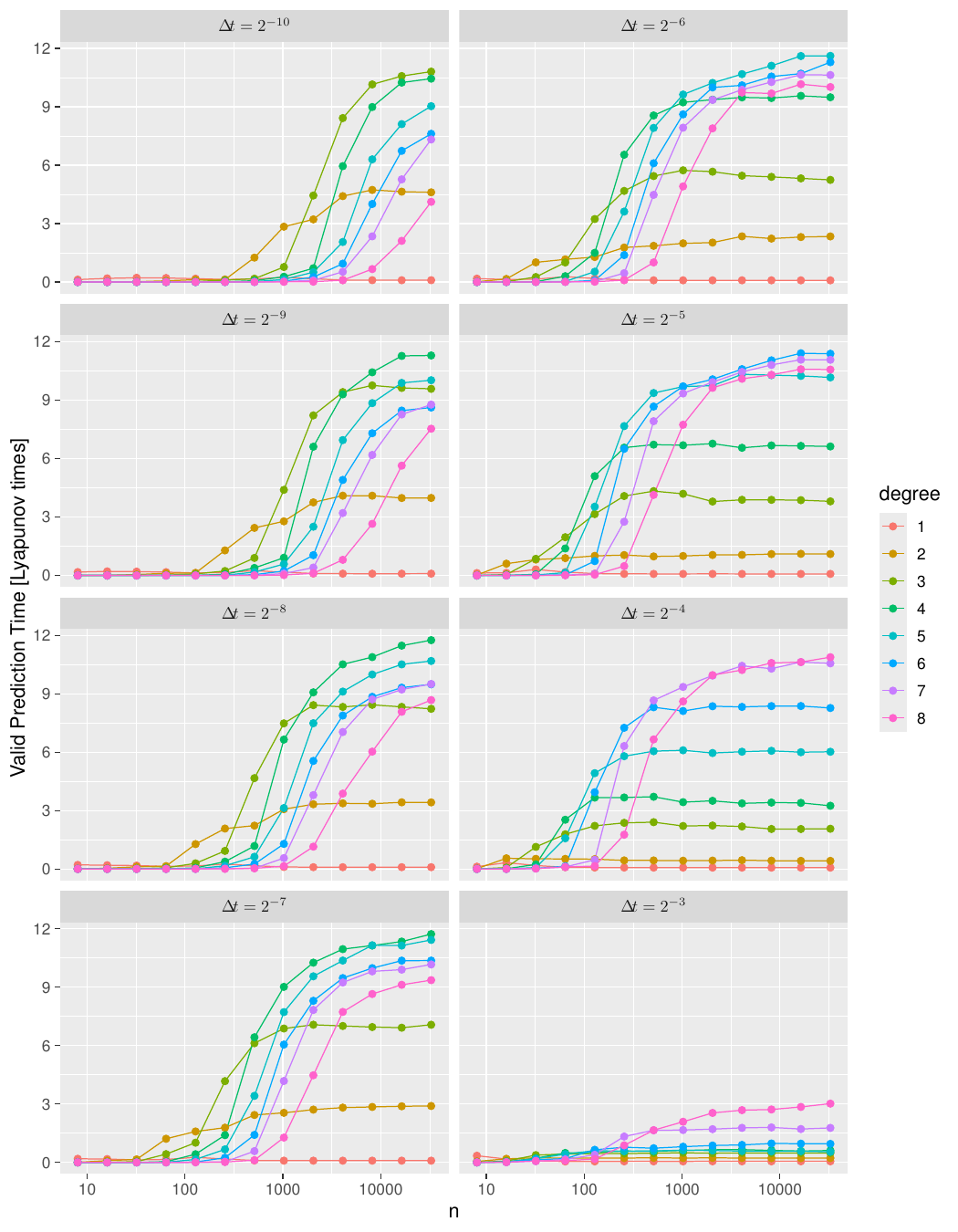}
\end{center}
\caption{\textbf{All Plot for L63, ddd, normalize full, test sequential, noise-free init.cond., noise level 1e-05}. See the beginning of \cref{app:sec:details} for a description.}
\end{figure}

\clearpage
\subsection{L63, ddd, normalize full, test sequential, noise-free init.cond., noise level 1e-03}
\begin{figure}[ht!]
\begin{center}
\includegraphics[width=\textwidth]{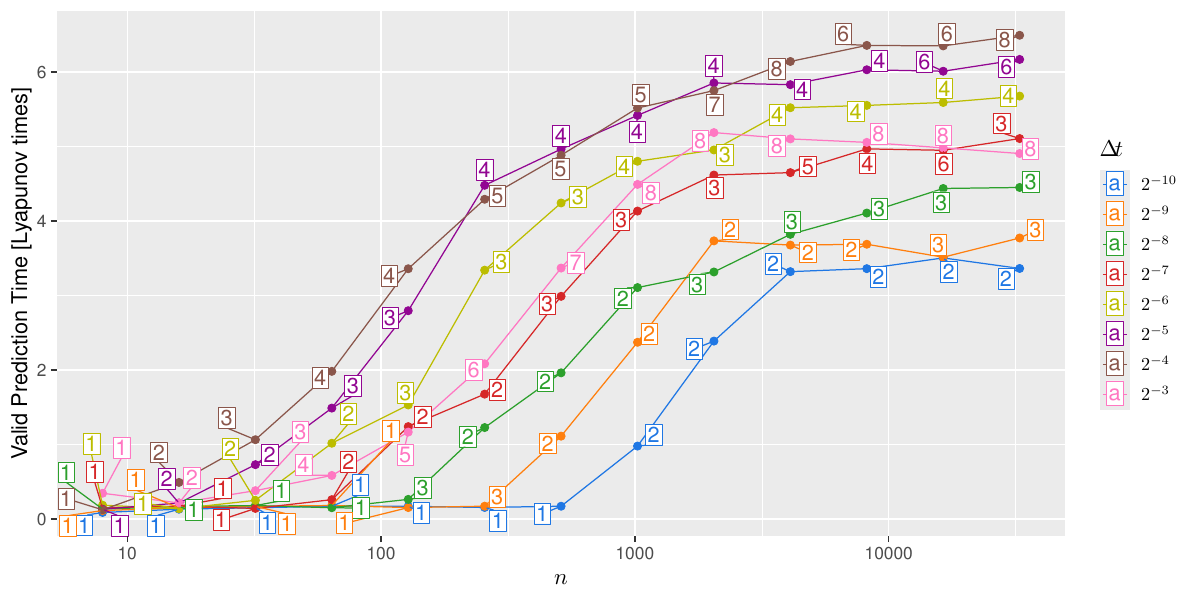}
\end{center}
\caption{\textbf{Best Plot for L63, ddd, normalize full, test sequential, noise-free init.cond., noise level 1e-03}. See the beginning of \cref{app:sec:details} for a description.}
\end{figure}
\begin{table}[ht!]
\input{tbl/L63_ddd_f_s_VPT_FALSE_3_best_table.tex}
\caption{\textbf{Best Table for L63, ddd, normalize full, test sequential, noise-free init.cond., noise level 1e-03}. See the beginning of \cref{app:sec:details} for a description.}
\end{table}
\begin{figure}[ht!]
\begin{center}
\includegraphics[width=\textwidth]{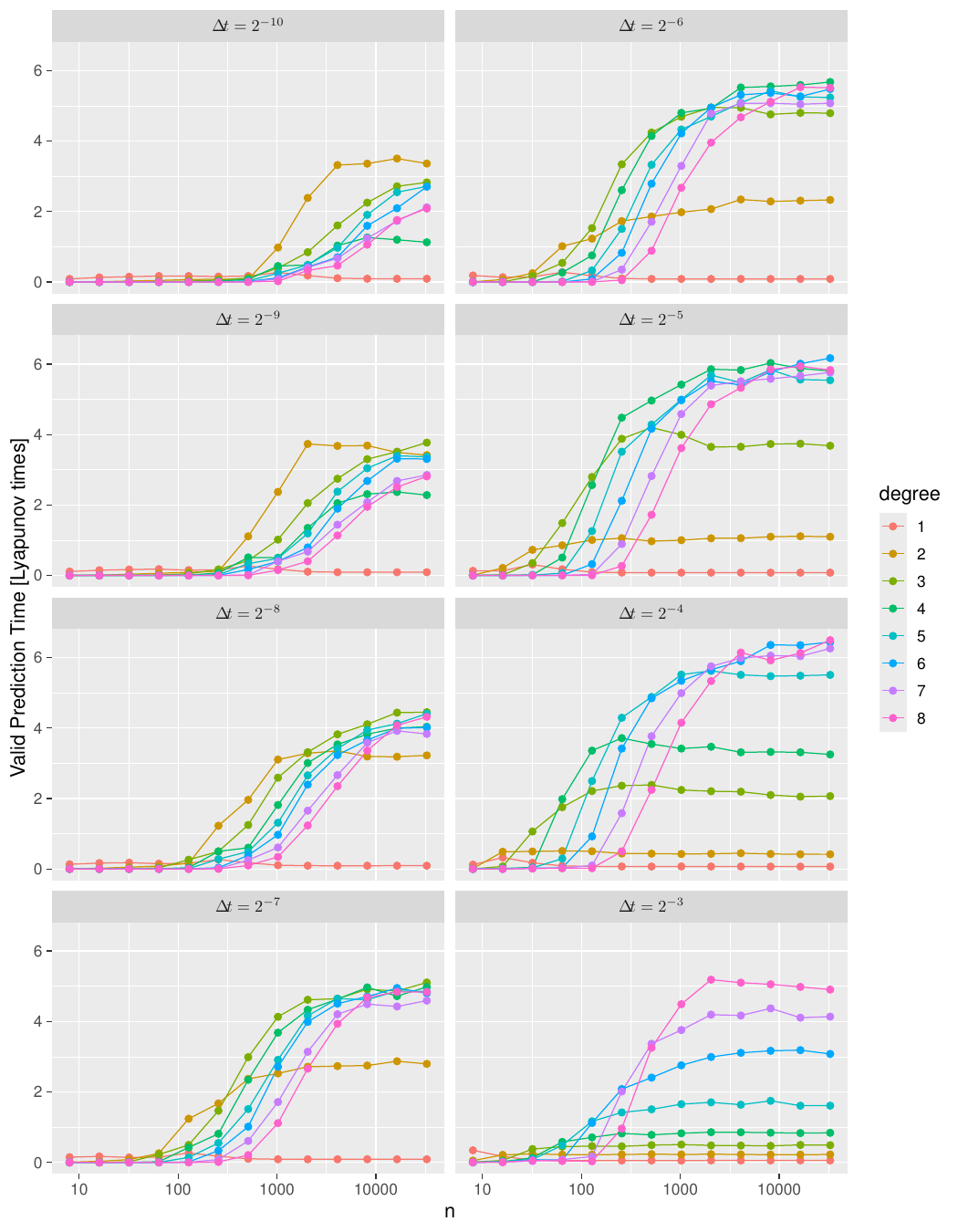}
\end{center}
\caption{\textbf{All Plot for L63, ddd, normalize full, test sequential, noise-free init.cond., noise level 1e-03}. See the beginning of \cref{app:sec:details} for a description.}
\end{figure}

\clearpage
\subsection{L63, ddd, normalize full, test sequential, noise-free init.cond., noise level 1e-01}
\begin{figure}[ht!]
\begin{center}
\includegraphics[width=\textwidth]{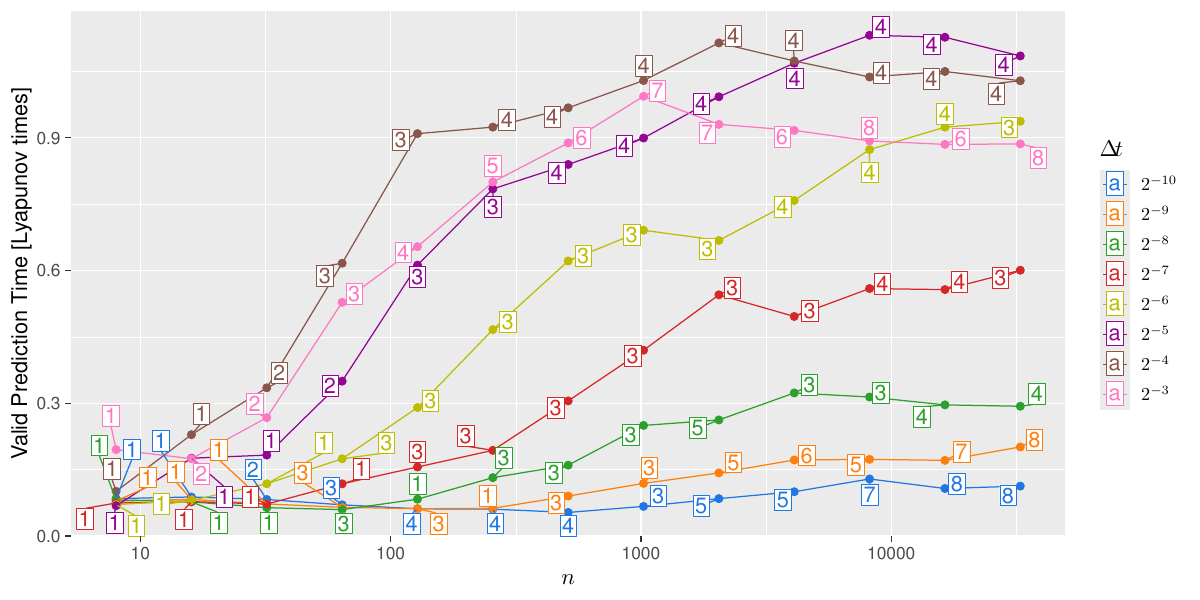}
\end{center}
\caption{\textbf{Best Plot for L63, ddd, normalize full, test sequential, noise-free init.cond., noise level 1e-01}. See the beginning of \cref{app:sec:details} for a description.}
\end{figure}
\begin{table}[ht!]
\input{tbl/L63_ddd_f_s_VPT_FALSE_1_best_table.tex}
\caption{\textbf{Best Table for L63, ddd, normalize full, test sequential, noise-free init.cond., noise level 1e-01}. See the beginning of \cref{app:sec:details} for a description.}
\end{table}
\begin{figure}[ht!]
\begin{center}
\includegraphics[width=\textwidth]{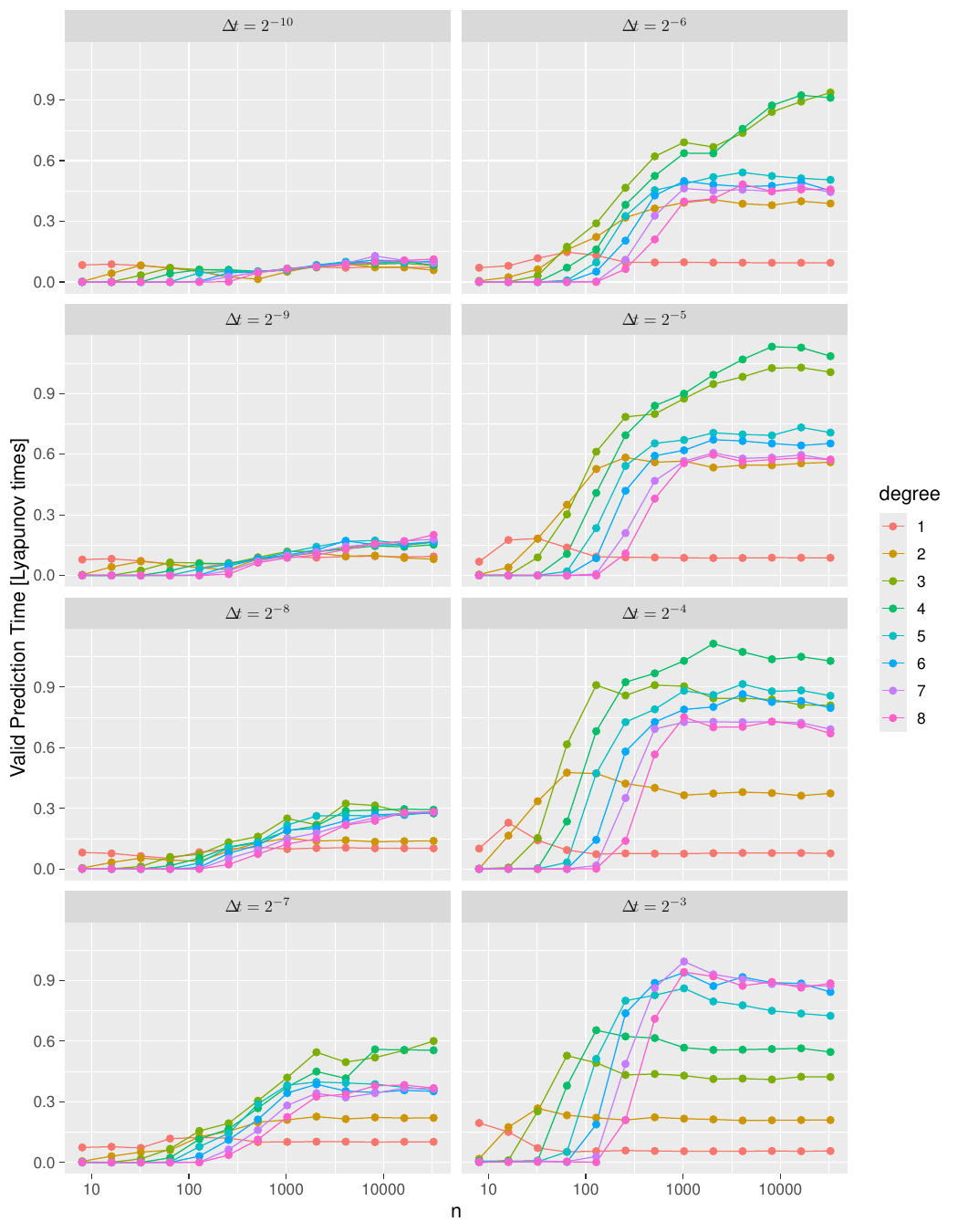}
\end{center}
\caption{\textbf{All Plot for L63, ddd, normalize full, test sequential, noise-free init.cond., noise level 1e-01}. See the beginning of \cref{app:sec:details} for a description.}
\end{figure}

\clearpage
\subsection{L63, ddd, normalize full, test sequential, noisy init.cond., noise level 1e-09}
\begin{figure}[ht!]
\begin{center}
\includegraphics[width=\textwidth]{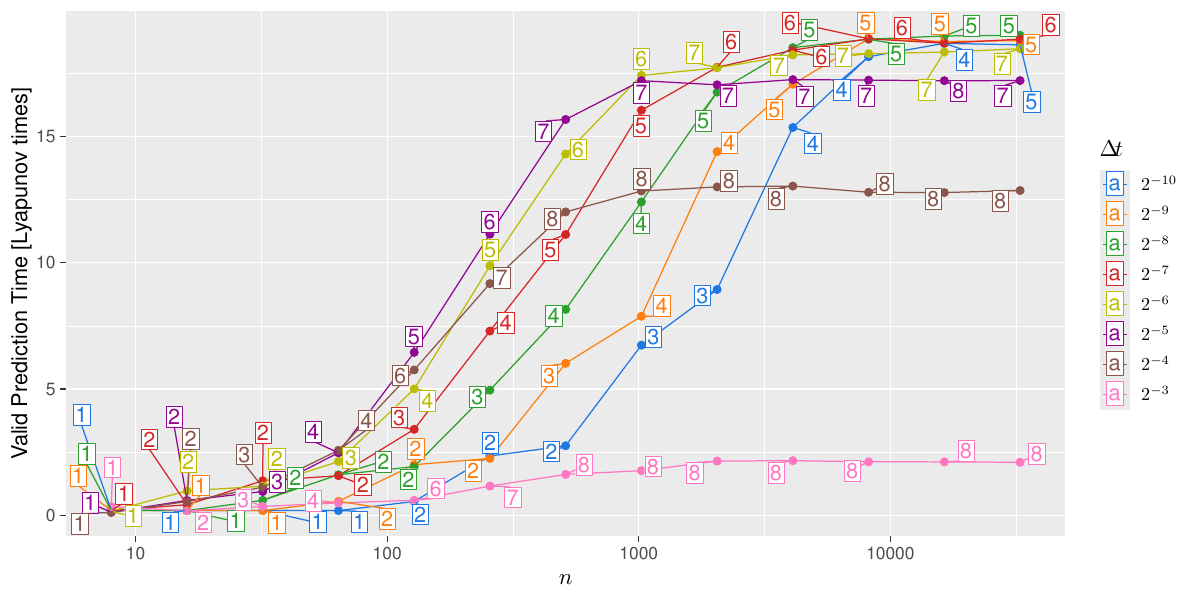}
\end{center}
\caption{\textbf{Best Plot for L63, ddd, normalize full, test sequential, noisy init.cond., noise level 1e-09}. See the beginning of \cref{app:sec:details} for a description.}
\end{figure}
\begin{table}[ht!]
\input{tbl/L63_ddd_f_s_VPT_TRUE_9_best_table.tex}
\caption{\textbf{Best Table for L63, ddd, normalize full, test sequential, noisy init.cond., noise level 1e-09}. See the beginning of \cref{app:sec:details} for a description.}
\end{table}
\begin{figure}[ht!]
\begin{center}
\includegraphics[width=\textwidth]{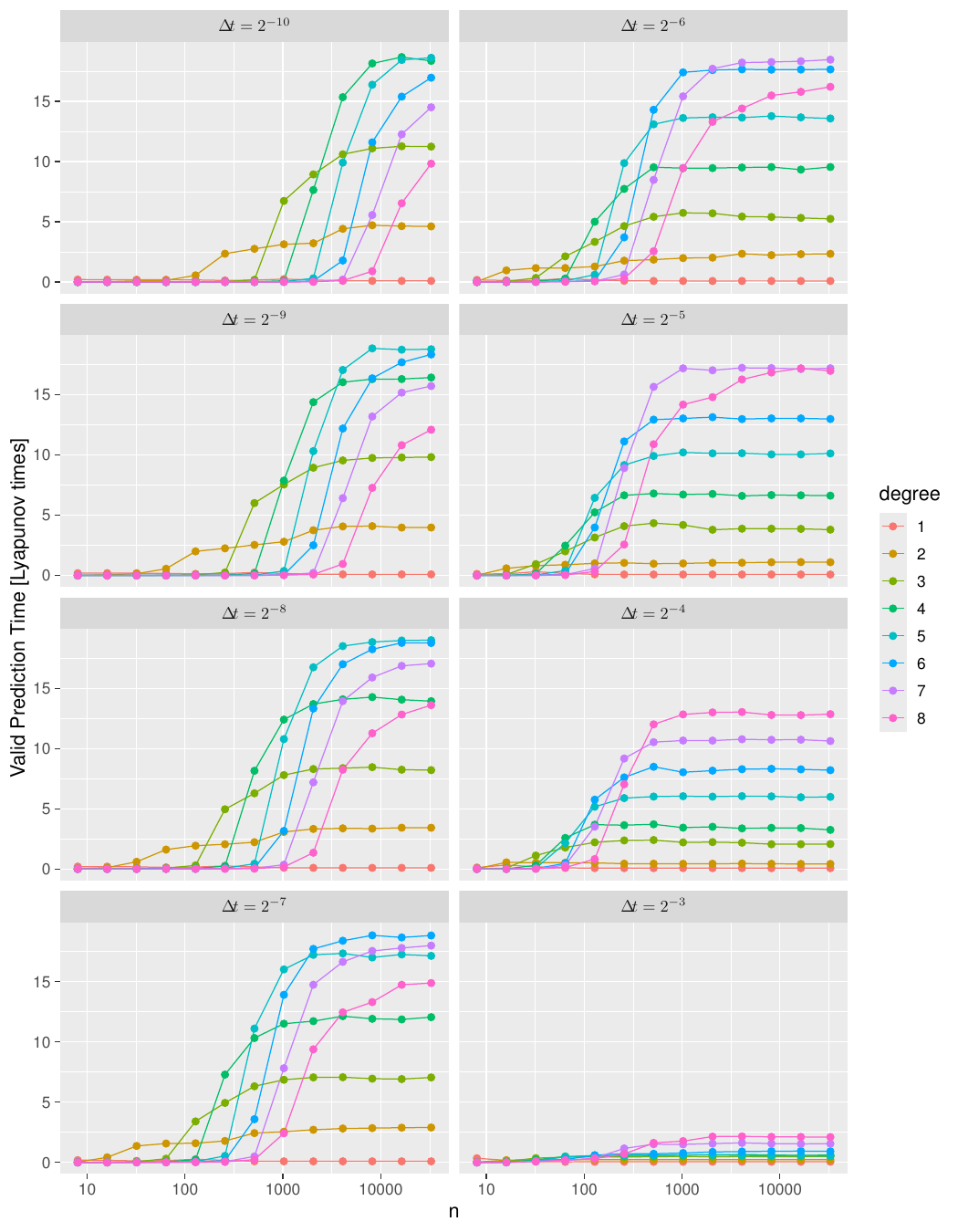}
\end{center}
\caption{\textbf{All Plot for L63, ddd, normalize full, test sequential, noisy init.cond., noise level 1e-09}. See the beginning of \cref{app:sec:details} for a description.}
\end{figure}

\clearpage
\subsection{L63, ddd, normalize full, test sequential, noisy init.cond., noise level 1e-07}
\begin{figure}[ht!]
\begin{center}
\includegraphics[width=\textwidth]{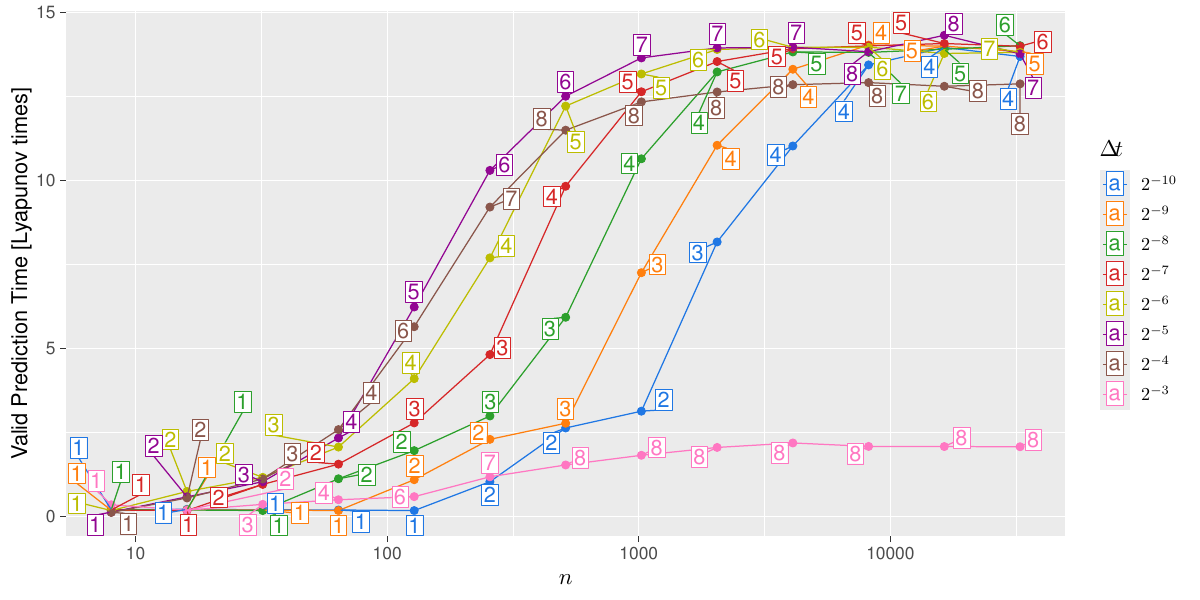}
\end{center}
\caption{\textbf{Best Plot for L63, ddd, normalize full, test sequential, noisy init.cond., noise level 1e-07}. See the beginning of \cref{app:sec:details} for a description.}
\end{figure}
\begin{table}[ht!]
\input{tbl/L63_ddd_f_s_VPT_TRUE_7_best_table.tex}
\caption{\textbf{Best Table for L63, ddd, normalize full, test sequential, noisy init.cond., noise level 1e-07}. See the beginning of \cref{app:sec:details} for a description.}
\end{table}
\begin{figure}[ht!]
\begin{center}
\includegraphics[width=\textwidth]{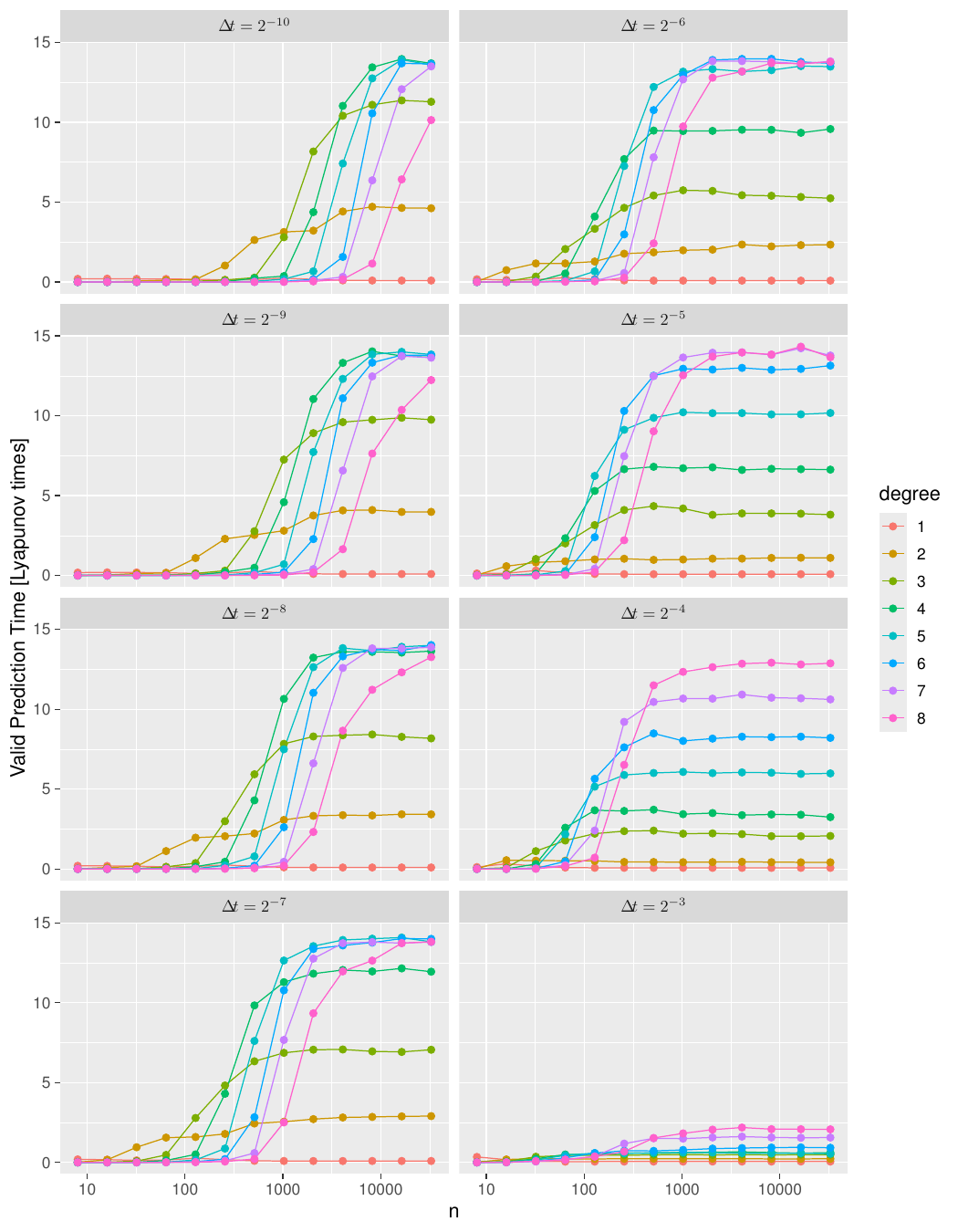}
\end{center}
\caption{\textbf{All Plot for L63, ddd, normalize full, test sequential, noisy init.cond., noise level 1e-07}. See the beginning of \cref{app:sec:details} for a description.}
\end{figure}

\clearpage
\subsection{L63, ddd, normalize full, test sequential, noisy init.cond., noise level 1e-05}
\begin{figure}[ht!]
\begin{center}
\includegraphics[width=\textwidth]{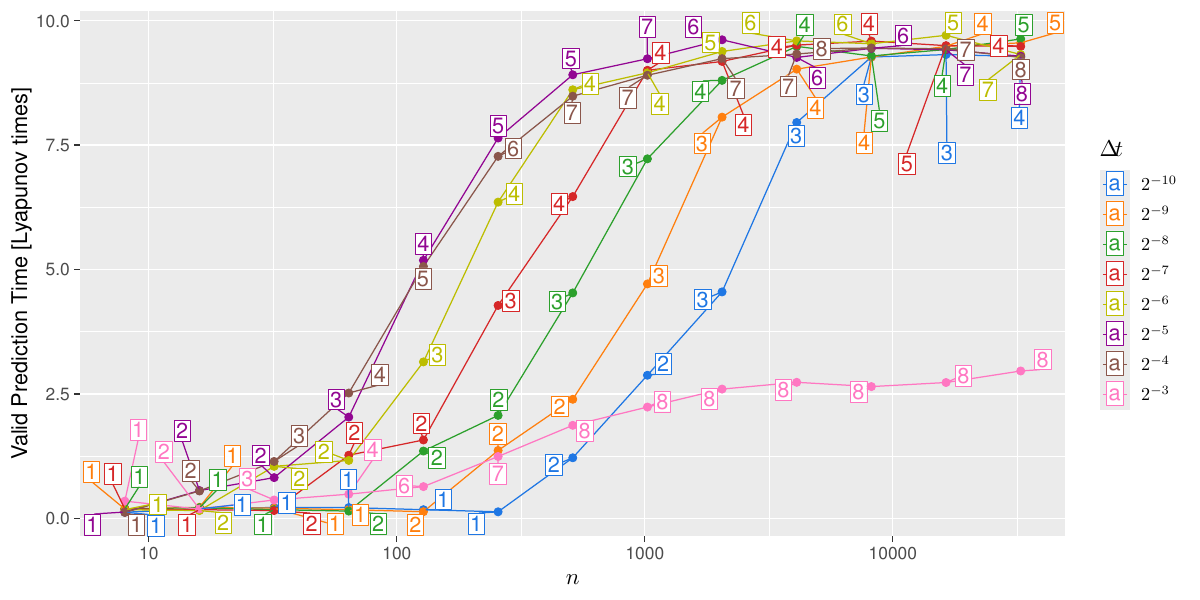}
\end{center}
\caption{\textbf{Best Plot for L63, ddd, normalize full, test sequential, noisy init.cond., noise level 1e-05}. See the beginning of \cref{app:sec:details} for a description.}
\end{figure}
\begin{table}[ht!]
\input{tbl/L63_ddd_f_s_VPT_TRUE_5_best_table.tex}
\caption{\textbf{Best Table for L63, ddd, normalize full, test sequential, noisy init.cond., noise level 1e-05}. See the beginning of \cref{app:sec:details} for a description.}
\end{table}
\begin{figure}[ht!]
\begin{center}
\includegraphics[width=\textwidth]{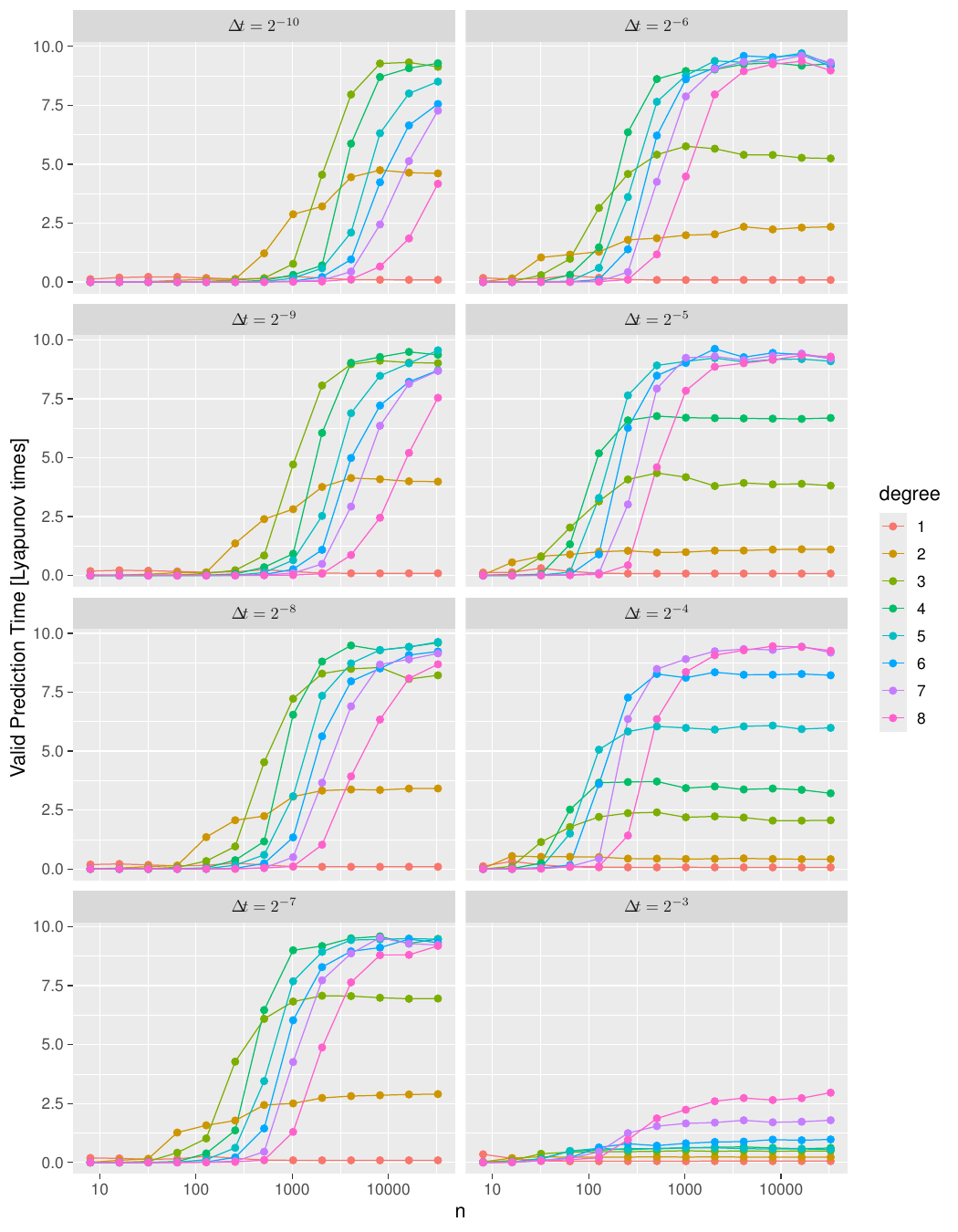}
\end{center}
\caption{\textbf{All Plot for L63, ddd, normalize full, test sequential, noisy init.cond., noise level 1e-05}. See the beginning of \cref{app:sec:details} for a description.}
\end{figure}

\clearpage
\subsection{L63, ddd, normalize full, test sequential, noisy init.cond., noise level 1e-03}
\begin{figure}[ht!]
\begin{center}
\includegraphics[width=\textwidth]{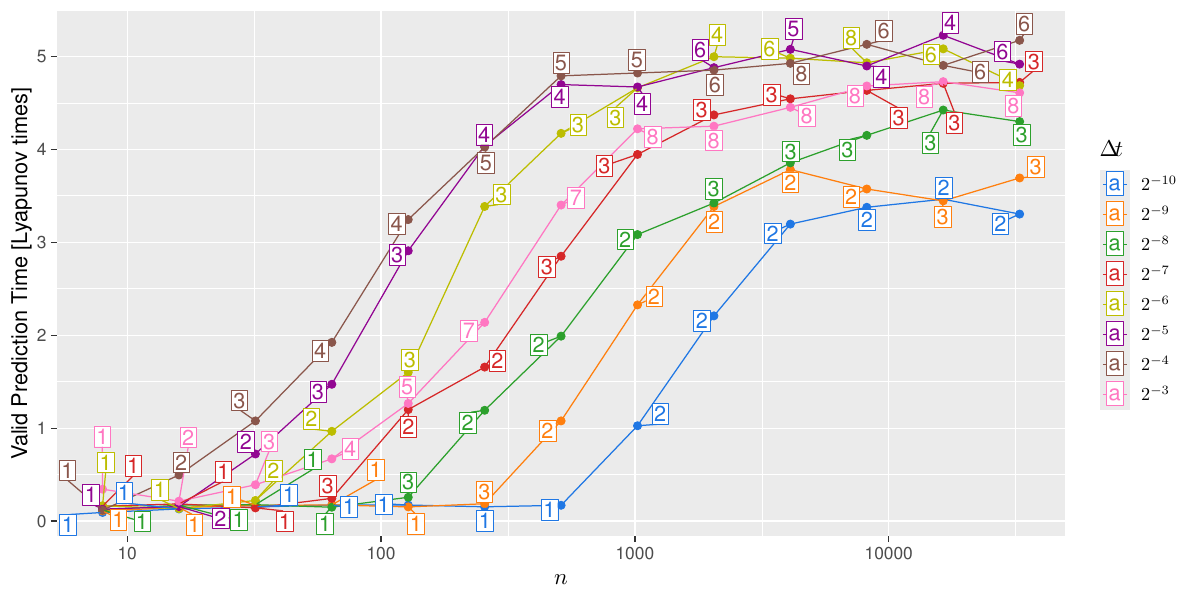}
\end{center}
\caption{\textbf{Best Plot for L63, ddd, normalize full, test sequential, noisy init.cond., noise level 1e-03}. See the beginning of \cref{app:sec:details} for a description.}
\end{figure}
\begin{table}[ht!]
\input{tbl/L63_ddd_f_s_VPT_TRUE_3_best_table.tex}
\caption{\textbf{Best Table for L63, ddd, normalize full, test sequential, noisy init.cond., noise level 1e-03}. See the beginning of \cref{app:sec:details} for a description.}
\end{table}
\begin{figure}[ht!]
\begin{center}
\includegraphics[width=\textwidth]{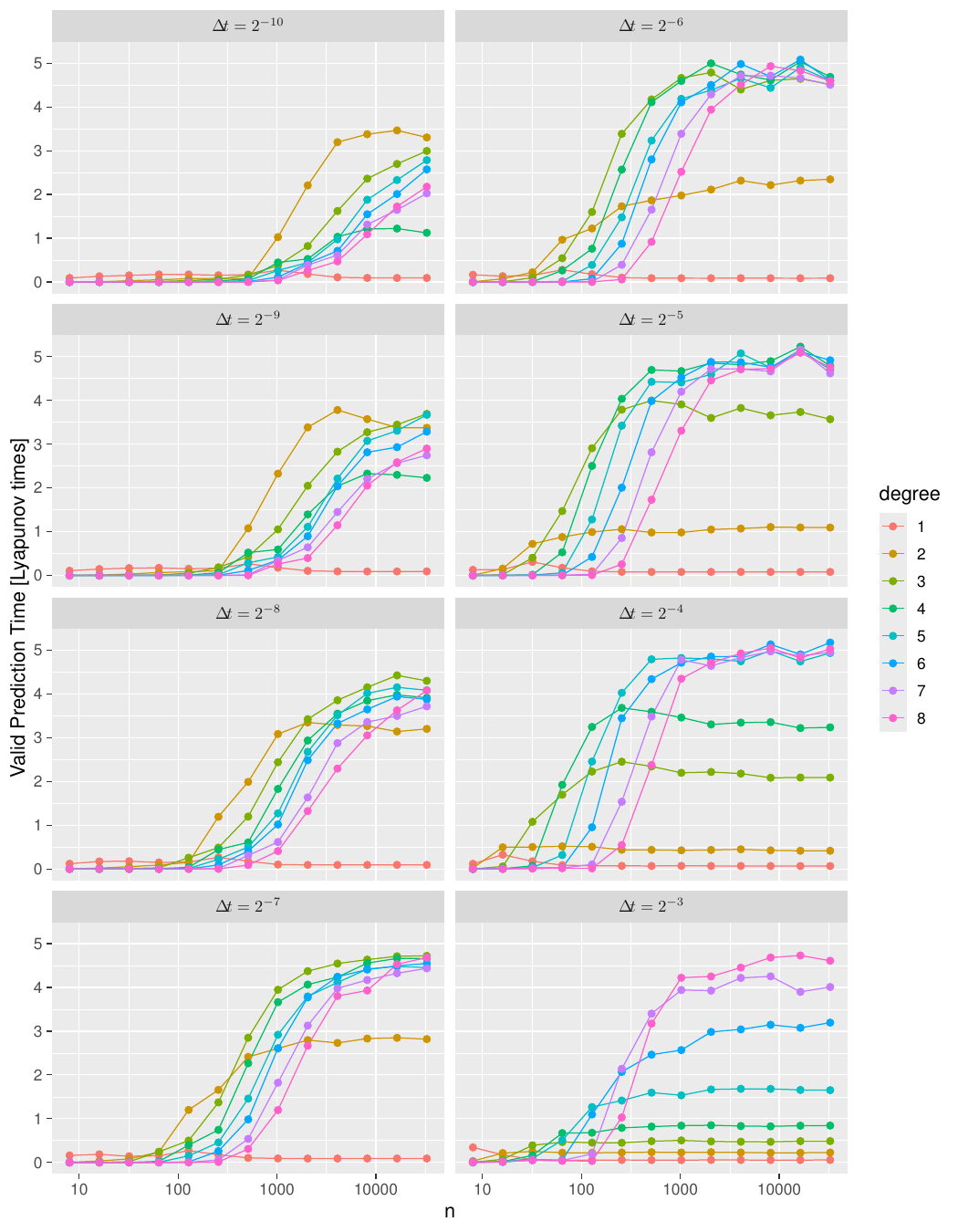}
\end{center}
\caption{\textbf{All Plot for L63, ddd, normalize full, test sequential, noisy init.cond., noise level 1e-03}. See the beginning of \cref{app:sec:details} for a description.}
\end{figure}

\clearpage
\subsection{L63, ddd, normalize full, test sequential, noisy init.cond., noise level 1e-01}
\begin{figure}[ht!]
\begin{center}
\includegraphics[width=\textwidth]{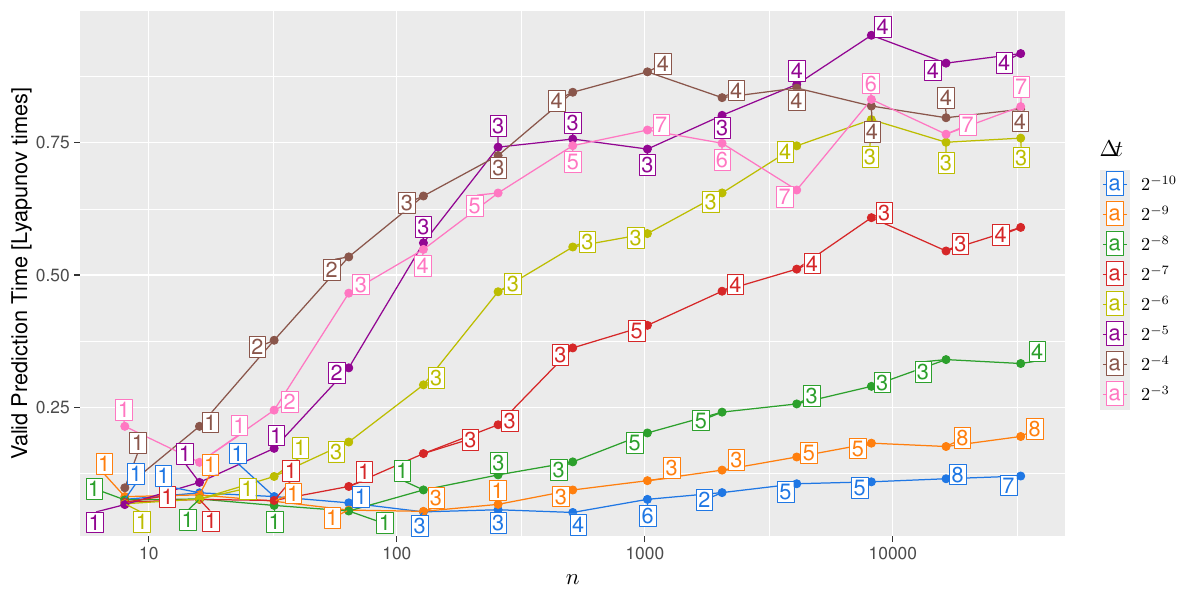}
\end{center}
\caption{\textbf{Best Plot for L63, ddd, normalize full, test sequential, noisy init.cond., noise level 1e-01}. See the beginning of \cref{app:sec:details} for a description.}
\end{figure}
\begin{table}[ht!]
\input{tbl/L63_ddd_f_s_VPT_TRUE_1_best_table.tex}
\caption{\textbf{Best Table for L63, ddd, normalize full, test sequential, noisy init.cond., noise level 1e-01}. See the beginning of \cref{app:sec:details} for a description.}
\end{table}
\begin{figure}[ht!]
\begin{center}
\includegraphics[width=\textwidth]{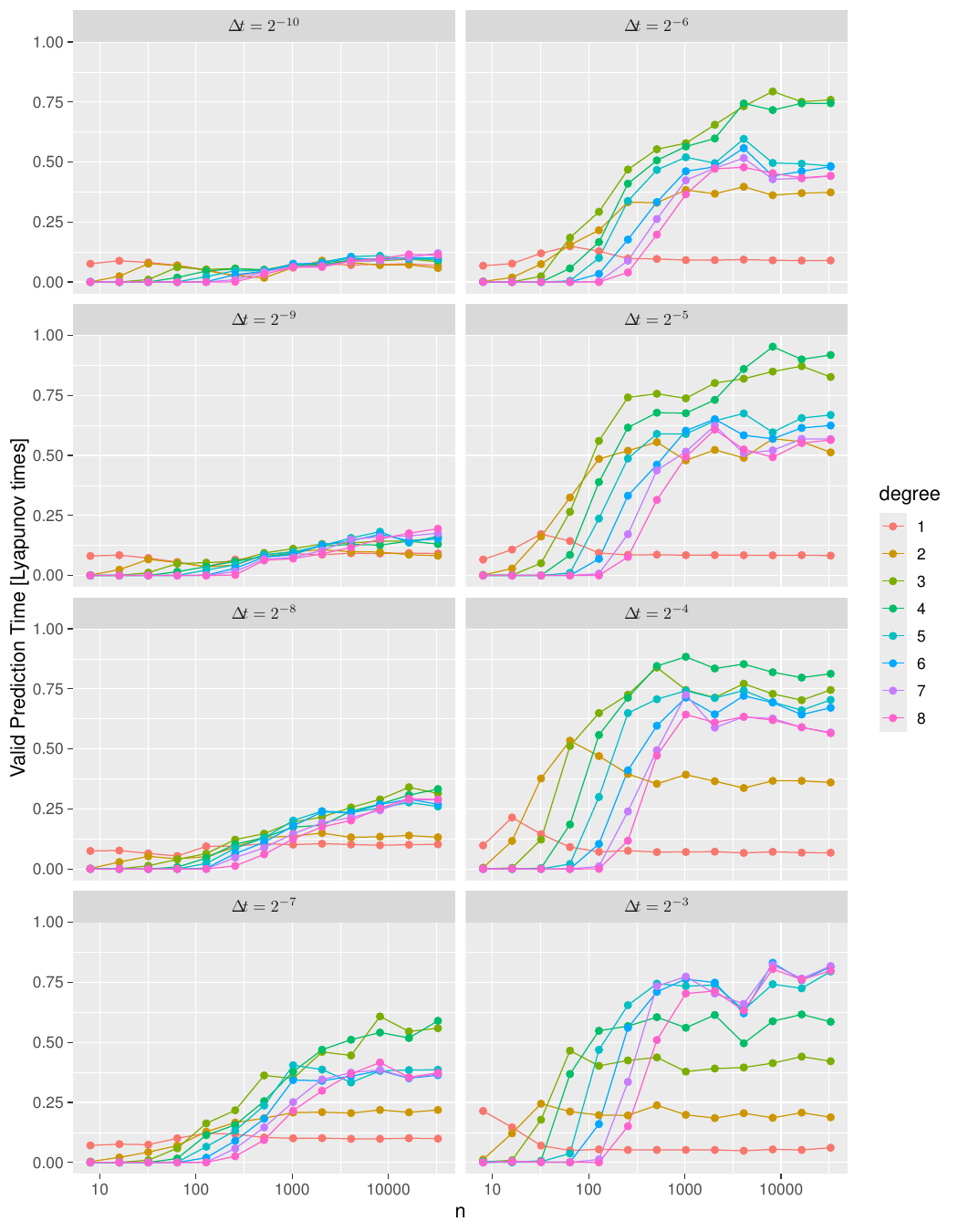}
\end{center}
\caption{\textbf{All Plot for L63, ddd, normalize full, test sequential, noisy init.cond., noise level 1e-01}. See the beginning of \cref{app:sec:details} for a description.}
\end{figure}

\clearpage
\subsection{L96D5, sdd, normalize full, test sequential}
\begin{figure}[ht!]
\begin{center}
\includegraphics[width=\textwidth]{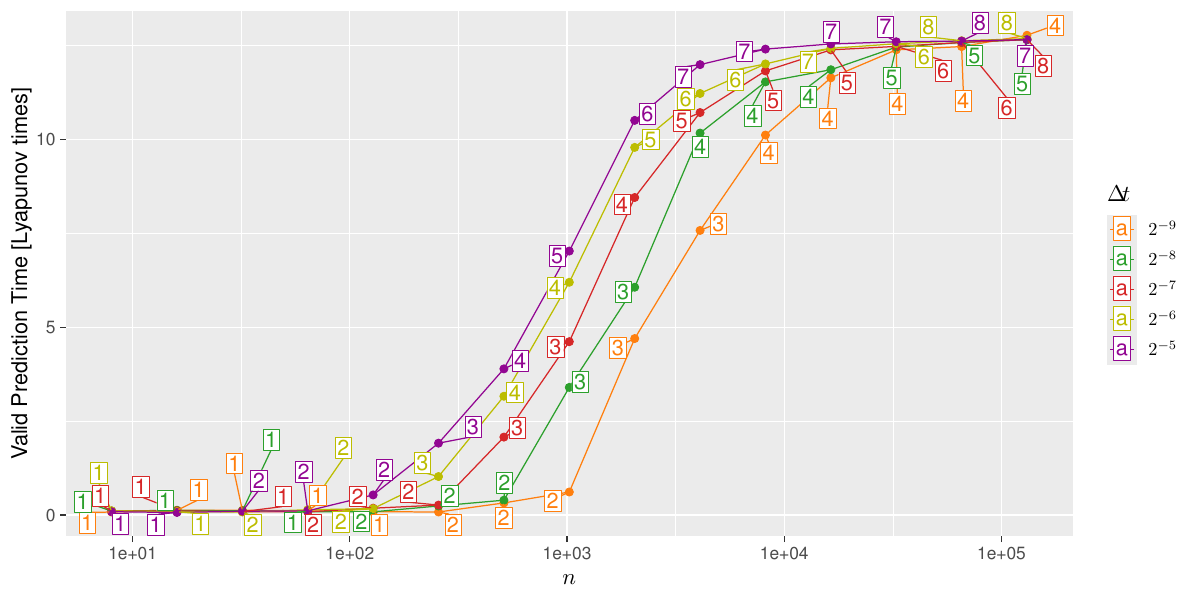}
\end{center}
\caption{\textbf{Best Plot for L96D5, sdd, normalize full, test sequential}. See the beginning of \cref{app:sec:details} for a description.}
\end{figure}
\begin{table}[ht!]
\input{tbl/L96D5_sdd_f_s_VPT_best_table.tex}
\caption{\textbf{Best Table for L96D5, sdd, normalize full, test sequential}. See the beginning of \cref{app:sec:details} for a description.}
\end{table}
\begin{figure}[ht!]
\begin{center}
\includegraphics[width=\textwidth]{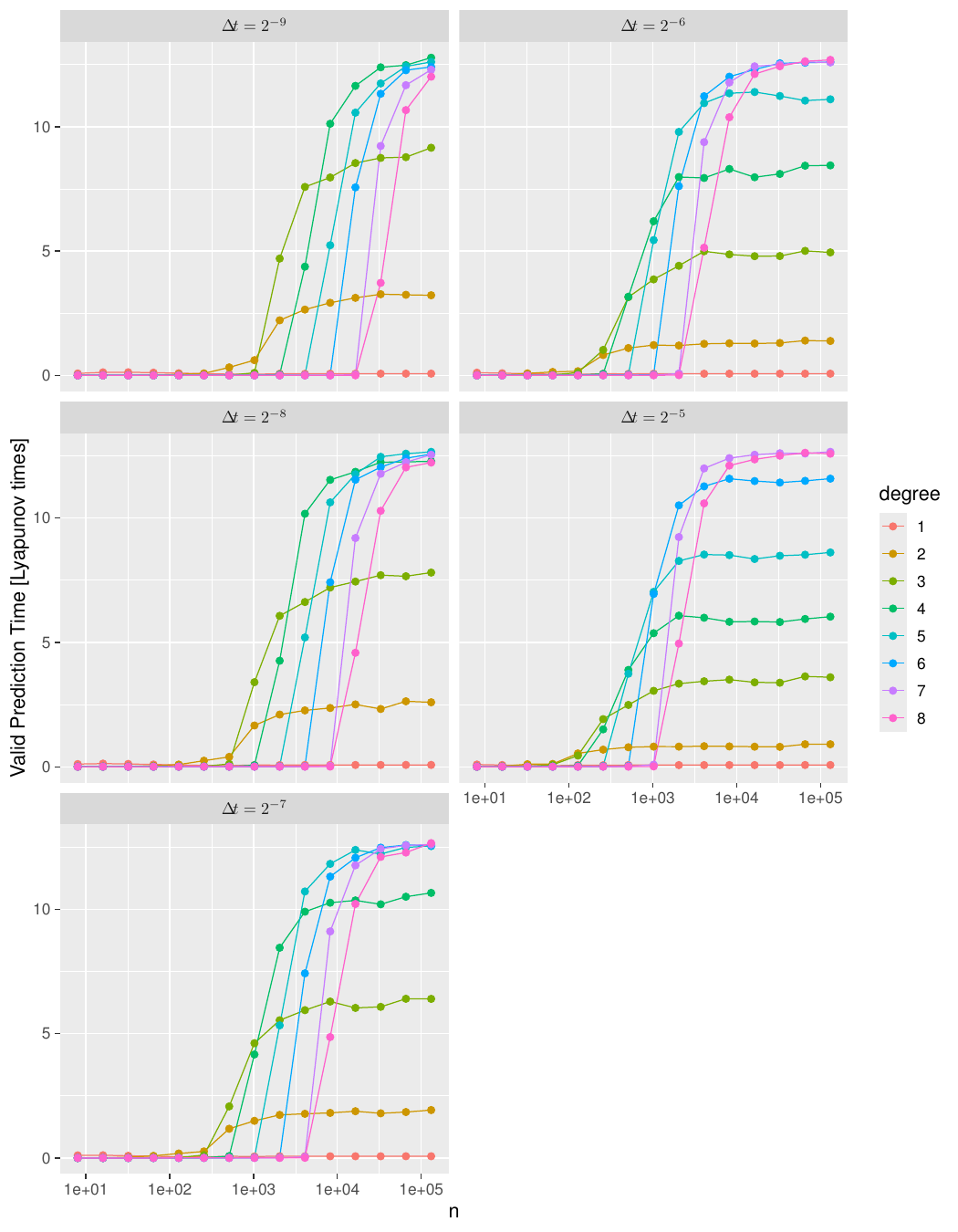}
\end{center}
\caption{\textbf{All Plot for L96D5, sdd, normalize full, test sequential}. See the beginning of \cref{app:sec:details} for a description.}
\end{figure}

\clearpage
\subsection{L96D6, sdd, normalize full, test sequential}
\begin{figure}[ht!]
\begin{center}
\includegraphics[width=\textwidth]{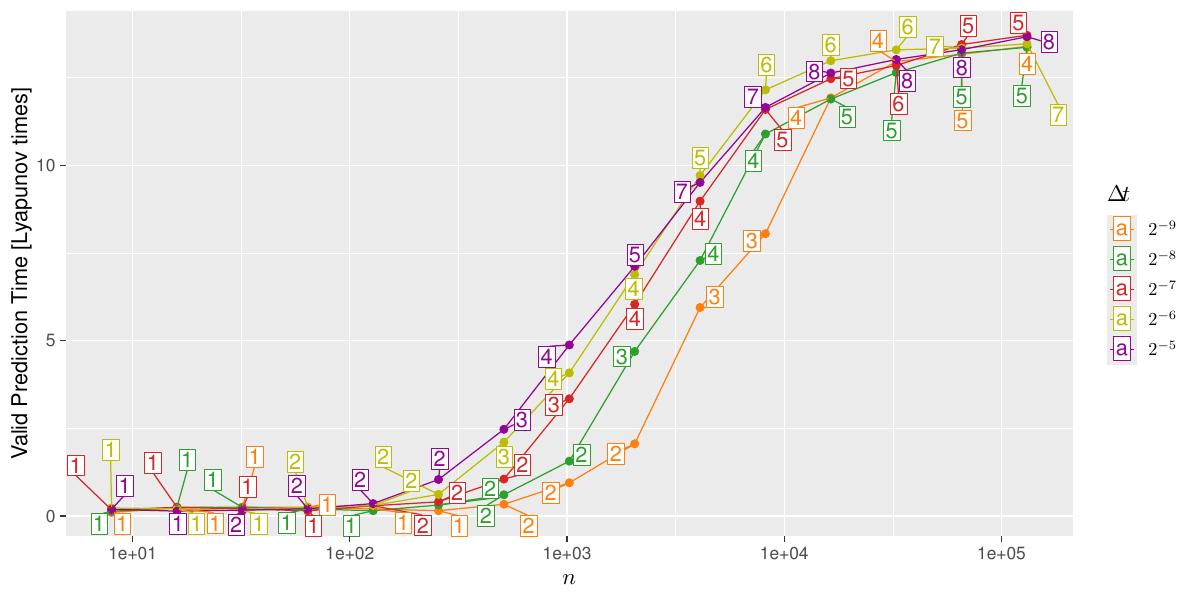}
\end{center}
\caption{\textbf{Best Plot for L96D6, sdd, normalize full, test sequential}. See the beginning of \cref{app:sec:details} for a description.}
\end{figure}
\begin{table}[ht!]
\input{tbl/L96D6_sdd_f_s_VPT_best_table.tex}
\caption{\textbf{Best Table for L96D6, sdd, normalize full, test sequential}. See the beginning of \cref{app:sec:details} for a description.}
\end{table}
\begin{figure}[ht!]
\begin{center}
\includegraphics[width=\textwidth]{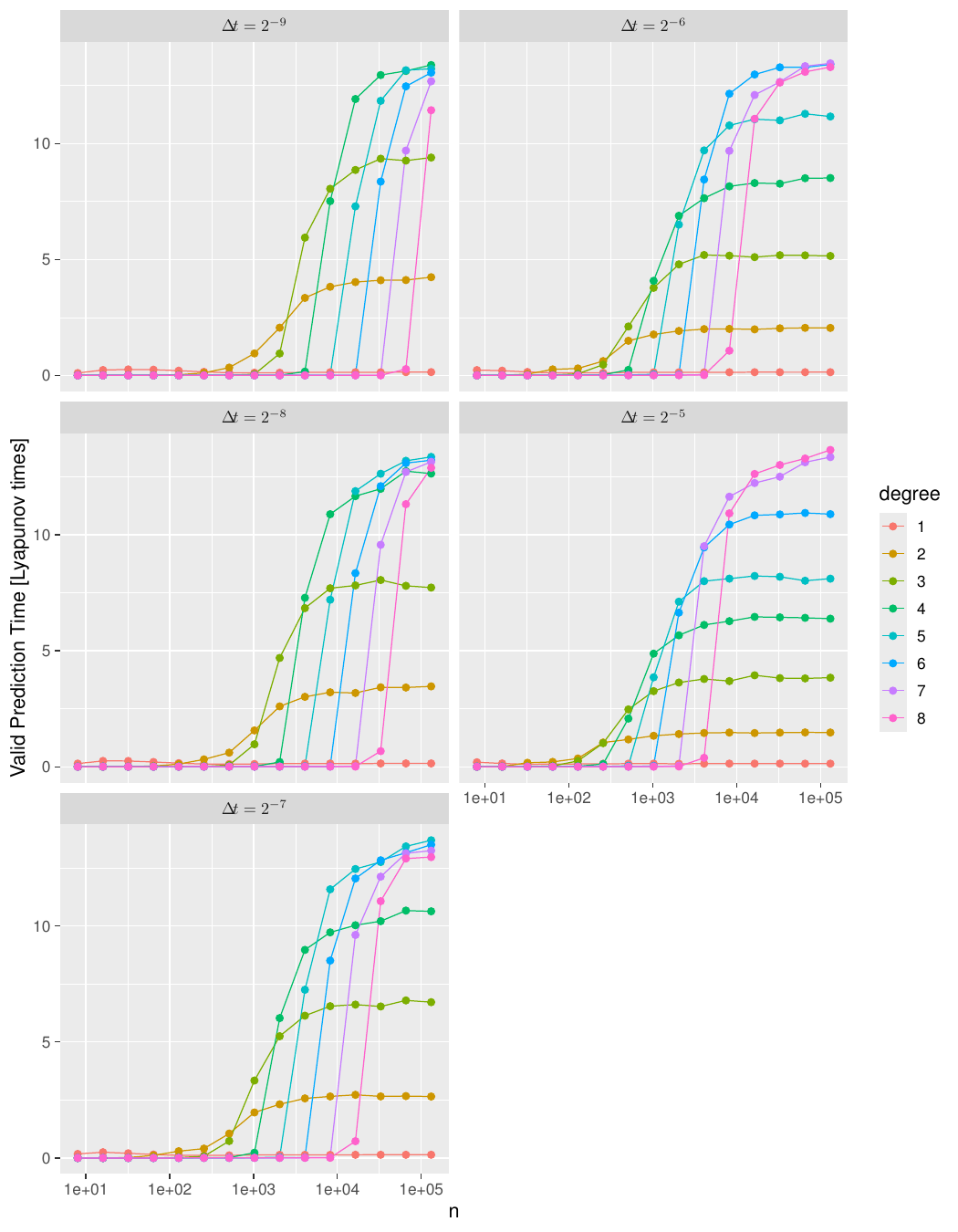}
\end{center}
\caption{\textbf{All Plot for L96D6, sdd, normalize full, test sequential}. See the beginning of \cref{app:sec:details} for a description.}
\end{figure}

\clearpage
\subsection{L96D7, sdd, normalize full, test sequential}
\begin{figure}[ht!]
\begin{center}
\includegraphics[width=\textwidth]{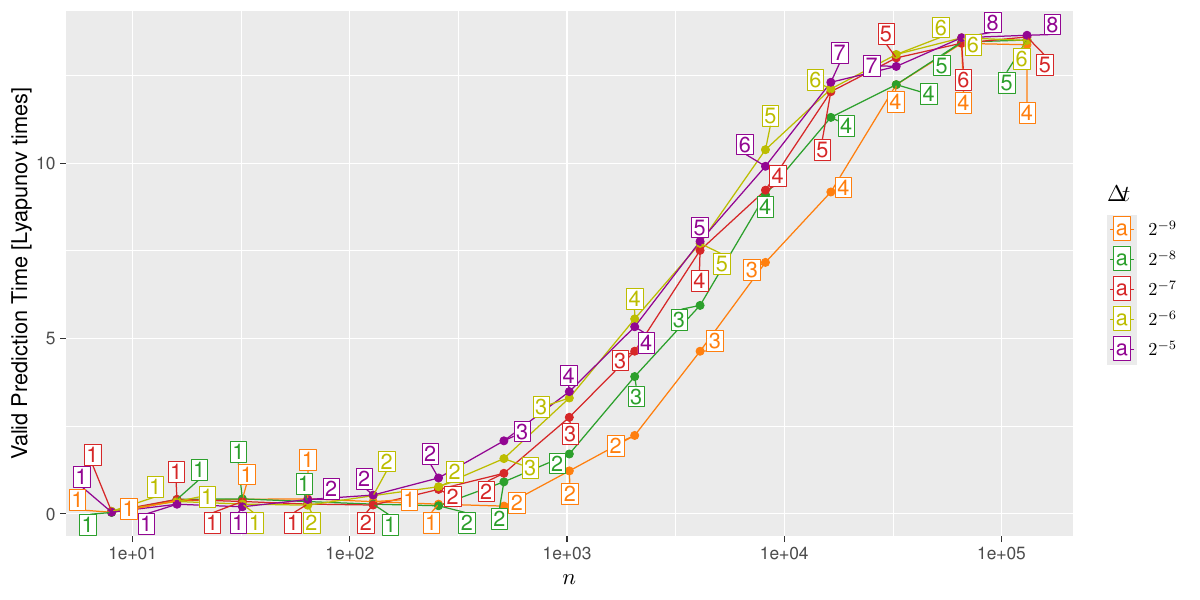}
\end{center}
\caption{\textbf{Best Plot for L96D7, sdd, normalize full, test sequential}. See the beginning of \cref{app:sec:details} for a description.}
\end{figure}
\begin{table}[ht!]
\input{tbl/L96D7_sdd_f_s_VPT_best_table.tex}
\caption{\textbf{Best Table for L96D7, sdd, normalize full, test sequential}. See the beginning of \cref{app:sec:details} for a description.}
\end{table}
\begin{figure}[ht!]
\begin{center}
\includegraphics[width=\textwidth]{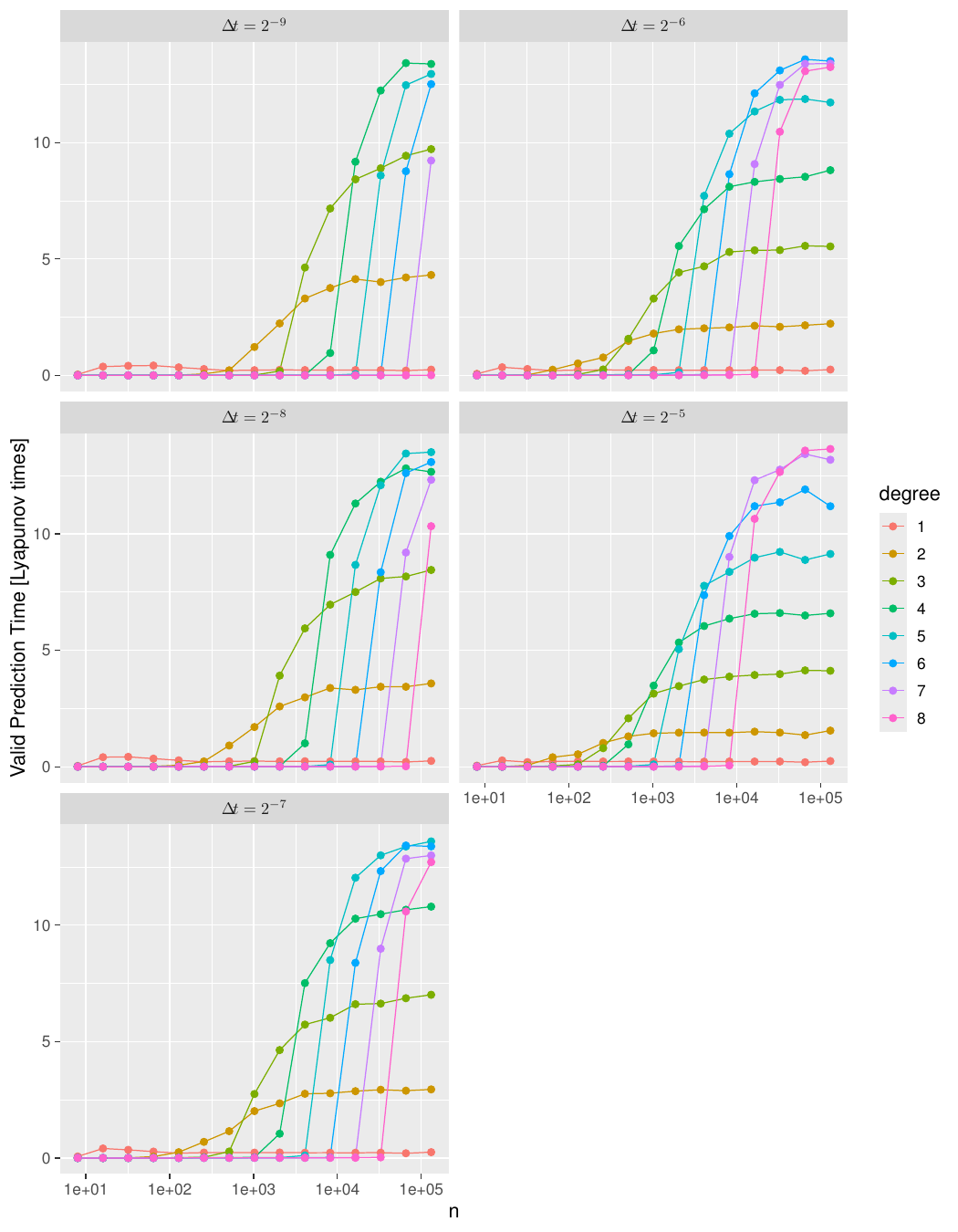}
\end{center}
\caption{\textbf{All Plot for L96D7, sdd, normalize full, test sequential}. See the beginning of \cref{app:sec:details} for a description.}
\end{figure}

\clearpage
\subsection{L96D8, sdd, normalize full, test sequential}
\begin{figure}[ht!]
\begin{center}
\includegraphics[width=\textwidth]{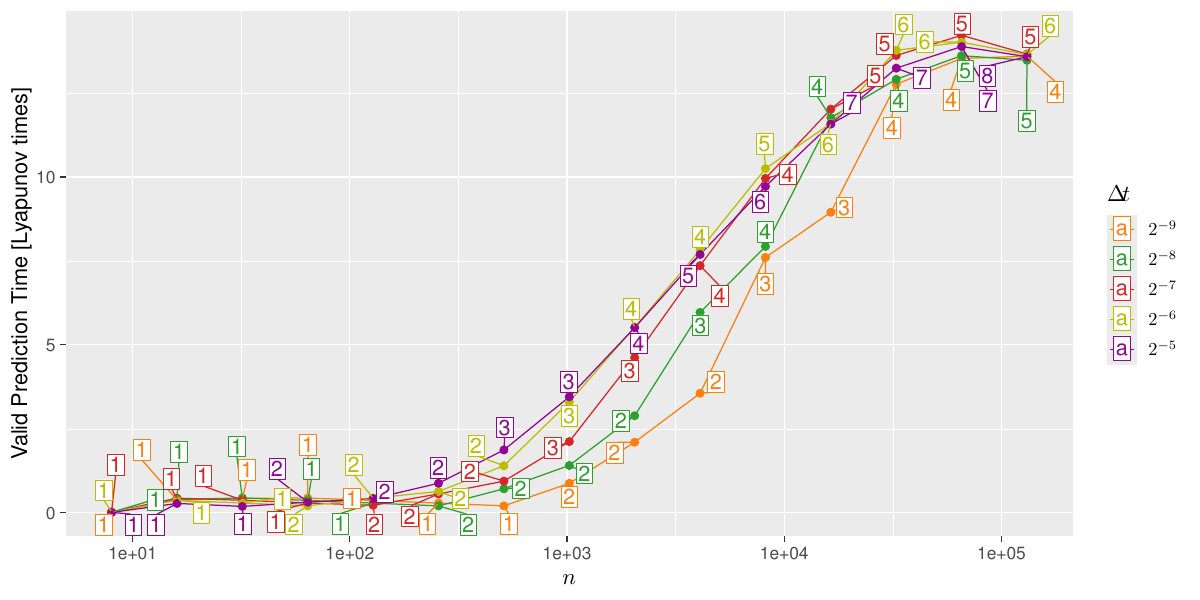}
\end{center}
\caption{\textbf{Best Plot for L96D8, sdd, normalize full, test sequential}. See the beginning of \cref{app:sec:details} for a description.}
\end{figure}
\begin{table}[ht!]
\input{tbl/L96D8_sdd_f_s_VPT_best_table.tex}
\caption{\textbf{Best Table for L96D8, sdd, normalize full, test sequential}. See the beginning of \cref{app:sec:details} for a description.}
\end{table}
\begin{figure}[ht!]
\begin{center}
\includegraphics[width=\textwidth]{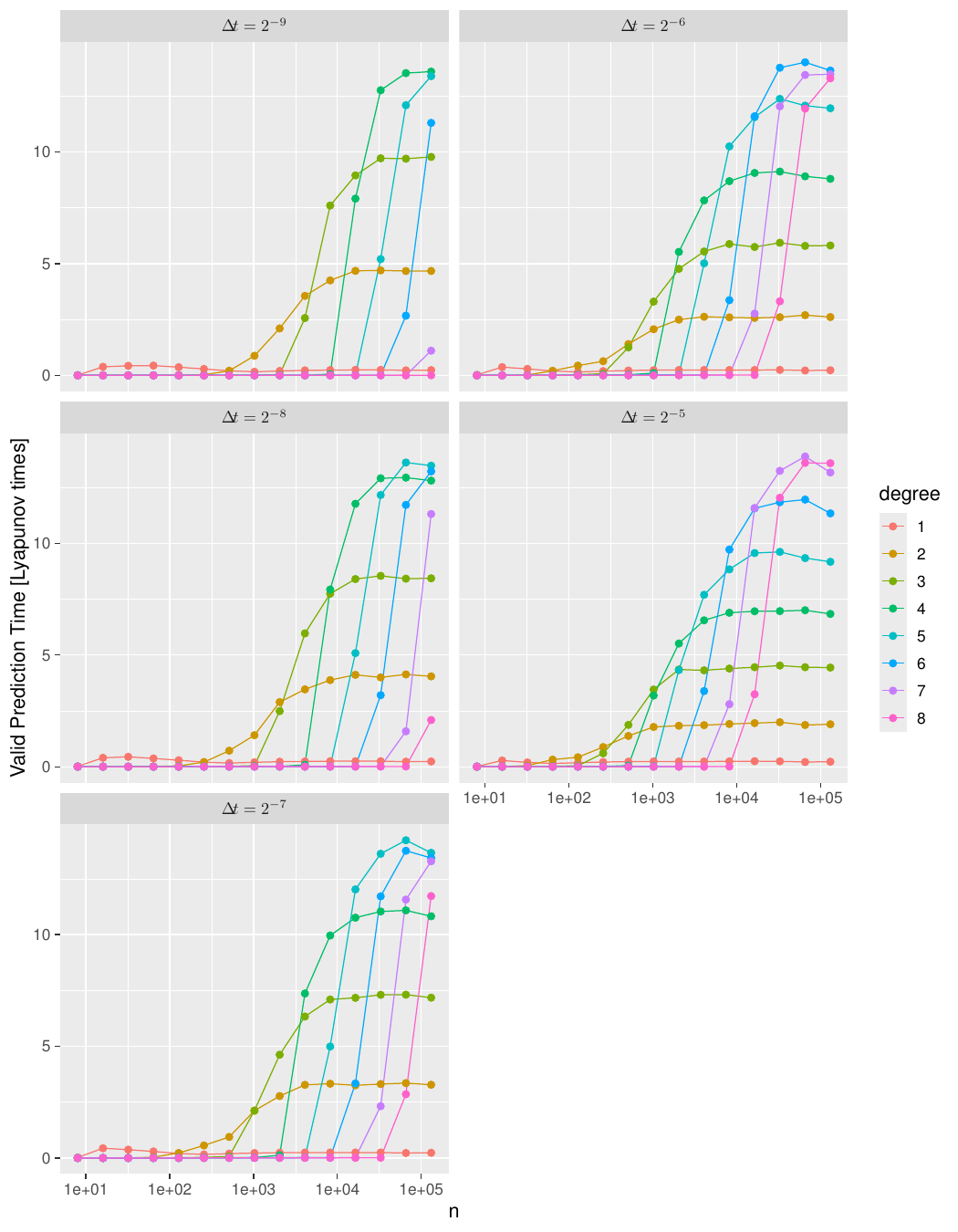}
\end{center}
\caption{\textbf{All Plot for L96D8, sdd, normalize full, test sequential}. See the beginning of \cref{app:sec:details} for a description.}
\end{figure}

\clearpage
\subsection{L96D9, sdd, normalize full, test sequential}
\begin{figure}[ht!]
\begin{center}
\includegraphics[width=\textwidth]{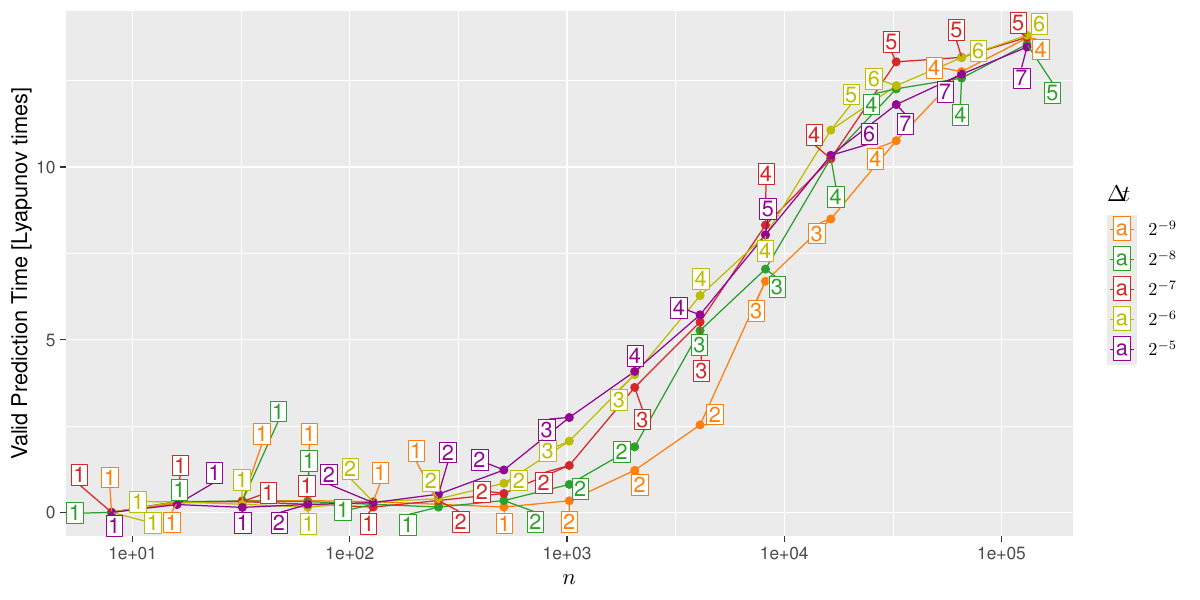}
\end{center}
\caption{\textbf{Best Plot for L96D9, sdd, normalize full, test sequential}. See the beginning of \cref{app:sec:details} for a description.}
\end{figure}
\begin{table}[ht!]
\input{tbl/L96D9_sdd_f_s_VPT_best_table.tex}
\caption{\textbf{Best Table for L96D9, sdd, normalize full, test sequential}. See the beginning of \cref{app:sec:details} for a description.}
\end{table}
\begin{figure}[ht!]
\begin{center}
\includegraphics[width=\textwidth]{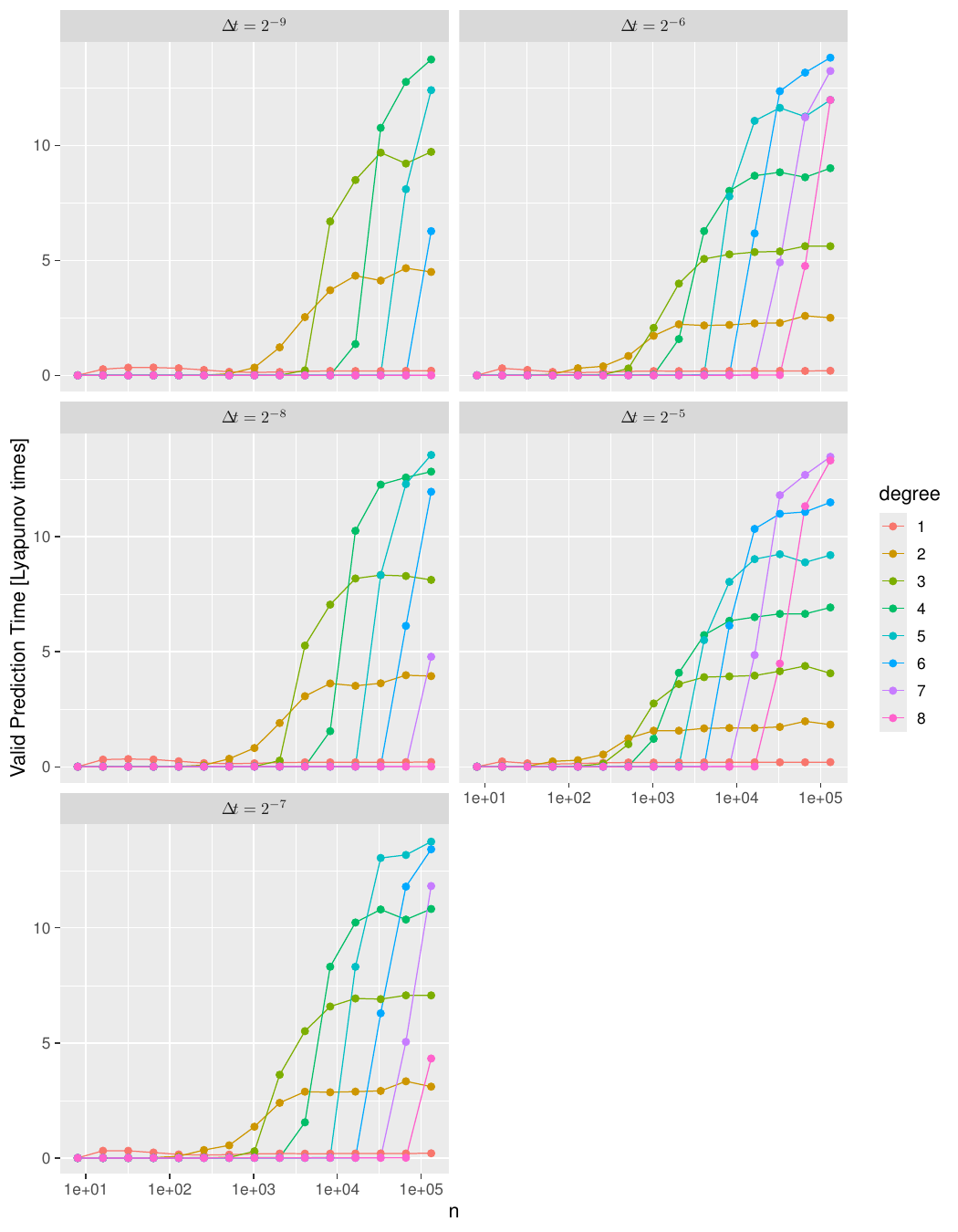}
\end{center}
\caption{\textbf{All Plot for L96D9, sdd, normalize full, test sequential}. See the beginning of \cref{app:sec:details} for a description.}
\end{figure}

\clearpage
\subsection{L96D5, dsd, normalize full, test sequential}
\begin{figure}[ht!]
\begin{center}
\includegraphics[width=\textwidth]{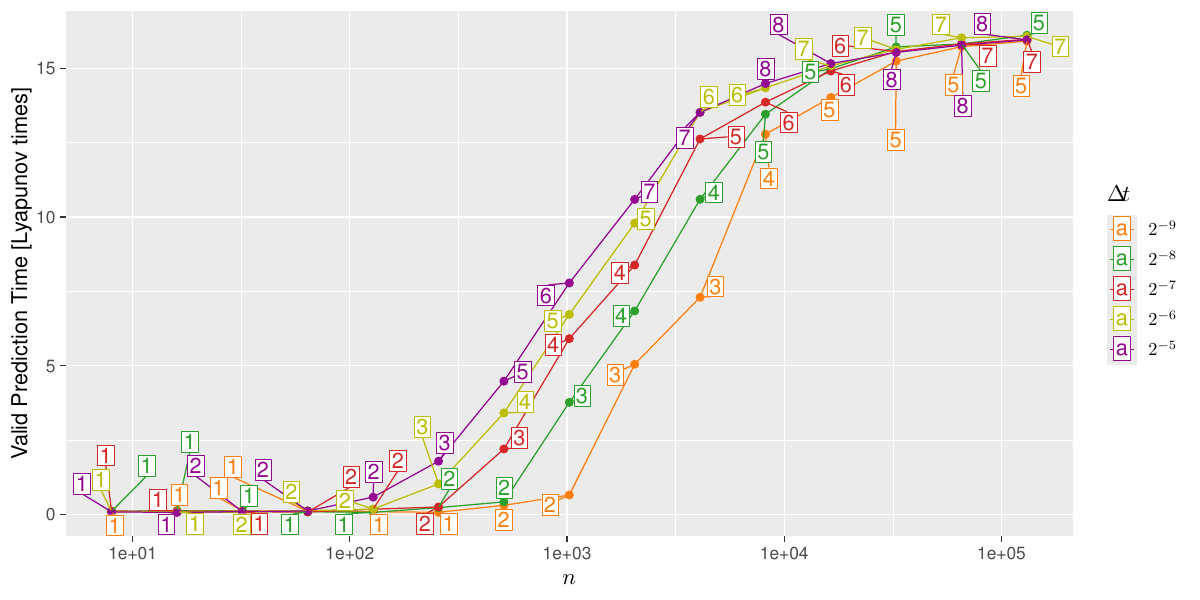}
\end{center}
\caption{\textbf{Best Plot for L96D5, dsd, normalize full, test sequential}. See the beginning of \cref{app:sec:details} for a description.}
\end{figure}
\begin{table}[ht!]
\input{tbl/L96D5_dsd_f_s_VPT_best_table.tex}
\caption{\textbf{Best Table for L96D5, dsd, normalize full, test sequential}. See the beginning of \cref{app:sec:details} for a description.}
\end{table}
\begin{figure}[ht!]
\begin{center}
\includegraphics[width=\textwidth]{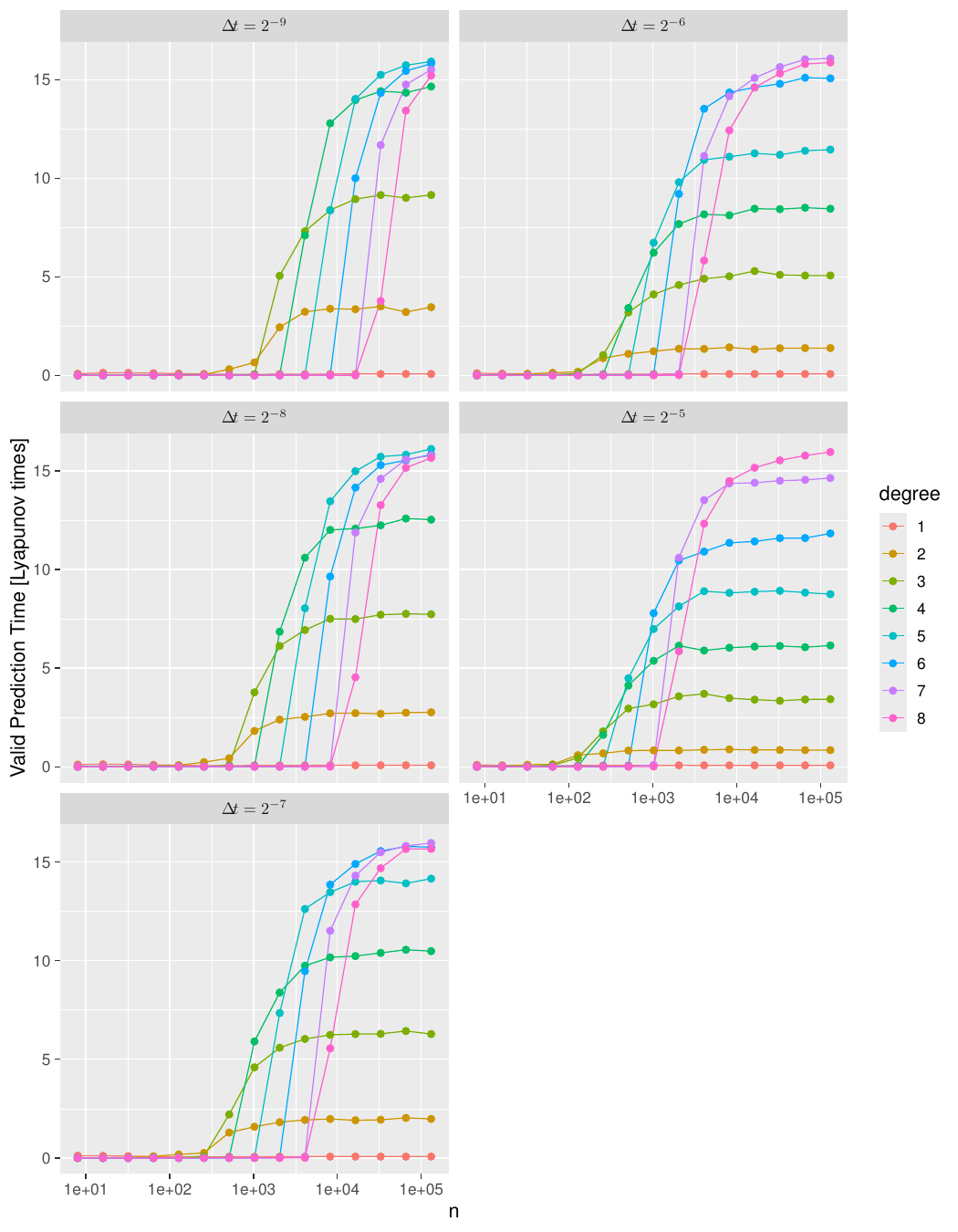}
\end{center}
\caption{\textbf{All Plot for L96D5, dsd, normalize full, test sequential}. See the beginning of \cref{app:sec:details} for a description.}
\end{figure}

\clearpage
\subsection{L96D6, dsd, normalize full, test sequential}
\begin{figure}[ht!]
\begin{center}
\includegraphics[width=\textwidth]{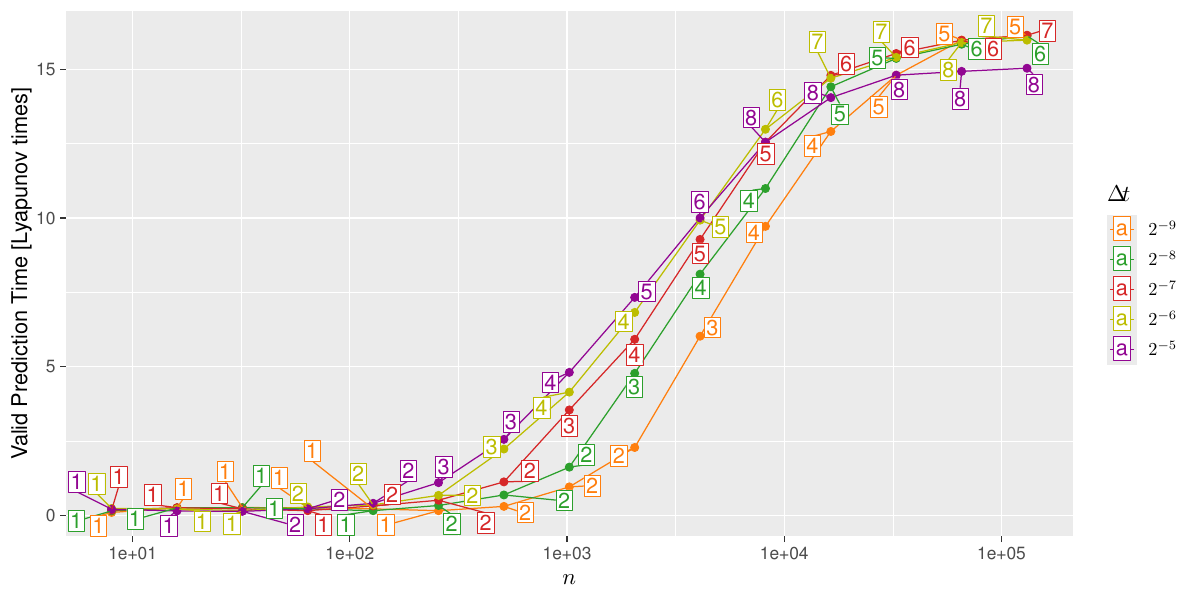}
\end{center}
\caption{\textbf{Best Plot for L96D6, dsd, normalize full, test sequential}. See the beginning of \cref{app:sec:details} for a description.}
\end{figure}
\begin{table}[ht!]
\input{tbl/L96D6_dsd_f_s_VPT_best_table.tex}
\caption{\textbf{Best Table for L96D6, dsd, normalize full, test sequential}. See the beginning of \cref{app:sec:details} for a description.}
\end{table}
\begin{figure}[ht!]
\begin{center}
\includegraphics[width=\textwidth]{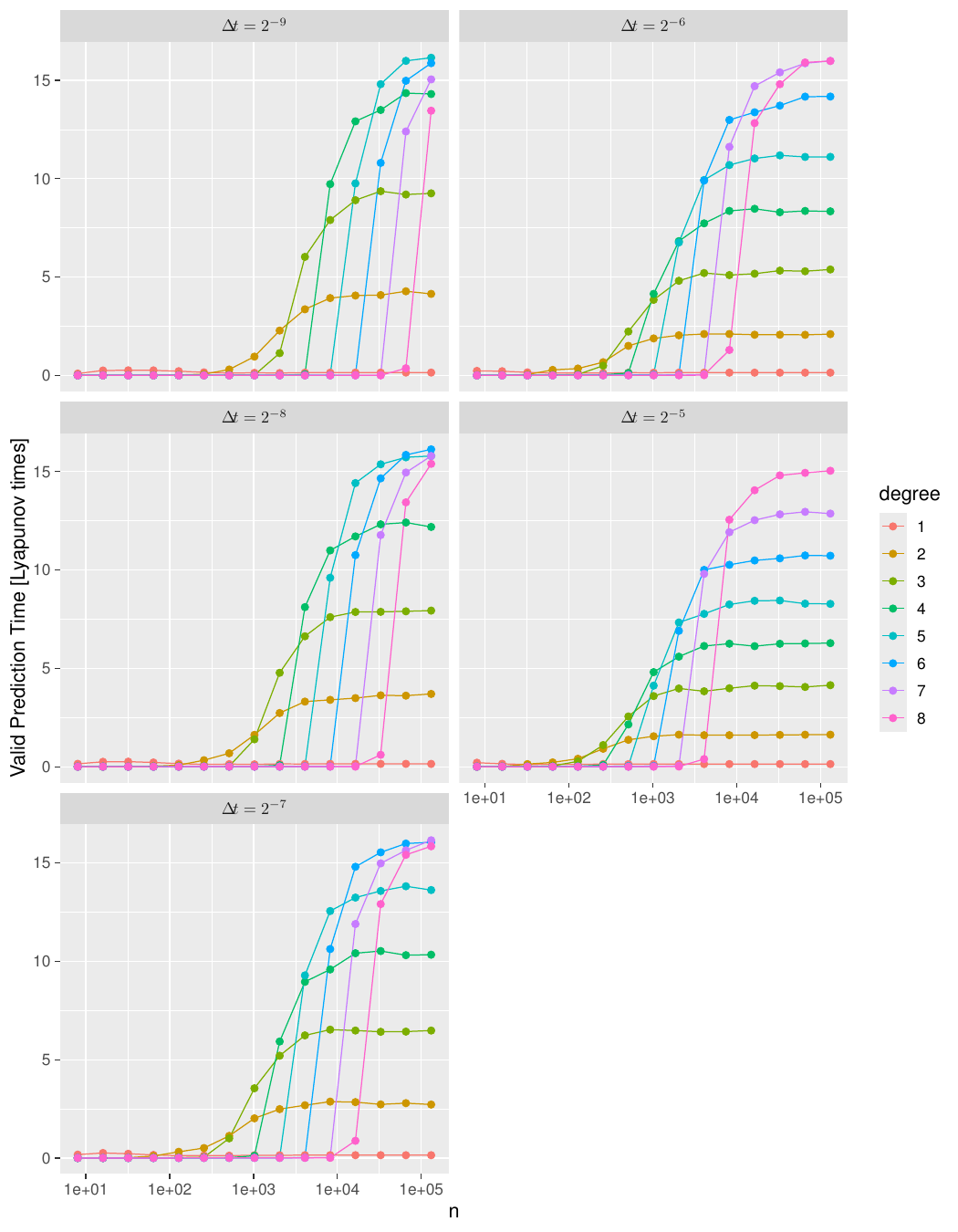}
\end{center}
\caption{\textbf{All Plot for L96D6, dsd, normalize full, test sequential}. See the beginning of \cref{app:sec:details} for a description.}
\end{figure}

\clearpage
\subsection{L96D7, dsd, normalize full, test sequential}
\begin{figure}[ht!]
\begin{center}
\includegraphics[width=\textwidth]{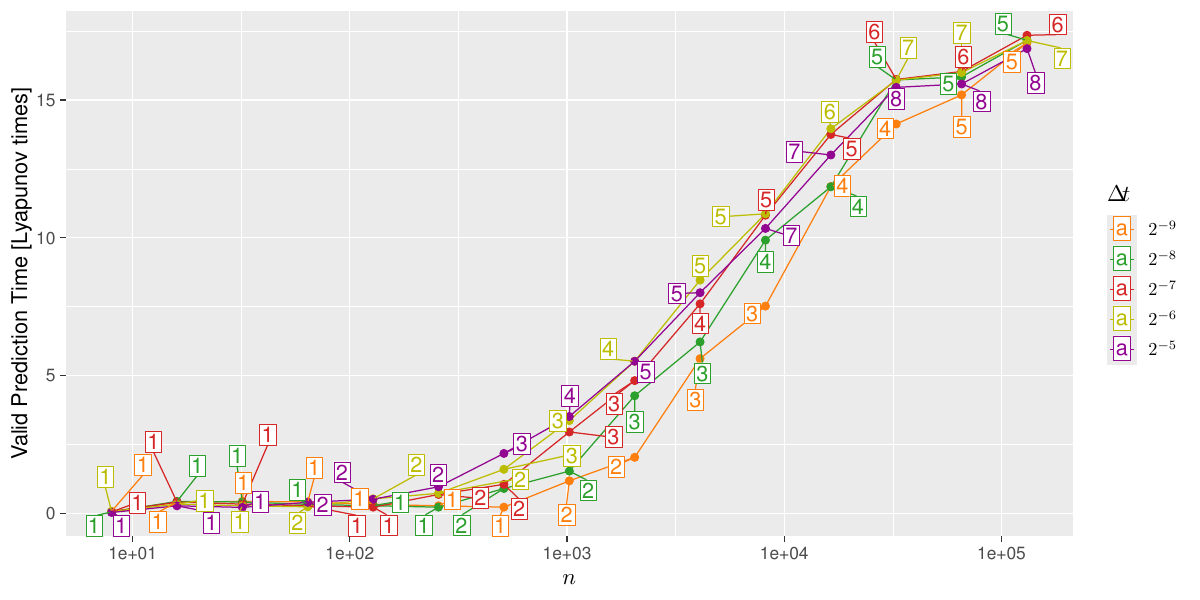}
\end{center}
\caption{\textbf{Best Plot for L96D7, dsd, normalize full, test sequential}. See the beginning of \cref{app:sec:details} for a description.}
\end{figure}
\begin{table}[ht!]
\input{tbl/L96D7_dsd_f_s_VPT_best_table.tex}
\caption{\textbf{Best Table for L96D7, dsd, normalize full, test sequential}. See the beginning of \cref{app:sec:details} for a description.}
\end{table}
\begin{figure}[ht!]
\begin{center}
\includegraphics[width=\textwidth]{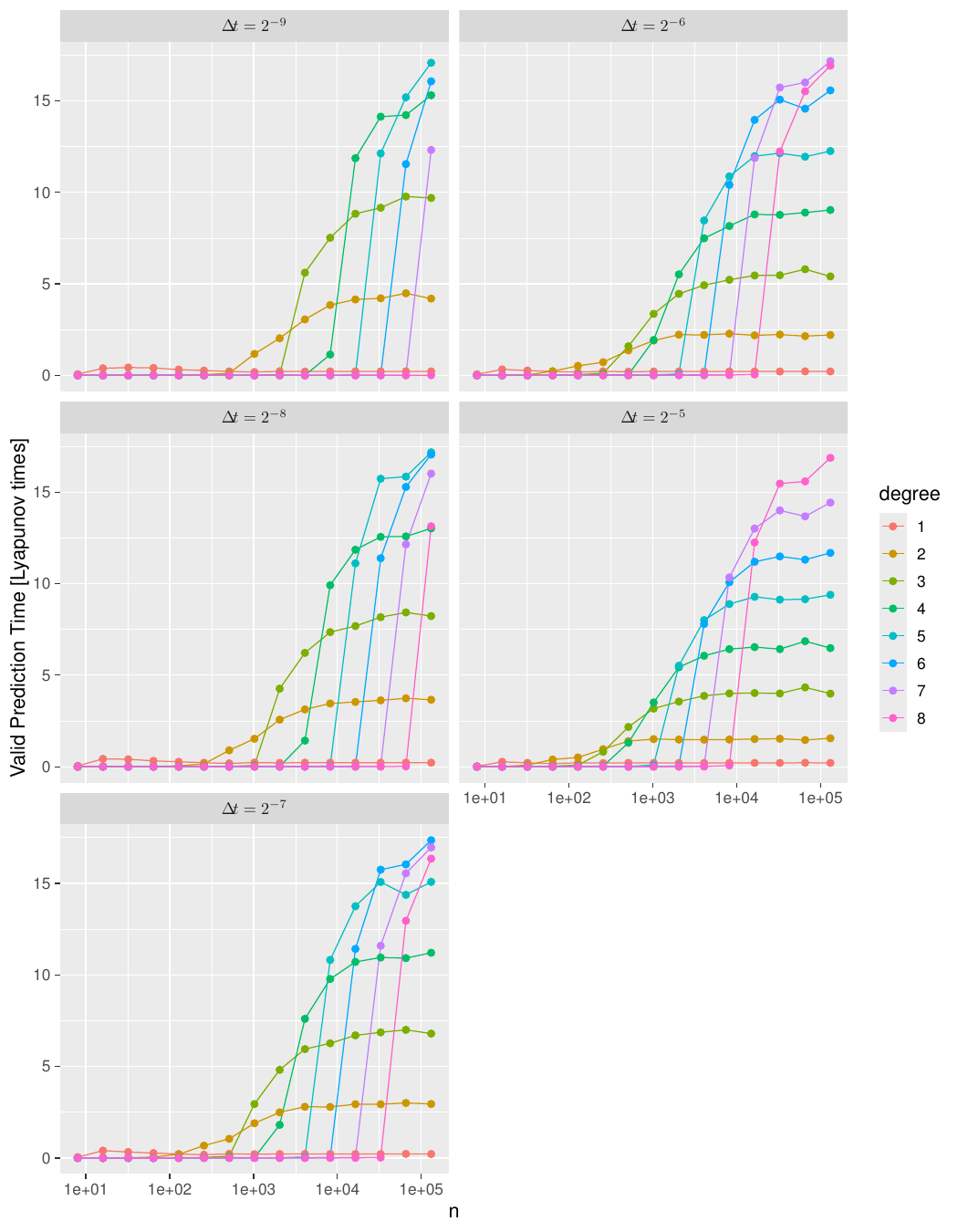}
\end{center}
\caption{\textbf{All Plot for L96D7, dsd, normalize full, test sequential}. See the beginning of \cref{app:sec:details} for a description.}
\end{figure}

\clearpage
\subsection{L96D8, dsd, normalize full, test sequential}
\begin{figure}[ht!]
\begin{center}
\includegraphics[width=\textwidth]{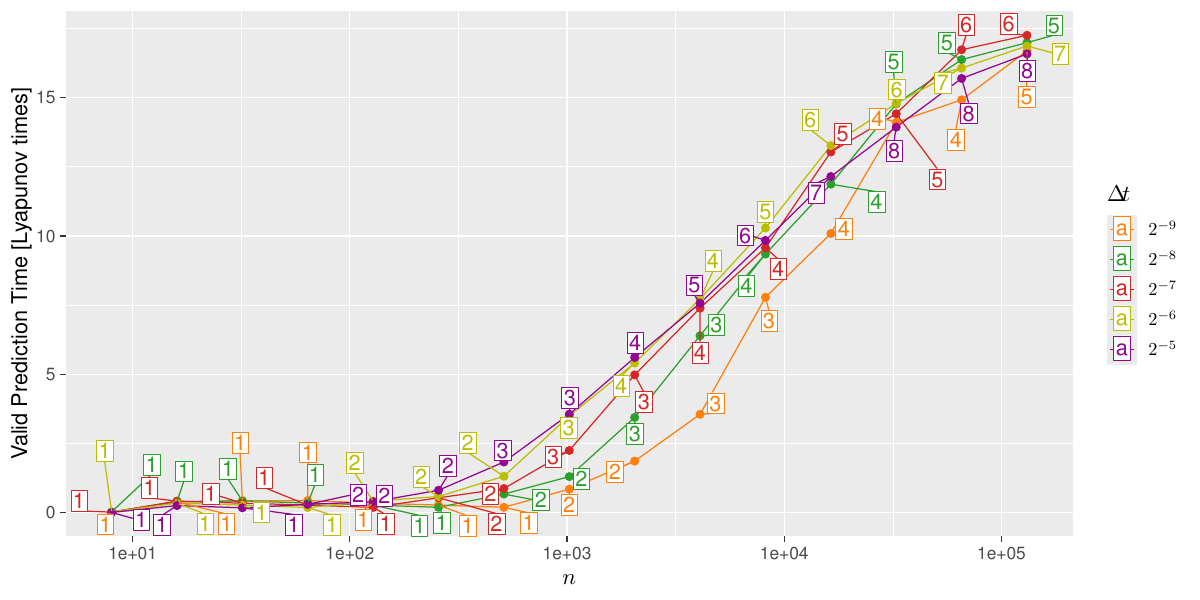}
\end{center}
\caption{\textbf{Best Plot for L96D8, dsd, normalize full, test sequential}. See the beginning of \cref{app:sec:details} for a description.}
\end{figure}
\begin{table}[ht!]
\input{tbl/L96D8_dsd_f_s_VPT_best_table.tex}
\caption{\textbf{Best Table for L96D8, dsd, normalize full, test sequential}. See the beginning of \cref{app:sec:details} for a description.}
\end{table}
\begin{figure}[ht!]
\begin{center}
\includegraphics[width=\textwidth]{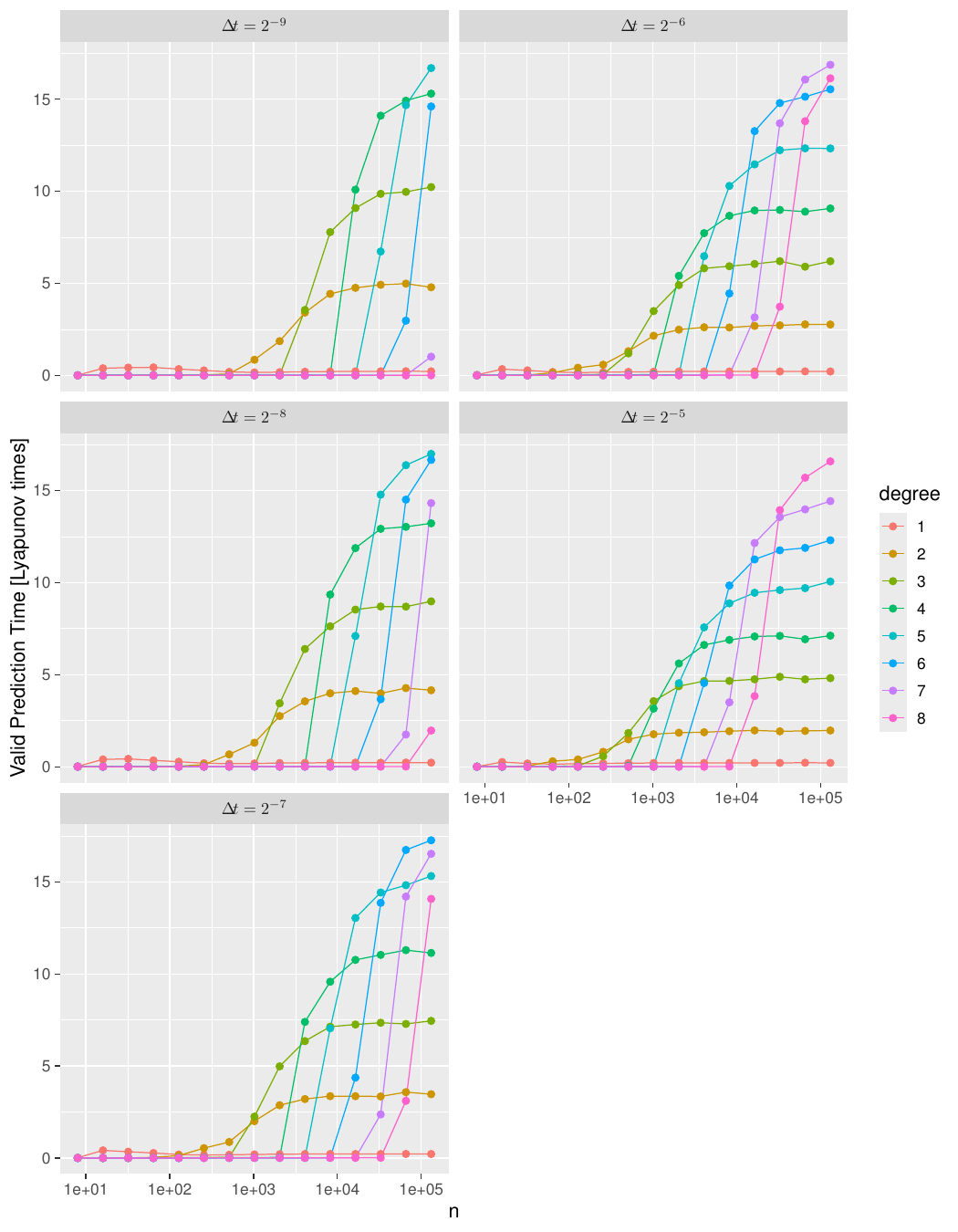}
\end{center}
\caption{\textbf{All Plot for L96D8, dsd, normalize full, test sequential}. See the beginning of \cref{app:sec:details} for a description.}
\end{figure}

\clearpage
\subsection{L96D9, dsd, normalize full, test sequential}
\begin{figure}[ht!]
\begin{center}
\includegraphics[width=\textwidth]{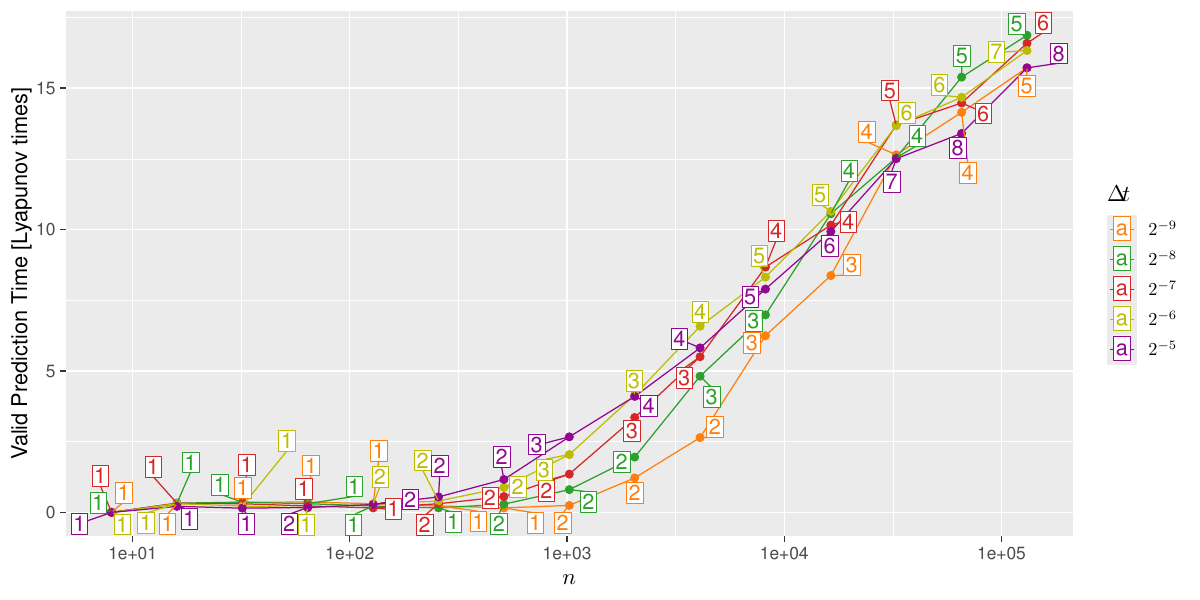}
\end{center}
\caption{\textbf{Best Plot for L96D9, dsd, normalize full, test sequential}. See the beginning of \cref{app:sec:details} for a description.}
\end{figure}
\begin{table}[ht!]
\input{tbl/L96D9_dsd_f_s_VPT_best_table.tex}
\caption{\textbf{Best Table for L96D9, dsd, normalize full, test sequential}. See the beginning of \cref{app:sec:details} for a description.}
\end{table}
\begin{figure}[ht!]
\begin{center}
\includegraphics[width=\textwidth]{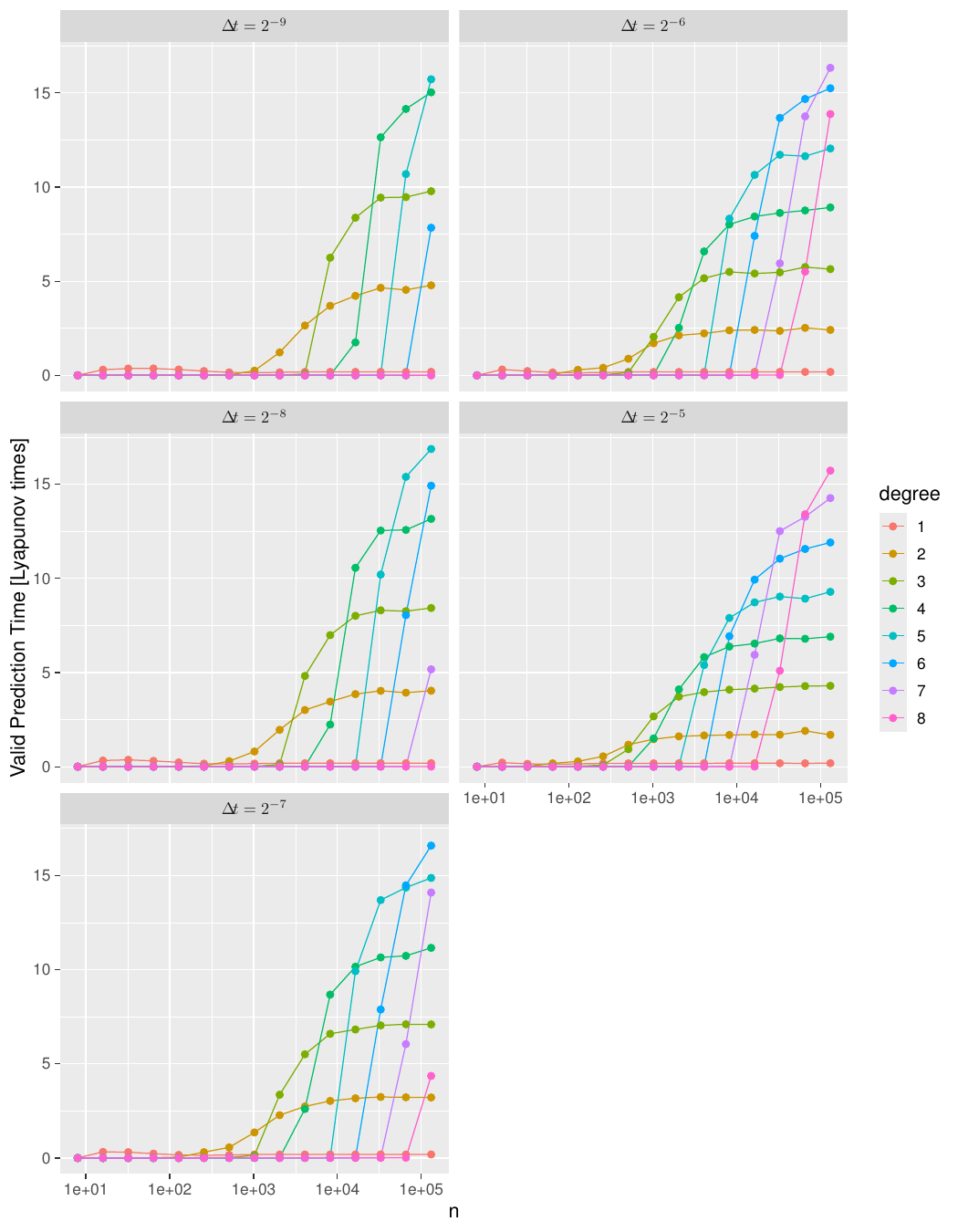}
\end{center}
\caption{\textbf{All Plot for L96D9, dsd, normalize full, test sequential}. See the beginning of \cref{app:sec:details} for a description.}
\end{figure}

\clearpage
\subsection{TCSA, mdm, normalize none, test sequential}
\begin{figure}[ht!]
\begin{center}
\includegraphics[width=\textwidth]{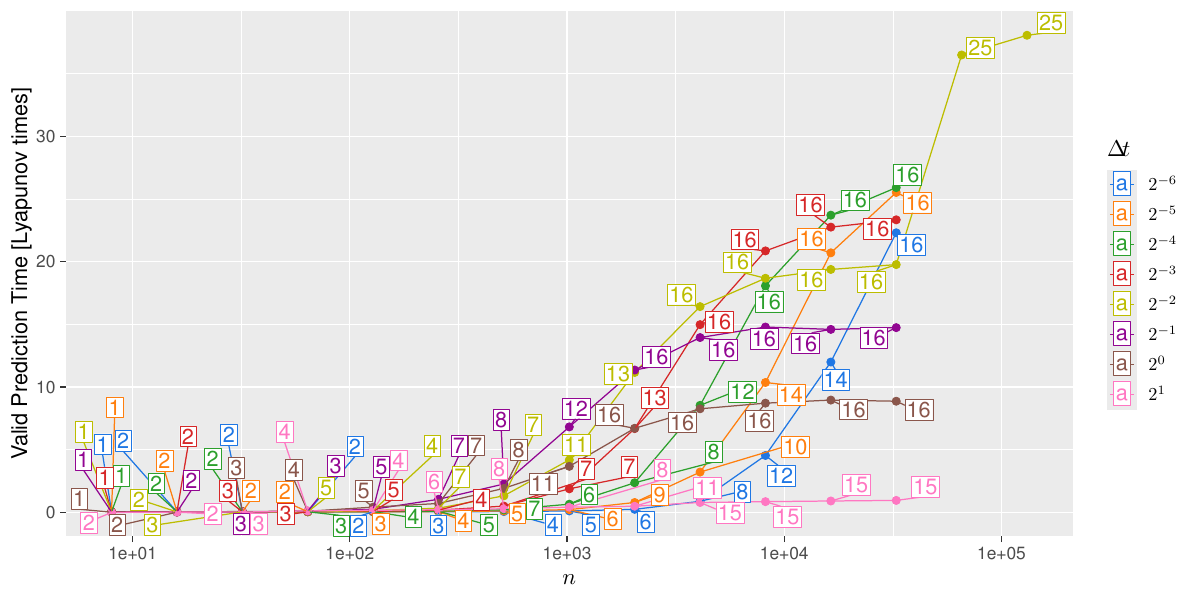}
\end{center}
\caption{\textbf{Best Plot for TCSA, mdm, normalize none, test sequential}. See the beginning of \cref{app:sec:details} for a description.}
\end{figure}
\begin{table}[ht!]
\input{tbl/TCSA_mdm_n_s_VPT_best_table.tex}
\caption{\textbf{Best Table for TCSA, mdm, normalize none, test sequential}. See the beginning of \cref{app:sec:details} for a description.}
\end{table}
\begin{figure}[ht!]
\begin{center}
\includegraphics[width=\textwidth]{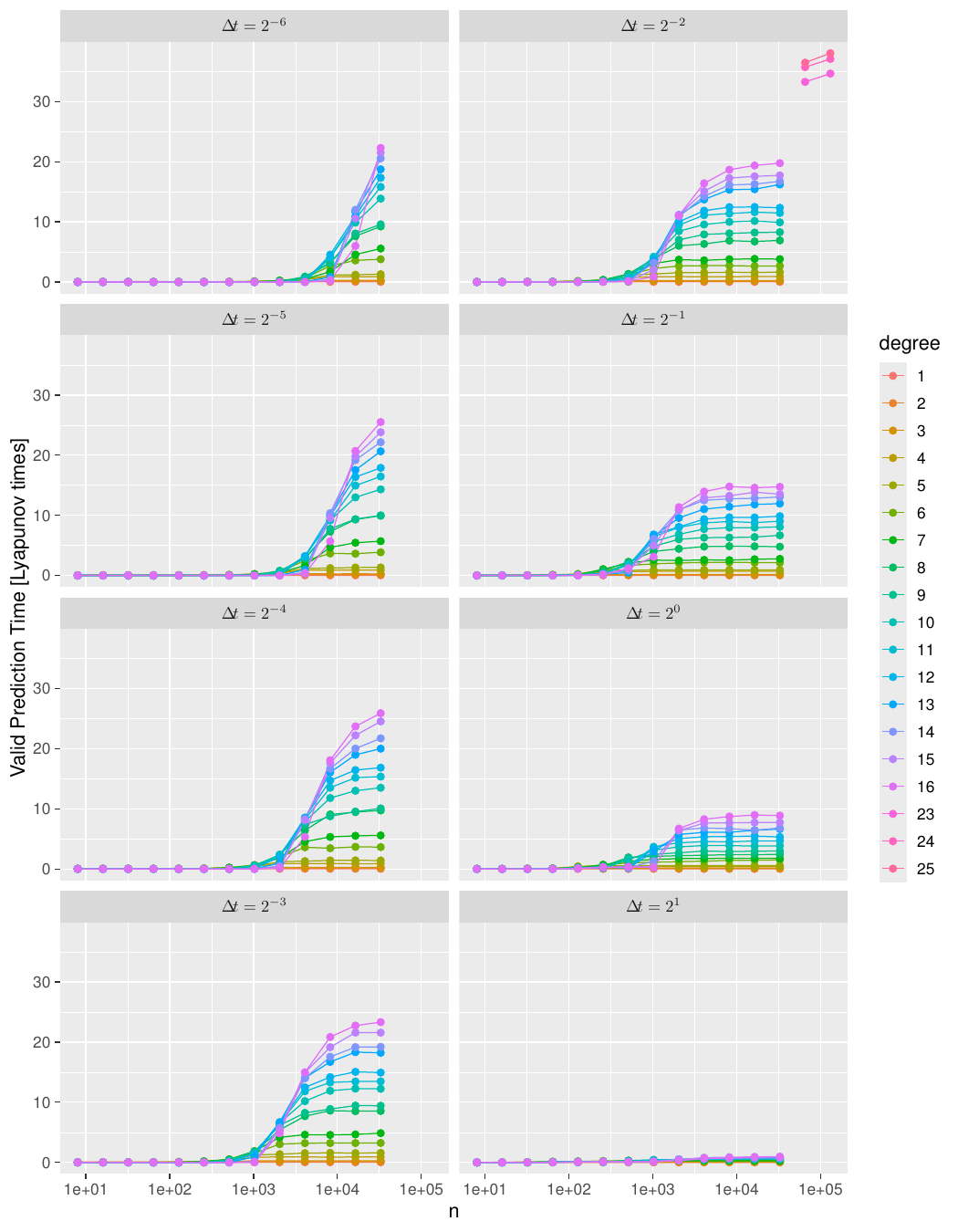}
\end{center}
\caption{\textbf{All Plot for TCSA, mdm, normalize none, test sequential}. See the beginning of \cref{app:sec:details} for a description.}
\end{figure}

\end{document}